\documentclass[graybox,envcountchap,envcountsame,sectrefs]{svmono}

\usepackage{mathptmx}
\usepackage{helvet}
\usepackage{courier}
\usepackage{type1cm}         

\usepackage{makeidx}                
\usepackage{graphicx}                
\usepackage{multicol}                 
\usepackage[bottom]{footmisc}   

\makeindex                          
                                             
\newif\ifidxskippage 
\newcommand*\skipindex[1]{%
\index{#1@\protect\idxskippagetrue~|skippageentry}%
}
\newcommand*\skippageentry[1]{\idxskippagefalse}

\usepackage{amsfonts,amssymb,amsmath,amscd}
\usepackage{mathtools}
\usepackage{centernot}
\usepackage[TS1,T1]{fontenc}
\usepackage{chapterbib}

\usepackage[pdfpagelabels,colorlinks]{hyperref}
\hypersetup{
    linkcolor= {blue},
    citecolor= {blue}}

\usepackage{comment}

\usepackage{enumitem}

\usepackage{tikz}
\usetikzlibrary{shapes.misc,arrows,decorations.markings}

\newcounter{mnotecount}[section]
\renewcommand{\themnotecount}{\thesection.\arabic{mnotecount}}
\newcommand{\mnote}[1]
{\protect{\stepcounter{mnotecount}}$^{\mbox{\footnotesize
$
\bullet$\themnotecount}}$ \marginpar{
\raggedright\tiny\em
$\!\!\!\!\!\!\,\bullet$\themnotecount: #1} }

\usepackage{relsize}

\newtheorem{hypothesis}{Hypothesis}[chapter]

\DeclareMathOperator{\Diff}{Diff}               
\DeclareMathOperator{\Dom}{Dom}                 
\DeclareMathOperator{\End}{End}                 
\DeclareMathOperator{\Erf}{Erf}            
\DeclareMathOperator{\id}{id}                   
\DeclareMathOperator{\Ima}{Im}            
\DeclareMathOperator{\Ker}{Ker}           
\DeclareMathOperator{\Lspan}{span}                      
\DeclareMathOperator{\OP}{OP}           
\DeclareMathOperator{\op}{op}           
\DeclareMathOperator{\ord}{ord}            
\DeclareMathOperator{\Res}{Res}         
\DeclareMathOperator{\Sd}{Sd}                   
\DeclareMathOperator{\sign}{sign}         
\DeclareMathOperator{\spec}{spec}           
\DeclareMathOperator{\supp}{supp}      
\DeclareMathOperator{\tr}{tr}             
\DeclareMathOperator{\Tr}{Tr}             
\DeclareMathOperator{\Vol}{Vol}            
\DeclareMathOperator{\WRes}{WRes}         

\newcommand{\<}{\left\langle}         
\renewcommand{\>}{\right\rangle}                
\newcommand{\A}{\mathcal{A}}              
\newcommand{\abD}{\abs{\DD}}			  
\newcommand{\aD}{\abD}			  
\newcommand{\abs}[1]{\left\lvert#1\right\rvert}			
\newcommand{\awD}{\abs{D}}				
\newcommand{\awDA}{\abs{D_{\Ag}}}				
\newcommand{\Ag}{\mathbb{A}}                    
\newcommand{\ahd}{(\mathcal{A,H,D})}            
\newcommand{\ahda}{(\A,\H,\DA)}         
\newcommand{\Aop}{\A^{\text{op}}}		
\newcommand{\aphi}{\vert\phi\vert}		
\newcommand{\B}{\mathcal{B}}              
\renewcommand{\bar}[1]{\overline{#1}}		
\renewcommand{\bbbone}{{\text{\usefont{U}{dsss}{m}{n}\char49}}} 		
\newcommand{\bD}{\bar{D}}			
\newcommand{\bra}[1]{\langle #1|}                  

\newcommand{\C}{\mathcal{C}}              
\newcommand{\capB}[1]{\underset{#1}{\mathlarger{\mathlarger{\cap}}}} 	
\newcommand{\CC}{\mathbb{C}}              
\newcommand{\Cl}{{\text{\usefont{OMS}{cmsy}{m}{n}C}}}    
\newcommand{\CM}{{\text{\usefont{OMS}{cmsy}{m}{n}CM}}}              
\newcommand{\Coo}{C^\infty}             
\newcommand{\cv}{= \vcentcolon}             

\newcommand{\DA}{\DD_{\Ag}}                     
\newcommand{\Dcl}{{\text{\usefont{OMS}{cmsy}{m}{n}D}}} 
\newcommand{\Dclp}{\Dcl\,'} 				
\newcommand{\Dclpp}{\Dclp_+}			 
\newcommand{\DD}{\mathcal{D}}           
\newcommand{\del}{\delta}			
\newcommand{\Dq}{\DD_q}				
\newcommand{\DqS}{\Dq^S}				
\newcommand{\Dslash}{{\DD \mkern-11.5mu/\,}} 
\newcommand{\dt}{\partial}                    
\renewcommand{\E}{\mathcal{E}}			
\newcommand{\eps}{\varepsilon}          
\newcommand{\eq}{\, \vc} 			 
\newcommand{\F}{\mathfrak{F}}			
\newcommand{\floor}[1]{\lfloor#1\rfloor} 
\newcommand{\ga}{\gamma}                 
\renewcommand{\H}{\mathcal{H}}          
\newcommand{\half}{{\mathchoice{\thalf}{\thalf}{\shalf}{\shalf}}}               
\newcommand{\thalf}{\tfrac{1}{2}} 
\newcommand{\shalf}{{\scriptstyle\frac{1}{2}}} 
\newcommand{\hD}{\hP}               
\newcommand{\hH}{\Tr \, e^{-t \,H}}               
\newcommand{\hKH}{\Tr \, K e^{-t \,H}}               
\newcommand{\hP}{\Tr \,e^{-t\, \abD}}               
\newcommand{\hTD}{\Tr \, T e^{-t \, \awD}}               
\newcommand{\K}{{\text{\usefont{OMS}{cmsy}{m}{n}K}}}             
\newcommand{\KKK}{\mathcal{K}}             
\newcommand{\ket}[1]{| #1 \rangle}                  

\newcommand{\ketQS}{\ket{l,m}_{\pm}}
\let\LL\L
\renewcommand{\L}{\mathcal{L}}          
\newcommand{\la}{\lambda}                   
\newcommand{\Lc}{\mathfrak{L}}			
\newcommand{\Linf}{\; \underset{\Lambda \to +\infty}{\sim} \;}  
\newcommand{\lmo}{l \! - \! 1}
\newcommand{\lpo}{l \! + \! 1}
\newcommand{\M}{\Mellin}
\newcommand{\Mellin}{\mathfrak{M}}		
\newcommand{\mmo}{m \! - \! 1}
\newcommand{\mpo}{m \! + \! 1}
\newcommand{\N}{\mathbb{N}}            

\newbox\ncintdbox \newbox\ncinttbox
\setbox0=\hbox{$-$} \setbox2=\hbox{$\displaystyle\int$}
\setbox\ncintdbox=\hbox{\rlap{\hbox
    to \wd2{\kern-.1em\box2\relax\hfil}}\box0\kern.1em}
\setbox0=\hbox{$\vcenter{\hrule width 4pt}$}
\setbox2=\hbox{$\textstyle\int$}
\setbox\ncinttbox=\hbox{\rlap{\hbox
    to \wd2{\kern-.14em\box2\relax\hfil}}\box0\kern.1em}
\newcommand{\ncint}{\mathop{\mathchoice{\copy\ncintdbox}
    {\copy\ncinttbox}{\copy\ncinttbox}
    {\copy\ncinttbox}}\nolimits}

\newcommand{\ncintd}[1]{\ncint^{[#1]}}          
\newcommand{\ncintA}{\ncint_{\hspace{-0.15cm}\Ag}}

\newcommand{\norm}[1]{\left\lVert#1\right\rVert}    
\newcommand{\OA}{\Omega_{\DD}^1(\A)}			
\newcommand{\oh}{{\tfrac{1}{2}}}          
\newcommand{\Oinf}{\OO_{\infty}}             
\newcommand{\oinf}{\oo_{\infty}}             
\def\oq{\frac{1}{q}}
\newcommand{\OO}{{\text{\usefont{OMS}{cmsy}{m}{n}O}}}   
\newcommand{\oo}{{\displaystyle\mathsmaller{\mathsmaller{\text{\usefont{OMS}{cmsy}{m}{n}O}}}}} 
\newcommand{\ox}{\otimes}                  
\newcommand{\Oz}{\OO_0}             
\newcommand{\oz}{\oo_0}             
\newcommand{\Pc}{\mathcal{P}}                    
\newcommand{\PDO}{\Psi(\A)}                             
\newcommand{\PDOC}{\Psi^{\CC}(\A)}                             
\newcommand{\PDOCA}{\Psi^{\CC}_{\Ag}(\A)}                             
\newcommand{\PDOk}[1]{\Psi^{#1}(\A)}            
\newcommand{\PDOz}{\PDOk{0}}            
\newcommand{\PP}{\mathfrak{P}}
\newcommand{\QQ}{\mathcal{Q}}
\newcommand{\R}{\mathcal{R}}             
\newcommand{\Rez}[1]{\underset{#1}{\Res} \,}		
\newcommand{\Rp}{\RR^+}                                 
\newcommand{\RR}{\mathbb{R}}             
\newcommand{\Scl}{{\text{\usefont{OMS}{cmsy}{m}{n}S}}}   
\newcommand{\Sclp}{\Scl\,'}   
\newcommand{\Sclpp}{\Sclp_+}   
\newcommand{\SA}{S(\DD,f,\Lambda)}		
\newcommand{\set}[1]{\{\,#1\,\}}              
\renewcommand{\SS}{\mathcal{S}}        
\newcommand{\stirling}{\genfrac{[}{]}{0pt}{}}

\newcommand{\suq}{\mathrm{SU}_q(2)}		
\newcommand{\Th}{\Theta}
\newcommand{\Tp}{\mathcal{T}^p}                
\newcommand{\Tpp}[1]{\mathcal{T}^{#1}}                
\newcommand{\Trdix}{\Tr_{\mathrm{Dix}}}    
\newcommand{\tstirling}[2]{\mathlarger{\left[ \begin{smallmatrix} #1 \\ #2 \end{smallmatrix} \right]}} 
\newcommand{\TT}{\mathbb{T}}                
\newcommand{\tzero}{\; \underset{t \downarrow 0}{\sim} \;}      
\newcommand{\UDqS}{\ZZ_{\DqS}}			
\newcommand{\UKH}{\ZZ_{K,H}}				

\newcommand{\UqSU}{\mathcal{U}_q(\mathfrak{su}(2))}     
\newcommand{\vc}{\vcentcolon =}             
\newcommand{\wD}{D}					
\newcommand{\wh}{\widehat}                  
\newcommand{\wt}{\widetilde}                 
\newcommand{\x}{\times}                      
\newcommand{\Z}{\mathbb{Z}}                 
\newcommand{\zD}{\zeta_{\wD}}                  
\newcommand{\zDA}{\zeta_{\wD_{\Ag}}}                  
\newcommand{\zH}{\zeta_{H}}                  
\newcommand{\zKH}{\zeta_{K,H}}                  
\newcommand{\zPD}{\zeta_{T,D}}                  
\newcommand{\zTD}{\zeta_{T,D}}                  
\newcommand{\zqS}{\zeta_{\DqS}}                 
\newcommand{\ZZ}{{\text{\usefont{OMS}{cmsy}{m}{n}Z}}}                
\newcommand{\EEE}{{\text{\small\usefont{OMS}{cmsy}{m}{n}E}}}          

\makeatletter
\DeclareRobustCommand\widecheck[1]{{\mathpalette\@widecheck{#1}}}
\def\@widecheck#1#2{%
    \setbox\z@\hbox{\m@th$#1#2$}%
    \setbox\tw@\hbox{\m@th$#1%
       \widehat{%
          \vrule\@width\z@\@height\ht\z@
          \vrule\@height\z@\@width\wd\z@}$}%
    \dp\tw@-\ht\z@
    \@tempdima\ht\z@ \advance\@tempdima2\ht\tw@ \divide\@tempdima\thr@@
    \setbox\tw@\hbox{%
       \raise\@tempdima\hbox{\scalebox{1}[-1]{\lower\@tempdima\box
\tw@}}}%
    {\ooalign{\box\tw@ \cr \box\z@}}}
\makeatother

\newcommand{\ePDO}{\wt{\Psi}(\A)}              
\newcommand{\ePDOk}[1]{\wt{\Psi}^{#1}(\A)}             
\newcommand{\ePDOC}{\wt{\Psi}^{\CC}(\A)}              
\newcommand{\PDOqk}[1]{\wt{\Psi}^{#1}(\A_q)}           

\begin{document} 

\author{Micha{\l} Eckstein, Bruno Iochum}
\title{Spectral Action \\ \text{ } in \\ Noncommutative Geometry}
\pagenumbering{Alph}
\maketitle

\frontmatter

\setcounter{page}{1}
%
%

\preface
\addcontentsline{toc}{chapter}{Preface}

%
%
%

\vspace{-5cm}

\begin{quotation}
\begin{flushright}
\large\emph {This book is dedicated to Alain Connes,\\ whose work has always been a fantastic \\ source of inspiration for us.}
\end{flushright}
\end{quotation}

\vspace{1.3cm}

The Least Action Principle is among the most profound laws of physics. The action --- a functional of the fields relevant to a given physical system --- encodes the entire dynamics. Its strength stems from its universality: The principle applies equally well in every domain of modern physics including classical mechanics, general relativity and quantum field theory. Most notably, the action is a primary tool in model-building in particle physics and cosmology.

The discovery of the Least Action Principle impelled a paradigm shift in the methodology of physics. The postulates of a theory are now formulated at the level of the action, rather than the equations of motion themselves. Whereas the success of the `New Method' cannot be overestimated, it raises a big question at a higher level: ``Where does the action come from?'' A quick look at the current theoretical efforts in cosmology and particle physics reveals an overwhelming multitude of models determined by the actions, which are postulated basing on different assumptions, beliefs, intuitions and prejudices. Clearly, it is the empirical evidence that should ultimately select the correct theory, but one cannot help the impression that our current models are only effective and an overarching principle remains concealed.

A proposal for such an encompassing postulate was formulated by Ali Chamseddine and Alain Connes in 1996 \cite{ConnesSA}. It reads \cite[(1.8)]{ConnesSA}:

\begin{center}
\textit{The physical action should only depend upon the spectrum of $\DD$,}
\end{center}

\noindent where $\DD$ is a certain unbounded operator of geometrical origin. The incarnation of the \emph{Spectral Action Principle} is very simple indeed:
\begin{align*}
\SA = \Tr f( \abs{\DD}/\Lambda),
\end{align*}
with a given energy scale $\Lambda$ and a positive cut-off function $f$. Such a formulation provides a link with the current effective actions employed in field theoretic models and allows for a confrontation against the experimental data. The striking upshot of the spectral action is that, with a suitable choice of the operator $\DD$, it allows one to retrieve the full Standard Model of particle physics in curved (Euclidean) spacetime \cite{ConnesPRL1996,Almost1,WalterBook}. This result attracted considerable interest in both physical and mathematical communities and triggered a far-reaching outflow of theoretical research. The most recent applications include Grand Unified Theories \cite{ConnesGUT}, modified Einstein gravity \cite{MarcolliCosmoBook} and quantum gravity \cite{ConnesGeomQuant}, to name only a few. \pagebreak


The formulation of the Least Action Principle dates back to the 18$^{\text{th}}$ century and the seminal works of Pierre de Maupertuis, Gottfried Leibniz and Leonhard Euler. 
The quest for its rigorous verbalisation sparked the development of the calculus of variations along with the Lagrangian and Hamiltonian formalisms. The modern formulation is expressed in the language of differential geometry.

The Spectral Action Principle is embedded in an even more advanced domain of modern mathematics -- \emph{noncommutative geometry}, pioneered and strongly pushed forward by Alain Connes \cite{Connes80,ConnesNCG}. The idea that spaces may be quantised was first pondered by Werner Heisenberg in the 1930s (see \cite{QFT_Moyal_review} for a historical review) and the first concrete model of a `quantum spacetime' was constructed by Hartland Snyder in 1949, extended by Chen-Ning Yang shortly afterwards. 
However, it took almost half a century for the concepts to mature and acquire the shape of a concrete mathematical structure. By now noncommutative geometry is a well-established part of mathematics. 

Noncommutative geometry \`a la Connes sinks its roots not only in the Riemannian geometry, but also in the abstract framework of operator algebras. Its conceptual content is strongly motivated by two fundamental pillars of physics: general relativity and quantum mechanics, explaining why it has attracted both mathematicians and theoretical physicists. It offers a splendid opportunity to conceive `quantum spacetimes' turning the old Heisenberg's dream into a full-bodied concept.

In this paradigm, geometry is described by a triplet $\ahd$, where $\A$ is a not necessarily commutative algebra, $\DD$ is an operator (mimicking the Dirac operator on a spin manifold) both acting on a common scene -- a Hilbert space $\H$. Thus, by essence, this geometry is spectral. The data of a spectral triple $\ahd$ covers a huge variety of different geometries. The classical (i.e. commutative) case includes primarily the  Riemannian manifolds, possibly tainted by boundaries or singularities, but also discrete spaces, fractals and non-Hausdorff spaces and when $\A$ is noncommutative, the resulting `pointless' geometries, with the examples furnished by the duals of discrete groups, dynamical systems or quantum groups to mention but a few.

At this point one should admit that the simple form of the spectral action is deceiving --- an explicit computation would require the knowledge of the full spectrum of the operator $\DD$, which is hardly ever the case. Nevertheless, one can extract a great deal of physically relevant information by studying the asymptotics of $\SA$ when $\Lambda$ tends to infinity. The key tool to that end is the renowned heat kernel method fruitfully employed in classical and quantum field theory, adapted here to the noncommutative setting. Beyond the (almost) commutative case the latter is still a vastly uncharted water. It is our primary intent to provide a faithful map of the mathematical aspects behind the spectral action. Whereas the physical motivation will be present in the backstage, we leave the potential applications to the Reader's invention. To facilitate the latter, we recommend to have a glimpse into the textbooks \cite{MarcolliCosmoBook,WalterBook} and references therein. 
\pagebreak

The plan of our guided tour presents itself as follows:

In the first chapter the basics of noncommutative geometry \`a la Connes are laid out. Chapter \ref{chap:tools} is designed to serve as a toolkit with several indispensable notions related to spectral functions and their functional transforms. Therein, the delicate notion of an asymptotic expansion is carefully detailed, both in the context of functions and distributions. With Chapter \ref{chap:asymptotic} we enter into the hard part of this book, which unveils the subtle links between the existence of asymptotic expansions of traces of heat operators and meromorphic extensions of the associated spectral zeta functions. While trying to stay as general as possible, we illustrate the concepts with friendly examples. Therein, the large-energies asymptotic expansion of the spectral action is presented in full glory. Chapter \ref{chap:perturbations} is dedicated to the important concept of a fluctuation of the operator $\DD$ by a `gauge potential' and its impact on the action. In terms of physics, this means a passage from `pure gravity' to a full theory vested with the all admissible gauge fields. In terms of mathematics, it involves rather advanced manipulations within the setting of abstract pseudodifferential operators, which we unravel step by step. We conclude in Chapter \ref{chap:open} with a list of open problems, which --- in our personal opinion --- constitute the main stumbling blocks in the quest of understanding the mathematics and physics of the Spectral Action Principle. We hope that these would inspire the Reader to have his own take on the subject. The bulk of the book is complemented with a two-part Appendix. Section \ref{classical tools} contains further auxiliary tools from the theory of pseudodifferential operators, including a detailed derivation of the celebrated heat kernel expansion. In Section \ref{chap:appendix} we present examples of spectral geometries of increasing complexity: spheres, tori, noncommutative tori and a quantum sphere.

\smallskip 

This book is devoted to the spectral action, which is only a small offspring in the vast domain of noncommutative geometry. Therefore, when introducing the rudiments of Connes' theory, we are bound to be brief and focus on the specific aspects related to the spectral action. We refer the Reader to the textbooks for a complete introduction on noncommutative geometry \cite{ConnesNCG,Elements,BasicNCG,Landi,Varilly}.

Let us also warn the Reader that, although we have designed the book to be as self-contained as possible, some mathematical prerequisites are indispensable to grasp the presented advanced concepts. The Reader should be acquainted with the basics of functional analysis, including, in particular, the spectral theory of unbounded operators on Hilbert spaces (e.g. \cite{Spectral,SimonReed,Rudin}) and the rudiments of operator algebras (e.g. \cite{Bratteli2012,KadisonBook}). Some intuitions from global differential geometry (e.g. \cite{Nakahara}) and the theory of pseudodifferential operators (e.g. \cite{Shubin}) may also prove useful.

Our ultimate purpose is not only to provide a rigid first course in the spectral action, but also to charm  the Reader with the marvellous interaction between mathematics and physics encapsulated in this apparently simple notion of spectral action. Let \textit{res ipsa loquitur}\ldots
 
\smallskip 

During the years spent in the realm of noncommutative geometry we have collaborated with a number of our close colleagues: Driss Essouabri, Nicolas Franco, Victor Gayral, Jos\'e Gracia-Bond\'ia, Michael Heller, Cyril Levy, Thierry Masson, Tomasz Miller, Andrzej Sitarz, Jo Varilly, Dmitri Vassilevich, Raimar Wulkenhaar, Artur Zaj\k{a}c. We also took benefits from discussions with Alain Connes. Moreover, Tomasz was a scrupulous proof reader and  Thierry was a great help with the \LaTeX\, typesetting. It is our pleasure to cordially thank all of them, as without their kind support this book could not come into being.

Finally, we are greatly indebted to our families for their constant support.

We acknowledge the financial support of the Copernicus Center for Interdisciplinary Studies in Krak\'ow Poland through the research grant ``Conceptual Problems in Unification Theories'' (No. 60671) awarded by the John Templeton Foundation, and the COST Action MP1405 ``Quantum Spacetime''.

\medskip

\begin{flushright}\noindent
Krak\'ow and Marseille,\hfill {\it Micha{\l} Eckstein and Bruno Iochum}\\
May 2018\hfill \text{ }\\
\end{flushright}

%
%
 \bibliographystyle{spmpsci}
 \bibliography{SA_BIB}
\tableofcontents

\mainmatter
%
%
%
\chapter{The Dwelling of the Spectral Action}
\label{chap:ST}

\abstract{
The natural habitat of the spectral action is Connes' noncommutative geometry. Therefore, it is indispensable to lay out its rudiments encoded in the notion of a spectral triple. We will, however, exclusively focus on the aspects of the structure, which are relevant for the spectral action computations. These include i.a. the abstract pseudodifferential calculus, the dimension spectrum and noncommutative integrals, based on both the Wodzicki residue and the Dixmier trace.
}

\section{Spectral Triples}
\label{sec:axioms}

The basic objects of noncommutative geometry \`a la Connes are spectral triples. As the name itself suggests, they consist of three elements: an algebra $\A$, a Hilbert space $\H$ and an operator $\DD$ acting on $\H$. These three constituents are tied together with a set of conditions, which could be promoted to the axioms of a new --- not necessarily commutative --- geometry. 

In the following, we shall denote successively by $\L(\H)$\index{_zL0h@$\L(\H)$},  $\B(\H)$\index{_zB@$\B(\H)$}, $\KKK(\H)$\index{_zKh@$\KKK(\H)$}, $\L^1(\H)$\index{_zL1h@$\L^1(\H)$} the sets of linear, bounded, compact and trace-class operators on $\H$. As for the latter, $\Tr$ will always stand for 
$\Tr_\H$, unless stated explicitly.
\begin{definition}
\label{def:triple}
A  \emph{spectral triple} \index{spectral triple} $\ahd$ \index{_zAz@$\ahd$} consists of a unital involutive algebra $\A$, \index{_z1@$\A$}  with a faithful representation $\pi: \A \to \B(\H)$ \index{_a8p@$\pi(a)$} on a separable Hilbert space $\H$,  \index{_zH@$\H$} and $\DD\in \L(\H)$ \index{_zD0@$\DD$} such that:
\begin{itemize}
\item $\DD$ is a (possibly unbounded) selfadjoint operator on $\H$,
\item $[\DD,\pi(a)]$ extends to a bounded operator on $\H$ for all $a \in \A$,
\item $\DD$ has a compact resolvent -- i.e. $(\DD-\la)^{-1} \in \KKK(\H)$ for $\la \notin \spec(\DD)$.
\end{itemize}
\end{definition}
Remark that the second assumption requires that $\pi(a) \Dom \DD \subset \Dom \DD$. \\
It is standard to omit the symbol $\pi$ of the representation when it is given once and for all.

This flexible definition is tailored to encompass the largest possible spectrum of different geometries. However, in order to have workable examples one often has to take into account the topology of $\A$. To this end, one can, for instance, demand that $\A$ is a dense $^*$-subalgebra of a $C^*$-algebra or, more restrictively, a pre-$C^*$-algebra --- if one desires to have a holomorphic functional calculus.

The set of axioms adopted in Definition \ref{def:triple} is at the core of Connes' noncommutative geometry. Shortly, we will discuss two additional properties of a spectral triple: $p$-summability and regularity, which are crucial for the sake of explicit spectral action computations. Before we do so, let us illustrate Definition \ref{def:triple} with the canonical example.
\begin{example}
\label{ex:commutative}
Let $M$ be a compact Riemannian manifold without boundary and let $P$ be any elliptic selfadjoint pseudodifferential operator (pdo) of order one on a vector bundle $E$ over $M$ endowed with a hermitian structure. \\Then, $(C^\infty(M),L^2(M,E),P)$ is a spectral triple. In fact, $P$ has a purely discrete spectrum and its singular values grow to infinity \cite[Lemma 1.6.3]{Gilkey1}, so its resolvent is compact. Moreover, since $a\in C^\infty(M)$ can be seen as a pdo of order zero, $[P,a]$ is bounded as a pdo of order 0.
\\
The archetype of such situation is when $M$ is spin: Let $\SS$\index{_zS@$\SS$}  be a spinor bundle over $M$. Let moreover $\A = \Coo(M)$, $\H = L^2(M,\SS)$ and let $\DD = \Dslash$ with $$\Dslash \vc -i \gamma^{\mu} \nabla_{\mu}^{\SS}$$ \index{_zD1@$\Dslash$}-- the standard Dirac operator on $(M,\SS)$ (cf. \cite{Friedrich}). Then, $(\A,\H,\DD)$ is a spectral triple.
\hfill$\blacksquare$
\end{example}
As the algebra $\A$ in the above example is commutative, the associated spectral triple is also called \emph{commutative}\index{spectral triple!commutative}. A natural question arises: Given a spectral triple with a commutative algebra $\A$, can one recover the underlying manifold? The positive answer is the content of the famous Connes' Reconstruction Theorem \cite{ConnesRecon}. It requires several additional assumptions on the spectral triple (see \cite[Chapter 3]{Varilly} for a pedagogical explanation of these). However, in the noncommutative realm, there are known examples of perfectly workable noncommutative geometries for which some of these additional assumptions are not satisfied \cite{AllPodles,EquatorialPodles,DiracSUq2,PodlesSA}.

If $M$ is a locally compact Riemannian spin manifold, then the natural associated $C^*$-algebra $\A=C_0(M)$ of functions on $M$ vanishing at infinity does not have a unit. Moreover, the Dirac operator on $M$ does not have a compact resolvent. \\
Thus, when $\A$ is not unital \index{spectral triple!nonunital}, the last item of Definition \ref{def:triple} needs to be replaced by 
\begin{itemize}
\item $a(\DD-\lambda)^{-1}$ is compact  for all $a\in\A$ and $\lambda \notin \RR$. 
\end{itemize}
Equivalently, one can require $a(\DD^2+\epsilon^2)^{-1/2}$ to be compact for any $\epsilon>0,\,a\in\A$.

Again, in practice one needs to consider the topological issues. This can be done (see, for instance, \cite{IochumMoyal}) by demanding that $\A$ and a preferred unitisation  $\wt{\A}$  \index{_zAA@$\protect\widetilde{\A}$} of $\A$ be pre-$C^*$-algebras, which are faithfully represented on $\H$ in terms of bounded operators and $[\DD,a]$ extends to a bounded operator for every $a \in \wt{\A}$. See also \cite{RennieSmooth}.

To simplify, we stick to the original Definition \ref{def:triple} and hence assume from now on that the algebra $\A$ is unital, unless explicitly written.\pagebreak

\begin{example}
\label{ex:noncommutative}
A basic noncommutative spectral triple is defined by
\begin{gather*}
\A_F={\cal M}_n(\CC) \,(\text{complex } n\times n\text{-matrices})\index{_zMn1@${\cal M}_{n}(\CC)$}, \quad \H_F=\CC^n, \quad \DD_F = \DD_F^{*} \in \A_F.
\tag*{$\blacksquare$}
\end{gather*}
\end{example}
Since the Hilbert space $\H_F$ is finite dimensional, finite direct sum of $(\A_F,\H_F,\DD_F)$'s as in Example \ref{ex:noncommutative} are called \emph{finite} spectral triples\index{spectral triple!finite} (see \cite{Krajewski} for a classification).

Taking a tensor product of a commutative spectral triple with a finite one results in an \emph{almost commutative} geometry \index{almost commutative geometry}\cite{AlmostName}. There exists an analogue of the Reconstruction Theorem for almost commutative spectral triples \cite{CacicReconstruction} allowing one to retrieve a smooth manifold $M$ together with a vector bundle and a connection. \\
Almost commutative geometries are extensively employed in building physical models of fundamental interactions \cite{SUSYNCG,
Almost1,MarcolliCosmoBook,MairiReview,WalterBook}.

The purpose of this book is, however, to study the spectral action in full generality of noncommutative geometry, beyond the almost commutative realm. We refer the Reader to the textbook \cite{WalterBook} for a friendly introduction to almost commutative geometries and their physical applications.
\begin{example}
Other illustrative examples are given by the noncommutative tori and the Podle\'s spheres -- see Sections \ref{sec:torus} and \ref{sec:Podles} of Appendix \ref{chap:appendix}.
\hfill$\blacksquare$
\end{example}
\begin{remark}\label{rem:pert}
Given a spectral triple $\ahd$ one can always obtain a new one $(\A,\H,\DD_V)$, with $\DD_V = \DD + V$, $V = V^* \in \B(\H)$ and $\Dom \DD_V \vc \Dom \DD$. Indeed, $(\A,\H,\DD_V)$ is a spectral triple as $[\DD + V,a]$ extends to a bounded operator for any $a \in \A$ whenever $[\DD,a]$ does so. Furthermore, if $z$ is in the resolvent of $\DD_V$ and $z'$ is in the resolvent of $\DD$ then, 
$$(\DD + V - z)^{-1} = (\DD - z')^{-1}\left[ \bbbone -(V+z'-z)(\DD + V - z)^{-1}\right]$$
is compact since the first term is compact and the second one is bounded.

However, the geometry of $(\A,\H,\DD_V)$ need not a priori be related to the one of $(\A,\H,\DD)$. If we want to obtain a geometry which is in a suitable sense equivalent, the perturbation $V$ has to acquire a precise form (see Section \ref{sec:fluc}).
\hfill$\blacksquare$
\end{remark}

On the technical side, we need to take into account the fact that $\DD$ may be non-invertible. We adopt the following convention 
\begin{align}
\label{absD}
D \vc \DD + P_0,  \index{_zD2@$D$}
\end{align}
where $P_0$\index{_zP0@$P_0$} is the projection on $\Ker \DD \subset \H$. The operator $P_0$ is a finite-rank (i.e. $\dim \Ima P_0 < \infty$) selfadjoint operator on $\H$ and $D$ is an invertible operator with a compact resolvent. Thus, by the previous remark, $(\A,\H,D)$ is also a spectral triple. Moreover, notice that $\abs{D} = \abs{\DD} + P_0$ and $\abs{D}^{-1}$ is compact.

Another possibility is to define (see for e.g. \cite{ConnesModular}) an invertible selfadjoint operator $\bD$\index{_zD2b@$\bD$} as the restriction of $\DD$ to the Hilbert subspace $(\bbbone-P_0) \H$ . Yet another option, chosen in \cite{CPRS}, is to work with the invertible operator $(\bbbone+\DD^2)^{1/2}$. \\
The selection of a prescription to cook up an invertible operator for $\DD$ is only a matter of convention (cf. \cite[Section 6]{CPRS}, \cite[Remark 3.2]{TorusSA}). However, one has to stay vigilant, as different choices may affect the associated spectral functions (see, for instance, Formula \eqref{heat_ker}).

\smallskip

The first vital property of a spectral triple we shall need is a `finite dimensionality' condition:
\begin{definition}
\label{def:p_summability}
A spectral triple $\ahd$ is {\it finitely summable} \index{spectral triple!finitely summable} or, more precisely, $p$-\emph{summable} \index{spectral triple!p-summable@$p$-summable} for some $p\geq 0$ if $\Tr \vert D\vert^{-p}<\infty$.  \\
Note that $p$-summability yields $(p+\epsilon)$-summability but not $(p-\epsilon)$-summability for $\epsilon>0$, so we define: 

The triple is said of \emph{dimension} $p$  (or $p$-\emph{dimensional}) \index{spectral triple!p-dimensional@$p$-dimensional} when 
\begin{align*}
p~\vc~\inf \, \{q \geq 0 \; \vert \; \Tr \vert D\vert^{-q} < \infty\} < \infty.
\end{align*}
\end{definition}
This notion of dimension will be refined in Definition \ref{def:dimspec}.

Remark that finite spectral triples (cf. Example \ref{ex:noncommutative}) are always 0-dimensional. But spectral triples with $\dim \H = \infty$ can also have dimension 0 --- like, for instance, the standard Podle\'s sphere (see Appendix \ref{sec:Podles}). 
\begin{example}
\label{ex:dim_manifold}
Let $\ahd$ be a commutative spectral triple based on a Riemannian manifold of dimension $d$, then $\ahd$ is $d$-dimensional \cite[p. 489]{Elements}.
\hfill{$\blacksquare$}
\end{example}
A slightly weaker notion is the one of $\theta$-summability:\index{spectral triple!T-summable@$\theta$-summable}\label{theta_summ} $\Tr e^{-t\DD^2}<\infty$ for $t>t_0 \geq 0$ \cite[Chapter 4, Section 8]{ConnesNCG}. Every finitely summable spectral triple is $\theta$-summable with $t_0 = 0$, but the converse is not true (cf. \cite{connes1989} and \cite[Chapter 4]{ConnesNCG}). \\
The geometry of spectral triples which are not finitely summable is, however, too poor to accommodate various analytical notions (cf. Section \ref{sec:dimsp}), which are indispensable for the spectral action computations. On the other hand, note that we do not require $p$ to be an integer. For instance, fractal spaces are fruitfully described via $p$-summable spectral triples with $p$ being the (irrational) Hausdorff dimension of the fractal \cite{MarcolliBall,Skalski,Christensen3,Christensen2,Christensen1,Cipriani,ConnesGarden,Guido1,Guido2,Guido3,Kellendonk,Lapidus,LapSar}.

\smallskip

The next key property encodes the notion of smoothness:
\begin{definition}
\label{def:regularity}
A spectral triple $\ahd$ is \emph{regular} \index{spectral triple!regular} if
\begin{align}\label{regularity}
\forall a \in \A \quad a, [\DD,a] \in \capB{n \in \N} \,\,\Dom \delta'^{n}, \,\, \text{ where } \delta' \vc [\abD,\,\cdot]. \index{_a2d0@$\delta'(T)$}
\end{align}
\end{definition}
The map $\delta'$ is an unbounded derivation of the algebra $\B(\H)$. Recall that the notation $T \in \Dom \delta'$ means that $T$ preserves $\Dom\,\abD$ and $\delta'(t) = [\abD,T]$ extends (uniquely) to a bounded operator on $\H$. We note that some authors refer to assumption \eqref{regularity} as smoothness \cite[Definition 11]{RennieSmooth} or $QC^{\infty}$ ($Q$ for ``quantum'') \cite[Definition 2.2]{CPRS}.

\label{topology_reg}
The regularity assumption allows one to equip $\A$ with a topology generated by the seminorms $a \mapsto \norm{\del'^k(a)}$ and $a \mapsto \norm{\del'^k([\DD,a])}$. The completion of $\A$ in this topology yields a Fr\'echet pre-$C^*$-algebra $\A_{\delta'}$ \index{_zAAA@$\A_{\delta'}$}and $(\A_{\delta'},\H,\DD)$ is again a regular spectral triple \cite{RennieSmooth} (see also \cite[p. 469]{Elements}, \cite[Section 3.4]{Varilly} and \cite{CPRS}).

A commutative spectral triple $\ahd$ based on a Riemannian manifold is regular. Condition \eqref{regularity} assures that the functions constituting $\A$ are indeed smooth. In this case, $ \A_{\delta'} \cong C^{\infty}(M)$ also as topological spaces \cite[Proposition 20]{RennieSmooth}.

\smallskip

In some of the approaches (in particular, in almost commutative geometries) it is desirable to encode in the axioms of a spectral triple the fact that the classical Dirac operator is a first order differential operator \cite[Section 3.3]{Varilly}. Such a demand was originally used to restrict admissible Dirac operators for almost commutative geometries \cite{Almost1,Almost2} and restrain the free parameters of the underlying physical models. On the other hand, in recent studies, it is argued that one does not actually need the first-order condition defined below and there exist examples of spectral triples for which it is not satisfied (see \cite{ConnesFirst} and references therein). Consequently, we will not assume the first-order condition to hold throughout this book. Nevertheless, we shall state it and explain its origin. To this end we need the following notions:

\begin{definition}
\label{def:even}
A spectral triple $\ahd$ is \emph{even} \index{spectral triple!even/odd} if there exists a selfadjoint unitary operator $\gamma$ \index{_a1gamma@$\gamma$} on $\H$ such that $\gamma^2 = 1$, $\gamma \DD = - \DD \gamma$ and $\gamma a = a \gamma$ for all $a \in \A$. Otherwise, the spectral triple is \emph{odd}.
\end{definition}

\begin{definition}
\label{def:real}
A spectral triple $\ahd$ is \emph{real} \index{spectral triple!real} of \emph{KO-dimension} \index{ko-dimension@$KO$-dimension} $d \in \Z/8$ if there is an antilinear isometry $J:\H \to \H$ called the \emph{reality operator} such that
$$
J\DD=\epsilon \,\DD J,\quad J^2=\epsilon',\quad \text{and }\, J \gamma=\epsilon'' \,\gamma J \,\,\text{ when the triple is even},
$$
with the following table for the signs $\epsilon, \epsilon',\epsilon''$
\begin{align}
\label{commu}
\begin{tabular}{|c| cccccccc|}
\hline
d & 0 & 1 & 2 & 3 & 4 & 5 & 6 & 7 \\
\hline
$\epsilon$ & 1&-1&1&1&1&-1&1&1\\
$\epsilon'$ &1&1&-1&-1&-1&-1&1&1\\
$\epsilon''$ &1& &-1 & &1 & &-1&\\
\hline
\end{tabular}
\end{align}
and the following commutation rule (see \cite{ConnesMarcolli} or \cite[Section 9.5]{Elements} for the details)
\begin{align}
\label{oppcommut}
[a,\ J b^*J^{-1}]=0,\quad\forall a,b\in \A.
\end{align}
\end{definition}

For a spectral triple based on a Riemannian manifold of dimension $d$ the $KO$-dimension is just $d$ mod 8. It encodes the fact that the Dirac operator is a square root of the Laplacian, what generates a sign problem corresponding to the choice of a spin structure (and orientation). In this context the operator $J$ plays the role of a charge conjugation for spinors (see \cite[Section 5.3]{Elements} or \cite{VarillyDirac}) and it encodes the nuance between spin and spin$^\CC$ structures. On the other hand, given a spectral triple $\ahd$ with a noncommutative algebra one can usually find different reality operators leading to real spectral triples with different $KO$-dimensions (see \cite{Krajewski}).

The operator $J$ takes its origin in the modular theory of von Neumann algebras (see \cite[Chapter 1, Section 3 \& Chapter 5]{ConnesNCG}) and has interesting applications in the algebraic quantum field theory \cite[Chapter V]{Haag}. \\
In the context of spectral triples, with the help of $J$ one can define a representation on $\H$ of the opposite algebra $\Aop$ \index{_zAAAA@$\Aop$}, which is isomorphic to $\A$ as a vector space, but the multiplication in $\Aop$ is inverted, i.e. $a \bullet_{\Aop} b \vc b \bullet_{\A} a$. Given a representation $\pi: \A \to \B(\H)$, define the representation $\pi^o : \Aop \to \B(\H)$ by $$\pi^o(a) \vc J a^* J^{-1}.$$ The condition $[\pi(a),J \pi(b^*) J^{-1}] = 0$ for $a, b \in \A$ means that the two representations commute $[\pi(\A), \pi^o(\Aop)] = 0$, hence $\pi^o(\Aop)$ is in the commutant of $\pi(\A)$ in $\B(\H)$. If $\A$ is commutative, $\A = \Aop$ and this requirement becomes trivial. See, for instance, \cite[Chapter 3]{Varilly} for more details.

We are now ready to formulate the announced first-order condition.
\begin{definition}
\label{def:first_order}
A real spectral triple $\ahd$ meets the \emph{first-order condition}\index{first-order condition} if
\begin{align}\label{first_order}
[[\DD,a],J b^* J^{-1}] = 0, \quad \forall \; a, b \in \A.
\end{align}
\end{definition}
As mentioned above, it encodes the fact that for a commutative spectral triple $\DD = \Dslash$ is a first order differential operator. The reason for the appearance of the representation of $\Aop$ in \eqref{first_order} is that $[\DD,\cdot]$ is not a derivation from $\A$ to itself, but rather to the commutant of $\Aop$ \cite{ConnesFirst}. The first-order condition plays an important role in the study of the spectral action fluctuations --- see Section \ref{sec:fluc}.

Finally, we remark that each of the properties of being even, real, regular or $p$-summable can be extended to the non-unital framework  \cite{IochumMoyal,RennieSmooth,RennieSummable}.

\section{\texorpdfstring{Some Spaces Associated with $\DD$}{Some spaces associated with D}}

In this section we discuss the pseudodifferential calculus suitable for noncommutative geometries. We essentially base on the classical papers \cite{ConnesSPV,ConnesMoscovici}, however some of the definitions are taken from more recent works \cite{CPRS,PodlesSA,IochumNotes,Tadpole}.

Let us first define the following scale of spaces for a parameter $s \in \RR$ 
\begin{align}
\label{Hs}
&\H^s\vc \Dom\,\vert D\vert^s.\index{_zH1@$\H^s$}
\end{align}
If $s\leq 0$, $\H^s=\H^0=\H$ and for $s\geq 0$, $\H^s$ are Hilbert spaces for the Sobolev norm 
\begin{align*}
\norm{\xi}_s^2 \vc \norm{\xi}^2 + \Vert \vert D\vert^{s} \xi \Vert^2\index{_An0@$\norm{\cdot}_{s}$}
\end{align*}
and $$\H^{s+\epsilon}\subset \H^s \subset \H^0 \text{ for }s,\epsilon \geq 0$$ since the injection $\H^{s+\epsilon} \hookrightarrow\H^s$ is continuous. 

Let us note, that we do not need the invertibility of $\DD$ since $\H^s = \Dom \,\abs{\DD}^s$ for $s \geq 0$ as $\DD P_0 = P_0 \DD = 0$ and $P_0 \in \B(\H)$. Actually, for $s\geq0$ we also have $ \H^s = \Dom \,(\bbbone+\DD^2)^{s/2}$ as shown in \cite[Section 6]{CPRS}. On the other hand, the Sobolev norms will \textit{not} be the same if we swap $\vert D\vert$ for $\abs{\DD}$ or $(\bbbone+\DD^2)^{1/2}$. In the second case the norms will not coincide even if $\Ker \DD = \{0\}$. However, the precise form of the Sobolev norms for $\H^s$ is not relevant and, as already stressed, the choice of $\vert D\vert$ instead of $(\bbbone+\DD^2)^{1/2}$ is only a matter of convention.

We also define the \emph{domain of smoothness of} \index{domain of smoothness} $\DD$ as 
\begin{align}
\label{Hinfty}
\H^{\infty} \vc \underset{s\geq 0}{\mathlarger{\mathlarger{\cap}}}\, \, \H^s = \underset{k\in\N}{\mathlarger{\mathlarger{\cap}}}\,\,\H^k.
\end{align}
$\H^{\infty}$ \index{_zH11@$\H^{\infty}$} is dense in $\H$ and is in fact a core (see \cite[p. 256]{SimonReed} for a precise definition) for $\DD$ \cite[Theorem 18]{RennieSmooth}. Actually, it is sufficient to consider $\H^k$ with $k \in \N$ \cite[p. 467]{Elements} (see also \cite[Definition 6.11]{VarillyDirac}). $\H^{\infty}$ can be equipped with a topology induced by the seminorms $\norm{\cdot}_k$, which makes it a Fr\'echet space.

With the spaces $\H^s$ at hand, we define the following classes of unbounded operators on $\H$ for any $r \in \RR$ (compare \cite[Section 1]{ConnesSPV}, \cite[Definition 6.1]{CPRS}): 
\begin{multline}
\label{op}
\op^r \vc \big\{ T \in \L(\H) \; \big\vert \; \Dom \,T \supset \H^\infty \text{ and for any }  s \geq r \\
T \text{ maps }(\H^\infty, \norm{\cdot}_s) \text{ continuously into } (\H^\infty, \norm{\cdot}_{s-r}) \big\}. \index{_zOp@$\op^r$}
\end{multline}

For instance $\forall r \in \RR,\,\abs{D}^r \in \op^r$. Any operator in $\op^r$ extends to a bounded map from $\H^s$ to $\H^{s-r} \,\forall s\geq r$ and since $\H^0=\H$, we have $$\op^0 \subset \B(\H).$$ Moreover, 
\begin{align}
\label{opProp}
 \op^r \subset \op^s \text{ if } r \leq s ,  \quad \op^r\cdot \op^s \subset \op^{r+s} \text{ for all } r,\,s \in \RR.
\end{align}

In the previous section we met a map $\delta'$ on $\B(\H)$, the domain of which played a pivotal role in the axiom of regularity. For further convenience we also define:
\begin{align*} 
& \del (\cdot)\vc [\vert D \vert,\, \cdot], && \nabla(\cdot) \vc [D^2,\,\cdot], &&
\sigma(\cdot)\vc\vert D\vert \cdot \vert D\vert^{-1}\, \\
& \EEE(\cdot)\vc\del(\cdot)\vert D\vert^{-1} && \E (\cdot)\vc \nabla(\cdot)\,\vert D\vert^{-2}.\index{_a2d1@$\delta(T)$}\index{_zE@$\protect\EEE$}\index{_AAnabla@$\nabla$}\index{_zEE@$\E$} &&
\end{align*}
Let $\Xi \in \{\delta,\nabla,\sigma,\EEE, \E\}$. Recall that if $T \in \Dom \Xi \subset \B(\H)$, then $\Xi(T) \in \B(\H)$.  Actually, one can extend any $\Xi$ to a map on $\{T \in \L(\H) \, \vert \, \Dom T \supset \H^{\infty}\}$. By the habitual abuse of notation we denote the extensions of $\Xi$'s with the same symbols. However, one has to stay vigilant as in the following $\Dom \Xi$ will always mean a subset of $\B(\H)$.


Furthermore, for any $T \in \L(\H)$ with $\Dom T \supset \H^{\infty}$ we define the following one-parameter group of (unbounded) operators:
\begin{align}
\label{OP_z}
\sigma_{z}(T) \vc \vert D\vert^{z}\, T\, \vert D\vert^{-z},\,\text{ for } z \in \CC,
\end{align}
\index{_a9sigmaz@$\sigma_z$}using the Cauchy integral along a curve $\C \subset \CC$ (see \eqref{heat as integral}) to define
\begin{align}
\label{def: D^z as integral}
\vert D\vert^{-z}=\tfrac{1}{i2\pi} \int_{\lambda \in \C} \lambda^{-z/2}\,(\lambda-D^2)^{-1}\, d\lambda.
\end{align}
Note that for $T \in \Dom \sigma$ we have $\sigma^n(T) = \sigma_n(T)$ for all $n \in \N$.

\begin{lemma}
[\cite{ConnesMoscovici,CPRS}]
\label{intersection of domains}
We have 
$\capB{n \in \N}\, \Dom\,\del'^n =\capB{n \in \N}\, \Dom\,\del^n\subset \op^0$.
\end{lemma}

\begin{proof}
Since $\awD=\aD+P_0$, $\Dom \del'=\Dom \del$ because $P_0$ is bounded and we get the equality of the lemma.

Let now $T\in \capB{n \in \N}\, \Dom\,\del^n$. One checks that $\sigma=\id+\EEE$ and, for every $n\in \N$,
\vspace*{-0.2cm}
\begin{align}
\label{eq:epsilon et sigma}
\EEE\,^n(T)=\del^n(T)\,\awD^{-n}\, \text{ and }\,\sigma^n(T)=(\id+\EEE)^n(T)=\sum_{k=0}^n \left ( \begin{smallmatrix} n \\ k \end{smallmatrix} \right ) \,\del^k(T)\,\awD^{-k}.
\end{align}
Hence, $\EEE\,^n(T)$ and $\sigma^n(T)$ are bounded. Similarly, 
\begin{align}
\label{eq:sigma moins}
\sigma^{-n}(T)=\sum_{k=0}^n (-1)^k\tbinom{n}{k} \, \awD^{-k}\,\del^k(T)
\end{align}
is bounded too. Thus, for $n\in \Z$ and $\xi \in \H^{\infty}$,
\begin{align*}
\norm{T\xi}_n^2 &= \norm{T\xi}^2+\norm{\vert D\vert^n T\xi}^2 = \norm{T\xi}^2+\norm{\sigma^n(T) \awD^n \xi}^2  \\
& \leq c(\norm{\xi}^2+\norm{\awD^n\xi}^2) =c \norm{\xi}^2_{n}.
\end{align*}
The case of an arbitrary $n\in \RR$ follows by the interpolation theory of Banach spaces (cf. \cite[Formula (10.65)]{Elements} and \cite[Chapter 4, Appendix B]{ConnesNCG}).
\hfill $\Box$
\end{proof}

We now introduce yet another class of operators on $\H$ (cf. \cite{ConnesMoscovici}, \cite[Def. 6.6]{CPRS}):
\begin{align}
\label{OP}
\OP^0 \vc \capB{n\in\N}\, \Dom\,\del^n, \quad \OP^r  \vc \{T\in \L(\H) \; \vert \;\, \vert D\vert^{-r} T \in \OP^0\}\, \,\text{ for } r \in \RR. \index{_zOpr@$\OP^r$}
\end{align}
\vspace*{-0.1cm}
The definition of $\OP^r$ is symmetric: When $r\in \N$, Equations \eqref{eq:epsilon et sigma} yields
\begin{align}
\label{eq:symmetry}
\OP^r=\awD^r \OP^0=\awD^r \OP^0 \awD^{-r} \awD^r \subset \OP^0 \awD^r.
\end{align}
When $T\in \OP^r$, we say that {\it the order of $T$ is (at most) $r$} and write $\ord T \vc r$.\index{order of an operator}

Since $\H^{\infty}$ is dense in $\H$, the operators in $\op^r$ (and a fortiori in $\OP^r$) are densely defined and we can define $$(\OP^r)^* \vc \{ T^* \; | \; T \in \OP^r\}.$$
We have
$$(\OP^r)^*  =\OP^r,$$
what follows from the observation that $(\del^n(T))^* = (-1)^n \del^n(T^*)$, so $(\OP^0)^*=\OP^0$ for $n\in \N$, and the symmetry \eqref{eq:symmetry} of $\OP^r$.

Note also that 
\begin{align}
\label{eq:P0 smoothing}
P_0 \in \OP^{-k} \text{ for all } k \in \RR,
\end{align} 
because $P_0$ is trivially in $\OP^0$ and $P_0 \vert D\vert^{k} = P_0 \in \B(\H)$ for any $k\geq 1$.

Obviously, we have $\awD \in \OP^1$ (but not in $\OP^0$, because $\awD \notin \B(\H)$!). Also, $D \in \OP^1$ since with $F = D \awD^{-1} \in \B(\H)$ we have $\delta(F) = 0$ as $D$ is selfadjoint. Consequently, $\DD, \abD \in \OP^1$, but again $\DD, \abD \notin \OP^0$ for $\DD, \abD \notin \B(\H)$.

Note that the regularity condition \eqref{regularity} is equivalent to requiring that
\begin{align}
\label{regularityOP}
\A \subset \OP^0, \quad [\DD,\A] \subset \OP^0.
\end{align}
In that case, for instance, $a\abD[\DD,b]\wD^{-4} \in \OP^{-3}$ with $a,b\in \A$. 

In \cite{ConnesMoscovici} it was proved that operators of order $\leq 0$ admit another characterisation:
\begin{align*}
\OP^0=\{T\,\vert\, t \mapsto F_t(T) \in C^\infty(\RR,\B(\H))\},
\end{align*}
where $F_t(T) \vc e^{it\,\vert D\vert}Te^{-it\vert D \vert}$ for $t\in \RR$, which is reminiscent of the geodesic flow \cite{ConnesSPV,ConnesMarcolli}, \cite[Chapter 8]{Cordes}. Recall that $$\delta(T)=[\abs{D},T]=\text{strong-}\lim_{t\to 0} \tfrac{F_t(T)-F_0(T)}{i\,t}.$$

\begin{proposition}
\label{prop:OP}
Let $r,\,s\in \RR$ and $z\in \CC$. Then

i) $\OP^r \subset \op^r$.

ii) $\sigma_z(\OP^r)\subset \OP^r$. 

iii) $\OP^r\OP^s \subset \OP^{r+s}$.

iv) $\OP^r \subset \OP^s$ when $r\leq s$.

v) $\del(\OP^r) \subset \OP^{r}$.

vi) $\nabla(\OP^{r}) \subset \OP^{r+1}$ and $\E(\OP^r)\subset \OP^{r-1}$.

vii) $[\awD^r,\OP^s] \subset \OP^{r+s-1}$.

viii) If $T\in \OP^r$, $r\geq 0$ has an inverse with $T^{-1}\awD^r\in \B(\H)$, then $T^{-1}\in \OP^{-r}$.

ix) $\OP^0 \subset \B(\H)$ and $\OP^{-1}\subset \KKK(\H)$.
\end{proposition}

\begin{proof}
$i)$ This follows from $\OP^0\subset \op^0$ (Lemma \ref{intersection of domains}), $\awD^r \in \op^r$ and \eqref{opProp}.

$ii)$ When $T\in \OP^0$, since $\awD^{-k}$ and $\del^k(T)$ are also in $\OP^0$ for any $k\in \N$, Formulae \eqref{eq:epsilon et sigma} and \eqref{eq:sigma moins} tell us that for any $n\in \Z$, $\sigma_n(T)=\sigma^n(T) \in \OP^0$. \\
Let us fix $m\in \N$ and let $F_m: \,z \in \CC \mapsto  \delta^m(\sigma_z(T))\in \L(\H)$. Since $F_m(n)$ is bounded when $n\in \Z$, a complex interpolation shows that $F_m(z)$ is bounded for $z \in \CC$, so that $\sigma_z(T) \in \OP^0$.\\
When $T\in \OP^r$, $T=\abs{D}^r T'$ with $T'\in \OP^0$, thus $\sigma_z(T)=\abs{D}^r\sigma_z(T')\in \OP^r$.

$iii)$ Let $T\in \OP^r$ and $T'\in \OP^s$. Then $\awD^{-r}T$ and $\awD^{-s}T'$ are in $\OP^0$ and by $ii)$ $\awD^{-s}(\awD^{-r}T)\awD^{s} \in \OP^0$. So $\awD^{-(r+s)}TT'= (\awD^{-(r+s)}T\awD^{s})(\awD^{-s}T')\in \OP^0$.

$iv)$ When $s\geq r$,  $\awD^{r-s}$ is a bounded operator and is in $\OP^0$. Thus if $T \in \OP^r$, $\awD^{-s}T=\awD^{r-s}( \awD^{-r} T)\in \OP^0$ by $iii)$.

$v)$ If $T\in \OP^r$, then $T=\awD^r S$ with $S \in \OP^0$ and $\delta(T)=\delta(\awD^r)S+\awD^r\delta(S)$, so the result follows from $\del(\OP^0)\subset \OP^0$ and $iii)$.

$vi)$ Let $T\in \OP^r$. Since $\nabla(T)=\del(T)\awD +\awD\del(T)$, we get $\nabla(T)\in \OP^{r+1}$ from properties  $iii)$ and $v)$. \\
Moreover, $\E(T)= \nabla(T)\awD^{-2}\in \OP^{r+1}\OP^{-2} \subset \OP^{r-1}$ using $v$) and $iii)$.

$vii)$ By $iii)$ it is sufficient to prove that $[\awD^r,\OP^0]\in \OP^{r-1}$. This is true for $r\in \N$ using $v)$ and for $r\in -\N$ using $[A^{-1},B]=-A^{-1}[A,B]A^{-1}$. Finally, this also holds true for an arbitrary $r\in \RR$ by interpolation --- as in the proof of $ii)$.

$viii)$ Assume $r=0$. Remark that $\delta(T^{-1})=-T^{-1}\delta(T)T^{-1}$ is bounded. Thus 
 $\delta^2(T^{-1})=-T^{-1}\delta^2(T)T^{-1}-\delta(T^{-1})\delta(T)T^{-1}-T^{-1}\delta(T)\delta(T^{-1})$ is bounded and, by induction, $\delta^n(T^{-1})$ is bounded for any $n\in\N^*$. Thus $T^{-1}\in \OP^0$.\\
Now for $r>0$, $T=\awD^r S$ where $S\in \OP^0$ is such that $S^{-1}=T^{-1}\awD^r$ is bounded by hypothesis, thus in $\OP^0$ by previous argument and the claim follows from $iii)$.

$ix)$ By $i)$, we have $\OP^0 \subset \op^0 \subset \B(\H)$. Since $\abs{D}^{-1}$ is a compact operator, the last assertion  follows.
\hfill $\Box$
\end{proof}
Let us stress that the $\op$ classes are, in general, strictly larger than $\OP$'s. E.g. on the standard Podle\'s sphere (cf. Appendix \ref{sec:Podles}), which is \emph{not} regular, we have $a \in \op^0$, but $a \notin \OP^0$ for any $a \in \A_q$, $a \neq c \bbbone$.

\smallskip

The following, innocent-looking lemma, will prove very handy in practice:
\begin{lemma}\label{lm:polynomial rest}
Let $T \in \L(\H)$. Assume that there exist  $r \in \RR$ and $0\leq \epsilon < 1$ such that for any $N \in \N$, 
$T = \sum_{n=0}^N P_n + R_N$,  with $P_n \in \OP^{r-n}$ and $R_N \in \OP^{r-N-1+\epsilon}$.\\
Then, actually, $R_N \in \OP^{r-N-1}$.
\end{lemma}

\begin{proof}
We have $R_{N} = P_{N+1} + R_{N+1}$. Because $R_{N+1} \in \OP^{r-N-2+\epsilon} \subset \OP^{r-N-1}$ and $P_{N+1} \in \OP^{r-N-1}$  we conclude that $R_N$ is actually in $\OP^{r-N-1}$.
\hfill $\Box$
\end{proof}

We will often need to integrate operators (like in Formula \eqref{def: D^z as integral}) from a given $\OP$ class and we will need a control of the order of the integral. A typical example is the following: For $i\in \{1,\dotsc,n\}$, let us be given $A_i(\lambda)\in \OP^{-a_i}$ with $a_i\in \RR$, commuting with $\abs{D}$, and $B_i\in \OP^{b_i}$, $C(\la) \in \OP^{-c}$ with $b_i,c\in \RR$, which do not necessarily commute with $\abs{D}$. For some $I \subset \RR$, we will need to show that $\int_I R(\lambda)\,d\lambda\in\OP^r$ for some $r \in \RR$, when $R(\la)=A_1(\lambda)B_1 \, \dotsm \, A_n(\lambda)B_n\,C(\la)$.

Let us define $$a\vc a_1+\dotsc+a_n\text{ and }b\vc b_1+\dotsc+ b_n.$$ Assuming $b\leq a+c$, we first remark that the integrand $R(\la)$ is in $\OP^{-a-c+b}\subset \OP^0$, so it is bounded. 
We now decompose $\abs{D}^{a+c-b}R(\la)$ in the following way: \\For $\alpha_k=\sum_{j=0}^{k}a_j,\, \beta_k=\sum_{j=0}^{k}b_j$, with $a_0=b_0=0$, we have
\begin{align*}
&\hspace{2cm} \abs{D}^{a+c-b}R(\la)  =E_1(\la)F_1 \, \dotsm \, E_n(\la)F_n \,G(\la),\\
& \hspace{-0.3cm}\text{where}\\
& \quad E_k(\la)  \vc \abs{D}^{(a+c-\alpha_{k-1})-(b-\beta_{k-1})} \,A_k\,(\la) \abs{D}^{-(a+c-\alpha_k)+(b-\beta_{k-1})}=\abs{D}^{a_k} A_k(\lambda),\\
& \quad F_k  \vc \abs{D}^{(a+c-\alpha_k)-(b-\beta_{k-1})} \,B_k\, \abs{D}^{-(a+c-\alpha_k)+(b-\beta_{k})}, \\
& \quad G(\la)\vc \vert D\vert^c C(\la).
\end{align*}
Remark that all operators $E_k(\la),\,F_k, G(\la)$ are in $\OP^0$, and hence bounded.

Typically, the operators $A_k(\lambda)$ would be of the form $(D^2 + \lambda)^{-1}$, which gives $A_k(\lambda) \in \OP^{-2}$ and $A_k(\lambda) = \Oinf(\lambda^{-1})$. So, to assure the integrability of $R(\lambda)$ for $I = \RR^+ \vc [0,\infty)$\index{_zRplus@$\RR^+ = [0,\infty)$} we usually have to sacrifice a few orders of $\OP$. \\ For example, with $R(\lambda) = (D^2 + \lambda)^{-2}\in \OP^{-4}$, we get $\norm{ \abs{D}^{r} \int_{\RR^+} R(\lambda)\,d\la} < \infty$ only for $r < 2$. 

\begin{theorem}
\label{thm:integral of OP}
Let $I \subset \RR$. Assume that there exist $r \in \RR$ and $n \in \N^*$ such that for any $\lambda \in I$ a given operator $R(\lambda) \in \L(\H)$ can be decomposed as
\begin{align}
\abs{D}^r\, R(\la) = E_1(\la)F_1 \, \dotsm \, E_n(\la)F_n \,G(\la), \label{hyp:int of OP}
\end{align}
with $E_k(\la), F_k, G(\la) \in \OP^0$ and $[\abs{D},E_k(\lambda)] = 0$.\\
If $$\displaystyle \int_I \Vert E_1(\la)\Vert \dotsm \Vert E_n(\la) \Vert  \cdot  \Vert \delta^m(G(\la)) \Vert\,d\la < \infty \text{  for }m \in \N,$$
then $\displaystyle \int_I R(\la) d\la \in \OP^{-r}$.
\end{theorem}
\begin{proof}
We need to show that $$\abs{D}^r\,\R \vc \abs{D}^r \int_I R(\la) d\la \in \OP^0=\cap_{m=0}^\infty \Dom \delta^m.$$
For $m=0$, we have $$\Vert \abs{D}^r\,\R\Vert \leq (\prod_{k=1}^n \norm{F_k}) \,\int_I \prod_{k=1}^n \Vert E_k(\la)\Vert\cdot\Vert G(\la)\Vert\,d\la<\infty$$ by hypothesis. Hence, $\abs{D}^r\,\R$ is bounded.\\
An application of $\delta$ to $\abs{D}^r\,\R$ generates from the integrand $\abs{D}^r\,R(\la)$ a finite linear combination of terms like those in \eqref{hyp:int of OP} with either one of the $F_k$'s replaced by $\delta(F_k)$ or $G(\la)$ replaced by $\delta(G(\la))$, as $\delta(E_k)=0$. Since $F_k \in \OP^0$, the norm of $\delta^m(F_k)$ is finite for any $m \in \N$ and the norm of $\delta^m(\abs{D}^r\,\R)$ can be estimated along the same lines as for the case $m=0$.
\hfill{$\Box$}
\end{proof}

As the first application of Theorem \ref{thm:integral of OP} we derive an operator expansion \eqref{OPexpansion}, which will play a pivotal role in the construction of an abstract pseudodifferential calculus. We will need the following technical lemma (cf. \cite[Lemma 6.9]{CPRS}):

\begin{lemma}
\label{lm:OP_exp}
Let $T \in \OP^{\,r}$ for an $r \in \RR$ and let $\lambda \notin \spec D^2$. Then, $\forall\, n \in \N^*,\, N \in \N$,
\begin{align}
\label{OP_exp_int}
&\!\!\!\!(D^2 - \lambda)^{-n} T \! = \!\sum_{j=0}^N (-1)^j \tbinom{n+j-1}{j} \nabla^j(T) (D^2 - \lambda)^{-(j+n)} + (-1)^{N+1}R_N(\lambda,n), \\
&\!R_N(\lambda,n)  \vc  \sum_{k=1}^n \tbinom{k+N-1}{N} (D^2-\lambda)^{k-n-1} \nabla^{N+1}(T) (D^2 - \lambda)^{-(k+N)} \in \OP^{\,r-2n-N-1}.\notag
\end{align}
\end{lemma}

\begin{proof}
The proof follows by pure combinatorics with an induction on $n$ and $N$.

\textit{1)} Let us first take $n=1$ and $N=0$. We have
\begin{align}\label{resolvent_trick}
(D^2 - \lambda)^{-1} T & = T (D^2 - \lambda)^{-1} + [(D^2 - \lambda)^{-1},T]  \notag \\
& = T (D^2 - \lambda)^{-1} - (D^2 - \lambda)^{-1} \nabla(T) (D^2 - \lambda)^{-1},
\end{align}
which is precisely the Formula \eqref{OP_exp_int} for $n=1$, $N = 0$.

\textit{2)} Let us now assume that Formula  \eqref{OP_exp_int} holds for $N=0$ and a fixed $n \in \N^*$, i.e.
\begin{align}\label{OP_n}
(D^2 - \lambda)^{-n} T & = T (D^2 - \lambda)^{-n} - \sum_{k=1}^n (D^2-\lambda)^{k-n-1} \nabla(T) (D^2 - \lambda)^{-k}.
\end{align}
We show that Equation \eqref{OP_n} holds also for $n+1$: Applying Equation \eqref{resolvent_trick} again
\begin{align*}
(D^2 - \lambda)^{-n-1} T & = (D^2 - \lambda)^{-1} \big[ T (D^2 - \lambda)^{-n} - \sum_{k=1}^n (D^2-\lambda)^{k-n-1} \nabla(T) (D^2 - \lambda)^{-k} \big] \\ 
& =  T (D^2 - \lambda)^{-n-1} - (D^2 - \lambda)^{-1} \nabla(T) (D^2 - \lambda)^{-n-1}  \\
& \hspace*{2.63cm}- \sum_{k=1}^n (D^2-\lambda)^{k-n-2} \nabla(T) (D^2 - \lambda)^{-k}.
\end{align*}

\textit{3)} If \eqref{OP_exp_int} holds for any $n \in \N^*$ and a fixed $N \in \N$, it is sufficient to show that
\begin{align*}
R_N(\lambda,n) = \tbinom{n+N}{N+1} \nabla^{N+1}(T) (D^2 - \lambda)^{-(N+n+1)} - R_{N+1}(\lambda,n).
\end{align*}
Indeed, by assumption we have
\begin{align*}
R_N(\lambda,n) &= \sum_{k=1}^n \tbinom{k+N-1}{N} (D^2-\lambda)^{k-n-1} \nabla^{N+1}(T) (D^2 - \lambda)^{-(k+N)} \\
& =  \sum_{k=1}^n \tbinom{k+N-1}{N} \nabla^{N+1}(T) (D^2 - \lambda)^{-(n+N+1)}  \\
& \qquad - \sum_{k=1}^n \tbinom{k+N-1}{N}  \sum_{j=1}^{n+1-k} (D^2-\lambda)^{j+k-n-2} \nabla^{N+2}(T) (D^2 - \lambda)^{-(j+k+N)} \\
& =  \tbinom{n+N}{N+1} \nabla^{N+1}(T) (D^2 - \lambda)^{-(N+n+1)} \\
& \qquad - \sum_{k=1}^n \sum_{j=1}^{k} \tbinom{j+N-1}{N}  (D^2-\lambda)^{k-n-1} \nabla^{N+2}(T) (D^2 - \lambda)^{-(k+N+1)},
\end{align*}
which is the claimed equation, since $\sum_{j=1}^{\ell} \tbinom{j+N-1}{N} = \tbinom{\ell+N}{N+1}$ and in the third equality we commuted the sums and changed the summation index $j+k-1 \rightsquigarrow k$.

Lastly, the use  of Proposition \ref{prop:OP} \textit{iii)}, \textit{vi)} shows that $R_N(\lambda,n) \in \OP^{\,r-2n-N-1}$.
\hfill $\Box$
\end{proof}

\begin{theorem}
\label{thm:expansion of sigma2z}
Let $T \in \OP^{\,r}$ for some $r \in \RR$. Then, for any $z \in \CC$ and any $N \in \N$,
\begin{align}
\label{OPexpansion}
\sigma_{2z}(T) = \sum_{\ell = 0}^{N} \tbinom{z}{\ell}\, \nabla^{\ell}(T)\, \vert D\vert^{-2\ell} + R_N(z), \quad R_N(z) \in \OP^{\,r-(N+1)}.
\end{align}
\end{theorem}
Here $\binom{z}{n}\vc  z(z-1)\cdots(z-n+1)(n!)^{-1}$ with the convention $\binom{z}{0}=1$,\index{_Azbinom@$\binom{z}{n}$}
\begin{proof} 
For $z \in \CC$ we decompose $z = n + s$, with $n = \lfloor \Re(z) \rfloor \in \Z$ and $\Re(s) <1$ yielding $\vert D \vert^{-2z} \,T =\vert D \vert^{-2s} \,\vert D \vert^{-2n}\,T$.

Let us first assume that $\Re(z) > 0$ and $\Re(s) > 0$. Hence, $n \in \N$. For any $N \in \N$ Formula \eqref{OP_exp_int} with $\lambda = 0$ implies 
\begin{align*}
\vert D \vert^{-2z}\, T = \sum_{j=0}^N (-1)^j \tbinom{n+j-1}{j} \vert D \vert^{-2s} \nabla^j(T) \vert D \vert^{-2(j+n)} +(-1)^{N+1} \vert D \vert^{-2s} R_N(0,n).
\end{align*}
Now, for $0< \Re(s)<1$, we invoke the operator identity
$$\vert D \vert^{-2s} = \tfrac{\sin( \pi s)}{\pi} \int_0^{\infty} (D^2+ \lambda)^{-1}\, \lambda^{-s}\, d\lambda.$$
For any $j \in \N$, Formula \eqref{OP_exp_int} with $-\lambda \in -\RR^+ \not\subset \spec D^2$, $n = -1, N-j \in \N$ yields
\begin{align*}
\vert D \vert^{-2s}\, \nabla^j(T)  & = \tfrac{\sin( \pi s)}{\pi} \int_0^{\infty} (D^2+ \lambda)^{-1} \nabla^j(T)\, \lambda^{-s} \,d\lambda \\
 & = \tfrac{\sin( \pi s)}{\pi} \sum_{k = 0}^{N-j} (-1)^k \nabla^{j+k}(T)  \int_0^{\infty} (D^2+ \lambda)^{-k-1} \lambda^{-s} \,d\lambda + R'_{N,j}(s)\\
 & = \tfrac{\sin( \pi s)}{\pi} \sum_{k = 0}^{N-j} (-1)^k \, \tfrac{\Gamma(1-s) \Gamma(s+k)}{\Gamma(k+1)} \, \nabla^{j+k}(T) \, \vert D \vert^{-2(s+k)} + R'_{N,j}(s),
\end{align*}
with $$R'_{N,j}(s) = (-1)^{N-j+1} \tfrac{\sin(\pi s)}{\pi} \int_0^{\infty} (D^2+\lambda)^{-1} \nabla^{N+1}(T) (D^2 +\lambda)^{-N+j-1} \lambda^{-s} d\lambda.$$

Combining the formulae from above we obtain
\begin{align*}
\!\vert D \vert^{-2z} T  & \\
& \hspace*{-0.6cm} = \sum_{j=0}^N \,\sum_{k=0}^{N-j} (-1)^{j+k} \tbinom{n+j-1}{j} \tfrac{\sin( \pi s)}{\pi} \, \tfrac{\Gamma(1-s) \Gamma(s+k)}{\Gamma(k+1)} \nabla^{j+k}(T) \, \vert D \vert^{-2(s+k+j+n)} + \R(z) \\
& \hspace*{-0.6cm} = \sum_{j=0}^N \,\sum_{k=0}^{N-j} (-1)^{j+k} \tbinom{n+j-1}{j} \tfrac{\Gamma(s+k)}{\Gamma(k+1)\Gamma(s)} \nabla^{j+k}(T)  \vert D \vert^{-2(k+j+z)} + \R(z) \\
& \hspace*{-0.6cm} = \sum_{\ell=0}^N \,\sum_{k=0}^{\ell} (-1)^{\ell} \tbinom{n+\ell-k-1}{\ell-k} \tbinom{k+s-1}{k} \nabla^{\ell}(T)  \vert D \vert^{-2(\ell+z)} + \R(z) \\
& \hspace*{-0.6cm} = \sum_{\ell=0}^N (-1)^{\ell} \tbinom{z+\ell-1}{\ell} \nabla^{\ell}(T)  \vert D \vert^{-2(\ell+z)} + \R(z)  = \!\sum_{\ell=0}^N \tbinom{-z}{\ell} \nabla^{\ell}(T)  \vert D \vert^{-2(\ell+z)} + \R(z),
\end{align*}
where along the road we employed Euler's reflection formula \cite[(6.1.17)]{Abram} and the Chu--Vandermonde identity \cite[(24.1.1)]{Abram}, and the global remainder reads
\begin{align}\label{Rz}
\R(z)\vc \sum_{j=0}^N (-1)^j \tbinom{n+j-1}{j} R'_{N,j}(s)  \vert D \vert^{-2(j+n)}+(-1)^{N+1} \vert D \vert^{-2s} R_{N}(0,n).
\end{align}

We have shown that $\sigma_{-2z}(T) = \sum_{\ell = 0}^{N} \binom{-z}{\ell}\, \nabla^\ell(T) \vert D\vert^{-2\ell} + \R(z) \awD^{2z}$,  under the assumption that $\Re(z) > 0$ and $\Re(s) > 0$, but the result can be extended to any $z\in \CC$ on the strength of the uniqueness of the holomorphic continuation. Thus, a  swap $z \rightsquigarrow -z$ would complete the proof, provided we can control the remainder $\R(z)$.

The second term of \eqref{Rz} is easy to handle: From Lemma \ref{lm:OP_exp} and Proposition~\ref{prop:OP} \textit{iii)} we deduce $\vert D \vert^{-2s} R_{N}(0,n) \in \OP^{\,r-2\Re(n+s)-N-1}$.

The first term of $\R(z)$,
\begin{align*}
&\R_1(z) \vc(-1)^{N+1} \tfrac{\sin( \pi s)}{\pi} \sum_{j=0}^N \tbinom{n+j-1}{j} \R_{1,j}(z)\\
&\text{with }\,\,\,\,\R_{1,j}(z) \vc \int_0^{\infty} (D^2+\lambda)^{-1} \nabla^{N+1}(T) (D^2+\lambda)^{-N-1+j} \vert D \vert^{-2(n+j)} \lambda^{-s} d\lambda,
\end{align*}
is more delicate. We will handle it with the aid of Theorem \ref{thm:integral of OP}.

For any fixed $j \in \{0,1,\dotsc,N\}$ and with $x=\Re(s)$, we show that the operator  $\abs{D}^{2x-2\epsilon} \R_{1,j}(z) \abs{D}^{N+1-r+2n}$ is bounded for some $\epsilon>0$: 
\begin{align*}
& \big\Vert \awD^{2x-2\epsilon} \int_0^{\infty} (D^2+\lambda)^{-1} \nabla^{N+1}(T) (D^2+\lambda)^{-N-1+j} \awD^{N+1-r-2j} \lambda^{-s} d\lambda \big\Vert \\
& \hspace*{1cm} \leq c_1 + \int_1^{\infty} \Vert (D^2+\lambda)^{-1} \Vert^{1-x+\epsilon} \, \big\Vert (D^2+ \lambda)^{-1} D^2 \big\Vert^{x - \epsilon} \x \\
& \hspace*{3cm} \x \, \Vert \nabla^{N+1}(T) \awD^{-N-1-r} \Vert \, \Vert (D^2+ \lambda)^{-1} D^2 \Vert^{N+1-j} \,\abs{\lambda^{-s}} d\lambda \\
& \hspace*{1cm} \leq c_1 + c_2 \int_1^{\infty} \lambda^{-1+x-\epsilon}\, \lambda^{-x} d\lambda < \infty.
\end{align*}

This computation is sound if $1-x+\epsilon>0$ and $x - \epsilon>0$. The former is always true, since $x = \Re(s) < 1$, whereas in the second one $x > 0$ and we can always tune the $\epsilon$ to be small enough. Note also that $\nabla^{N+1}(T) \awD^{-N-1-r} \in \OP^0$ for any $N$ and the convergence of the integral of norms at $0$ is automatic since $0<(D^2+\la)^{-1}D^2 \leq \bbbone$.

As we kept $j$ free and $\R_1(z)$ involves only a finite sum over $j$, we have actually proven that $\big\Vert \awD^{2x-2\epsilon} \R_1(z) \awD^{N+1-r+2n} \big\Vert < \infty$ for an arbitrarily small $\epsilon >0$. \\ To show that $\R_1(z) \in \OP^{\,r-2 \Re(z)-(N+1)+2\epsilon}$, so that $R_N(-z)\in \OP^{r-(N+1) +\epsilon}$, we invoke Theorem \ref{thm:integral of OP} with
\begin{align*}
&E_1(\lambda)  = \awD^{2x-2\epsilon} (D^2 + \lambda)^{-1},\\ 
&E_2(\lambda)  = (D^2 + \lambda)^{-N-1+j} D^{2(N+1-j)},  \\
& F_1  = \nabla^{N+1}(T) \awD^{-N-1-r}, \\
& G(\lambda)  = \lambda^{-s}. 
\end{align*}

Finally, to get rid of the $\epsilon$ we invoke the handy Lemma \ref{lm:polynomial rest}.
 \hfill $\Box$
\end{proof}
This theorem was first proven in \cite{ConnesSPV,ConnesMoscovici} for $T$ being a pseudodifferential operator (see Definition \ref{def:PDO}). But it only uses the fact that $T \in \OP^{\,r}$ for $r \in \RR$. This formulation was first given by Nigel Higson \cite{Higson} with another proof (see also \cite{CPRS}).

\section{Abstract Pseudodifferential Calculus}
\label{sec:pdo}
In the previous section we were only concerned with an abstract unbounded operator $\DD$ on a Hilbert space $\H$. We now let the algebra $\A$ enter into the game and describe the abstract pseudodifferential calculus associated with a spectral triple $\ahd$.

We shall follow the conventions adopted in \cite{Tadpole} (see also \cite{PodlesSA,IochumNotes}). Define first
\begin{center}
$\Pc(\A)\vc$ \index{_zPz1@$\Pc(\A)$} polynomial algebra generated by $\A,\,\DD,\,|\DD|$ and $J \A J^{-1}$ when $J$ exists.
\end{center}
\begin{definition}
\label{def:PDO}
Given a regular spectral triple $\ahd$, one defines the set of \emph{pseudodifferential operators} (pdos) \index{pseudodifferential operator (pdo)} as
\begin{multline}
\label{PDO}
\PDO \vc \big\{T \in \L(\H) \; \big\vert \; \forall \, N \in \N \;\; \exists \, P \in \Pc(\A),\, \,R \in \OP^{-N}, \,p \in \N, \\
\text{such that } T = P \vert D\vert^{-p} + R \big\}.\index{_apDO@$\PDO$}
\end{multline}
\end{definition}
In particular $\abs{D}^r\in \Psi(\A)$ for any $r\in \RR$.

An operator $T$ is \emph{smoothing} \index{pseudodifferential operator (pdo)!smoothing} if $T \in \OP^{-N}$ for all $N \in \N$. Remark that any smoothing operator is automatically in $\Psi(\A)$ with $P=0$ for each $N$ in Formula \eqref{PDO}.

\begin{example}
\label{example:pdo}
Here are few other examples of smoothing abstract pdos:

\textit{i)} The operator $P_0$ defined in \eqref{absD} is smoothing thanks to property \eqref{eq:P0 smoothing}. For a generalisation see \eqref{eq:PA smoothing}.

\textit{ii)} If $f$ is a Schwartz function, then $f(\abs{D})$ is a smoothing pdo: We claim that for any $N\in \N$, $\abs{D}^N\,f(\abs{D})\in \cap_{n\in\N}\, \Dom\,\del^n$, i.e. $f(\abs{D})\in \OP^{-N}$ for all $N \in\N$. Indeed, the function $x\in \RR^+ \to x^N \,f(x)$ is bounded, whereas $\delta^n(f(\abs{D})) = 0$ for $n \geq 1$.

\textit{iii)} Let $P$ be a positive invertible operator in $\OP^r$  with $r>0$ and such that $P^{-1} \awD^r \in \B(\H)$. Then, the operator $A=e^{-t\,P}$ is a smoothing pdo for all $t>0$.

To prove this we use Formula \eqref{heat as integral} to write $$A=\tfrac{i}{2\pi}\int_\C e^{-t\la}\,(P-\la)^{-1}\,d\la.$$
We first observe that $\Vert \abs{D}^N A\Vert\leq\Vert \abs{D}^N P^{-N/r}\Vert \,\Vert P^{N/r} A\Vert<\infty$ for any $N\in \N$ since $ \abs{D}^N P^{-N/r} \in \OP^0$ by Proposition \ref{prop:OP} \textit{iii)}, \textit{viii)} and the function $x\in \RR^+\to x^s\,e^{-tx}$ is bounded for any $s\geq 0,\,t>0$. \\
Secondly, we need to show that $A$ is in $\Dom \delta^n$ for any $n \in \N$. To this end, we observe that $$\delta(A)=\tfrac{-i}{2\pi}\int_\C e^{-t\la}\,(P-\la)^{-1}\delta(P)(P-\la)^{-1}\,d\la$$ converges in norm because $$\Vert\del(P)(P-\la)^{-1}\Vert\leq \Vert\del(P)P^{-1} \Vert \cdot \Vert P(P-\la)^{-1} \Vert$$ and moreover $\Vert P(P-\la)^{-1} \Vert \leq \abs{\la}\Im (\la)^{-1}$ for $\Re(\la)\geq 0$ (or $\leq 1$ if $\Re(\la)<0$) cf. \eqref{eq:estimate resolvent for positive}, so that we can choose the curve $\C$ such that $\Vert P(P-\la)^{-1} \Vert$ is uniformly bounded for $\la\in \C$. Thus the norm of $\delta(A)$ is estimated in the same way as the norm of $A$. The argument extends to any $\delta^n(A)$ showing that $e^{-t P}\in \OP^0$ for any $t>0$.

To see that $A$ is smoothing amounts to showing that the operator $\abs{D}^N A$ is in $\OP^0$ for any $N\in \N$ but we have already checked that it is bounded.

 It is in the domain of $\delta$ because $$\del(\abs{D}^N A)=\abs{D}^N \del(A)=[\abs{D}^N P^{-N/r}] \,[P^{N/r} \del(A)]$$ and both terms are bounded: In particular for any $s\geq 0$
\begin{align*}
\norm{P^s \del(A)} \leq \Vert P^s \abs{D}P^{-s-1/r}\Vert \Vert P^{s+1/r}A\Vert +\Vert P^sAP^{1/r}\Vert \Vert P^{-1/r}\abs{D}\Vert  <\infty.
\end{align*}
Similar arguments show that $\abs{D}^N A\in \underset{n\in\N}{\mathlarger{\mathlarger{\cap}}}\Dom \del^n$.
\hfill$\blacksquare$
\end{example}

For $T,T'\in\Psi(\A)$, we define the equivalence $T \sim T' \,\text{ if }\, T-T' \,\text{ is smoothing}.\index{_ASim@$\sim$}$

\begin{lemma}
The set $\PDO$ is a $\Z$-graded involutive algebra with 
\begin{align}
\label{PDO_grading1}
\PDOk{k} \vc \PDO \cap\, \OP^k, \, \text{ for } k \in \Z. \index{_apDOk@$\PDOk{k}$} \index{pseudodifferential operator (pdo)!of order $k$}
\end{align}
Moreover, $\delta(\Psi^k(\A))\subset \Psi^k(\A)$ and $\Pc(\A) \subset \Psi(\A) \subset \cup_{k\in \Z}\,\,\OP^k$.
\end{lemma}

\begin{proof}
Only the stability under the product deserves a bit of attention. Let us take the pdos $T\in \Psi^k(\A)$ and $T' \in \Psi^{k'}(\A)$ with $k,k'\in \Z$. For any $N,N'\in \N$ we can write $T=P\abs{D}^{-p} +R$ and $T'=P'\abs{D}^{-p'}+R'$, where $p,p'\in \N$, $P,P'\in \Pc(\A)$ and $P\abs{D}^{-p}\in \OP^k,\, P'\abs{D}^{-p'}\in \OP^{k'},\, R\in \OP^{-N}, R'\in \OP^{-N'}$.\\ We claim that, for any $N''\in \N$, there exist $p''\in \N,P''\in \Pc(\A)$ and $R''\in \OP^{-N''}$ such that we have $TT'=P'' \abs{D}^{-p''} +R''$. \\Using \eqref{OPexpansion} we get, for any $M\in \N$,
\begin{align*}
TT'\!&=P\sum_{n=0}^M \! \!\tbinom{-p/2}{n} \nabla^n(P')\abs{D}^{-p-p'-2n} \!+PR_M\!\abs{D}^{-p-p'}\!\!+RP'\!\abs{D}^{-p'}\!+P\!\abs{D}^{-p}\!R' \!+RR'\\
&=\big[P\sum_{n=0}^M \! \tbinom{-p/2}{n} \nabla^n(P')\abs{D}^{2M-2n}\big] \abs{D}^{-p-p'-2M}+R'' = P'' \abs{D}^{-p-p'-2M}+R''.
\end{align*}
We have $P'' \in \Pc(\A)$ and we set $p''=p+p'+2M\in \N$. Let us focus on the remainder $$R''=PR_M\abs{D}^{-p-p'}\!+RP'\abs{D}^{-p'}\!+P\abs{D}^{-p}R' \!+RR',$$ with $R_M\in \OP^{p'-M-1}$. 
Since $N,N',M$ are arbitrary, we can choose $N=N''+\abs{k'}$, $N'=N''+\abs{k}$ and $M=\max\{k+N''-1,0\}$. Thus $PR_M\abs{D}^{-p'} \in \OP^{-N''}$, because we have $k-M-1\leq-N''$. \\
Similarly, the orders of the other terms in $R''$ are successively $-(N'' +\abs{k'}-k')$, $-(N'' +\abs{k}-k), -(2N'' +\abs{k}+\abs{k'})$, thus less than $-N''$ and hence $R''\in \OP^{-N''}$ proving the claim. \\
It is immediate that $TT'\in \Psi^{k+k'}(\A)$.
\hfill $\Box$
\end{proof}

Under the assumption of regularity of $\ahd$, the algebra $\PDO$ can be seen as the set of all operators with asymptotics of the form $\sum{ }_{n\in\N}\,\, P_n\,\vert D\vert^{d-n}$ with $P_n \in \Pc(\A)$, i.e. \\
\centerline{if $T\in \PDO$ we have  $T - \sum{ }_{n=0}^N\,\, P_n\,\vert D\vert^{d-n} \in \OP^{-N}$ for all $N \in \N$.}

If $\Pc'(\A)$\index{_zPz1p@$\Pc'(\A)$} is the polynomial algebra generated by $\A,\,\DD$ and $J \A J^{-1}$, when $J$ exists, then $\Pc'(\A) \subset \Pc(\A)$ and one can construct a subalgebra $\Psi'(\A)$\index{_apDOp@$\Psi'(\A)$} of $\Psi(\A)$ by taking $P\in \Pc'(\A)$ in \eqref{PDO}. In particular, $\vert\DD \vert^k$ is not necessarily in $\Pc'(\A)$ for $k$ odd. On the other hand, since $\abs{\DD}=\abs{D}-P_0 \in \Pc(\A)$, we could have defined $\Pc(\A)$ with $\abs{D}$ in place of $\abD$ and we would arrive at the same definition of $\PDO$.

Let us remark that a key point about pdos on a manifold is that they are integral operators (see Appendix \ref{classical tools}). In the abstract framework this notion is missing even if it is reminiscent in a few specific cases, like the noncommutative torus \cite{ConnesTretkoff}.

\begin{example}
\label{example:inverse}
To show that the inverse of a classical pseudodifferential operator on a manifold is also a pdo is not immediate. Let us have a look at an abstract example: Let $D,\,X \in \Psi^1(\A)$, so $D^2+X \in \Psi(\A)$ and assume that $D$ and $D^2+X$ are invertible. 
We claim that $(D^2+X)^{-1} \in \Psi(\A)$: For any $N \in \N$, we have the expansion
\begin{align}
\label{eq:expansion of D2+tXV}
(D^2+X)^{-1}=\sum_{k=0}^N (-\vert D\vert^{-2}X)^k D^{-2}+ (-\vert D\vert^{-2}X)^{N+1}(D^2+X)^{-1}.
\end{align}
The first term is in $\Psi(\A)$, because $\vert D\vert^{-2}, X \in \Psi(\A)$ and the remainder is in $\OP^{-N-3}\subset \OP^{-N}$, provided $(D^2+X)^{-1} \in \OP^{-2}$. But since $D^2+X\in \OP^2$ and 
$(D^2+X)^{-1} D^2=(\bbbone+\awD^{-2}X)^{-1}$ is bounded, because $\awD^{-2}X\in \OP^{-1}$ is compact by Proposition \ref{prop:OP} $ix)$, the part $viii)$ of the latter yields $(D^2+X)^{-1} \in \OP^{-2}$.
\hfill$\blacksquare$
\end{example}

It is desirable to have an extension of the pseudodifferential calculus, which takes into account the complex powers of $\abs{D}$ (cf. \cite[p. 239]{ConnesMoscovici}):
\begin{align}
\label{def:psic}
\Psi^\CC(\A) \vc \{ T\abs{D}^z\,\vert\, T \in \Psi(\A),\,z\in \CC\}.\index{_apDOz@$\Psi^\CC(\A)$}
\end{align}
By construction, $T \vert D\vert^z \in \OP^{k+\Re(z)}$ when $T\in \Psi^k(\A)$. Since Formula \eqref{OPexpansion} is true for any $z \in \CC$, the set $\Psi^\CC(\A)$ is in fact an $\RR$-graded algebra.

The link between the algebraic and geometric definitions of pseudodifferential operators is provided by the following fact:

\begin{proposition}
\label{prop:PDO}
For a commutative spectral triple $\big( \Coo(M), L^2(M,\SS\ox E),\Dslash_E \big)$, with $\Dslash_E = -i\gamma^{\mu} \nabla_{\mu}^{\SS \ox E}$ \index{_zD1e@$\Dslash_E$} (cf. Example \ref{ex:commutative}), we have a natural inclusion
\begin{align*}
\Psi(\Coo(M)) \subset \Psi^{\Z}(M,\SS\ox E),
\end{align*}
where $\Psi^{\Z}(M,\SS\ox E)$ is the space of classical pseudodifferential operators \index{pseudodifferential operator (pdo)!classical} of integer order defined over a vector bundle $\SS \ox E$.
\end{proposition}
This is a consequence of the observation that $\Dslash_E, \abs{\Dslash_E} \in \Psi^1(M,\SS\ox E)$ and also $a \in \Psi^0(M,\SS\ox E)$ for any $a \in \Coo(M)$. \\
For the full story on classical pseudodifferential operators see the references in Appendix \ref{sec:Def of pdo} and for a noncommutative vision see \cite{Elements} and \cite[Section 5.1]{VarillyDirac}.

\section{Dimension Spectrum}
\label{sec:dimsp}

In Section \ref{sec:axioms} we discussed the $p$-summability of a spectral triple, which encodes the dimension of the underlying manifold when $\ahd$ is commutative. This notion does not, however, capture the whole richness of noncommutative geometry. A more adequate description turns out to be provided by a, possibly infinite, discrete subset of complex numbers.

\begin{definition}
\label{def:dimspec}
A regular spectral triple $\ahd$ of dimension $p$ has a \emph{dimension spectrum} \index{dimension spectrum} $\Sd$\index{_zSzd@$\Sd$}  if $\Sd \subset \CC$ is discrete and for any $T \in \PDOk{0}$, the function
\begin{align}
\label{zetaB}
\zPD(s) \vc \Tr\, T \vert D\vert^{-s}  \index{_a3zetaTD@$\zPD(s)$}, \quad \text{ defined for } \, \Re(s) > p,
\end{align}
extends meromorphically to $\CC$ with poles located in $\Sd$. \footnote{One could in principle also allow for essential singularities of $\zPD$, as long as they are isolated.}

We say that the dimension spectrum is of \emph{order $k \in \N^*$} \index{dimension spectrum!of order $k$} if all of the poles of functions $\zPD$ are of the order at most $k$ and \emph{simple} \index{dimension spectrum!simple} when $k = 1$.
\end{definition}

Note that all functions $\zTD$ are well defined for $\Re(s) > p$, since for such $s$ the operator $\abs{D}^{-s}$ is trace-class and, as $T \in \PDOk{0} \subset \B(\H)$, so is the operator $T \abs{D}^{-s}$. By a standard abuse of notation we shall denote the (maximal) meromorphic extension of $\zeta_{T,D}$ with the same symbol. 

Observe also that if $\ahd$ has a dimension spectrum, then $\zTD$ is actually meromorphic on $\CC$ for any $T \in \PDOC$. Indeed, if $T \in \PDOC$ then there exist ${S\in\Psi^0(\A)}$ and $z\in \CC$ such that $T=S\abs{D}^z$. Thus, for $\Re(s)> p+\Re(z)$, the function $\zeta_{T,D}$ is defined by $$\zeta_{T,D}(s)\vc \Tr T\abs{D}^{-s}=\Tr S \abs{D}^{-s+z}=\zeta_{S,D}(s-z).$$
By the uniqueness of the meromorphic extension $\zeta_{T,D}$ is meromorphic on $\CC$.

If, however, $\ahd$ is not finitely summable, then the zeta function $\zeta_{\bbbone,D}$ is nowhere defined and the notion of a dimension spectrum does not make sense.

Let us stress that the algebra does play an important role in Definition \ref{def:dimspec}. Generally, given the meromorphic extension of the function $\zeta_{\bbbone,D}$, it is not clear how to get one for $\zTD$, even when $T\in \A$. Although in the case of a commutative spectral triple the poles of $\zeta_{T,D}$ do coincide with the poles of $\zeta_{\bbbone,D}$ for $T\in \A$, it is no longer to be expected for general noncommutative geometries. 

\begin{example}\label{ex:dim_sp}
Let us consider the spectral triple $(\A_q,\H_q,\DD^S_q)$ of the standard Podle\'s sphere (cf. Appendix \ref{sec:Podles}). The function $\zeta_{\bbbone,D^S_q} = 4 \tfrac{(1-q^2)^s}{\abs{w}^s (1-q^s)^{2}}$ (where $D^S_q=\DD^S_q$ since the kernel is trivial) is regular at $s = -2$. On the other hand, using the explicit formulae for the representation (\ref{rep}--\ref{Defpi}), one can check that, for instance,
\begin{align*}
\Rez{s=-2} \, (s+2) \zeta_{A,D^S_q}(s) = \tfrac{2 q (1+q^2) \abs{w}^2}{(\log q)^2}.
\tag*{$\blacksquare$}
\end{align*}
\end{example}

The definition presented above is the one of \cite{Tadpole} (see also \cite{IochumNotes}) and it differs from the original one presented in \cite{ConnesSPV,ConnesMoscovici}. There, $\Sd$ was defined \cite[Definition II.1]{ConnesMoscovici} as the set of poles of functions $\zPD$, but with $T \in \B_0$ -- the algebra generated by $\delta'^n(a)$ with $a \in \A$ and $n \in \N$. This was tailored to prove the local index theorem in noncommutative geometry \cite{ConnesMoscovici}. On the other hand, in the context of spectral action \cite{ConnesMarcolli}, the operator $T$ from the definition was taken from a bigger algebra $\B_1$ generated by $\delta'^n(a)$ and $\delta'^n([\DD,a])$ \cite[Definition 1.133]{ConnesMarcolli}.

The difference between the definition of the dimension spectrum adopted here and the one of \cite{ConnesMarcolli} (see also \cite{CPRS,RennieSummable}) is that $\vert D\vert^{-1} \in \PDO \cap \OP^0$, but $\vert D\vert^{-1} \notin \B_1$. \\This fact led to eventual variations concerning the dimension spectrum of a commutative spectral triple (compare for instance \cite[Example 13.8]{ConnesGarden} with \cite[p. 22]{ConnesCIME}). 

All these definitions of $\Sd$ satisfy $\Sd$(disjoint sum of spaces) = $\cup$ $\Sd$(spaces).

There is a natural way to define the tensor product of two even spectral triples (see \cite{ConnesReality,FarnsworthProduct}) using $$\A\vc \A_1 \otimes \A_2,\,\,\H\vc \H_1 \otimes \H_2 \,\text{ and }\, \DD\vc \DD_1\otimes \bbbone + \gamma_1\otimes \DD_2.$$ It would be tempting to conclude that $\Sd\ahd = \Sd_1 + \Sd_2$, but it is not obvious how to obtain a meromorphic extension of $\zeta_{T_1 \ox T_2,D}$ given $\zeta_{T_1,D_1}$ and $\zeta_{T_2,D_2}$. 

With Definition \ref{def:dimspec} the following result holds \cite[Proposition A.2]{Tadpole}:
\begin{example}
\label{ex:dimsp_manifold}
Let $\ahd$ be the commutative spectral triple associated with a $d$-dimensional Riemannian manifold $M$ and $\DD$ is a first order differential operator, then
$\Sd\ahd = d - \N$ and it is simple.
\hfill$\blacksquare$
\end{example}

An interesting result of Jean-Marie Lescure \cite{Lescure} shows that for spectral triples describing manifolds with conical singularities there appear poles of second order in the dimension spectrum. On the other hand, the dimension spectra of fractal spaces studied via noncommutative geometry \cite{Christensen2,Christensen1,Cipriani,Guido1,Guido2,Guido3} encode the Hausdorff and spectral dimensions of the fractal along with its self-similarity structure. The latter is signalled by the appearance of complex numbers outside the real axis in $\Sd$ \cite{Kellendonk}.\label{p:fractals_dimsp}

We see that the dimension spectrum carries much more information about the underlying geometry than a single number, for instance $p$-summability. However, there is no systematic procedure to compute the dimension spectrum for general spectral triples. Beyond the almost commutative geometry it is not even clear under what conditions $\ahd$ has a dimension spectrum. In concrete examples, one has to dutifully prove the existence of the meromorphic extensions of the whole family of spectral functions and identify the poles. This was done only for few specific spectral triples like the noncommutative torus \cite{TorusSA}, quantum group  $SU_q(2)$ \cite{ConnesSU2,DiracSUq2} and quantum Podle\'s spheres \cite{EquatorialDimSp,AllPodles,PodlesSA,PalSundar}. \label{dim_sp examples} See also Problem \ref{prob:dim_sp} in Chapter~\ref{chap:open}.

Let us note that in the definition of dimension spectrum, both the original one of \cite{ConnesMoscovici} and the more recent one of \cite{Tadpole}, there is an (usually unspoken) assumption about the regularity of the spectral triple. Indeed, if one does not have control on the order of pseudodifferential operators, the definition of the dimension spectrum might turn out to be inconsistent with the $p$-summability (c.f. \cite[Section 3.2]{PodlesSA}). Clearly, if a regular $p$-summable spectral triple has a dimension spectrum $\Sd$, then $p \in \Sd$. However, the $p$-summability together with regularity is not sufficient to conclude that there is a pole of $\zeta_{\bbbone,D}$ at $p \in \CC$ --- there might be an essential singularity, a branch cut or the boundary of analyticity, in which case $\ahd$ does not have a dimension spectrum at all. All these pathologies certainly deserve further studies.

\section{Noncommutative Integral}\label{sec:ncint}

In a seminal paper \cite{Wodzicki} Mariusz Wodzicki has shown that the algebra of classical pseudodifferential operators of integer order $\Psi^{\Z}(M,E)$ admits a single (up to normalisation) trace when $M$ is connected and $\dim M>1$. It has been called the Wodzicki residue, since it can be expressed as a residue of a certain spectral function. Alain Connes established in \cite{ConnesAction} a link between the Wodzicki residue of a compact pseudodifferential operator and its Dixmier trace (see Definition \ref{def:Dixmier}), which is an exotic trace on a certain ideal in the algebra $\B(\H)$. \\
In this section we recollect some basic facts about the Wodzicki residue and its use for defining an integral suitable for noncommutative spaces. 
We start with the following definition:
\begin{definition}
\label{def:ncint}
Let $\ahd$ be a regular $p$-dimensional spectral triple with a dimension spectrum. For any $T \in \Psi^\CC(\A)$ and any $k \in \Z$ define  \footnote{Using the notation of \cite{ConnesMoscovici} we have $\ncintd{k} = 2^{k-1} \tau_{k-1}$ for $k \geq 1$.}
\begin{align}
\label{ncint}
\ncintd{k} T \vc \Rez{s=0} \, s^{k-1} \zPD(s),  \index{noncommutative integral} \index{_Ancint1@$\displaystyle\ncint^{[k]} $} &&
\ncint T \vc \ncintd{1} T= \Rez{s=0} \, \zPD(s) \index{_Ancint@$\displaystyle\ncint $}.
\end{align}
\end{definition}

If the dimension spectrum of $\ahd$ is of order $d$, then for $s$ in an open neighbourhood of any $z \in \CC$ we have the Laurent expansion
\begin{align*}
\zeta_{T,D}(s) = \sum_{k=-d \,}^{\infty} \,\ncintd{-k} T \abs{D}^{-z} (s-z)^{k}.
\end{align*}

The adopted notation is useful in the following result (cf. \cite[Proposition II.1]{ConnesMoscovici}):

\begin{theorem}
\label{thm:ncint_trace}
Let $\ahd$ be a regular $p$-dimensional spectral triple with a dimension spectrum. Then, for any $T_1, T_2 \in \PDOC$ and any $k \in \Z$ we have
\begin{align*}
\ncintd{k} T_1 T_2 = \ncintd{k} T_2 T_1 + \sum_{n=1}^{N} \tfrac{1}{n!} \sum_{j=0}^n \tstirling{n}{j} \tfrac{(-1)^j}{2^j} \ncintd{k+j} T_2 \nabla^n(T_1) \awD^{-2n},
\end{align*}
with $N= p + \ord T_1 + \ord T_2$ and $\tstirling{n}{j}$ the unsigned Stirling numbers of the first kind\index{_Azbstirling@$\stirling{n}{j}$}.\footnote{For any $s \in \CC$, $n \in \N$, $\tbinom{s}{n} = \sum_{j=0}^{n} \tstirling{n}{j} \tfrac{s^j}{n!}$ with the convention $\tstirling{n}{0} = \delta_{n,0}$. See \cite[Sec. 24.1.3]{Abram}.} 
\\
In particular, if the dimension spectrum has order $d$, then $\ncintd{d}$ is a trace on $\PDOC$.
\end{theorem}

\begin{proof}
Since the triple is regular \eqref{OPexpansion} yields, for any $T\in \PDOC, s\in \CC, M\in \N$,
\begin{equation}
\label{OP_T}
\!\!\!\!\vert D\vert^{-s}T =\sum_{n = 0}^{M}\tbinom{-s/2}{n} \nabla^n(T) \vert D\vert^{-2n-s} + R_N(s) \abs{D}^{-s}, \quad R_M(s) \in \OP^{\ord T-M-1}\!.
\end{equation}
Now, let $r_1 = \vert \ord T_1 \vert, r_2 = \vert \ord T_2 \vert$. For $\Re(s) > p + 2 r_1 + r_2$ we can use the cyclicity property of the trace,
\begin{align*}
\zeta_{T_1 T_2,D}(s) & = \Tr\, T_1 T_2 \abs{D}^{-s} = \Tr \,( T_1 \abs{D}^{-r_1} )( \abs{D}^{r_1} T_2 \abs{D}^{-s} )  \\
& = \Tr \,( \abs{D}^{r_1} T_2 \abs{D}^{-s} T_1 ) \abs{D}^{-r_1} = \Tr\, T_2 \abs{D}^{-s} T_1,
\end{align*}
since the operators $\abs{D}^{r_1} T_2 \abs{D}^{-s}, \awD^{r_1} T_2 \abs{D}^{-s} T_1$ are trace-class, whereas $T_1 \abs{D}^{-r_1}$ and $\abs{D}^{-r_1}$ are bounded. Then, with the help of Formula \eqref{OP_T}, we obtain
\begin{align}
\zeta_{T_1 T_2,D}(s) & = \Tr \,T_2 T_1 \abs{D}^{-s} + \sum_{n = 1}^{N} \tbinom{-s/2}{n}\, \Tr \,T_2 \nabla^n(T_1) \vert D\vert^{-2n-s} + \Tr T_2 R_N(s) \abs{D}^{-s} \notag \\
& = \zeta_{T_2 T_1,D}(s) + \sum_{n = 1}^{N} \tbinom{-s/2}{n}\, \zeta_{T_2 \nabla^n(T_1),D}(2n+s) + h(s), \label{zeta_commut}
\end{align}
where $h$ is holomorphic for $\Re(s) > p + \ord T_1 + \ord T_2 - N - 1 = -1$.\\ For any $n \in \N$, $T_2 \nabla^n(T_1) \in \PDOC$ and, since the spectral triple has a dimension spectrum, the functions $\zeta_{T_2 \nabla^n(T_1),D}$ admit meromorphic extensions to the whole complex plane. Consequently, equality \eqref{zeta_commut} holds true for $\Re(s) > -1$ and, for any $k \in \Z$, we have
\begin{align}
\label{eq:commutation under ncintd}
\ncintd{k} T_1 T_2 & = \ncintd{k} T_2 T_1 + \sum_{n=1}^{N} \Rez{s=0}\,  s^{k-1}  \tbinom{-s/2}{n}\, \zeta_{T_2 \nabla^n(T_1),D} (s+2n) \\
& = \ncintd{k} T_2 T_1 + \sum_{n=1}^{N} \tfrac{1}{n!} \sum_{j=0}^n \tstirling{n}{j} \tfrac{(-1)^j}{2^j} \ncintd{k+j} T_2 \nabla^n(T_1) \awD^{-2n}.
\end{align}

If the dimension spectrum of $\ahd$ has order $d$, for $k=d$ we obtain $$\ncintd{d} T_1 T_2 = \ncintd{d} T_2 T_1,$$ because $\ncintd{d+j} T = 0$ for any $T \in \PDOC$ and $j \in \N^*$, whereas for $j=0$, we have $\tstirling{n}{j} = 0$.
\hfill $\Box$
\end{proof}

In the case of a commutative spectral triple (as in Proposition \ref{prop:PDO}), $\ncint$ defined by \eqref{ncint} is related to the \emph{Wodzicki residue} \index{Wodzicki residue} \cite{Wodzicki,Wodzicki2,Guillemin}. The latter is defined in the following context: Let $D \in \Psi(M,E)$ be an elliptic pseudodifferential operator of order 1. For $P \in \Psi^{\Z}(M,E)$, define
\begin{align*}
\WRes P \vc \Rez{s=0} \zeta_{P,D}(s). \index{_zwres@$\WRes$}
\end{align*}
Mariusz Wodzicki has shown that 
\begin{align}\label{WRes}
\WRes P =  \tfrac{1}{(2\pi)^d} \int_{\mathbb{S}^*M} \tr \sigma^P_{-d}(x,\xi) \; d\mu(x,\xi),
\end{align}
where $\mathbb{S}^*M \vc \{(x,\xi) \in T^*M \; \vert \; g_x^{-1}(\xi,\xi) = 1\}$ is the cosphere bundle over $M$, $\tr$ is the matrix trace of the symbol $\sigma^P_{-d}(x,\xi)$ and 
\begin{align*}
& d\mu(x,\xi) \vc \abs{\sigma(\xi)} \abs{dx}, \\
& dx \vc dx^1\wedge\cdots\wedge dx^d,\\
& \sigma(\xi) \vc \sum_{j=1}^n\, (-1)^{j-1} \,\xi_j\,d\xi_1\wedge\cdots \wedge \widehat{d\xi_j}\wedge\cdots\wedge d\xi_d\,.
\end{align*}
is a suitable Riemannian measure on $\mathbb{S}^*M$ (the integrand of \eqref{WRes} is a $1$-density on $M$, so it is coordinate-independent, see \cite[Section 7.3]{Elements} or \cite[Section 5.3]{Varilly}). Moreover, $\WRes$ is the unique (up to normalisation) trace on the algebra $\Psi^{\Z}(M,E)$. 

For a commutative spectral triple as in Proposition \ref{prop:PDO} with $D = \Dslash_E+P_0$, we simply have $\ncint P = \WRes P$ for any $P \in \Psi^{\Z}(M,\SS\ox E)$.

\begin{example}
\label{ex:ncint}
For $a \in \A = \Coo(M)$ of Example \ref{ex:commutative}, Formula \eqref{WRes} yields
\begin{align*}
\ncint a \abs{D}^{-d} = \tfrac{2^{\lfloor d/2 \rfloor} \Vol(S^{d-1})}{(2 \pi)^d} \int_M a(x) \,d\mu(x),
\end{align*}
where $d\mu(x) = \sqrt{g} \, dx$ is the standard Riemannian measure on $M$, since we have $\sigma^{a \abs{D}^{-d}}_{-d} \!(x,\xi) =  a(x) \norm{\xi}^{-d}$.\\
 It is worthy to store the following: In a commutative spectral triple, since $\zeta_{\bbbone,D}$ is regular at 0 (see \cite[page 108, eq. (1.12.16)]{Gilkey1}), we get
\begin{align}
\label{int 1=0}
&\hspace{4cm}  \ncint \bbbone= 0.\quad \hspace{4cm} \blacksquare
\end{align}
\end{example}

The noncommutative integral defined in \eqref{ncint} turns out, moreover, to be related to the Dixmier trace, the construction of which we briefly describe below.

Given a selfadjoint operator $T$ with purely discrete spectrum, we denote by $\lambda_n(T)$\index{_a5lambda_n(T)@$\lambda_n(T)$}, with $n\in \N$, its eigenvalues ordered increasingly. 
If $T \in \KKK(\H)$, then for $n \in \N$ we define $\mu_n(T)$ \index{_a7mu@$\mu_n(T)$} as the $(n+1)$-th singular value of $T$, i.e. $(n+1)$-th eigenvalue of $\abs{T}$ sorted in decreasing order, with the corresponding multiplicity $M_n(T)$.\index{_zMe@$M_n(T)$} Note that since $T$ is a compact operator, we have $\lim_{n \to \infty} \mu_n(T) = 0$.

Now, let us denote the partial trace $$\Tr_N(T) \vc \sum_{n=0}^{N} M_n(T) \,\mu_n(T).$$
If $T \in \L^1(\H)$, $\lim_{N \to \infty} \Tr_N(T) < \infty$, but for a general compact operator the sequence $\Tr_N$ is unbounded. The role of a Dixmier trace is to capture the coefficient of the logarithmic divergence of $\Tr_N(T)$.

There exists a very convenient formula \cite[Proposition 7.34]{Elements} for $\Tr_N(T)$ which uses a decomposition of $T$ into a trace-class and a compact part:
\begin{align}
\label{partial_trace_N}
\Tr_{N}(T) = \inf \big\{ \Tr \abs{R} + N \norm{S} \; \big\vert \; R \in \L^{1}(\H), \, S \in \KKK(\H) , \; T = R + S \big\}.  \index{_zTr_T@$\Tr_N(T)$}
\end{align}
The partial trace, $\Tr_N(T)$ may be seen as a trace of $\abs{T}$ cut-off at a scale $N$. But this scale does not need to be a natural number and indeed, Formula \eqref{partial_trace_N} still makes sense for any positive $N$, hence we define for any $\lambda \in \RR^+$:
\begin{align}
\label{partial_trace_l}
\Tr_{\lambda}(T) \vc \inf \big\{ \Tr \abs{R} + \lambda \norm{S} \; \big\vert \; R \in \L^{1}(\H), \, S \in \KKK(\H) , \; T = R + S \big\}.
\end{align}

With the help of Formula \eqref{partial_trace_l} we define the subspace of compact operators, the partial trace of which diverges logarithmically 
\begin{align}
\label{def:Mideal}
\L^{1,+}(\H) \eq \{ T \in \KKK(\H) \; \vert \; \norm{T}_{1,+} \eq \sup_{\lambda \geq e} \tfrac{\Tr_{\lambda}(T)}{\log \lambda} < \infty\}. \index{_zL1h1@$\L^{1,+}(\H)$} \index{_An@$\norm{\cdot}_{1,+}$}
\end{align}
One has $\norm{ATB}_{1,+} \leq \norm{A}\,\norm{T}_{1,+}\,\norm{B}$  for every $A,B \in \B(\H)$ and the space $\L^{1,+}(\H)$ (sometimes denoted by $\L^{1,\infty}(\H)$, see \cite{DixZeta4} for historical notes) is a $C^*$-ideal of $\B(\H)$ \cite[Lemma 2.5]{IochumNotes}.

Now, for $T \in \KKK(\H)$ and $\lambda \geq e$, define
\begin{align}\label{cesaro}
\tau_{\lambda}(T) \vc \tfrac{1}{\log \lambda} \int_{e}^{\lambda} \tfrac{\Tr_{u}(T)}{\log u} \, \tfrac{d u}{u}.
\end{align}

The functional $\tau_{\lambda}$ is not additive on $\L^{1,+}$, but the defect is controllable \cite[Lemma 7.14]{Elements}, i.e. for positive operators $T_1, T_2 \in \L^{1,+}$
\begin{align*}
\tau_{\lambda}(T_{1} + T_{2}) - \tau_{\lambda}(T_{1}) - \tau_{\lambda}(T_{2}) = \OO_{\la \to \infty} ( \tfrac{\,\log \log \lambda}{\log \lambda}).
\end{align*}

To obtain a useful additive functional on $\L^{1,+}$, two more steps are needed. First note that the map $\lambda \mapsto \tau_{\lambda}(T)$ is in $C_b([e,\infty))$ for $T\in \L^{1,+}$. Define also a quotient $C^*$-algebra $$\QQ \vc C_b([e,\infty)) / C_0([e,\infty))$$\index{_zQ@$\QQ$}and let $[\tau(T)] \in \QQ$ be the class of $\lambda \mapsto \tau_{\lambda}(T)$.  Then, $T \mapsto [\tau(T)]$ extends to a linear map from $\L^{1,+}$ to $\QQ$, which is a trace, i.e. $[\tau(T_1T_2)] -[\tau(T_2T_1)]= [0]$. To get a true linear functional on $\L^{1,+}$, we need to apply to $[\tau(\cdot)]$ a state $\omega$ (i.e. $\omega$ is a positive linear functional of norm 1) on $\QQ$.

\begin{definition}
\label{def:Dixmier}
A \emph{Dixmier trace} \index{Dixmier trace} \cite{Dixmier} associated with a state $\omega$ on $\QQ$ is defined as $\Tr_{\omega}(\cdot) \vc \omega \circ [\tau(\cdot)]$.
\end{definition}

In fact, a Dixmier trace is not only a trace, but a \emph{hypertrace} \index{hypertrace} on $\L^{1,+}(\H)$, i.e. $\Tr_{\omega}(TS) = \Tr_{\omega}(ST)$ for any $T \in \L^{1,+}$ and any $S \in \B(\H)$ \cite[p. 45]{Varilly} (compare also \cite[Theorem 10.20 and Corollary 10.21]{Elements}).

The definition of a Dixmier trace, although powerful, is not completely satisfactory as it involves an arbitrary state on the commutative  algebra $\QQ$, which is not separable (and $\L^{1,+}(\H)$ is also not separable). Thus, in practice, one cannot construct an explicit suitable state and there exist myriads of `singular traces', different than the one of Dixmier, exploiting various notions of generalised limits \cite{DixZeta4}.
 
Moreover, if one insists on defining the noncommutative integral via the Dixmier trace, one faces a notorious problem related to the existence of measurable (and non measurable!) sets in Lebesgue's theory. In particular, there exists a class of operators in $\L^{1,+}(\H)$, for which a Dixmier trace does not depend on $\omega$.

\begin{definition}
An operator $T \in \L^{1,+}(\H)$ is \emph{measurable} \index{measurable operator} if $\Tr_{\omega}(T)$ does not depend on $\omega$. Then one speaks about \emph{the} Dixmier trace of $T$ and denote it by $\Trdix(T)$.
\end{definition}

The following proposition gives a convenient characterisation of measurable operators, see \cite{LSS}, \cite[Theorem 9.7.5]{DixZeta4} (and \cite{SUZ2017} for a link with the $\zeta$-function)
\begin{proposition}
An operator $T \in \L^{1,+}$ is measurable if and only if $\,\lim_{N \to \infty} \tfrac{\Tr_N(T)}{\log N}$ exists, in which case it is equal to $\Trdix(T)$.
\end{proposition}

The link between Definition \ref{def:ncint} of $\ncint$ and $\Trdix$ was provided by Connes via his famous Trace Theorem \cite[Theorem 1]{ConnesAction}.

\begin{theorem}
\label{thm:Connes_trace}
Let $M$ be a compact Riemannian manifold of dimension $d$, $E$ a vector bundle over $M$ and $T \in \Psi^{-d}(M,E)$. Then, $T \in \L^{1,+}$, $T$ is measurable and
\begin{align}
\label{eq:equality of traces}
\Trdix(T) = \tfrac{1}{d} \WRes(T) = \tfrac{1}{d} \ncint T.
\end{align}
\end{theorem}

\begin{remark}
\label{rem:Dix_vs_ncint}
Despite the nice equality \eqref{eq:equality of traces}, we cannot hope for proportionality between the Dixmier trace and the noncommutative integral for an arbitrary $P \in \PDO$, even in the commutative case: Consider for instance $M=S^2$ -- the two-dimensional sphere, endowed with the standard Dirac operator $\Dslash$, which has singular values $\mu_n = n+1$ with multiplicities $4(n+1)$ (see Appendix \ref{sec:spheres}). Now take $T = \bbbone \in \Psi^0(C^{\infty}(S^2))$. We have, $\zeta_{\bbbone,\Dslash}(s) = \Tr \abs{\Dslash}^{-s} = 4 \zeta(s-1)$, which is regular at $s=0$, hence $\ncint \bbbone = 0$. On the other hand, $\bbbone$ is not a compact operator, hence its Dixmier trace does not make sense at all. 

We have privileged the noncommutative integral $\ncint$, which --- as we have just explained --- is more suitable in the noncommutative-geometric framework. However, on the technical side there are possible variants: Instead of $\zPD$ , we could have chosen the more symmetrised $\Tr \,(T^{1/2}\vert D\vert T^{1/2})^{-s}$ when $T$ is positive. The not obvious links between these functions and others are investigated in \cite{CGVRInt}.
\hfill$\blacksquare$
\end{remark}

\section{Fluctuations of Geometry}
\label{sec:fluc}

Given a spectral triple $\ahd$ it is natural to ask whether there exist other spectral triples describing a noncommutative geometry which is equivalent in some sense to the one determined by $\ahd$. To discuss this issue we firstly need the definition of noncommutative \emph{one-forms} \index{one-form}:
\begin{align}
\label{def:one-forms}
\OA \vc \Lspan \{ a \, \mathrm{d}b \; \vert \; a,b \in \A \}, \quad \text{ with } \mathrm{d}b \vc [\DD,b].  \index{_azzz@$\OA$}\skipindex{_azzz}
\end{align}
One can define accordingly the $n$-forms (modulo the so-called junk forms \cite[Section 8.1]{Elements}), which are building blocks of the Hochschild and cyclic homologies (see, for instance, \cite[Chapter 3]{BasicNCG} for more details). On the physical side, elements of $\OA$ are to be seen as gauge potentials of the theory.

A notion of equivalence suitable for spectral triples is that of Morita equivalence \cite[Chapter 1, Section 10.8]{ConnesMarcolli} (see also \cite[Section 2]{ConnesInner}, \cite[Section 4.5]{Elements} or \cite[Section 2.3]{BasicNCG}), which we now briefly describe. \\
Recall first that a (right) \emph{module} \index{module} $E$ over a unital algebra $\A$ is an Abelian group $(E,+)$ along with the (right) action of $\A$ on $E$, i.e. a map $E \x \A \to E$ such that for all $e,f \in E$ and any $a,b \in \A$
\begin{align*}
(e+f)a = ea + fa, \quad e(a+b) = ea + eb, \quad e(ab) = (ea)b, \quad e \bbbone_{\A} = e.
\end{align*}
A finitely generated module $E$ over $\A$ is \emph{free} \index{module!free} if there exists $N \in \N$ such that $E$ is isomorphic to $\A^{\ox N}$. \\
A module $E$ is called \emph{projective} \index{module!projective} \label{projective}if there exists another (right) module $F$ over $\A$ such that $E \oplus F$ is a free module. 

The Morita equivalence of algebras can be characterised as follows (cf. \cite[Proposition 6.12]{WalterBook}).

\begin{definition}
Let $\A, \A'$ be two unital associative algebras. We say that $\A'$ is \emph{Morita equivalent} \index{Morita equivalence} to $\A$ if there exists a finitely generated projective (right) module $E$ over $\A$, such that $\A' \simeq \End_{\A} E$.
\end{definition}

Having fixed a representation $\pi$ of $\A$ on $\H$, $\A'$ acts on $\H'\vc E \otimes_\A \H$. The space $\H'$ is endowed with the scalar product: 
$$\langle r \otimes \eta , s \otimes \xi \rangle \vc \langle \eta, \pi(r\vert s)\xi \rangle,$$
where $(\cdot \vert \cdot)$ is a 
pairing (or an $\A$-valued inner product) $E \times E \to \A$, which is $\A$-linear in the second variable and satisfies $$(r \vert s)=(s\vert r)^*,\quad (r \vert sa)=(r\vert s)a \,\,\text{ and }\,(s\vert s)>0 \,\text{ for }r\in E,0\neq s \in E.$$
Thus, each representation of $\A$ on $\H$ gives a representation of $\A'$ on $\H'$.
\\
For a given projective (right) module $E$ over $\A$, one can choose a \emph{Hermitian connection}, \index{Hermitian connection} i.e. a linear map $\nabla: E \to E \ox_{\A} \OA$, which satisfies the Leibniz rule $\nabla (r a) = (\nabla r)a + r \ox da$, for all $r \in E, \; a \in \A$. \\
With the help of a Hermitian connection one can define a selfadjoint operator $\DD' \in \L(\H')$ by
\begin{align*}
\DD'(r \ox \xi) \vc r \ox \DD \xi + (\nabla r)\xi, \quad \text{ with } \quad r\in E,\, \xi\in\Dom(\DD).
\end{align*}
Then, $(\A',\H',\DD')$ is a spectral triple (for the compatibility with a real structure see e.g. \cite[Theorem 6.16]{WalterBook}).

We say, by definition, that the spectral triple $(\A',\H',\DD')$ {\it is equivalent to} the original one $\ahd$. This is motivated by the following observation:
\\
The algebra $\A$ is obviously Morita equivalent to itself (with $E = \A$), in which case $\H' = \H$ and  Ad$_\DD:\,a \to [\DD,a]\in \Omega^1_\DD(\A)$ is a natural hermitian connection for $E$. Thus, any $\DD_A = \DD + A$ with $A = 
A^* \in \OA$ would provide a spectral triple $(\A,\H,\DD_A)$ equivalent to $\ahd$. \\ The operator $\DD_A$ is called a(n) (inner) \emph{fluctuation}\index{fluctuation} of $\DD$ and $A$ (following physicists conventions) -- a \emph{gauge potential}\index{gauge potential}. If $\DD_A = \DD + A$ with $A = A^* \in \OA$, then for any $B = B^* \in \Omega_{\DD_A}^1(\A)$ we have $\DD_A + B = \DD + A'$ with $A' = A + B \in \OA$ \cite[Proposition 1.142]{ConnesMarcolli}. In other words: ``Inner fluctuations of inner fluctuations are inner fluctuations''.

As we want to consider real spectral triples, we should require $(\A,\H,\DD_A)$ to have the same $KO$-dimension (see p. \pageref{def:real}) as $\ahd$. If $\DD J = \epsilon J \DD$, then we should require $\DD_A J = \epsilon J \DD_A$. Therefore, we define
\begin{align}
\label{def:D_A}
\DA \vc \DD + \Ag, \,\,\text{ with }\,\, \Ag \vc A + \epsilon J A J^{-1}, \quad \text{for} \quad A=A^* \in \OA. \index{_zD0A@$\DA$}
\end{align}
\vspace{-0.5cm}
\begin{example}[{cf. \cite[p. 734]{ConnesSA} and \cite[p. 22]{Almost2}}]\label{ex:no fluctuations on M}
Let $M$ be an even dimensional spin manifold as in Example \ref{ex:commutative}. For $a,b \in \A$ we have $a[\Dslash,b] = - i  \gamma^{\mu} a \dt_{\mu} b$. \\Since $A$ is Hermitian, $A_{\mu} \vc -i a \dt_{\mu} b \in C^{\infty}(M,\RR)$ and $[J,A_{\mu}] = 0$. On the other hand, $J\gamma^{\mu} = - \gamma^{\mu} J$ because $J$ commutes with $\Dslash$ and anticommutes with $i$. Finally,
\begin{align*}
\Dslash_{\Ag} = \Dslash + \Ag = \Dslash + A + J A J^{-1} = \Dslash + A - A J J^{-1} = \Dslash.
\end{align*}
Hence, there are no inner fluctuations in commutative geometry.
\hfill $\blacksquare$
\end{example}

The above considerations rely on the first-order condition \eqref{first_order}. The latter can be relaxed, what induces in a more general form of fluctuations \cite{ConnesFirst}. Whereas the exact shape of $\Ag$ does affect the physical content of the theory, the mathematics detailed in Chapter \ref{chap:perturbations} is not afflicted. Therefore, we shall only assume that $\Ag = \Ag^* \in \PDOk{0}$, which holds also for the more general form of perturbations, cf. \cite[Eq. (10)]{ConnesFirst}. As we will see, for such an $\Ag$, $(\A,\H,\DA)$ is still a spectral triple, which inherits the regularity, $p$-summability and dimension spectrum properties of $\ahd$. \label{A without first}

\section{\textit{Intermezzo}: Quasi-regular Spectral Triples}\label{sec:quasi_reg}

Before we proceed, let us come back for a moment to the assumption of regularity of a spectral triple --- recall Definition \ref{def:regularity}. As we have witnessed, it played an essential role in the construction of abstract pdos along with the dimension spectrum and the noncommutative integral. It turns out that at least some of these properties survive under a weaker assumption of \emph{quasi-regularity}\index{spectral triple!quasi-regular} (cf. \cite[Chapter 4]{PhD}):
\begin{align*}
\A, [\DD,\A] \subset \op^0.
\end{align*}
Clearly, every regular spectral triple is quasi-regular. An example of non-regular, though quasi-regular spectral triple, is provided by the standard Podle\'s sphere (cf. Appendix \ref{sec:Podles}).

Observe that if $T \in \op^0$, then $\sigma_z(T) \in \op^0$ for any $z \in \CC$ by properties \eqref{opProp}. On the other hand, the expansion \eqref{OPexpansion} does not hold in general. Typically, it is substituted by a `twisted' version (see \cite[Lemma 4.3]{PodlesSA} for an illustration).\\
Similarly to the regular case, given a quasi-regular $\ahd$ one can furnish the algebra $\A$ with a locally convex topology determined by the family of seminorms $a \mapsto \norm{\sigma_n(a)}$, $a \mapsto \norm{\sigma_n([\DD,a])}$. The completion $\A_\sigma$ is a Fr\'echet pre-$C^*$-algebra and $(\A_\sigma,\H,\DD)$ is again a quasi-regular spectral triple. The proof of this fact closely follows \cite[Lemma 16]{RennieSmooth} and the details can be found in \cite[Chapter 4]{PhD}.

In such a context one can define the algebra of pdos $\ePDO$\index{_apDOzz@$\ePDO$} as the polynomial algebra generated by $a \in \A$, $\DD$, $\aD$ and $\awD^{-1}$, complemented by $J\A J^{-1}$ if the triple is real and $P_0$ if $\ker \DD \neq \{0\}$. \\In the same vein, one can consider the algebra $\ePDOC$ as being generated by the elements of $\ePDO$ sandwiched with $\awD^{z_i}$ for $z_i \in \CC$. The grading is now provided by the $\op$ classes, $\ePDOk{r} \vc \ePDO \cap \op^r$.

The dimension spectrum of a quasi-regular spectral triple can be considered naturally with $T \in \ePDOk{0}$ in place of $\PDOz$ in Definition \ref{def:dimspec}. Being fairly involved, this notion is nevertheless workable and gives reasonable results for the Podle\'s sphere --- cf. Theorem \ref{thm:Podles_dimsp}.

Given a quasi-regular spectral triple with a dimension spectrum of order $d$, the noncommutative integral \eqref{ncint} still makes sense, however, the pleasant feature of $\ncintd{d}$ being a trace on the algebra of pdos is, in general, lost. Also, the study of fluctuations in the quasi-regular framework --- although possible in principle --- is rather obnoxious, as the entire Chapter \ref{chap:perturbations} of this book bases on Formula \eqref{OPexpansion}, hence on the regularity.

\section{The Spectral Action Principle}
\label{sec:sa}

Having prepared the ground we are now ready to present the concept of the spectral action. In a seminal paper  \cite{ConnesSA} Chamseddine and Connes put forward the following postulate: 
\begin{center}
,,\textit{The physical action only depends upon the spectrum of $\DD$.}''
\end{center}
\vspace*{0.2cm}

\noindent A mathematical implementation given in \cite{ConnesSA} led to the following definition.

\begin{definition}
\label{def_SA}
Let $\ahd$ be a spectral triple and let $\Lambda > 0$. The (bosonic) \emph{spectral action} \index{spectral action} associated with $\DD$ reads:
\begin{align}
\label{SA}
S(\DD,f,\Lambda) \vc \Tr \,f \left( \abs{\DD} / \Lambda \right), \index{_zS_aa@$S(\DD,f,\Lambda) $}
\end{align}
where $f: \RR^+\to \RR^+$ is a positive function such that $f \left( \abs{\DD} / \Lambda \right)$ exists and is a trace-class operator.
\end{definition}

Let us emphasise (cf. \cite[(1.23)]{ConnesSA}) that the operator $\DD$ in Formula \eqref{SA} should be the one dressed with gauge potentials, i.e.
\begin{align*}
\DD = \DD_0 + \Ag.
\end{align*}
This is because $(\A,\H,\DD_0)$ and $(\A,\H,\DD)$ cook up equivalent geometries and we must take into account all of the available degrees of freedom when constructing the action. See also Problem \ref{The space of Dirac operators} in Chapter \ref{chap:open}.

Before we pass on to the hard mathematics hidden behind Definition \ref{def_SA}, to which this book is primarily devoted, we provide a glimpse into its physical content. 

\paragraph{The physics of the spectral action}

As pointed out in \cite{ConnesUncanny}, the spectral action has several conceptual advantages:

-- Simplicity: If $f$ is the characteristic function $\chi_{[0,1]}$, then Formula \eqref{SA} simply counts the singular values of $\DD$ smaller or equal than $\Lambda$.

-- Positivity: When $f$ is positive, the action is manifestly positive --- as required for a sound physical interpretation.

-- Invariance: The invariance group of the spectral action is the unitary group of the Hilbert space $\H$, which is vast.

The role of the parameter $\Lambda$ is to provide a characteristic cut-off scale, at which \eqref{SA} is a bare action and the theory is assumed to take a geometrical form. It should have a physical dimension of length$^{-1}$, as the operator $\DD$ does. Within the almost commutative models the value of $\Lambda$ is typically taken to be in the range $10^{13}$--$10^{17} \, \mathrm{GeV}\cdot(\hbar c)^{-1}$.

Although $f = \chi_{[0,1]}$ is a privileged cut-off function providing the announced simplicity, it seems that for physical applications one should stay more flexible and allow it to depart from the sharp characteristic function. This is desirable in the almost commutative context, as the moments of $f$ provide free parameters of the theory, which can be tuned to fit the empirical data --- cf. \cite[Chapter 12]{WalterBook}. The actual physical role of $f$ in full generality of noncommutative geometry is more obscure --- see Problem \ref{The role of f} in Chapter \ref{chap:open}.

Let $\mathcal{U}(\A) \vc \{u \in \A \: \vert \: uu^*=u^*u=\bbbone\}$ \index{_zUa@$\mathcal{U}(\A)$} and $U \vc u JuJ^{-1}$, with $u \in \mathcal{U}(\A)$. Now, if the first-order condition \eqref{first_order} is satisfied, then the operator $\DA$ transforms under the action of $U$ as follows
\begin{align*}
U \DA U^* = U (\DD + A + \epsilon J A J^{-1}) U^* = \DD + A^u + \epsilon J A^u J^{-1},
\end{align*}
with $A^u \vc uAu^* + u[\DD,u^*]$. \index{_z0A_u@$A^u$}
This ensures that  $\mathcal{U}(\A)$ is a subgroup of the symmetry group of the spectral action. \\
In the almost commutative realm, the full symmetry group of \eqref{SA} is the  semi-direct product, $\mathfrak{G} \rtimes \Diff(M)$\index{_zGgauge@$\mathfrak{G}$}, of the group $\mathfrak{G} \vc \{u JuJ^{-1} \: \vert \: u \in \mathcal{U}(\A)\}$ of local gauge transformations and the group of diffeomorphisms of the manifold $M$ --- see \cite{WalterBook} for the details.\\
If one abandons the first-order condition \eqref{first_order}, then the inner fluctuations of the spectral action form a semi-group, which extends $\mathcal{U}(\A)$ \cite{ConnesFirst}.

To conclude this paragraph we exhibit the powerfulness of the spectral action: \\
Let $M$ be a compact Riemannian 4-manifold with a spin structure, as in Example~\ref{ex:commutative} and let $\A_F = \CC \oplus \mathbb{H} \oplus \mathcal{M}_3(\CC)$, with $\mathbb{H}$\index{_zHz@$\mathbb{H}$} denoting the algebra of quaternions. Then,
\begin{align}\label{SA_SM}
\!\!\!\!\SA = \int_M \sqrt{g}  \,d^4 x \Big( & \gamma_0 +\tfrac{1}{2 \kappa_0^2} R +  \alpha_0 \, C_{\mu \nu \rho \sigma} C^{\mu \nu \rho \sigma} + \tau_0 \, R^{\star} R^{\star} + \delta_0 \, \Delta R  \notag\\
& + \tfrac{1}{4} G^i_{\mu\nu} G^{\mu\nu i} + \tfrac{1}{4} F_{\mu\nu}^{a} F^{\mu\nu a} + \tfrac{1}{4} B^{\mu\nu}B_{\mu\nu} \\
& + \tfrac{1}{2} \vert D_{\mu} H \vert^2 - \mu_0^2 \vert H \vert^2 + \lambda_0 \vert H \vert^4 - \tfrac{1}{12} R \vert H \vert^2 \color{black} \Big) + \Oinf(\Lambda^{-1}).\notag
\end{align}
One recognises respectively: the cosmological constant, the Einstein--Hilbert term, the modified gravity terms, the dynamical terms of gauge bosons and the Higgs sector. The coefficients $\gamma_0$, $\kappa_0^{-2}$, $\alpha_0$ etc. depend on the powers of the energy scale $\Lambda^4, \Lambda^2, \Lambda^0$ and also of the moments of the cut-off function $f$: 
$$f_4 = \int_0^{\infty} x^3 f(x) dx,\, \,\,\,f_2 = \int_0^{\infty} x f(x) dx \,\,\,\text{ and }\, f(0).$$ 
Furthermore, they depend upon the fermionic content of the model, which is fixed by the choice of the Hilbert space $\H_F$ and the matrix $\DD_F$, which encodes the Yukawa couplings of elementary particles (cf. \cite{Almost1,WalterBook}).

To relish the full panorama of spectral physics, one supplements Formula \eqref{SA} with the \emph{fermionic action}\index{fermionic action} $S_F \vc \<J \psi, \DD \psi \>$, for $\psi \in \H^+ = \tfrac{1}{2}(1+\gamma) \H$ (cf. \cite[Definition 7.3]{WalterBook} and \cite{Barrett,Almost1}). Although the physical content of $S_F$ is exciting, its mathematics is rather mundane. Therefore, shamefully, we shall ignore it in the remainder of the book.

\paragraph{Many faces of the asymptotic expansions}

How on earth the simple Formula \eqref{SA} can yield the knotty dynamics of the full Standard Model and gravity? Actually, there is no mystery --- just asymptotics. Roughly (see Section \ref{sec:exp_SA} for the full story): $S(\DD,f,\Lambda)= \int_0^\infty \Tr\,(e^{-s \vert\DD\vert/\Lambda })\, d\phi(s)$, when $f$ is the Laplace transform of a measure $\phi$. But, $\Tr\,e^{-t P}$ is the celebrated heat trace associated with a (pseudo)differential operator $P$. The latter is known to enjoy a small-$t$ asymptotic expansion \cite{Gilkey2,VassilevichReport}: For instance, when $P = \Dslash^2$ and $d = \dim M$
\begin{align}
\label{eq:exp for D2}
\Tr e^{-t\,\Dslash^2} \tzero\,\sum_{k=0}^\infty \,a_{k}(\Dslash^2) \,t^{(d-k)/2}.
\end{align}
The alchemy is concealed in the coefficients $a_k$ known under the nickname of Seeley--deWitt coefficients. They are expressible as integrals over $M$ of local quantities polynomial in the curvature of $M$ and, if we happen to work with $\Dslash_E$ (cf. Proposition \ref{prop:PDO}), in the curvature of $E$. 

The Seeley--deWitt coefficients of a differential operator are \emph{local}\index{locality}. The principle of locality lies at the core of the concept of a field, which asserts that every point of spacetime is equipped with some dynamical variables \cite{Haag}. Concretely, a local quantity in quantum field theory is precisely an integral over the spacetime manifold $M$ of some frame-independent smooth function on $M$, which is polynomial in the field and its derivatives.

When $P$ is a pseudodifferential operator, an asymptotic expansion similar to \eqref{eq:exp for D2} is still available (cf. Appendix \ref{sec: the heat asymptotics of exp(-tP)}). On the other hand, the coefficients $a_k(P)$ are, in general, \emph{nonlocal} --- see \cite{GilkeyGrubb} for an explicit example. This means that pdos over classical manifolds belong already to the noncommutative world, which is prevailingly nonlocal.

In the full generality of noncommutative geometry the existence of an asymptotic expansion \eqref{eq:exp for D2} is no longer guaranteed. Nevertheless, one may hope to deduce it from a meromorphic extension of the associated zeta function $\zeta_D$. The bulk of Chapter \ref{chap:asymptotic} is devoted to this enjoyable interplay.\\ 
If $\ahd$ is a regular spectral triple with a simple dimension spectrum one can hope to obtain the following formula (see Corollary \ref{cor:SA_simple} and Section \ref{sec:consequences}):
\begin{align}\label{SA_simple}
\SA \!=\! \!\!\!\!\sum_{\alpha \in \Sd^+}\!\! \!\!\Lambda^{\alpha} \!\!\int_0^{\infty}\!\! \!\!x^{\alpha - 1} f(x) dx  \,\ncint \vert D \vert^{-\alpha} + f(0) \zeta_{D}(0) + \oinf(1).
\end{align}
Astonishingly, the abstract Formula \eqref{SA_simple} does yield the Standard Model action \eqref{SA_SM}, once the suitable almost commutative spectral triple has been surmised.

\smallskip

\label{asymptotics vs exact}
A careful reader have spotted the $\oinf(1)$ term in \eqref{SA_simple} and (s)he might wonder what does this symbol hide. Actually, if one has at one's disposal the \emph{full} asymptotic expansion (see Definition \ref{def:asym}) of the form \eqref{eq:exp for D2} one can expand Formula \eqref{SA_simple} to the order $\Lambda^{-N}$ for arbitrarily large $N$ (vide Theorem \ref{thm:f_expansion}). On the other hand, one must be aware of the fact that the explicit computation of the coefficients $a_k(P)$ for $k > 4$ is arduous, even if $P$ is a friendly differential operator. For a Laplace type operator general formulae are available up to $a_{10}$ --- see \cite{Avramidi} and \cite{VassilevichReport}. From the perturbative standpoint one might argue that the terms in the action, which vanish at large energies can safely be neglected. Then, however, one risks overlooking some aspects of physics. For example, it has been argued \cite{SitarzNeutrino} that a contribution of the order $\Lambda^{-2}$ might affect neutrino physics and the study of cosmic topology requires the knowledge of all $a_k$'s. To the latter end, several authors \cite{MarcolliBall,ConnesFLRW,ConnesUncanny,Lai-Teh,MarcolliCosmoBook,Marcolli1,Marcolli2,Piotrek1,TehPHD,Teh2013} employed the Poisson summation formula, which we discuss in Section~\ref{sec:Poisson}.

The summation formulae (sometimes dubbed not quite correctly ``nonperturbative methods'') give the spectral action modulo a reminder $\OO_\infty(\Lambda^{-\infty})$, which is usually disrespected. However, the devil often sits in the details: There exists an extensive catalogue of physical phenomena, which are `exponentially small'. An enjoyable account on this issue was produced by John P. Boyd \cite{Boyd}.

\paragraph{Highlights on the research trends}

The literature on the spectral action is abundant and growing fast. On top of the references already quoted, we list below some highlights on current research. The list is admittedly subjective and far from being complete --- we would be pretentious to claim to possess full knowledge on the topic.

To consider gravity as a low-energy effect of quantised matter fields, rather than a fundamental force is a long-standing idea: Yakov Zeldovich considered the cosmological constant as an effect of quantum matter fluctuations and Andrei Sakharov suggested that the structure of the quantum vacuum encodes the Einstein--Hilbert action (see \cite{NovozhilovV} for a short review).

The heat kernel methods were successfully applied in 1960's by Bryce S. deWitt in order to derive a series expansion for the Feynman propagator of quantum fields on a curved background. A more recent summary of the spectral approach to quantum field theory can be found in the textbooks \cite{AnalyticQFT} and \cite{ZetaQFT}.

But the upshot of the spectral action is that it provides an exact --- i.e. truly non-perturbative --- formula for the bare action at the unification scale. Expression \eqref{SA} as it stands is nonlocal, which means that it encompasses both local and global (topology, in particular) aspects of the physical world. For an approach balanced between the local and global aspects of the spectral action, which also goes beyond the weak field approximation, see \cite{ILVGlobal,ILVWeak}.

The impact of a boundary of a manifold on the spectral action was studied in \cite{CC2007,CC2011,ILVTorsion}. The main difficulty is related to the choice of the boundary conditions for the operator $\DD$, which would guarantee a selfadjoint extension, and then to define a compatible algebra $\A$, see \cite{Tadpole1,Tadpole}.

Some aspects of spectral geometry, including the heat trace expansion, on manifolds with conical singularities were studied in \cite{Lescure} (cf. also \cite{LeschFuchs}).

The role of the torsion in the spectral context has also been explored \cite{ILVTorsion,PS1,PS2,SitZaj}. Surprisingly enough, the spectral action for a manifold with torsion turns out to embody the Holst action, well known in the Loop Quantum Gravity approach \cite{PS3}.

On the physical side, a programme on the inflationary scenarios compatible with the spectral action has been launched \cite{KLSW,MarcolliCosmSurv,MarcolliSACosmo,Marcolli1,Marcolli2,NS,MairiReview,Sakellariadou}.

Also, the spectral action was approached via quantum anomalies and Higgs--dilaton interactions \cite{AKL,AL}.

\paragraph{Variations on the definition}

Usually, one encounters the operator $\DD^2$ instead of $\abs{\DD}$ in Formula \eqref{SA} for the spectral action. The reason for which $\DD^2$ is favoured in the literature is that for a commutative spectral triple $\Dslash^2$ is a differential operator (of Laplace type), whereas $\abs{\DD}$ is a priori only a pseudodifferential one. In the full generality of a noncommutative geometry, however, we are bound anyway to work with abstract operators, which are not even classical pdos and working with $\abs{\DD}$ allows for more flexibility. Clearly, we can restore the presence of $\DD^2$ by taking $f(x) = g(x^2)$ (cf. Remark \ref{rem:D_vs_Dsq}), but caution is needed, as the cut-off functions $f$ and $g$ belong to different classes (see Section \ref{subsec:Laplace}).

One could also consider the action of the form $\Tr f(\DD/\Lambda)$. If $f$ is even, as assumed in \cite{ConnesSA}, this is equivalent to \eqref{SA}. But one can take into account the asymmetry of the spectrum of $\DD$: When $f$ is odd, then $f(\DD/\Lambda) = \DD \abs{D}^{-1} g (\abs{\DD} / \Lambda)$ for an even function $g$, see for instance \cite{Piotrek1}. Technically, the parity of $f$ is not innocent and does play a role in the Poisson formula, cf. Section \ref{sec:Poisson}.

In \cite{MoyalSA} a formulation of the spectral action for nonunital spectral triples has been proposed, see also \cite{CCscale,Wulkenhaar}. Alternatively, one can simply consider a compactified spacetime manifold --- see, for instance \cite{ConnesFLRW,ConnesUncanny}, for the casus of Friedman--Lema\^{i}tre--Robertson--Walker universe. 

It should also be recognised that the spectral action (and its entire dwelling) works under the assumption of a positive-definite metric. Hence, the action \eqref{SA} is an Euclidean one and its physical applications require a Wick rotation \cite{LizziWick}. A truly Lorentzian approach is a challenging programme --- cf. Problem \ref{sec:Lorentz} in Chapter \ref{chap:open}.

\smallskip

We conclude with a derived notion of action: When the triple is even with a grading $\gamma$, we define the {\it topological spectral action} by 
\begin{align}
\label{toplogical action}
S_{\text{top}}(\DD,f,\Lambda)\vc \Tr \,\gamma \, f(\vert \DD \vert/\Lambda).
\end{align}
\index{spectral action!topological}We us now show as in \cite[Section 10.2.3]{WalterBook} that $S_{\text{top}}(\DD,f,\Lambda)=f(0)\text{ index} (\DD)$:
\\
The McKean--Singer formula \index{McKean@McKean--Singer formula} $$\text{index} (\DD) =\Tr \,\gamma \,e^{-t \,\DD^2}$$ holds true for any $t>0$: Indeed, let $\H_n$ be the eigenspace associated to the eigenvalue $\la_n$ of $\DD$ and $P_n$ be the eigenprojection on $\H_n$. Then,
\begin{align*}
\Tr \,\gamma \,e^{-t\, \DD^2}=\Tr(\gamma \,P_0)+\sum_{\lambda_n >0} \,e^{-t\lambda_n^2}\,\Tr(P_n-P_{-n})=\Tr(\gamma \,\,P_0)=\text{index }(\DD).
\end{align*}
Thus, if the function $x\in \RR^+ \mapsto f(x^{1/2})$ is a Laplace transform of a finite measure $\phi$, so that $f(x)=\int_0^\infty e^{-t \,x^2} \,d\phi(t)$, then the topological spectral action is simply 
$$
S_{\text{top}}(\DD,f,\Lambda)=\int_0^\infty \Tr \gamma \,e^{-t\,\DD^2/\Lambda}\,d\phi(t)=\text{index} (\DD) \int_0^\infty d\phi(t)=f(0)\text{ index} (\DD).
$$
Similarly its fluctuation is $S_{\text{top}}(\DD_\Ag,f,\Lambda)=f(0) \text{ index} (\DD_\Ag)$.

%
%
 \bibliographystyle{spmpsci}
 \bibliography{SA_BIB}
%

%
%
%

\chapter{The Toolkit for Computations}
\label{chap:tools}

\abstract{In this chapter we introduce a number of mathematical tools, which will prove useful in the spectral action computations. Firstly, we consider the basic properties of some spectral functions via the functional calculus and general Dirichlet series. Next, we study the interplay between these provided by the functional transforms of Mellin and Laplace. The remainder of the chapter is devoted to various notions from the theory of asymptotic behaviour of functions and distributions.}
\medskip

Before we start off let us recall the big-$\OO$ and small-$\oo$ notation: 

Let $X$ be a topological space and let $x_0$ be a non-isolated point of $X$. Let $U$ be a neighbourhood of $x_0$ and $V = U \setminus \{x_0\}$ -- a \emph{punctured neighbourhood}\index{punctured neighbourhood} of $x_0$. For two functions $f,g : V \to \CC$ we write
\begin{align*}
f(x) & = \OO_{x\to x_0} \left( g(x) \right) \; \text{ if }\,\limsup_{x \to x_0} \abs{f(x)/g(x)} < \infty,
\index{_zOoz@$\OO_{x_0}$} \\
f(x) & = \oo_{x\to x_0} \left( g(x) \right) \;\, \text{ if }\,\lim_{x \to x_0} \abs{f(x)/g(x)} = 0.
\index{_zOo@$\protect\oo_{x_0}$}
\end{align*}
We use $\OO_{x_0} \left( g(x) \right)$ and $\oo_{x_0} \left( g(x) \right)$ when no mistake concerning the variable can arise.

We will mostly be concerned with the cases $X = \RR^+ \cup \{\infty\}$ or $X = \CC \cup \{\infty\}$ and $x_0 = 0$ or $x_0 = \infty$. \\
The notations $\OO_0(x^\infty)$ and $\OO_{\infty}(x^{-\infty})$ \index{_zOozz@$\OO_{\infty}(x^{-\infty})$} will stand, respectively, for $\OO_{x\to 0}(x^k)$ and $\OO_{x\to +\infty}(x^{-k})$, for all $k>0$, and similarly for $\oo$.

\begin{example}
We have $\sin x = \Oinf(1)$, but $\sin x \neq \oinf(1)$ and $1 \neq \Oinf(\sin x)$.
\\
For any $n > 0$, $\log^n x = \oz(x^{-\epsilon}) = \oinf(x^{\epsilon})$ for all $\epsilon > 0$.
\hfill$\blacksquare$
\end{example}
For further examples and properties of $\OO$ and $\oo$ symbols see, e.g. \cite[Section 1.2]{Estradabook}.

\section{Spectral Functions}
\label{sec:spec}

The spectral action \eqref{SA} is par excellence a \emph{spectral function}\index{spectral function}, i.e. a (possibly complex valued) function on the spectrum of some operator.

We now take a closer look at various spectral functions from the perspective of general Dirichlet series \cite{Hardy}. As this part only involves the properties of $\DD$ and not the full force a spectral triple, we shall work with a general operator $H$ acting on an infinite dimensional separable Hilbert space $\H$. We shall need the following classes of positive densely defined unbounded operators, for $p \in \RR^+$,
\begin{align*}
\Tp \vc \big\{ H \in \L(\H) \; \big\vert \; H > 0\, \text{ and }\, \forall \epsilon > 0 \; \Tr\, H^{-p-\epsilon} < \infty,\text{ but } \Tr\, H^{-p+\epsilon} = \infty \big\}.\index{_zTp@$\Tp$}
\end{align*}

If $H \in \Tp$ for some $p$, then it is invertible and $H^{-1} \in \KKK(\H)$. Moreover, with $\L^r(\H) \vc \{T \in \L(\H) \, \vert \, \Tr \abs{T}^r < \infty \}$\index{_zLp@$\L^r(\H)$} (the so-called $r$-th Schatten ideal\index{Schatten ideal}), we have
\begin{align*}
p = \inf \{ r \in \RR^+ \,\, \vert\, \, H^{-1} \in \L^r(\H)\}.
\end{align*}

If $\ahd$ is a not finite (i.e. with $\dim \H = \infty$ --- cf. Example \ref{ex:noncommutative}) $p$-dimensional spectral triple then $\awD \in \Tp$ and $D^2 \in \Tpp{p/2}$ (cf. \eqref{absD}). If $\ahd$ is a regular spectral triple, then also $\awDA \in \Tp$ and $D_{\Ag}^2 \in \Tpp{p/2}$ with a suitable fluctuation $\Ag$ --- see Section \ref{sec:fluc} and Chapter \ref{chap:perturbations}. Note, however, that the requirement $\DD \in \Tp$ for some $p$ rules out the finite spectral triples, since $\Tp \cap \L^1(\H) = \emptyset$. The latter situation is trivial from an analytic point of view as $\Tr f(\awD/\Lambda)$ is finite and explicitly computable for any bounded measurable function $f$ and any $\Lambda > 0$. \label{finite triple case}

As the primary example of a spectral function let us consider
\begin{align*}
N_H(\Lambda) \vc \sum_{n: \lambda_n(H) \leq \Lambda} M_n(H), \text{ for } \Lambda > 0,\index{_zNH@$N_H(\Lambda)$}
\end{align*}
titled the \emph{spectral growth function}. We get  $$N_{\abs{D}}(\Lambda) = \Tr \chi_{[0,\Lambda]}(\abs{D}) = S(D, \chi_{[0,1]}, \Lambda),$$ which is the archetype of the spectral action \cite{ConnesAction}.

Via the unbounded functional calculus (see, for instance \cite[Chapter 13]{Rudin}) we define an operator 
$$f(H) \vc \int_{\lambda \in \spec(H)} f(\lambda)\,  d P_\lambda(H) = \sum_{n=0}^{\infty} f(\lambda_n(H)) \,  P_n(H),$$\index{_zPH@$P_n(H)$}with the spectral projections $P_n(H) \vc P_{\lambda_n}(H)$, for any bounded Borel (possibly complex) function $f$ on $\RR^+$. The operator $f(H)$ is trace-class if and only if
\begin{align}\label{SA_H}
\Tr f(H) = \sum_{n=0}^{\infty} M_n(H) \, f(\lambda_n(H)) < \infty.
\end{align}
More generally, if $H\in \Tp$, $K$ is any operator in $\B(\H)$ and $\Tr f(H) < \infty$, then 
$\Tr K f(H) = \sum_{n=0}^{\infty} \Tr ( P_n(H) K ) \,f( \lambda_n(H)) < \infty$. In particular, for $f(x) = x^{-s}$ with $\Re(s) > p$, we obtain the \emph{spectral zeta function}\index{spectral zeta function}
\begin{align}
\label{zetaKH}
\zKH(s) \vc \Tr \, K H^{-s} = \sum_{n=0}^{\infty} \Tr ( P_n(H) K ) \,\lambda_n(H)^{-s}, \text{ for } \Re(s) > p.\index{_a3zetaTDH1@$\zKH(s)$}
\end{align}
If $H$ is not positive, but nethertheless $\vert H \vert \in \Tp$ for some real number $p$, we define $\zKH(s) \vc \Tr \,K\vert h\vert ^{-s}$ in accordance with the notation adopted in \eqref{zetaB}. \\Moreover if $K = \bbbone$, we shall simply write $$\zH \vc \zeta_{\bbbone,H}.$$\index{_a3zetaTH@$\zH(s)$}
The function $\zKH$ will often admit a meromorphic extension to a larger region of the complex plane, in which case it is customary to denote the extension with the same symbol. Note that $H\in \Tp$ guarantees that $\zKH$, if meromorphic, has at least one pole located at $s = p$. \label{zeta_one_pole}

Yet another propitious spectral function is the \emph{heat trace}\index{heat trace}, which results from $f(x) = e^{-t x}$ with $t>0$,
\begin{align}
\label{heat_KH}
\hKH \vc \sum_{n=0}^{\infty} \Tr ( P_n(H) K ) \,e^{-t \lambda_n(H)}.
\end{align}
The operator $e^{-t \,H}$ is called the \emph{heat operator}\index{heat operator}. This nickname comes from physics: When $-H=\Delta$ is the standard Laplacian on a Riemannian manifold then $e^{-t\,H} \psi$ solves the heat equation, $\partial_t \phi +H\phi=0$, with the initial condition $\phi(0)=\psi$.

On the side, we note that for any positive (unbounded) operator $P$ the function $t\in \RR^+\mapsto e^{-t\,P}\in \B(\H)$ admits a holomorphic extension to the right half-plane via $
e^{-t\,P}=\tfrac{1}{i2\pi}\int_{\lambda\in\C} e^{-t\,\lambda}(P-\lambda)^{-1}\,d\lambda$,
where $\C$ is a contour around $\RR^+$. This can be proved by differentiating under the integral --- see Appendix \ref{About the exponential of an unbounded operator}. 
It implies, in particular, that if $e^{-t \,H}$ is trace-class for any $t > 0$, then the function $t \mapsto \hH$ is smooth on $(0,\infty)$ and so is $t \mapsto \hKH$ for any $K \in \B(\H)$. We will give an alternative proof of this fact in Proposition \ref{prop:KH}.

When $H\in \Tp$, both heat traces and spectral zeta functions (for $\Re(s) > p$) are instances of \emph{general Dirichlet series} \index{Dirichlet series} \cite{Hardy}. The latter are defined as
\begin{align}\label{gen_Dirichlet_series}
\sum_{n=0}^\infty a_n \,e^{-s\,b_n},
\end{align}
for $a_n\in \CC$, $b_n \in \RR$ with $\lim_{n \to \infty} b_n = +\infty$ and some $s\in\CC$. The region of their convergence constitutes a half-plane $\Re(s) > L$ for some $L \in \RR$, the latter being called the \emph{abscissa of convergence}
\index{abscissa of convergence}. In contradistinction to Taylor series, the domain of absolute convergence of \eqref{gen_Dirichlet_series}, which is also a half-plane $\Re(s) > L'$, can be strictly smaller, i.e. $L' \geq L$ (cf.  \cite[Section 11.6]{Apostol}). The regions of convergence of general Dirichlet series can be determined from the following theorem:

\begin{theorem}[{\cite[Theorem 7 with the footnote]{Hardy}}] \label{thm:abscissa}
If the series $\sum_{n=0}^\infty a_n$ is not convergent, then the series \eqref{gen_Dirichlet_series} converges for $\Re (s)>L$ and diverges for $\Re (s)<L$, with
\begin{align}\label{L}
L=\limsup_{n\to\infty}\, b_n^{-1}\,\log\abs{a_0+\ldots+a_n} \geq 0.
\end{align}
\end{theorem}
The non-negativity of $L$ follows from the fact that for $s=0$ the expression \eqref{gen_Dirichlet_series} equals to $\sum_{n=0}^\infty a_n$, which is not convergent. To compute the abscissa of absolute convergence one simply needs to take $\abs{a_i}$ in Formula \eqref{L}.\\
For $K, H > 0$, we always have $L' = L$ in both Dirichlet series \eqref{heat_KH} and \eqref{zetaKH}. 

\medskip

For given operators $K$ and $H$ the associated heat trace and zeta functions are closely interrelated, which will be explored in great detail in Chapter \ref{chap:asymptotic}. We now offer the first glimpse into their intimate relation.

Theorem \ref{thm:abscissa} implies that the operator $K e^{-t H}$ is trace-class for all $t>0$ iff
\begin{align}\label{heat_abscissa}
\log \big\vert \sum_{n=0}^{N} \Tr ( P_n(H) K ) \big\vert = \oo_{N \to \infty}(\lambda_N(H)).
\end{align}
On the other hand, $\zKH$ is defined (i.e. \eqref{zetaKH} converges for $\Re(s) > p$ for some $p\geq 0$) if and only if (cf. \cite[Proposition 2]{HeatEZ})
\begin{align}
\label{zeta_abscissa}
\sum_{n=0}^{N} \Tr ( P_n(H) K ) = \OO_{N \to \infty}(\lambda_N(H)^r), \quad \text{ for all } \; r > p.
\end{align}

Formula \eqref{zeta_abscissa} implies \eqref{heat_abscissa}, but the converse is not true. As a counterexample consider $K = \bbbone$ and $\lambda_n(H) = \log^2 n$, $M_n(H) = 1$. 
In particular, if a spectral triple $\ahd$ is $\theta$-summable with $t_0 = 0$, but not finitely summable (recall p.~\pageref{theta_summ}), the heat traces $\Tr T e^{-t \awD}$ exist, while the spectral zeta functions $\zTD$ do not.


Given the abscissa of convergence of the zeta function $\zKH$ we can easily deduce the behaviour of $\hKH$ as $t \downarrow 0$.

\begin{proposition}
\label{prop:KH}
Let $H\in \Tp$ and $K \in \B(\H)$. The function $t \mapsto \hKH$ is smooth on $(0,\infty)$ and $\hKH=\Oz(t^{-r})$, for all $r > p$.
\end{proposition}

\begin{proof}
Let us consider the function $x \mapsto x^{\alpha} e^{-x}$, which is bounded on $\RR^+$ for any $\alpha \geq 0$. For any $x >0$, we thus have $x^{\alpha} e^{-x} \leq c(\alpha)$, 
with some positive constant $c(\alpha)$. By invoking this inequality with $x/2$ and multiplying it by $e^{-x/2}$, we obtain $x^{\alpha} e^{-x} \leq 2^{\alpha}c(\alpha)  \,e^{-x/2}$.\\
Let us fix any $r > p \geq 0$ and use $x^{\alpha} e^{-x} \leq c(\alpha)$ again. For any $r > p$ we have
\vspace*{-0.1cm}
\begin{align*}
0 \leq t^{r} \hH & = \sum_{n=0}^{\infty} M_n(H) \, t^{r} e^{-t \lambda_n(H)}   \leq c(r) \sum_{n=0}^{\infty} M_n(H) \lambda_n(H)^{-r} = c(r)\, \zH(r) < \infty.
\vspace*{-0.1cm}
\end{align*}
This shows in particular that $e^{-t\,H}$ is trace-class for any $t>0$. \\Moreover, the function $t \mapsto t^{r} \,\hH$ is bounded for $t \in [0,\infty)$ for every $r > p$ and the limit superior of $t^{r} \hH$ as $t \downarrow 0$ exists and is finite, giving $\hH=\Oz(t^{-r})$ for all  $r > p$.

Furthermore, for any $t>0$ and any $k \in \N$ we have
\begin{align*}
\vert\tfrac{d^k}{d \, t^k} \Tr K\,e^{-t\,H}\vert &= \vert\Tr \, K\,H^k e^{-t H}\vert \leq \norm{K}\Tr \,H^k \,e^{-t H}  \leq \norm{K}2^k c(k)  \, \Tr \, e^{-t\, H/2} < \infty,
\end{align*}
which shows that the function $t\to \Tr K e^{-t\,H}$ is indeed smooth on $(0,\infty)$. \\
Finally, $0 \leq t^{r} \vert \hKH \vert \leq \norm{K} t^{r} \hH$ and one concludes as above.
\hfill $\Box$
\end{proof}

\begin{remark}
\label{rem:bounded_from_below}
The heat trace \eqref{heat_KH} can also be defined when $H$ is a selfadjoint operator bounded from below with a compact resolvent. In this case one writes, for any $t>0$,
\begin{align*}
\hKH = \sum_{n=0}^{N} \Tr ( P_n(H) K ) \,e^{-t \,\lambda_n(H)} + \sum_{n=N+1}^{\infty} \Tr ( P_n(H) K ) \, e^{-t\, \lambda_n(H)},
\end{align*}
for an $N \in \N$ such that $\lambda_{N+1}(H) > 0$. Since the first finite sum is manifestly a smooth bounded function on $\RR^+$, Proposition \ref{prop:KH} holds.
\hfill$\blacksquare$
\end{remark}
In particular, in the context of spectral triples we have (recall p. \pageref{absD})
\begin{align}
\hD = \Tr e^{-t \,\vert D \vert} + (1- e^{-t}) \dim \ker \DD = \Tr_{(\bbbone-P_0)\H}\, e^{-t\,\vert\bD\vert} + \dim \ker \DD. \label{heat_ker} 
\end{align}
Also, for $\Re(s) > p$, we have $$\zeta_{\bD} (s)= \Tr_{(\bbbone-P_0)\H}\,\bD^{\,-s} = \zeta_D(s) - \dim \ker \DD. \index{_a3zetaTHDb@$\zeta_{\bD}$}$$

\begin{example}
\label{ex:heat_S1}
Let $\Dslash$ be the Dirac operator associated with the trivial spin structure on $S^1$ equipped with the round metric (cf. Appendix \ref{sec:spheres}). Then,
\begin{align*}
\zeta_D(s) = \sum_{n=0}^{\infty} M_n(\abs{D}) \, \lambda_n(\vert D\vert)^{-s} = 1^{-s} + 2 \sum_{n=0}^{\infty} n^{-s} = 1 + 2 \zeta(s),
\end{align*}
where $\zeta$\index{_a3zeta_@$\zeta(s)$} is the familiar Riemann zeta function. The latter has a single simple pole at $s=1$, so the triple is $1$-dimensional. The corresponding heat trace reads, for $t > 0$, $$\hD = 1+2 \sum_{n=1}^\infty e^{-n\,t} =1+ \tfrac{2e^{-t}}{1-e^{-t}} = \coth\tfrac t2$$ and, indeed, $\hD = \Oz(t^{-1})$ in accord with Proposition \ref{prop:KH}.
\hfill$\blacksquare$
\end{example}

\begin{remark}
Let us note that Proposition \ref{prop:KH} extends in a straightforward way to $\Tr e^{-t \,H^{\alpha}} = \Oz(t^{-r})$, for all $\; r > p/\alpha$, 
for any $\alpha>0$, as $\zeta_{H^{\alpha}}(s) = \zH(\alpha s)$. \\
When applied to Example \ref{ex:heat_S1} it implies a non-trivial fact: 
\begin{align*}
\sum_{n=1}^\infty e^{-n^{r} t} = \Oz(t^{-\alpha}) \text{ for all }\alpha > 1/r. \tag*{$\blacksquare$}
\end{align*}
\end{remark}

We conclude this section with a sufficient condition for the well-definiteness of the spectral function \eqref{SA_H}, which is of our primary interest.

\begin{lemma}
\label{traceclass}
Let $H\in \Tp$ and let $f$ be a positive function defined on $\RR^+$ such that  $f(x)=\Oinf(x^{-p -\epsilon})$ for some $\epsilon >0$. Then, $f(H)$ is trace-class.
\end{lemma}

\begin{proof}
With $\{\psi_n\}_{n \in \N}$ being the orthonormal basis on $\H$ composed of the eigenvectors of $H$ we obtain $\Tr f(H) = \sum_{n=0}^{\infty} \< \psi_n, f(H) \psi_n \> = \sum_{n=0}^{\infty} M_n(H) \, f(\lambda_n(H))$. By the hypothesis on the decay of $f$, there exist $c_2,N>0$, such that for any $x\geq 0$,  we can estimate $\abs{f(x)}\leq c_1+c_2 \,\chi_{[N,\infty)}(x) x^{-p -\epsilon}$, with $c_1=\sup_{\abs{x}<N}\abs{f(x)}$. Then,
\begin{align*}
\Tr f(H) \leq c_1 \!\!\sum_{n<N} M_n(H) +c_2 \!\! \sum_{n\geq N} M_n(H) [\lambda_n(H)]^{-p-\epsilon} \leq c_3+c_2\Tr \,H^{-p-\epsilon} < \infty.
\tag*{$\Box$}
\end{align*}
\end{proof}

\section{Functional Transforms}
\label{sec:func}

In this section we recall some basic notions and facts on two functional transforms, which will allow us to express some spectral functions in terms of others. These tools will help us to eventually establish an asymptotic expansion of the spectral action in Chapter \ref{chap:asymptotic}.

\subsection{Mellin Transform}
\label{subsec:Mellin}

We recall the following definition (see e.g. \cite[Section 3.1.1]{Paris} or \cite[Definition 1]{Flajolet}):

\begin{definition}
\label{def:Mellin}
The \emph{Mellin transform} \index{Mellin transform} of a locally (Lebesgue) integrable function $f: (0,\infty) \to (0,\infty)$ is a complex function $\Mellin[f]$ defined by
\begin{align}
\label{Mellin}
\Mellin[f](s)=\int_0^\infty f(t)\,t^{s-1}\,dt, \index{_zMf@$\Mellin[f]$}
\end{align}
for all $s \in \CC$ for which the integral converges.
The \emph{inverse Mellin transform} $\Mellin^{-1}$ of a meromorphic function $g$ is defined by
\begin{align}\label{Mellin_inv}
\Mellin^{-1}[g](t)=\tfrac1{2\pi i}\int_{c-i\infty}^{c+i\infty}g(s)\,t^{-s}\,ds,
\end{align}
for some $c \in \RR$ such that the integral exists for all $t>0$.
\end{definition}
The domain of definition of a Mellin transform turns out to be a strip, called the fundamental strip (see \cite[Definition 1]{Flajolet}). If $f(t)=\Oz(t^\alpha)$ and $f(t)=\Oinf(t^\beta)$ for some $\alpha > \beta$, then $\Mellin[f](s)$ exists at least in the strip $-\alpha<\Re (s)<-\beta$ \cite[Lemma~1]{Flajolet}. The invertibility of Mellin transform is addressed by the following \cite[Theorem 2]{Flajolet}:

\begin{theorem}
\label{thm:Mellin_inverse}
Let $f$ be a continuous function. If $c \in (0,\infty)$ belongs to the fundamental strip of $\Mellin[f]$ and the function $\RR \ni y \mapsto\Mellin[f](c+iy)$ is integrable, then, for any $t > 0$
\begin{align*}
f(t)=\Mellin^{-1}\big[\Mellin[f]\big](t)=\tfrac1{2\pi i}\int_{c-i\infty}^{c+i\infty}\Mellin[f](s)\,t^{-s}\,ds.
\end{align*}
\end{theorem}

In the realm of spectral functions the usefulness of the Mellin transform is attested by the following result:

\begin{proposition}
\label{prop:hk_Mellin}
Let $H\in \Tp$ and let $K \in \B(\H)$, then
\begin{align}\label{MellinFund}
\Mellin[t \mapsto \hKH](s)=\Gamma(s)\,\zKH(s), \quad \text{for } \; \Re(s)>p.
\end{align}
\end{proposition}

\begin{proof}
Let us pick any $r$ such that $\Re(s)>r>p$. Proposition \ref{prop:KH} guarantees that $\hKH(t)=\Oz(t^{-r})$ and the integral $\Mellin[t \mapsto \hKH](s)=\int_0^\infty \hKH \,t^{s-1}\,dt$ converges (absolutely) at 0.
It also converges absolutely at $\infty$ since
\begin{align*}
\abs{\hKH} & = e^{-t\lambda_0(H)} \,\big{\vert}  \Tr ( P_0(H) K ) + \sum_{n=1}^\infty \Tr ( P_n(H) K ) \, e^{-t (\lambda_n(H) - \lambda_0(H))}  \big{\vert}  \\
& = \Oinf(e^{-t\lambda_0(H)})
\end{align*}
and the first eigenvalue $\lambda_0(H)$ is positive.

Since the map $(t,n) \mapsto \Tr ( P_n(H) K ) \, e^{-t \lambda_n(H)}$ is in $L^1((0,\infty) \x \N)$ we can use the Fubini theorem to swap the integral and the sum in the following:
\begin{align*}
\Mellin&[t \mapsto \hKH](s)\\
 &=\int_0^\infty \sum_{n=0}^{\infty} \Tr ( P_n(H) K ) \, e^{-t \lambda_n(H)} \,t^{s-1}\,dt  \notag  =\sum_{n=0}^{\infty} \Tr ( P_n(H) K ) \int_0^\infty e^{-t \lambda_n(H)} \,t^{s-1}\,dt\notag\\
& =\sum_{n=0}^{\infty} \Tr ( P_n(H) K ) \,\lambda_n(H)^{-s}\int_0^\infty e^{-y}\,y^{s-1}\,dy =\zKH(s)\,\Gamma(s).
\tag*{$\Box$}
\end{align*}
\end{proof}

The domain of definition is an important ingredient of the Mellin transform: Even if both functions $f,g$ do have Mellin transforms, $\Mellin[f+g]$ might not exist if $\Dom \Mellin[f] \cap \Dom \Mellin[g] = \emptyset$. Let us illustrate this feature by inspecting the Mellin transform of Formula \eqref{heat_ker}.

\begin{example}
Let $\ahd$ be a $p$-summable spectral triple. If $\DD$ has a non-trivial kernel, then the Mellin transform of the heat trace $\hD$ does not exist. Indeed, when $\mu_0(\DD) = 0$, $\lim_{t \to \infty} \hD = \dim \ker \DD$ and the integral \eqref{Mellin_inv} converges at $\infty$ only if $\Re(s)<0$, whereas the convergence at 0 requires $\Re(s) > p \geq 0$. On the other hand, $\hD = \Tr e^{-t \,\vert D \vert} - (e^{-t}-1) \dim \ker \DD$ using \eqref{heat_ker}. Notice that both functions of $t$ on the RHS of \eqref{heat_ker} do have Mellin transforms, but with disjoint domains. In particular, we have $$\Mellin[t \mapsto (e^{-t} - 1)](s) = \Gamma(s)\,\text{ for }\,\Re(s) \in (-1,0)$$ (see \cite[p. 13]{Flajolet}).
\hfill$\blacksquare$
\end{example}

The relation between the spectral functions unveiled by Proposition \ref{prop:hk_Mellin} can be inverted with the help of Theorem \ref{thm:Mellin_inverse}.\pagebreak

\begin{corollary}
\label{cor:hk_Mellin_inv}
Let $H\in \Tp$ and let $K \in \B(\H)$. 
Then, for any $c > p$ and any $t > 0$, $$\hKH = \tfrac1{i2\pi}\int_{c-i\infty}^{c+i\infty}\,\Gamma(s) \,\zKH(s)\,t^{-s}\,ds.$$
\end{corollary} 

\begin{proof}
Recall that the function $t \mapsto \hKH$ is smooth (cf. Proposition \ref{prop:KH}). In order to apply Theorem \ref{thm:Mellin_inverse} we need the function $y \mapsto \Gamma(c + iy) \,\zKH(c + iy)$ to be Lebesgue integrable for any $c > p$. This is indeed the case as if $c > p \geq 0$, we have $\Gamma(c+ i y) = \Oinf \big(\abs{y}^{c-1/2} e^{-\pi \abs{y}/2} \big)$ (cf. Lemma \ref{lm:Gamma_vert}), and moreover we estimate
\begin{align*}
\abs{\zKH(c+iy)} = \abs{ \Tr K H^{-c-iy} } \leq \norm{K} \Tr \abs{H^{-c-iy}} \leq \norm{K} \Tr H^{-c}.
\tag*{$\Box$}
\end{align*}
\end{proof}

\subsection{Laplace Transform}
\label{subsec:Laplace}

The Mellin transform allows us to move back and forth between a spectral zeta function and the associated heat trace. In order to build a bridge from these to the spectral action, we will resort to the Laplace transform. 

\begin{definition}
\label{def:Laplace}
The \emph{Laplace transform} \index{Laplace transform} of a function $f \in L^1(\RR^+,\CC)$ is an analytic function $\Lc[f] : \RR^+ \to \CC$ defined by
\begin{align}
\label{Laplace}
\Lc[f](x) \vc \int_0^\infty f(s)\,e^{-s x}\,ds. \index{_zLx@$\Lc[f] \text{ and } \Lc[\phi]$}
\end{align}
\end{definition}
\noindent For our purposes it is sufficient to consider $x \in \RR^+$, although \eqref{Laplace} actually defines an analytic function in the whole half-plane $\Re(x)\geq 0$. Also, one can consider the Laplace transform of functions in $L^1_{\text{loc}}(\RR^+)$\index{_zLzloc@$L^1_{\text{loc}}(\RR^+)$} (locally integrable on $\RR^+$), as long as there exists $M \in \RR$ such that the integral converges (properly) for $\Re(x) > M$.

More generally, one can define the Laplace transform of a complex Borel measure $\phi$ on $\RR^+$\index{Laplace transform!of a measure} via the Lebesgue integral
\begin{align}
\label{Laplace_measure}
\Lc[\phi](x) \vc \int_0^{\infty} e^{-s x}\, d \phi(s), \quad \text{ for } x>0. 
\end{align}
Since we are interested in real functions $\Lc[\phi]$ we shall restrict ourselves to the case of Borel \emph{signed measures}\index{measure!signed} (see for instance \cite{Cohn2013}) on $\RR^+$, i.e. $\sigma$-additive maps $\phi: \B(\RR^+) \to \RR \cup \{\pm \infty \}$, on the $\sigma$-algebra of all Borel subsets of $\RR^+$. A signed measure may assume one of the infinite values $\pm \infty$, but not both. 
\\
The Hahn--Jordan decomposition allows us to uniquely write any signed measure as $\phi = \phi^+ - \phi^-$ for two non-negative measures, at least one of which is finite. It is customary to denote $\vert\phi\vert = \phi^+ + \phi^-$, which is a non-negative measure on $\RR^+$ called the \emph{variation} of $\phi$\index{measure!variation of}. \\
The support of a signed measure is $\supp \phi \vc \supp \aphi = \supp \phi^+ \cup \supp \phi^-$.

Let us recall that a function $f$ is \emph{completely monotonic (c.m.)} if \\\centerline{$f\in C^\infty((0,\infty),\RR)$ and $(-1)^n \,f^{(n)}(x) \geq 0$ for any $n\in \N$, and any $x>0$.} \\
The set of such functions, denoted by $\CM$\index{_zCM@$\CM$}, is well adapted as the range of the Laplace transform:

\begin{theorem}[Bernstein, see e.g. {\cite[p. 160]{Widder} or \cite[Theorem 1.4]{SRV}}]
\noindent Given a function $f \in \CM$ we have $f(x)=\Lc[\phi](x)$ for all $x>0$ for a unique non-negative measure $\phi$ on $\RR^+$. \\
Conversely, 
whenever $\Lc[\phi](x) < \infty, \,\forall x>0$, then $\Lc[\phi]\in \CM$.
\end{theorem}
The set $\CM$ is a convex cone which is stable under products, derivatives of even order and pointwise convergence. Moreover, the closure under pointwise convergence of Laplace transforms of finite measures on $\RR^+$ is exactly $\CM$ \cite[Corollary 1.6]{SRV}.

The Bernstein theorem naturally extends to the context of signed measures:

\begin{corollary}
If $f = f^+ - f^-$ with c.m. functions $f^\pm$, then $f(x)=\Lc[\phi](x)$ for all $x>0\,$ for a unique signed measure $\phi$ on $\RR^+$. Conversely, whenever $\Lc[\phi](x) < \infty$ for all $x>0$, then $\Lc[\phi]$ can  be uniquely written as a difference of two c.m. functions.
\end{corollary}

\begin{proof}
This is a straightforward consequence of the Bernstein theorem and the Hahn--Jordan decomposition.
\hfill$\Box$
\end{proof}

The Laplace transform of a measure $\phi$ is well defined on the open set $\RR^+\setminus\{0\}$. If the measure $\phi$ is sufficiently nice, then $f =\Lc[\phi]$ along with its derivatives can be extended to the whole $\RR^+$:

\begin{proposition}
\label{prop:moments}
Let $\phi$ be a signed measure on $\RR^+$ and let $f(x) = \Lc[\phi](x)$ for $x>0$. If $\phi$ has a finite $n$-th moment\index{measure!moment of}, i.e. $\int_0^{\infty} s^n d \phi(s) < \infty$ for an $n \in \N$, then $\lim_{x \to 0^+}  f^{(n)}(x)$ exists and is equal to $(-1)^n \int_0^{\infty} s^n d \phi(s)$.
\end{proposition}
\begin{proof}
This is a consequence of the Lebesgue dominated convergence theorem: For $s,x \geq 0$ we have $\abs{e^{-sx}} \leq 1$, thus $$\int_0^{\infty} \, d\phi(s) = \lim_{x \to 0^+} \, \int_0^{\infty} e^{-s x} \, d\phi(s) = \lim_{x \to 0^+} f(x).$$  
The statement for $n > 0$ follows from a general property of the Laplace transform (cf. \cite[Eq. (12)]{Zemanian1}): 
\begin{align*}
\Lc[s \mapsto s^n \phi(s)](x) = (-1)^n f^{(n)}(x),  \text{ for any } n \in \N.
\tag*{$\Box$}
\end{align*}
\end{proof}

\begin{remark}
\label{rem:moments}
The converse of Proposition \ref{prop:moments} is not in general true for signed measures. A counterexample is provided by the function $f(x) = 1- e^{-1/x}$, which is the Laplace transform of the function $\phi(s) = J_1(2 s^{1/2}) s^{-1/2}$, with $J_{\alpha}$\index{_zJalpha@$J_{\alpha}$} being the Bessel function of the first kind. Although $\lim_{x \to 0^+} f^{(n)}(x) = 0$ for any $n \in \N^*$, the function $s^n \phi(s)$ is not integrable on $\RR^+$ for $n \geq 1$.

If the measure $\phi$ is non-negative then $\lim_{x \to 0^+} f^{(n)}(x) < \infty$ does imply that $\phi$ has a finite $n$-th moment. This is because if $\phi \geq0$ one can use the monotone convergence theorem, which implies that $\lim_{x \to 0^+} e^{-s x}$ is integrable with respect to the measure $\phi$ (cf. \cite[Proposition 1.2]{SRV}). But for signed measures the monotone convergence fails --- see \cite[p. 177]{BogachevBook}.
\hfill$\blacksquare$
\end{remark}

Even more generally, on can consider the Laplace transform of distributions in a suitable class\index{Laplace transform!of a distribution} (cf. \cite{DijkBook,Estradabook,ZemanianBook}). Let $\Dcl$\index{_zDe@$\Dcl$} be the space of smooth compactly supported functions on $\RR$ endowed with the standard topology of uniform convergence with all derivatives (cf. \cite[Definition 2.1]{DijkBook}), with $\Dclp$ \index{_zDep@$\Dclp$} denoting its dual. Furthermore, let $\Scl$ \index{_zS_s@$\Scl$} denote the space of rapidly decreasing (Schwartz) functions on $\RR$ and its dual $\Sclp$ \index{_zS_sp@$\Sclp$} -- the space of tempered distributions. We have $\Dcl \subset \Scl \subset \Sclp \subset \Dclp$.\\
Recall that a distribution $T$ is said to be null on some open set $U \subset \RR$ if $\<T,\varphi\> = 0$ for all $\varphi$ with $\supp \varphi \subset U$. \\ The \index{distribution!support of} \emph{support of a distribution} $T$ is the complement of its null set (i.e. the union of all $U$, on which $T$ is null) (cf. \cite[Definition 2.5.2]{Estradabook}). \\We shall denote by $\Dclpp$\index{_zDepp@$\Dclpp$} the set of all \emph{right-sided} distributions, i.e. distributions in $\Dclp$\index{distribution!support of}, the support of which is contained in $\RR^+$. Consequently, $\Sclpp = \Dclpp \cap \Sclp$.\index{_zS_spp@$\Sclpp$}

Consider now a smooth function $\varphi$ on $\RR$, which has bounded support on the left (i.e. $\varphi(x) = 0$ for all $x < m$ for some $m \in \RR$) and equals 1 in a neighbourhood of $[0,\infty)$. Then, for any $T \in \Sclpp$, we define
\begin{align}\label{Laplace_distribution}
\Lc[T](x) \vc \< T, s \mapsto \varphi(s) e^{-s x} \>, \quad \text{ for } x>0.\index{_zLx1@$\Lc[T]$}
\end{align}
The definition is sound as for any $x > 0$, the function $s \mapsto \varphi(s) e^{-s x}$ is in $\Scl$ and since $T$ is supported on $\RR^+$ the choice of $\varphi$ does not play any role.\\
 It is common to use the --- somewhat sloppy --- notation $\Lc[T](x) = \<T_s, e^{-s x} \>$.

The range of the Laplace transform of distributions in $\Sclpp$ is uncovered by the following result:

\begin{theorem}[{\cite[Theorem 9]{Zemanian1}}]
\label{thm:Laplace_inv}
A function $f$ is a Laplace transform of a distribution in $\Sclpp$ if and only if

i) $f$ is analytic in the half-plane $\Re(s) > 0$;

ii) There exists a polynomial $p$ with $\abs{f(s)} \leq p(\abs{s})$ for any $s\in \CC$ with $\Re(s) > 0$.
\end{theorem}

\begin{example}\label{ex:Laplace_inv_exp}
For any $a \in \RR^+$ we have $\Lc[\delta_a](x) = e^{-ax}$. In particular, $\Lc[\delta](x) = 1$. Since $\delta_a$ is a legitimate measure on $\RR^+$, $\Lc[\delta_a]$ is a c.m. function.\\
Let now $p(x) = \sum_{j=0}^d c_j\,x^j$ for some $d \in \N$ and let $f(x) = p(x)\, e^{-ax}$ with $a \in \RR^+$. Then, $f$ is not c.m. unless $d = 0$. Nevertheless, its inverse Laplace transform exists and equals $\Lc^{-1}[f] = \sum_{j=0}^{d} c_j\, \delta_a^{(j)}$.
\hfill$\blacksquare$
\end{example}

\begin{remark}
\label{rem:Gauss_Laplace}
It would be highly desirable to utilise the Laplace transform to deduce the behaviour of $\Tr e^{-tH^2}$ knowing that of $\Tr e^{-tH}$. Unfortunately, the Gaussian function $f(x) = e^{-x^2}$ does not satisfy the bound demanded in Theorem \ref{thm:Laplace_inv} as $\abs{f(x+iy)} = e^{-x^2+y^2}$. Hence $\Lc^{-1}(f)$ does not exist, even as a distribution.

We also note that if the function $f$ has compact support in $\RR^+$ --- in particular if it is a smooth cut-off function as depicted in \cite[Fig. 1]{ConnesUncanny} --- then it is not complex analytic in the right half-plane, and hence cannot be a Laplace transform.
\hfill$\blacksquare$
\end{remark}

\smallskip

The pertinence of the Laplace transform in the context of the spectral action stands from the fact that if $f=\Lc[\phi]$ is the Laplace transform of a measure $\phi$, then for any positive selfadjoint (possibly unbounded) operator $H$, 
\begin{align}
\label{f(H/Lambda)}
f(H/\Lambda)=\int_0^\infty e^{-s\,H/\Lambda} \, d\phi(s).
\end{align}
The RHS of the above formula is well defined in the strong operator sense (cf. \cite[p. 237]{SimonReed2}). Moreover, since the trace is normal (i.e. if $H_{\alpha} \to H \in \B(\H)$ strongly with $H_{\alpha'} > H_{\alpha}$ for $\alpha'>\alpha$, then $\Tr \,H = \sup_{\alpha} \, \Tr \, H_{\alpha}$\,), and, if  $f(H/\Lambda)$ is trace-class, then
\begin{align}
\label{Tr f(H/Lambda)}
\Tr f(H/\Lambda)=\int_0^\infty \Tr e^{-sH/\Lambda} \,d\phi(s),
\end{align}

Formula \eqref{Tr f(H/Lambda)} will allow us to deduce the properties of the spectral action given the corresponding heat trace. Most notably, it will enable us to establish a large energy expansion of $\SA$ in Section \ref{sec:energy}. For our plan to succeed, we need to identify the suitable classes of cut-off functions. We define successively:
\begin{align*}
& \Cl \,\,\, \vc \,\, \{ f = f^+ - f^-, \, f^{\pm} \in \CM \;\; \vert \;\; f(x) \geq 0, \text{ for } x > 0 \}, \\
& \Cl_0\vc \,\, \{ f = \Lc[\phi] \in \Cl \;\; \vert \;\; \forall \, n \in \N ,\;\; \int_{0}^{\infty} s^n d\abs{\phi}(s)  < \infty \}, \\
& \Cl_c \vc \,\, \{ f = \Lc[\phi] \in \Cl \;\; \vert \;\; \supp\, \phi \text{ is compact} \}, \\
& \Cl^p \vc \,\, \{ f \in \Cl \;\; \vert \;\; f^{\pm}(x) = \Oinf(x^{-p}) \}, \quad \text{for } p >0, \\
& \Cl_0^p \vc \,\, \Cl_0 \cap \Cl^p,   \quad \text{for } p >0, \\
& \Cl_c^p \vc \,\, \Cl_c \cap \Cl^p,  \quad \text{for } p >0.
\end{align*} 
\index{_zCl@$\Cl$}\index{_zCl0@$\Cl_0$}\index{_zClc@$\Cl_c$}\index{_zClp&@$\Cl^p$}
\indent \label{f(0)}Observe that, if the operator $H$ is in $\Tp$ then Lemma \ref{traceclass} implies that for $f \in \Cl_0^r$ with $r>p$ we have $\Tr f(H) < \infty$. The property $f(0)<\infty$ is not necessary for $f(H)$ to be trace-class when $H$ is invertible, but $f(0)$ might (and usually does) pop-up as the coefficient in front of $t^0$ in the expansion of $\Tr f(tH)$.

Clearly, $$\CM \subset \Cl,\quad \Cl_c \subset \Cl_0 \quad\text{and} \quad \Cl^r \subset \Cl^p\,\text{ for }\,r > p.$$ Furthermore:

\begin{proposition}
\label{prop:cl}
For any $p,r>0$,  $$\Cl_0^p \cdot \Cl_0^r \subset \Cl_0^{p+r}\quad\text{and}\quad\Cl_c^p \cdot \Cl_c^r \subset \Cl_c^{p+r}.$$
Moreover,
\begin{align}
\label{moments_Lap}
f =\Lc[\phi] \in \Cl_0^p \,\,\,\, \Longrightarrow \,\, \text{for all } m > -p, \;\; \int_0^{\infty} s^{m} \, d\vert\phi\vert(s) < \infty.
\end{align}
\end{proposition}

\begin{proof}
Firstly, we note that since the set $\CM$ is closed under multiplication, so is $\Cl$, and hence $\Cl^p \cdot \Cl^r \subset \Cl^{p+r}$. Secondly, recall that when two signed measures $\phi,\chi$ on $\RR^+$ are finite then their convolution\index{measure!convolution of} $\phi \ast \chi$\index{_Aoast@$\phi \ast \chi$} is defined as \cite[Definition 5.4.2]{Athreya2006}
\begin{align*}
\int_{0}^{\infty} f(s) \,d(\phi \ast \chi)(s) \vc \int_{0}^{\infty}  \int_{0}^{\infty} f(s+t)\, d\phi(s) d\chi(t), \quad \text{ for } \; f \in C_b(\RR^+).
\end{align*}
If $\Lc[\phi], \Lc[\chi] \in \Cl_0$, then $\phi$ and $\chi$ are finite and, moreover, for any $n \in \N$,
\begin{align*}
\int_{0}^{\infty} s^n d\vert\phi \ast \chi\vert(s) & = \int_{0}^{\infty}\int_{0}^{\infty}   (s+t)^n d\aphi(s) d\vert\chi\vert(t)  \\
&= \sum_{k=0}^{n} \tbinom{n}{k} \, \int_{0}^{\infty}  s^k\,d\aphi(s) \, \int_{0}^{\infty} t^{n-k}\, d\vert\chi\vert(t) < \infty.
\end{align*}
From the general properties of Laplace transform \cite[Theorem 10]{Zemanian1} we have \linebreak$\Lc[\phi \ast \chi](x) = \Lc[\phi](x) \Lc[\chi](x)$, hence we conclude that $$\Cl_0 \cdot \Cl_0 \subset \Cl_0 \,\,\text{ and }\,\,\Cl_c \cdot \Cl_c \subset \Cl_c.$$

Formula \eqref{moments_Lap} for $m \in (-p,0)$ follows from Fubini's theorem:
\begin{align*}
 \int_0^{\infty} s^{m} d\vert\phi\vert(s) & =\Gamma(-m) ^{-1} \int_0^{\infty} x^{-m-1} \int_0^{\infty} e^{-s x} \, d\vert\phi\vert(s) \, dx \\
 & =\Gamma(-m) ^{-1}\int_0^{\infty} x^{-m-1} \left(f^+(x) + f^-(x)\right) dx <\infty.
\end{align*}
On the other hand, for $m \geq 0$ the statement follows from the very definition of $\Cl_0$:
\begin{align*}
\int_0^{\infty} s^m d\aphi(s) \leq \int_0^{1} s^{\lfloor m \rfloor} d\aphi(s) + \int_1^{\infty} s^{\lceil m \rceil} d\aphi(s) \leq \int_0^{\infty} ( s^{\lfloor m \rfloor} + s^{\lceil m \rceil} ) d\aphi(s).
\tag*{$\Box$}
\end{align*}
\end{proof}

A pleasant consequence of Proposition \ref{prop:cl} is that if we have at our disposal a cut-off function $f \in \Cl_0^r$ for some $r >0$, then we immediately obtain a cut-off function in $\Cl_0^p$ for any $p \geq r$ by taking $f^n$ for $n\in\N,\, n \geq p/r$.

\begin{example}
\label{ex:ep-ax}
For any $a >0$ the c.m. function $f(x) = e^{-ax}=\Lc(\delta_{a})(x)$ is in $\Cl_c^p$ for all $p>0$. Thus, when $b>a>0$, the function $f(x) = e^{-ax} - e^{-bx}$ is in $\Cl_c^p$ for any $p>0$, while it is not in $\CM$.
\\If $b > a>0$ then also $f(x) = x^{-1} ( e^{-ax} - e^{-bx} )=\Lc( \chi_{[a,b]})(x)$ is in $\Cl_c^p$ for $p>0$.
\hfill$\blacksquare$
\end{example}

\begin{example}
\label{ex:(ax+b)-r}
Let $f(x) = (a x + b)^{-r}$ for $a,\,b,\,r>0$. Then, one computes that $\Lc^{-1}[f] (s)= \Gamma (r)^{-1} \, a^{-r} \,s^{r-1} e^{-b s/a}$ and hence, $f \in \Cl_0^r$, but $f \notin \Cl_c$.
\hfill$\blacksquare$
\end{example}

\begin{example}
Let $f(x) = e^{x^2} [ 1 - \Erf(x) ]$, where the error function $\Erf$ is given by $\Erf(x) = 2\pi^{-1/2} \int_0^x e^{-y^2} dy$. Then, $f\in \CM$ with $\Lc^{-1}[f](s) =\pi^{-1/2} e^{-s^2/4}$. Moreover $f^+(x) = f(x) = \Oinf(x^{-1})$, hence $f \in \Cl_0^1$.
\hfill$\blacksquare$
\end{example}

\begin{example}
\label{ex:exp_sqrt}
Let $f(x) = e^{-\sqrt{x}}$. Then, $f \in \CM$ since $\Lc^{-1}[f](s) = \frac{1}{2 \sqrt{\pi }} \,s^{-3/2}e^{-1/4 s}$. On the other hand, $\lim_{x \to 0^+} f'(x) = \infty$, thus $f \notin \Cl_0$.\\
Similarly, let $f(x) = \log( a + b / x)$ with $a \geq 1,\, b >0$. Then we have $f\in \CM$ with $\Lc^{-1}[f](s) = s^{-1}(1-e^{-b\, s/a})+\log (a) \,\delta (s)$, but $f \notin \Cl_0$ as $\lim_{x \to 0^+} f(x) = \infty$.
\hfill$\blacksquare$
\end{example}

Every signed measure on $\RR^+$ (including $\sum_i \delta_{i}$) can be seen as a Laplace transformable distribution, but the space $\Lc[\Sclpp]$ is strictly larger, as illustrated in Example \ref{ex:Laplace_inv_exp}. A careful reader might thus ask whether one could extend the class $\Cl$ and allow $f$ to be the Laplace transform of a distribution in $\Sclpp$. At the operatorial level it is possible to cook up an analogue of Formula \eqref{Tr f(H/Lambda)}. However, the topology of $\Sclp$ obliges us to control the derivatives of the test functions and the asymptotic expansions, which are our ultimate objective, do not behave well under differentiation. We come back to this point in Remark \ref{rem:Laplace_dist}.

\section{Weyl's Law and the Spectral Growth Asymptotics}\label{sec:Weyl}

Let $\Delta$ be the Dirichlet Laplacian $\Delta$ on a bounded domain $X \subset \RR^d$. In 1911 Hermann Weyl proved that $$N_\Delta(\Lambda) =\tfrac{\Vol (S^{d-1}) \,\Vol (X) }{(2\pi)^d}\,\Lambda^{d/2}\big(1+\oo(1)\big)\,\text{ as }\,\Lambda\to \infty.$$
This formula, dubbed \emph{Weyl's law}, extends to the case of $X$ being a compact connected Riemannian manifold and $\Delta$ the Laplace--Beltrami operator --- see \cite{Ivrii} for a nice overview. Beyond the commutative world one can deduce the asymptotic behaviour of the spectral growth function of an arbitrary positive operator $H$, via the Wiener--Ikehara Tauberian theorem, from the properties of the associated spectral zeta function $\zeta_H$ (cf. \cite{Aramaki} and references therein, and also \cite{Elizaldebook}). On the other hand, the control of $N_H(\Lambda)$ as $\Lambda$ tends to infinity yields the leading small-$t$ behaviour of the heat trace $\Tr e^{-t H}$ (cf. \cite[Corollary 3]{HeatEZ} \footnote{The cited result provides a better control than Proposition \ref{prop:KH}, but requires some non-trivial assumptions --- see \cite[Section 5]{HeatEZ} for a detailed discussion and (counter)examples.} or \cite{SUZ2017} and references therein).

The control of the sub-leading terms in the asymptotic behaviour of $N_H(\Lambda)$ is a formidable task, even in the context of classical elliptic pdos. In the latter case, the remainder is of the order $\Oinf(\Lambda^{(d-1)/2})$, but in full generality of noncommutative geometry one can only expect $\oinf(\Lambda^{d/2})$. The lower order terms, which are of great interest from the viewpoint of the spectral action, are accessible (and will be accessed in Chapter \ref{chap:asymptotic}) provided we trade the sharp cut-off $\chi_{[0,1]}$ for a gentler one.

Let now $\ahd$ be a $d$-dimensional spectral triple with a simple dimension spectrum $\Sd = d -\N$ and $\zeta_D$ regular at the origin. Then, one defines \cite[Eq. (25)]{ConnesUncanny}
\begin{align}
\label{Navr}
\langle N_\DD (\Lambda)\rangle \vc \sum_{k= 1}^{d} \tfrac{1}{k}\,\Lambda^{k}  \ncint \abs{D}^{-k} + \zeta_{D}(0 ) +\dim\ker \DD.
\end{align}
This could be justified via Formula \eqref{SA_simple}~\footnote{Strictly speaking, Formula \eqref{SA_simple} requires $f$ to be the Laplace transform of a signed measure, which is not the case for the counting function. Nevertheless, naively $\int_0^1 u^{k-1}\,du = 1/k$.} using $N_{\DD}(\Lambda) = \langle N_\DD (\Lambda)\rangle + \Oinf(\Lambda^{-1})$.

The computation of \eqref{Navr} is illustrated as follows \cite[Proposition~1]{ConnesUncanny}:

\begin{proposition}
\label{sa average}
Given $\DD$, assume that $\spec D= \Z^*$ and that the total multiplicity of eigenvalues $\{\pm n\}$ is $P(n)$ with a polynomial  $P(u)=\sum_{j=0}^{d-1} c_j\,u^j$. \\Then, $$\langle N_\DD (\Lambda)\rangle =\int_0^\Lambda P(u)\,du + \sum_{j=0}^{d-1} c_j\,\zeta(-j) +\dim\ker \DD.$$
\end{proposition}

\begin{proof}
The zeta function reads $\zeta_D(s)=\tfrac{1}{2}\sum_{n\in \Z^*} P(n) \,n^{-s}=\sum_{n=1}^\infty P(n) \,n^{-s}$, thus $$\zeta_D(s)=\sum_{j=0}^{d-1} c_j\, \zeta(s-j)\, \text{ and }\,\ncint \abs{D}^{-k} = \Res_{s=k}\, \zeta_D(s) = c_{k-1}.$$ Moreover we have $\zeta_D(0)=\sum_{j=0}^{d-1} c_j\, \zeta(-j)$ completing the proof with 
\begin{align*}
\int_0^\Lambda P(u)\,du  = \sum_{j=0}^{d-1} \tfrac{c_j}{j+1}\Lambda^{j + 1}.
\tag*{$\Box$}
\end{align*}
\end{proof}

This result will be extended in Theorem \ref{res-int} in Appendix \ref{app: residues of series} and applied in the computation of the dimension spectrum of noncommutative tori in Section~\ref{subsec:NC torus dim sp}.

\section{Poisson Summation Formula}
\label{sec:Poisson}

Given the Fourier transform\index{Fourier transform}\index{_zFz@$\F[g]$} of a function $g \in L^1(\RR^m)$, 
\begin{align*}
y\in \RR^m \mapsto \F[g](y)\vc \int_{\RR^m}g(x)\,e^{-i\,2\pi \,x.y}\,dx,
\end{align*}
the usual Poisson formula \index{Poisson formula}
\begin{align}
\label{Poisson}
\sum_{k\in \Z^m} g(t\, k+\ell)=t^{-m}\,\sum_{k \in \Z^m} e^{i\,2\pi t^{-1} k.\ell}\,\,\F[g](t^{-1}\,k), \quad \forall\, t>0,\,\forall\, \ell \in \RR^m, 
\end{align}
is valid under mild assumptions for the decay of $g$ and $\F[g]$ \cite[VII, Corollary 2.6]{SteinWeiss}. 

Formula \eqref{Poisson} is deeply rooted in complex analysis. We have (if, e.g., $g \in \Scl$)
\begin{align}
\label{eq:zetagamma}
\zeta(s)\,\M[g](s)=\int_0^\infty \,\sum_{k=1}^\infty g(t\,k)\,t^{s-1}\,dt=\M\big[\sum_{k=1}^\infty g(\cdot\,k)\big](s),\quad \text{ for } \; \Re(s)>1,
\end{align}
what can be deduced from $\int_0^\infty g(t\,k)\, t^{s-1}\,dt=k^{-s}\int_0^\infty g(t) \,t^{s-1} \,dt$. In particular, a proof of the Poisson formula (at least for $m=1$) can be given by applying the inverse Mellin transform to \eqref{eq:zetagamma} and invoking the Riemann functional equation~\cite{Feauveau}.

The Poisson summation formula is particularly useful for the spectral action computations if one knows explicitly the singular values $\mu_n(\DD)$ along with the multiplicities. In favourable cases (which include i.a. the classical spheres and tori) when $\mu_n(\DD)$ are indexed by $\Z^m$ for some $m\in\N$ one can rewrite the spectral action as 
$\Tr f(\abs{\DD}/\Lambda)=\sum_n M_n(\abs{\DD})\,f(\mu_n/\Lambda) = \sum_{k\in \Z^m} g_\Lambda(k)$, 
for some function $g_\Lambda$. Then, the next lemma shows that 
\begin{align}
\label{eq:asymptotics ghat(0)}
\Tr f(\abs{\DD}/\Lambda)= \F[g_\Lambda](0) +\OO_\infty(\Lambda^{-\infty}).
\end{align}
In other words: Given a Schwartz function $g$, the difference between $\int_{\RR^m} g(tx)$ and $\sum_{\Z^m} g(tx)$ is negligible as $t \downarrow 0$.

\begin{proposition}[Connes]
\label{prop:seriesintegral}
Given a function $g \in \Scl(\RR^m)$, define
\begin{align*}
&S(g)\vc\sum_{k\in \Z^m} g(k), \quad I(g) \vc\int_{\RR^m} g(x)\,dx, \quad g_t(x)\vc g(tx) \text{ for any }t>0.
\end{align*}
Then, $S(g_t)=I(g_t)+\OO_0(t^\infty)$.
\end{proposition}

\begin{proof} 
Since $t^{-m}\, \F[g](0)=I(g_t)$, the Poisson formula \eqref{Poisson} implies that it suffices to show that for any large $n\in \N^*$, $$S(g_t)-I(g_t)=\sum_{k\in \Z^m \backslash \{0\}} \F[g](t^{-1}\,k) =\Oz(t^n).$$ 
But $\F[g]\,$ is also a Schwartz function, $\vert \F[g](x) \vert \leq c \vert x\vert ^{-n}$ for $x\neq 0$ and every $n\in \N$, so for $n$ large enough, 
\begin{align*}
\big{\vert} \sum_{k\in \Z^m \backslash \{0\}}  \F[g](t^{-1}\,k) \big{\vert} \leq c \sum_{k\in \Z^m \backslash \{0\}} \vert t^{-1}\,k \vert^{-n} =c' t^n.
\tag*{$\Box$}
\end{align*}
\end{proof}

As an example, this shows that 
\begin{align}
\label{eq: gaussian series}
\sum_{k\in \Z^m}  e^{-t \norm{k}^2} = \int_{\RR^m} e^{-t\,\norm{x}^2}\,dx +\OO_0(t^\infty)= \pi^{m/2}\, t^{-m/2}+\OO_0(t^\infty).
\end{align}
The smoothness of $g$ is indispensable in Proposition \ref{prop:seriesintegral}: For instance, the function $g(x) = e^{-t \,\abs{x}}$ has a Fourier transform $\F[g](k) =  \frac{2\,t}{4\pi^2 k^2+t^2}$ and $\int_{-\infty}^{\infty} e^{-t \,\abs{x}} \,dx = \tfrac{2}{t}$, while $$\sum_{k=-\infty}^{\infty} e^{-t \,\abs{k}} = \tfrac{e^t+1}{e^t-1} =\tfrac{2}{t}+2\sum_{n=1}^\infty \tfrac{B_{2n}}{ (2n)!}\,t^{2n-1}$$ using Bernoulli numbers (cf. p. \pageref{eq:Euler Maclaurin}).

Proposition \ref{prop:seriesintegral} does not give a way to compute the asymptotics of $I(g_t)$ or $S(g_t)$. But it is hiddenly used in the computation of the heat trace asymptotics of a Laplace type operator $P$ on a compact manifold as in Appendix \ref{classical tools}: The asymptotics of $\Tr e^{-t\,P}$ is not computed directly via $\Tr e^{-t\,P}=\sum_{n=0}^\infty e^{-t\,\mu_n(P)}$ but, as reminded in Appendix \ref{sec: The operator exp(-iP) as a pdo}, using the integral kernel of $e^{-t\,P}$ and an approximation of the resolvent $(P-\lambda)^{-1}$ by a pseudodifferential operator $R(\lambda)$, the parametric symbols of which $r_n(x,\xi,\lambda)$ authorise the use of integrals over $x$ and $\xi$ (see \cite[Section 1.8]{Gilkey1}) without any reference to the discrete summation $\sum_n$.

\medskip

Let us now illustrate the usefulness of the Poisson summation formula for the spectral action computations:\pagebreak

\begin{example}\label{ex:Poisson_sphere}
If $\DD$ is the Dirac operator on the sphere $S^d$, then by Formula \eqref{eigenvalues_Sd}, 
$$\Tr f(\abs{\DD}/\Lambda)=\sum_{n\in \N} 2^{\floor{d/2}+1} \tbinom{n+d-1}{d-1}\,f((n+\tfrac{d}{2})/\Lambda).$$
Let us set $d=3$, as in \cite{ConnesUncanny}. Assuming that $f$ is an even Schwartz function and checking that $\pm1/2$ are not in the spectrum, we get, via Equation \eqref{Poisson},
\begin{align*}
\Tr f(\abs{\DD}/\Lambda) &= 2\sum_{n\in \N} (n+1)(n+2) f((n+3/2)/\Lambda)=\sum_{k\in \Z} k(k+1) f((k+1/2)/\Lambda)\\
& =\sum_{k\in \Z} g_\Lambda(k+1/2)  = \sum_{k\in \Z} (-1)^k \,\F[g_\Lambda](k),
\end{align*}
where $g_\Lambda(x)\vc(x-1/2)(x+1/2)f(x/\Lambda)$. As already seen in Formula \eqref{eq:asymptotics ghat(0)}, we only need to compute $\F[g_\Lambda](0)$ and the asymptotics of the action is
\begin{align}
\label{SA:S3}
\Tr f(\abs{\DD}/\Lambda) = \Lambda^3\int_\RR x^2 f(x)\,dx -\tfrac14 \Lambda \int_\RR f(x)\,dx+\OO_\infty(\Lambda^{-\infty}).
\quad \hspace{0.5cm} \blacksquare
\end{align}
\end{example}
\begin{example}
\label{ex:Poisson_tori}
Consider the spectral triple $(C^\infty(\TT^d),L^2(\TT^d,\SS),\Dslash)$ and assume again that $d=3$. By \eqref{eq:eigenvalues tori}, the spectrum of $\Dslash$ is the set of all values $\pm2\pi\,\Vert k+\ell\Vert$ where $k$ varies in $\Z^3$ while $\ell\in \RR^3$ is fixed and given by the chosen spin structure $(s_1,\dotsc,s_d)$. There is no degeneracy, apart from zero which is a double eigenvalue. Again, to simplify we assume that $f\in \Scl(\RR)$ and, as computed in \cite{Marcolli1},
\begin{align*}
\Tr f(\abs{\DD}/\Lambda)&=2\sum_{k\in \Z^3} g_\Lambda(k) = 2\sum_{k\in \Z^3} \F[g_\Lambda](k) =2\F(g)(0)+\OO_\infty(\Lambda^{-\infty}),
\end{align*}
with $g_\Lambda(x)\vc f(2\pi\norm{x+\ell}/\Lambda)$. Since $\F[g_\Lambda](0)=\int_{\RR^3} f(2\pi\norm{x+\ell}/\Lambda)\,dx$, we get
\begin{align}
\label{action 3-torus}
\Tr f(\abs{\DD}/\Lambda)=\tfrac{1}{4\pi^3}\Lambda^3\int_{\RR^3} f(\norm{x})\,dx +\OO_\infty(\Lambda^{-\infty}). \hspace{2.5cm}\blacksquare
\end{align}
\end{example}
In Example \ref{ex:Poisson_tori}, one can see that the spectrum of $\DD$ does depend on the choice of the spin structure and so does the spectral action. But this choice disappears in \eqref{action 3-torus}, so that the dependence on the spin structure is asymptotically negligible. Is this a general fact? See the associated Problem \ref{Beyond the asymptotic expansion} in Chapter \ref{chap:open}.

\medskip

The Poisson formula has been widely employed in the computation of the spectral action, as explained in Section \ref{sec:sa}. It requires, however, a complete knowledge of the spectrum of $\DD$. The latter is available for a few Dirac operators on Riemannian manifold, like the spheres, tori, Bieberbach manifolds, quotients of Lie groups, like $SU(n)$, etc. The reader should notice that the calculations can be tricky, especially when swapping between $f$ and $g$ in \eqref{eq:asymptotics ghat(0)}.

Although the Poisson summation formula is exact, we focused on its application for the study of asymptotics. Of course, since every summation formula is associated with an asymptotic formula \cite{Estradabook}, one can investigate the Taylor--Maclaurin or Euler--Maclaurin asymptotic formulae. For instance, the latter reads \cite{ConnesFLRW}, for any $2 \geq m \in \N$
\begin{align}
&\sum_{k=0}^N g(k) =\!\int_0^N \!\!g(x)\,dx +\tfrac{[g(0)+g(N)]}{2} +\!\sum_{j=2}^m \tfrac{B_j}{j!} \,[g^{(j-1)}(N)-g^{(j-1)}(0)]  + R_m , \label{eq:Euler Maclaurin} \\
&\hspace{-0.35cm}\text{with }\, R_m \vc\tfrac{(-1)^{m+1}}{m!}\int_0^N g^{(m)}(x)\, B_m(x-\floor{x})\,dx, \nonumber
\end{align}
where $B_j=-2j\zeta(1-2j)$ are the Bernoulli numbers, with $B_{2j+1}=0$ for $j\in \N^*$ so that  $B_2=-\tfrac16,\,B_4=-\tfrac{1}{30},B_6=\tfrac{1}{42}, \ldots$, and $B_j(x)$ are the Bernoulli polynomials defined by induction (cf. \cite[Chap.~23]{Abram}):  $$B_0(x)=1,\quad B'_j(x)=j\,B_{j-1}(x),\quad\int_0^1 B_j(x)\,dx=0.$$
When $N\to \infty$ and $7\leq m\in \N$,
$$\abs{R_m}\leq \tfrac{2}{(2\pi)^m}\,\zeta(m)\int_0^\infty \vert g^{(m)}(x)\vert \,dx \leq \int_0^\infty \vert g^{(m)}(x)\vert \,dx.$$

\begin{example}
\label{ex:Euler_sphere}
Consider the Dirac operator on the sphere $S^4$, as in \cite{ConnesFLRW}. According to Appendix \ref{sec:spheres}, we have $\{\mu_n(\abs{\Dslash}\,\vert\, n\in \N\}=\{k\in \N\}$ and the eigenvalue $k$ has multiplicity $\tfrac{4}{3} (k^3-k)$ (thus $0,\,1$ are not eigenvalues). The Euler--Maclaurin formula \eqref{eq:Euler Maclaurin} can be used for $\Tr f(\abs{\DD}/\Lambda)=\tfrac43 \sum_{k=0}^\infty g(k)$, with $g(x)=(x^3-x)f(x/\Lambda)$. Assuming again $f\in \Scl$, we get $R_m=\OO_\infty(\Lambda^{4-m})$ because $g^{(m)}(\Lambda u)=\Lambda^{3-m}P_m(u,\Lambda^{-1},f^{(k)}(u))$, where $P_m(u,\Lambda ^{-1},f^{(k)}(u))$ is a polynomial in $u,\,\Lambda^{-1}$ and a finite number of $f^{(k)}(u)$. 
Since $g^{(2j)}(0)=0$ for all $j\in \N$ and 
\begin{align*}
&g'(0)=-f(0),\,g^{(3)}(0)=6f(0)-6\Lambda^{-1}f'(0),\\
&g^{(5)}(0)=120\Lambda^{-1}f'(0)-60\Lambda^{-2}f''(0)
\end{align*}
etc., we obtain
\begin{align*}
\tfrac34 &\Tr f(\abs{\DD}/\Lambda)\\
&=\int_0^\infty (x^3-x)\,f(x/\Lambda)\,dx + (\tfrac{1}{12} + \tfrac{1}{120})f(0) 
+ (-\tfrac{1}{120} \Lambda^{-2}-\tfrac{1}{252}\Lambda^{-2})f'(0) +  \dotsc ,
\end{align*}
which means
\begin{align}
\Tr f(\abs{\DD}/\Lambda)=  \tfrac43\Lambda^4\,\int_0^\infty u^3\,f(u)\,du&- \tfrac43\Lambda^{2}\int_0^\infty u \,f(u)\,du + \tfrac{11}{90}f(0) \nonumber\\
&+\sum_{m=1}^\infty c_m\,\Lambda^{-2m}\,f^{(m)}(0) +\OO_\infty(\Lambda^{-\infty}). \label{action for S4}
\end{align}

The coefficients $c_m$ can be computed explicitly: $c_1=31/1890$, $c_2=41/7560$, $c_3=-31/11880$, etc.
\hfill $\blacksquare$
\end{example}

Similarly, the Euler--Maclaurin formula can be utilised in the same way when the singular values of $\DD$ are polynomials in $k\in \N$ with polynomial multiplicity in $k$. 

Some authors use the name ``non-perturbative spectral action'' for expressions like \eqref{action for S4} after erasing the remainder $\OO_\infty(\Lambda^{-\infty})$. We warn the Reader that the `exponentially small' term $\OO_\infty(\Lambda^{-\infty})$ can hide the devils, as argued in Section \ref{sec:sa}.

\section{Asymptotic Expansions}\label{sec:asympt}

The bulk of the physical information encoded in the spectral action can be read out from its asymptotic behaviour at large energies. The leading term established in Section \ref{sec:Weyl} provides rather scarce information. 
We would like to recognise the leading, subleading, etc. terms of the action and to control the remainder in a sensible way. It would be most desirable to have a Laurent expansion
\begin{align}\label{SA_Taylor}
\SA = \sum_{k = -N}^{\infty} a_{k}(f,\DD) \Lambda^{-k},
\end{align}
the convergence of which for $\Lambda$ larger than some $\Lambda_{\text{min}}$ would guarantee that considering higher and higher order terms provides a better approximation of the true, non-perturbative, action $\SA$. Unfortunately, we are granted such a level of precision only for rather specific geometries --- see Section \ref{sec:conv_exp}. Typically, the best what one could hope for is a control of the order of the remainder when the series \eqref{SA_Taylor} is truncated to the first, say, $M$ terms. This desire is formalised with the help of the notion of an \emph{asymptotic series}\index{asymptotic series}. It firstly requires the following definition:

\begin{definition}\label{def:asym_scale} 
Let $(\varphi_k)_k$ be a (finite or infinite) sequence of functions on a punctured neighbourhood of $x_0$. We call it an \emph{asymptotic scale}\index{asymptotic scale} at $x_0$ if, for any $k$, $\varphi_k(x) \neq 0$ for $x \neq x_0$ and $\varphi_{k+1}(x)=\oo_{x_0}(\varphi_k(x))$.
\end{definition}

\begin{example}
For any $x_0 \in \RR$ the sequence $\big((x-x_0)^n\big)_{n \in \N}$, known from the Taylor series, is an asymptotic scale at $x_0$. More generally, one could take $(x-x_0)^{r_n}$ for any complex $r_n$'s with $\Re (r_n) \nearrow \infty$.

The logarithmic terms can also be taken into account: For instance, the sequence  $\{x^{-1} \log^2 x, \,x^{-1} \log x,\, x^{-1}, \log^2 x, \,\log x,\, 1, \dotsc \}$ defines an asymptotic scale at $x_0 = 0$.

On the other hand, the sequence $\{\tfrac{\pi}{2}, \,x^{-1} \cos x, \,x^{-2} \sin x, \,x^{-3} \cos x, \dotsc\}$ is \emph{not} an asymptotic scale at $x_0 = \infty$, as, for instance, the limit of $\varphi_3(x) / \varphi_2(x)$ for $x \to \infty$ does not exist.
\hfill$\blacksquare$
\end{example}

\begin{definition}\label{def:asym}
Let $(\varphi_k)_{k \in \N}$ be an asymptotic scale at $x_0$ and $g$ a complex function on some punctured neighbourhood of $x_0$. \\ We say that $g$ has an \emph{asymptotic expansion with respect to $(\varphi_k)_k$}\index{asymptotic expansion} when there exists a sequence of functions $(\rho_k)_{k\in \N}$ such that 
\begin{align*}
&\rho_k(x) = \OO_{x_0}(\varphi_k(x)),\qquad g(x)-\sum_{k=0}^N \,\rho_k(x)=\oo_{x_0}(\varphi_{N}(x)), \quad \text{ for any } N \in \N.
\end{align*}
In this case, we write
\begin{align}
\label{eq:exp}
g(x) \; \underset{x \to x_0}{\sim} \; \sum_{k=0}^\infty \,\rho_k(x)
\index{_ASim1@$\protect\underset{x \to x_0}{\sim} $}
\end{align}
and the symbol $\sum_{k=0}^\infty \rho_k(x)$ is called an \emph{asymptotic series} of $g$ at $x_0$.
\end{definition}

\begin{remark}\label{rem:asymp_exp}
Let us warn the reader that the adopted Definition \ref{def:asym} is sometimes called in the literature an \emph{extended} asymptotic expansion (cf. \cite[Definition 1.3.3]{Estradabook}). In the standard approach (see for instance \cite{Erdelyi}) one assumes that $\rho_k(x) = c_k \varphi_k(x)$ with some complex constants $c_k$. An example of a function which does not have an asymptotic expansion at infinity in this standard sense, but does have one in the extended sense is $S(x) = \int_0^{x} t^{-1} \sin t \, dt$ --- see \cite[Example 8]{Estradabook}.

Observe also that Definition \ref{def:asym} uses a flexible, though somewhat slippery, notation -- two \emph{different} pairs of scale functions $(\varphi_k)_k$ and coefficients $(\rho_k)_k$ can give \emph{the same} expansion \eqref{eq:exp} around $x_0$ for a given function $g$ (cf. Example \ref{ex:Gilkey_pseudodiff}).

In our venture through the meanders of noncommutative geometry we will encounter situations in which the spectral action $\SA$ exhibits oscillations at large energies --- see Example \ref{ex:Podles_SA_asym}. To study its asymptotics, we will need the full force of Definition \ref{def:asym}.
\hfill$\blacksquare$
\end{remark}

Let $g: \RR \to \RR$ be a smooth function around $x=0$. Then $g$ always admits a Taylor expansion $g(x) \sim_{x\to 0}\,\sum_{k=0}^{\infty} \tfrac{g^{(k)}(0)}{k!}\,x^k$. However, if $g$ is not analytic, its Taylor series is only an asymptotic one. As an example consider $g(x) = \int_0^{\infty} e^{-s}  (1 + s^2 x^2)^{-1} ds$. It is smooth at $x=0$ and $g(x) \sim_{x\to 0}\,\sum_{k = 0}^{\infty} (-1)^k(2k)!\,x^{2k}$, but this expansion is clearly not a convergent series.

Let us also note that whereas the functions $\varphi_k$ of the asymptotic scale must be non-zero in some punctured neighbourhood $V$ of $x_0$, the coefficients of the expansion --- i.e. the functions $\rho_k$ --- can actually vanish in the whole $V$. A standard example is provided by the function $h(x)\vc e^{-1/x}$ for $x>0$, $h(x)\vc 0$ for $x\leq0$. The function $h$, being smooth at $x=0$, admits a Taylor expansion, i.e. an asymptotic expansion with respect to the asymptotic scale $(x^k)_k$. But, as $h^{(k)}(0) = 0$ for every $k \in \N$, we have $h(x) \sim_{x\to 0} 0$, i.e. $h(x) = \Oz(x^{\infty})$. Note, however, that we could take $\varphi_k(x) = e^{-1/x^k}$ as the asymptotic scale and write a trivial expansion $h(x) \sim_{x \to 0} e^{-1/x} + \oz(\varphi_N(x))$ for any $N \in \N^*$.

The casus of the latter function $h$ shows that an asymptotic expansion can never determine a function uniquely, even if it converges. Indeed, for any function $g$ smooth at zero, the functions $g$ and $g+h$ will have the same asymptotic expansions with respect to the scale $(x^k)_k$.

\begin{remark}
\label{rem:asympt_diff}
Asymptotic expansions can be combined by linearity and multiplication. They can also be integrated, both with respect to the variable $x$, as well as to some external parameters. On the other hand, asymptotic expansions (even the standard --- ``non-extended'' --- ones) do not behave well under differentiation. This is because the only information available about the remainder is its $\Oinf$ behaviour, which is scarce: Although $f(x) = \OO_{x_0}(g(x))$ but, in general, $f'(x) \neq \OO_{x_0}(g'(x))$, as exemplified by $$f(x) = x \sin(x^{-1}) = \Oz(x),\quad f'(x) = \Oz(x^{-1}) \neq \Oz(1).$$
For more details on the properties of asymptotic series see \cite[Section 1.4]{Estradabook} and references therein.
\hfill$\blacksquare$
\end{remark}

\begin{example}
\label{ex:Gilkey_diff}
As in Example \ref{ex:commutative} let us consider an elliptic selfadjoint {\it differential} operator $H$ of order $m$ acting on a vector bundle $E$ over a closed (i.e. compact without boundary) $d$-dimensional Riemannian manifold $M$ and let $K \in C^{\infty} (\End(E))$ be an auxiliary smooth endomorphism. If $H$ is positive, then \cite[Theorem 1.3.5]{Gilkey2}
\begin{align}
\label{eq:tr e-tp, P diff}
\Tr K e^{-t\,H} \tzero\,\sum_{k=0}^\infty \,a_k(K,H) \,t^{(k-d)/m}.
\end{align}
The natural asymptotic scale is $\varphi_k(t) = t^{(k-d)/m}$, whereas $\rho_k(t) = a_k(K,H) \,t^{(k-d)/m}$ in accordance with \cite[Definition (1.8.10)]{Gilkey1}.
\hfill$\blacksquare$
\end{example}

\begin{example}
\label{ex:Gilkey_pseudodiff}
More generally, if $P \in \Psi^m(M,E)$ is positive, elliptic with $m>0$, then
\begin{align*}
\Tr e^{-t\,P} \tzero \sum_{k=0}^\infty a_k(P) \,t^{(k-d)/m} + \sum_{\ell=0}^{\infty} b_\ell(P)  \,t^{\ell} \, \log t.
\end{align*}
See Corollary \ref{cor:heat for elliptic pdo} in Appendix \ref{sec: the heat asymptotics of exp(-tP)} for the full story and a detailed proof.\\
Note that this expansion can be read as the one with respect to the asymptotic scale $\{t^{-d/m} \log t, \,t^{-d/m}, \,t^{(-d+1)/m} \log t, \,t^{(-d+1)/m}, \,\ldots \}$, with constant coefficients $c_{2j+1} = a_j(P)$ and $c_{2j} = 0$ for $j<d$, $c_{2j} = b_j(P)$ for $j \geq d$. \\Alternatively, one can set the asymptotic scale $\varphi_j(t) = t^{(j-d)/m - \epsilon}$ for some small $\epsilon >0$ and take
\begin{align*}
\rho_j(t) =  \left\{\begin{array}{ll}
(a_j + b_{(j-d)/m} \, \log t)\, t^{(j-d)/m}&\text{ for }j-d \in m\N,\\
a_j \, t^{(j-d)/m} &\text{ for } j-d \notin m\N, 
\end{array}\right.
\end{align*}
to satisfy $\rho_j(t) = \Oz(\varphi_j(t))$. 

This can be further generalised: If $Q$ is a log-polyhomogeneous pdo of order $q$, 
\begin{align*}
\Tr \,Q \,e^{-t \,P} \tzero \sum_{j=0}^{\infty} \big[ \,\sum_{k=0}^{r+1} a_{j,k}(Q,P)\, \log ^k t\,\big] \,t^{(j-q-d)/m},
\end{align*}
where $r$ is the degree of log-polyhomogeneity of $Q$ (cf. \cite[Section 3]{Lesch}).
\\
In fact, as suggested on p. 168 of \cite{Lesch} (see also \cite[p. 488]{GrubbSeeley}), this result should survive even if $P$ is also log-polyhomogeneous rather than a classical one. 
\hfill$\blacksquare$
\end{example}

\begin{remark}
\label{rem:different operators with same asymp}
Let $P$ be an elliptic positive pseudodifferential and let $Q=P+R$, where $R$ is a smoothing pdo. \\
Then, $(P-\lambda)^{-1}$ and $(Q-\lambda)^{-1}=(P-\lambda)^{-1}[\bbbone-R(Q-\lambda)^{-1}]$ have the same parametrix. Moreover, if $P$ satisfies Hypothesis \ref{hyp: principal symbol with parameter}, so does $Q$. Thus the difference between the two pdo's $e^{-t\,P}$ and $e^{-t\,Q}$ (see Appendix \ref{sec: The operator exp(-iP) as a pdo}) is $\OO_0(t^{\infty})$. This means that $\Tr e^{-t\,P}$ and $\Tr e^{-t\,Q}$ have the same asymptotics. This coincidence also follows from the Duhamel formula \eqref{Duhamel}.
\hfill$\blacksquare$
\end{remark}

Finally, let us present an example of an asymptotic expansion of the spectral action on the standard Podle\'s sphere, which justifies the need of the full force of Definition \ref{def:asym}.

\begin{example}\label{ex:Podles_SA_asym}
Let $\DD_q$ be the Dirac operator on the standard Podle\'s sphere (cf. Appendix \ref{sec:Podles}). For a suitable cut-off function $f$ the spectral action admits a large-energy asymptotic expansion
\begin{align}\label{SA_Podles_asym}
S(\DD_q,f,\Lambda) & \Linf \sum_{k = 0}^{\infty} \,\, \sum_{n = 0}^{2} \,\, \sum_{j \in \Z} \, a_{-2k+ \tfrac{2 \pi i}{\log q}\, j,n} \, \x \\
& \qquad\quad\; \x \sum_{m = 0}^{n} (-1)^{n-m}  \tbinom{n}{m}  \,f_{-2k+ \tfrac{2 \pi i}{\log q}\, j,m} \,(\log \Lambda)^{n-m} \,\Lambda^{-2k+ \tfrac{2 \pi i}{\log q}\, j}. \notag
\end{align}
The numbers $f_{z}$ are the generalised moments of the cut-off function and the $a_{z,n}$'s can be computed explicitly in terms of the residues of $\zeta_{\DD_q}$ --- see \cite{PodlesSA} for the details.

As in Example \ref{ex:Gilkey_pseudodiff} the expansion \eqref{SA_Podles_asym} can be consider either with respect to the asymptotic scale $\{\log^2 \Lambda,\, \log \Lambda, \,1, \,\Lambda^{-2} \log^2 \Lambda, \,\Lambda^{-2} \log \Lambda, \,\Lambda^{-2},\, \ldots \}$ or with respect to $(\Lambda^{\epsilon - 2 k})_{k \in \N}$ for any $\epsilon >0$. In both cases, the coefficients $\rho_k(\Lambda)$ of the asymptotic expansion are rather involved functions of $\Lambda$ given in terms of absolutely convergent Fourier series in the variable $\log \Lambda$. 
\hfill$\blacksquare$
\end{example}

\section{Convergent Expansions}\label{sec:conv_exp}

Surprisingly enough, it turns out that the asymptotic series in Formula \eqref{SA_Podles_asym} is actually convergent for all $\Lambda > 0$ and, moreover, one could write a genuine equality in place of the asymptotic expansion symbol \cite[Theorem 5.2]{PodlesSA}. The following definition provides a general framework for such situations.

\begin{definition}
\label{def:conv}
Let $\sum_{k=0}^\infty \rho_k(x)$ be an asymptotic series of a function $g$ at $x_0$ with respect to an asymptotic scale $(\varphi_k)_k$. \\ We say that the asymptotic expansion of $g$ at $x_0$ is \emph{(absolutely\,/\,uniformly) convergent} \index{asymptotic expansion!absolutely convergent} \index{asymptotic expansion!uniformly convergent} if the series $\sum_{k=0}^\infty \rho_k(x)$ converges (absolutely\,/\,uniformly) in some punctured neighbourhood $V$ of $x_0$. In this case, for $x \in V$,
\begin{equation*}
g(x) = \sum_{k=0}^\infty \rho_k(x) + R_{\infty}(x), \quad \text{ with } \,\, R_{\infty}(x) = \oo_{x_0}(\varphi_k(x)) \,\, \text{ for all } \,\, k \in \N.
\end{equation*}
If, moreover, $R_{\infty} = 0$ then the expansion is called \emph{exact}.\index{asymptotic expansion!exact}
\end{definition}

Let us note that since $\sum_{k=0}^\infty \rho_k(x)$ is not, in general, a Taylor series, the domain of its absolute convergence can be strictly smaller than the domain of conditional convergence --- cf. page \pageref{gen_Dirichlet_series}.

In this light, the expansion presented in Example \ref{ex:Podles_SA_asym} is exact for all $\Lambda >0$. Let us illustrate Definition \ref{def:conv} with two further examples of geometric origin:

\begin{example}
\label{ex:conv_non}
Let $\Dslash$ be the standard Dirac operator on $S^1$ associated with the trivial spin structure (see Appendix \ref{sec:spheres}). Then, $$\Tr e^{-t \Dslash^2} = \sum_{n=-\infty}^{\infty} e^{-t \,n^2} = \vartheta_3(0; e^{-t})\text{  for }t>0$$ where 
$$
\vartheta_3(z; q = e^{i \pi \tau}) = \vartheta_3(z \, \vert \, \tau) \vc \sum_{n=-\infty}^{\infty} q^{n^2} e^{2 n i z}
$$
\index{_a4theta3@$\vartheta_3(z; q)$}for \index{_a4theta3v@$\vartheta_3(z\,\vert\,\tau)$} $z \in \CC$, $\abs{q} < 1$ is a Jacobi theta function (see \cite[p. 464]{WhittakerWatson}) which enjoys the Jacobi identity \cite[p. 475]{WhittakerWatson}
\begin{align}
\label{Jacobi_id}
\vartheta_3(z \, \vert \, \tau) = (- i \tau)^{-1/2} \,  e^{z^2/\pi i \tau} \, \vartheta_3 \left(z\tau^{-1} \, \vert  -\tau^{-1} \right).
\end{align}
Equation \eqref{Jacobi_id} implies that, for any $t>0$,
\begin{align*}
\Tr e^{-t\, \Dslash^2} = (\tfrac{\pi}{t})^{1/2} \, \vartheta_3(0; e^{-\pi^2/t}) =(\tfrac{\pi}{t})^{1/2} + 2 (\tfrac{\pi}{t})^{1/2}  \sum_{n=1}^{\infty} e^{-\pi^2 n^2/t} = (\tfrac{\pi}{t})^{1/2}  + \Oz(t^{\infty}).
\end{align*}
The above asymptotic expansion is finite (thus absolutely and uniformly convergent $\forall t>0$), but it is not exact. In fact, the asymptotic expansion of the heat trace associated with $\Dslash^2$ -- the square of the standard Dirac operator on any \emph{odd} dimensional sphere $S^d$ is finite, but not exact --- see Example \ref{ex:odd_sphere_2} and \cite[Corollary 2]{HeatEZ}.
\hfill$\blacksquare$
\end{example}

\begin{example}
\label{ex:exact}
Let again $\Dslash$ be the standard Dirac operator on $S^1$ associated with the trivial spin structure (see Appendix \ref{sec:spheres}). Then, for $t>0$,
\begin{align}\label{heat_S1}
\Tr e^{-t \,\abs{\Dslash}} = 1 + 2 \sum_{n=1}^{\infty} e^{-t \,n} = \coth \left( \tfrac{t}{2} \right).
\end{align}
The function $\coth$ is complex analytic in the annulus $\abs{s} \in (0,\pi)$ and its Laurent expansion at $s=0$ reads \cite[(4.5.67)]{Abram}: 
\begin{align}
\label{coth}
\coth s = s^{-1} + \sum_{k=1}^{\infty} \tfrac{2^{2k} \,B_{2k}}{(2k)!} \, s^{2k-1},
\end{align}
where $B_k$\index{_zBernoulli@$B_k$} are the Bernoulli numbers (cf. page \pageref{eq:Euler Maclaurin}). Hence, previous equality provides an asymptotic expansion for $\Tr e^{-t \,\abs{\Dslash}}$ as $t \downarrow 0$, which is convergent (absolutely and uniformly on compacts) for $t \in (0,2\pi)$ and, moreover, exact.

On the side we remark that \eqref{heat_S1} actually provides a meromorphic continuation of $\Tr e^{-t \,\abs{\Dslash}}$ to the whole complex plane.
\hfill$\blacksquare$
\end{example}

\section{On the Asymptotics of Distributions}
\label{sec: On the asymptotics of distributions}

As explained in Section \ref{subsec:Laplace}, the Laplace transform technique applied to the spectral action computations restricts the admissible cut-off function. In this section we briefly sketch an alternative approach, via the asymptotics of distributions following \cite{Estradabook} (see also \cite{Distributions,EstradaFulling1999} and \cite[Setion 7.4]{Elements}), which works for $f \in \Scl(\RR)$.

For $m \in \N$, let 
\begin{align*}
\K_m \vc \{f\in C^\infty(\RR^n, \CC)\,\vert\,\,\vert \partial^\alpha f(x)\vert\leq c_\alpha(1+\norm{x})^{m-\vert\alpha\vert}, \,\forall \text{ mult-ind. }\alpha\in \N^n\}
\end{align*}
be equipped with the topology defined by the family of seminorms $$\norm{f}_{k,m} \vc \sup_{x\in \RR^n} \{\max(1,\norm{x}^{k-m}) \,\vert \partial^\alpha f(x)\,\vert \text{ with } \vert \alpha \vert=k\}.
$$
In particular, $f\in \K_m$ implies $\partial^\alpha f \in \K_{m-\vert\alpha\vert}$.\\ Now, let $\K$ \index{_zKh1@$\K$}  be the inductive limit of the spaces $\K_m$ as $m$ tends to infinity. \\
Remark that all polynomials are in $\K$ and the Schwartz space $\Scl$ is dense in $\K$ so that we can consider the dual $\K'$ as a subset of tempered distributions.\\
We remark that Proposition \ref{prop:seriesintegral} still holds true for $f\in \K(\RR^n)$ and such that $\int_{\RR^n} f(x) dx$ is defined --- see \cite[Lemma 2.10]{Distributions}. More importantly, a distribution $T \in\Dclp(\RR)$ is in $\K'(\RR)$ if and only if (see \cite[Theorem 6.7.1]{Estradabook})
\begin{align}
\label{eq:asymptoticdistribution}
T(\lambda x) \,\underset{\lambda \to \infty}{\sim} \,\sum_{n=0}^\infty \tfrac{(-1)^n\,\mu_n}{n!\,\lambda^{n+1}}\,\delta^{(n)} (x),\, \text{with } \mu_n \vc \langle T,\,x^n \rangle \text{ -- the moments of } T
\end{align}
in weak sense: $$\text{For any }f\in \K, \langle T(\lambda x),\,f(x) \rangle = \sum_{n=0}^N \tfrac{\mu_n}{n!\,\lambda^{n+1}}f^{(n)}(0) + \OO_\infty(\lambda^{-N-2}) \text{ for any }N \in \N^*.$$

Notice that such a \emph{moment asymptotic expansion} of a distribution harmonises with the notion of the asymptotic expansion pondered in Section \ref{sec:asympt}, as for $f\in \K$, we have $\langle T(t^{-1} x), f(x) \rangle \tzero \sum_{n=0}^\infty \tfrac{\mu_n}{n!}\, f^{(n)}(0)\, t^{n+1}$.

Formula \eqref{eq:asymptoticdistribution} is directly related to the Ces\`aro behaviour of distributions (cf. \cite[Section 6.3]{Estradabook}): Given a distribution $T\in \Dclp(\RR)$ and $\beta \in \RR\setminus( -\N^*)$ one writes
$$
T(x)\vc\OO_\infty(x^\beta) \quad (C)
$$
\index{_zOozzz@$\OO_\infty(x^\beta) \quad (C)$}if there exist $N\in \N$, a primitive $T_N$ of $T$ of order $N$  and a polynomial $p$ of degree at most $N-1$, such that $T_N$ is locally integrable for large $x$ and we have, in the usual sense, $T_N(x)=p(x) + \Oinf(x^{N+\beta})$. Similarly, one defines $T(x)=\oinf(x^\beta)\quad (C)$ and, consequently, $T(x) \vc \oinf(x^{-\infty})\quad (C)$ means that $T(x) \vc \oinf(x^{-\beta})\quad (C)$ for all $\beta$. \\
It turns out that $T \in \K'$ if and only if $T(x) = \oo_{\pm \infty}(\vert x\vert^{-\infty})\quad (C)$ (cf. \cite[Theorem 6.7.1]{Estradabook}), which shows that the asymptotic behaviour \eqref{eq:asymptoticdistribution} at infinity is in fact equivalent to the Ces\`aro behaviour.

A good illustration of these notions is the following (see \cite[Example 164]{Estradabook}). 

\begin{example}
A distribution $S\in \Sclp(\RR)$ is said to be distributionally smooth at $0$ when $S^{(n)}(0)$ exists (in the sense of Stanis\l aw {\LL}ojasiewicz) for $n\in \N$ --- a property equivalent to the fact that its Fourier transform is distributionally small: $\F(S)\in \K'$.\\
Let now $\xi\in \CC$ with $\abs{\xi} = 1$ and assume first that $\xi\neq 1$. Let $T\in \K'$ be defined as $T(x)=\sum_{n\in \Z} \xi^n\,\delta_n(x)$. We have $\mu_k=\sum_{n\in \Z} \xi^n\,n^k=0\quad (C)$. Thus, if $S\in \Sclp$ is distributionally smooth at $0$, then $\sum_{n\in \Z} \xi^n S(nx)=\oo_0(x^\infty)$ in $\Dclp$. \\ But when $\xi=1$, the Dirac comb $\sum_{n\in \Z}\delta_n$ is not in $\K'$. Nevertheless, $\sum_{n\in \Z}\delta_n-1 \in \K'$ and, still in $\Dclp$: 
$$\sum_{n\in \Z} S(n\,x)=x^{-1} \int_\RR S(u)\,du+\oo_0(x^\infty),$$
provided that $\int_\RR S(u)\,du$ is defined.
\hfill $\blacksquare$
\end{example}

To justify this incursion within distributions, let us consider, as in \cite[Section~6.16]{Estradabook}, a (possibly unbounded) selfadjoint operator on $\H$. Its spectral decomposition reads: $H = \int_{-\infty}^{\infty} \lambda \,d P_{\lambda}(H)$ and determines the \emph{spectral density}\index{spectral density} $$d_H(\lambda) \vc \tfrac{d P_{\lambda}(H)}{d \lambda}$$ understood as a distribution from $\Dclp$ valued in $\B(\Dom H,\H)$. \\ Given any function $f \in \Dcl$ one gets $f(H)= \langle d_H(\lambda),\,f(\lambda)\rangle_\lambda$ and it is natural to use the following notation $$d_H(\lambda)\cv\delta(\lambda-H).$$ 
Observe now that, with $X\vc \cap_{n=1}^\infty \Dom \,H^n$, we have $\langle d_H(\lambda),\,\lambda^n\rangle = H^n$, which means that actually $\delta(\lambda-H) \in \K'(\RR,\B(X,\H))$. Hence, Formula \eqref{eq:asymptoticdistribution} yields $$\delta(\lambda \sigma-H) \,\underset{\sigma \to \infty}{\sim} \,\sum_{n=0}^\infty \tfrac{(-1)^n}{n!\,\sigma^{n+1}}\,H^n\,\delta^{(n)}(\lambda)$$ and, moreover, $\delta(\lambda-H)$ vanishes to infinite order at $\pm \infty$ in the Ces\`aro sense, namely $$\delta(\lambda-H) = \oo_{\pm \infty}(\vert \lambda \vert^{-\infty}) \quad (C).$$

We now take a detour through \cite{Distributions}: Let $T\in \Sclpp(\RR)$ and {\it assume} the following Ces\`aro asymptotic expansion:
\begin{align}
\label{hyp: asymptotic ex}
T(\lambda)  \,\underset{\la \to \infty}{\sim} \, \sum_{n=1}^\infty c_n \la^{\alpha_n}+\sum_{j=1}^\infty b_j \la^{-j} \quad (C),
\end{align}
where $\alpha_n \in \RR\backslash -\N$ constitute a decreasing sequence. Then, by \cite[Theorem 32]{Estradabook},
\begin{align*}
T(\la\Lambda)  \,\underset{\Lambda \to \infty}{\sim} \,\sum_{n=1}^\infty c_n\,(\la_+ \Lambda)^{\alpha_n } +\sum_{j=1}^\infty b_j \la^{-j} 
\text{Pf}[(\la\Lambda)^{-j}\,H(\la)]+\sum_{n=0}^\infty \tfrac{(-1)^n\,\mu_n\,\delta^{(n)}(\la)}{n!\,\Lambda^{n+1}},
\end{align*}
where $\mu_n= \langle T(x) -\sum_{n=1}^\infty c_n\,x_+^{\alpha_n } -\sum_{j=1}^\infty b_j\, \text{Pf }[x^{-j}\,H(x)],\,x^n\rangle$ are the ``generalised moments''. Here,  $\mathrm{Pf}$ is the Hadamard finite part: $$\langle \text{Pf }[h(x)],\,f\rangle \vc \text{F. p. }\int_0^\infty h(x)\,f(x)\,dx$$ and the index $+$ means that the integral defining the duality is restricted to $\RR^+$. Thus,
\begin{align*}
\langle \text{Pf }[x^{-j}\,H(x)],\, f(x)\rangle &=\text{F. p.}\int_0^\infty f(x)\, x^{-j}\,dx\\
&  \hspace{-0.7cm}=\int_1^\infty f(x)\,x^{-j} \,dx + \int_0^1[f(x)-\sum_{n=0}^{j-1} \tfrac{f^{(n)}(0)}{n!}\,x^n]\,x^{-j}\,dx-\sum_{n=0}^{j-2}\tfrac{f^{(n)}(0)}{k!\,(j-n-1)}\,
\end{align*}
and, consequently, $\text{Pf }[x^{-j}\,H(x)]$ is homogeneous of degree $-j$ up to a logarithmic term: 
$$\text{Pf }[(\la \Lambda)^{-j}\,H(\la \Lambda)]= \Lambda^{-j}\,\text{Pf }[\la^{-j}\,H(\la)]+ \tfrac{(-1)^j\delta^{(j-1)}(\la)}{(j-1)!\,\Lambda^j}\,\log \Lambda.$$
 Finally, 
\begin{align}\label{eq:development for heat}
 \langle T(\la), f(t\la)\rangle_\la \tzero & \; \sum_{k=1}^\infty \,c_k\,t^{-\alpha_k-1}\,\text{F. p. } \int_0^\infty \la^{\alpha_k}\,f(\la)\,d\la  \\
& + \sum_{j=1}^\infty b_j\,t^j\,\big[ \text{F. p. }\int_0^\infty \tfrac{f(\la)}{\la^j}\,d\la -\tfrac{f^{(j-1)}(0)}{(j-1)!}\,\log t \big] + \sum_{n=0}^\infty \tfrac{\mu_n\,f^{(n)}(0)}{n!}\,t^n. \nonumber 
\end{align}

Formula \eqref{eq:development for heat} can be utilised to obtain an asymptotic expansion of the spectral action computation, at least in the commutative case:\\
Let $P$ be a positive elliptic pdo of order $k>0$ on a $d$-dimensional compact manifold $M$. Then, the Schwartz kernel $d_P(x,y;\lambda)$ of the operator $\delta(\lambda-H)$ enjoys the following Ces\`aro expansion on the diagonal (cf. \cite[Formula (4.5)]{Distributions}):
\begin{align}
\label{eq:diagonal exp}
d_P(x,x;\lambda) \,\underset{\lambda \to \infty}{\sim} \,\sum_{n=0}^\infty a_n(x)\, \lambda^{(d-k-n)/k} \quad (C),
\end{align}
where $$\int_M a_0(x) \,d\mu(x)=\tfrac{1}{k} \,\WRes \,P^{-d/k},$$ see \eqref{WRes}. The existence of such an expansion is based on the following Ansatz: With $\sigma$ denoting the total symbol,
\begin{align*}
\sigma[\delta(\la-P)] \,\underset{\lambda \to \infty}{\sim} \,\sum_{n=0}^\infty c_n\,\delta^{(n)}(\la-\sigma(P)).
\end{align*}

It is important to stress that the asymptotics \eqref{eq:diagonal exp} cannot, in general, be integrated term by term in $\lambda$ (see \cite{Distributions}). But we are free to integrate over $x \in M$: Given an $f \in \Scl(\RR)$, we write
\begin{align*}
S(P,f,\Lambda) = \int_M \langle d_{P}(x,x;\la), f(\la \Lambda^{-1})\rangle_\la \,dx
\end{align*}
and plug in Formula \eqref{eq:development for heat} with \eqref{eq:diagonal exp} to obtain an asymptotic expansion of $S(P,f,\Lambda) $ for large $\Lambda$. 

In particular, with $f(x) = e^{-x}$ (thus $f\in \K$), one can recover as in \cite[Corollary 6.1]{Distributions} the celebrated heat kernel expansion for an elliptic pdo, which we derive from the scratch in Corollary \ref{cor:heat for elliptic pdo} (with a generalisation to pdos with non-scalar symbols).

\begin{example}
Let $(\A,\H,\Dslash)$ be as in Example \ref{ex:commutative} with $\dim M = 4$. For the generalised Laplacian $P=\Dslash^2$ we get
$$
d_{\Dslash^2}(x,x;\lambda)\, \underset{\lambda \to \infty}{\sim} \, c_0 \,\lambda + c_1(x) \quad (C),
$$
which gives 
\begin{align*}
\Tr f(\Dslash^2/\Lambda^2)\, \underset{\Lambda \to \infty}{\sim} \,c_0 \,\Lambda^4 \int_0^\infty d\lambda\, \lambda\,f(\lambda) &+c_2 \,\Lambda^2 \int_0^\infty d\lambda\, f(\lambda) \\
&+ \sum_{n=0}^\infty (-1)^n \,f^{(n)}(0)\,c_{2n+4} \,\Lambda^{-2n}
\end{align*}
(cf. \cite[p. 247]{Distributions}). Such an expansion works for $f\in \Scl$. If we would attempt to use $f=\chi_{[0,1]}$ --- the counting function --- we would discover that the expansion is not valid after the first term. Nevertheless, the full expansion, beyond the first term is valid in the Ces\`aro sense; see also \cite{EstradaFulling1999} for more details.
\hfill $\blacksquare$
\end{example}

It would be highly desirable to implement this approach beyond the commutative world --- see Problem \ref{expansion and distribution} in Chapter \ref{chap:open}.

%
%
 \bibliographystyle{spmpsci}
 \bibliography{SA_BIB}
%

%
%
%

\chapter{Analytic Properties of Spectral Functions}
\label{chap:asymptotic}

\abstract{In the previous chapter we have witnessed the interplay between the spectral zeta functions and the associated heat traces (cf. Proposition \ref{prop:hk_Mellin}). We have also learned, in Section \ref{subsec:Laplace}, how to exploit the Laplace transform to compute the spectral action from a given heat trace. In this chapter we further explore the connections between the spectral functions unravelling the intimate relationship between the meromorphic continuation of a zeta function and the asymptotic expansion of the corresponding heat trace. We utilise the latter to establish the sought asymptotic expansion of the spectral action at large energies. Finally, we ponder the possibility of obtaining convergent, rather than only asymptotic, formulae for this action.}
\medskip

Let us first fix some notations. Let $f$ be a meromorphic function on some domain $W \subset \CC$. For any subset $V \subset W$ we shall denote by $\PP(f,V)$\index{_zPPV@$\PP(f,V)$} the set of poles of $f$ contained in $V$. If $f$ is meromorphic on $\CC$ we abbreviate $\PP(f) \vc \PP(f,\CC)$\index{_zPPC@$\PP(f)$}.

As in the previous chapter, let $H\in \Tp$ and let $K \in \B(\H)$. For further convenience we adopt the following notation:
\begin{align*}
\UKH(s) \vc \Gamma(s)\, \zKH(s), \quad \text{for } \Re(s) > p.\index{_zzKH@$\UKH(s)$}
\end{align*}
In accordance with the conventions adopter after Formula \eqref{zetaKH}, we keep the simplified notation $\UKH$ for $H$ non-positive but enjoying $\vert H \vert \in \Tp$, and set $$\ZZ_H \vc \ZZ_{\bbbone,H}.$$

If $\zKH$ admits a meromorphic extension to a larger region of the complex plane, then so does $\UKH$ and we keep denoting the meromorphic extension of $\UKH$ by the same symbol. \\
Recall also that $\Gamma$ is holomorphic on $\CC\setminus(-\N)$ and has first order poles at $-n \in \N$, with $\Res_{s=-n} \Gamma(s) = (-1)^n(n!)^{-1}$. Hence, $\PP(\UKH) = \PP(\zKH) \cup (-\N)$.

\section{From Heat Traces to Zeta Functions}\label{sec:hk2zeta}

As illustrated by Examples \ref{ex:Gilkey_diff}, \ref{ex:Gilkey_pseudodiff} in the realm of pseudodifferential operators in classical differential geometry we are vested with an asymptotic expansion of the heat trace $\hKH$ for small $t$. It turns out that such an expansion can be utilised to establish a meromorphic extension of the corresponding zeta function $\zKH$. As an example, we quote the following result:

\begin{theorem}[{\cite[Theorem 1.12.2]{Gilkey1}}]
Let $K$ and $H$ be differential operators as in Example \ref{ex:Gilkey_diff}. Then, the function $\UKH$ admits a meromorphic extension to $\CC$ with (possibly removable) simple poles at $s = (d-k)/m$ for $k \in \N$ and
\begin{align*}
\Rez{s=(d-k)/m}\, \UKH(s) = a_{k}(K,H).
\end{align*}
\end{theorem}

This agreeable interplay, which essentially relies on the Mellin transform as captured by Proposition \ref{prop:hk_Mellin}, extends to the setting of abstract operators on a separable Hilbert space. In this context the asymptotic expansion of $\hKH$ would, in general, be more involved including terms in $\log^k t$ and/or proportional to $t^{i}$ --- exhibiting log-periodic oscillations as $t$ tends to 0.

\begin{theorem}
\label{thm:hk2zeta}
Let $H \in \Tp$ and let $K \in \B(\H)$. Assume that there exist $d \in \N$, a sequence $(r_k)_{k \in \N} \subset \RR$ strictly increasing to $+ \infty$ and a discrete set $X \subset \CC$ without accumulation points, such that
\begin{align}
\label{hk_general_exp}
\hKH \tzero \sum_{k=0}^{\infty} \rho_k(t),\quad \text{ with } \quad \rho_k(t) \vc \sum_{z \in X_k} \,\big[\sum_{n=0}^{d} a_{z,n}(K,H) \log^n t\, \big]\,t^{-z},
\end{align}
where $X_k \vc \{z \in X \, \vert \, - r_{k+1} < \Re(z) < - r_{k} \}$, $X = \cup_{k} X_k$ and the series defining $\rho_k(t)$ is absolutely convergent for any $t>0$ and any $k \in \N$. \\
Then, the function $\zKH$ admits a meromorphic extension to the whole complex plane with the poles of order at most $d+1$ and $\PP(\UKH) \subset X$.\\
Moreover, for any $z \in X$ and $n \in \{0,1,\ldots,d\}$,
\begin{align}\label{zKH_res}
\Rez{s=z}\, (s-z)^{n} \,\UKH(s) = (-1)^n n!\, \, a_{z,n}(K,H).
\end{align}
\end{theorem}
Before we prove the theorem, three comments concerning the notation are in order:

Firstly, the set $X_k = \{z_i\}_{i \in \Z}$ comes endowed with a lexicographical order: 
$$
z_i \preceq z_j\text{ iff }\Im(z_i) < \Im(z_j)\text{ or }\Im(z_i) = \Im(z_j)\text{ and }\Re(z_i) \leq \Re(z_j),
$$
what makes it isomorphic --- as a poset --- to $\Z$. Thus, the symbol $\sum_{z \in X_k} f(z)$ should be read as $\lim_{J \to \infty} \sum_{j=-J}^J f(z_j)$. Since the assumed absolute convergence of the series implies that the limit does not depend on the grouping or permutation of terms, the notation $\sum_{z \in X_k}$ is unambiguous.

Secondly, the asymptotic expansion \eqref{hk_general_exp} should be understood with the respect to the asymptotic scale $(t^{r_k})_k$, i.e. we should have $\rho_k(t) = \Oz(t^{r_k})$ (recall Definition~\ref{def:asym} and Remark \ref{rem:asymp_exp}). The latter is not an additional assumption and it follows from the assumed form of $\rho_k$, as we shall see in the course of the proof.  

Thirdly, the sequence $(r_k)_k$ is needed only to fix an asymptotic scale and can be adjusted at will, provided the absolute convergence of the resulting $\rho_k$'s remains unharmed. In particular, since $\max\{ \Re(z) \, \vert \, z \in X\} = p$ we can choose $r_0$ arbitrarily close to $-p$ (see Figure \ref{fig:ht2zeta1}). What matters in the final result \eqref{zKH_res} are the coefficients $a_{z,n}(K,H)$ and not how we distribute them into $\rho_k$'s (cf. Remark \ref{rem:asymp_exp} and Exam\-ple~\ref{ex:Gilkey_pseudodiff}) and the poles of $\UKH$ depend only on $X$ and not on its partition by $X_k$'s.

\begin{figure}[ht]
\begin{center}
\begin{tikzpicture}
%
%
\draw[-latex,thick] (-4,0) -- (4,0);
\draw[-latex,thick] (0,-4) -- (0,4);
\node[] at (4,-0.3) {$\Re(z)$};
\node[] at (0,4.2) {$\Im(z)$};
%
%
%
\filldraw (1.25945, -3.13148) circle (1pt); 
\filldraw (2.0582, 2.14066) circle (1pt);
\filldraw (1.67479, 2.80959) circle (1pt);  
\filldraw (2.60021, -2.17227) circle (1pt); 
\filldraw (2.54572, 1.46555) circle (1pt);
\filldraw (2.55604, 2.70598) circle (1pt); 
\filldraw (1.31823, 1.97422) circle (1pt);
\filldraw (2.62783, -2.37889) circle (1pt); 
\filldraw (1.84381, -0.924021) circle (1pt); 
\filldraw (1.57486, 1.59567) circle (1pt);
\filldraw (1.76111, -3.57055) circle (1pt);
\filldraw (2.69955, 1.33859) circle (1pt);
\filldraw (1.50166, 3.56093) circle (1pt);
%
%
\filldraw (0.734873, -2.222) circle (1pt); 
\filldraw (0.725679, -1.88359) circle (1pt);
\filldraw (0.37699, -1.94107) circle (1pt);
\filldraw (0.186981, 0.312932) circle (1pt);
\filldraw (-0.306402, -2.65595) circle (1pt);
\filldraw (-1.93375, -0.878248) circle (1pt); 
\filldraw (-0.158759, 3.49133) circle (1pt); 
\filldraw (0.892321, 1.08486) circle (1pt); 
\filldraw (0.211636, -1.10176) circle (1pt); 
\filldraw (-0.42055, 0.0254455) circle (1pt);
\filldraw (-1.87595, -1.35347) circle (1pt);
\filldraw (-0.687024, -3.74293) circle (1pt);
\filldraw (-1.3, 0) circle (1pt);
\filldraw (0, 2) circle (1pt);
%
%
\filldraw (-2.37965, 0.977474) circle (1pt); 
\filldraw (-2.81933, 0.34137) circle (1pt);
\filldraw (-2.6512, 1.29105) circle (1pt);  
\filldraw (-2.21576, 2.47371) circle (1pt); 
\filldraw (-2.81974, -2.34804) circle (1pt);
\filldraw (-2.88071, 3.71907) circle (1pt); 
\filldraw (-2.82043, 1.80348) circle (1pt);
\filldraw (-2.7722, -3.28512) circle (1pt); 
\filldraw (-2.29309, -2.80882) circle (1pt); 
\filldraw (-2.04921, -2.14728) circle (1pt);
\filldraw (-2.94202, 2.03689) circle (1pt);
\filldraw (-2.12347, 0.593084) circle (1pt);
\filldraw (-2.56238, -2.94058) circle (1pt);
\filldraw (-2.30414, -3.74829) circle (1pt);
\filldraw (-2.91118, -0.231629) circle (1pt);
%
%
\filldraw (-3.97142,-2.17682) circle (1pt); 
\filldraw (-3.10642, -3.08173) circle (1pt);
\filldraw (-3.72845, 3.68643) circle (1pt);  
\filldraw (-3.60357, 1.86766) circle (1pt); 
\filldraw (-3.83147, -0.360911) circle (1pt);
\filldraw (-3.33506, 2.75464) circle (1pt); 
\filldraw (-3.9993, -0.151525) circle (1pt);
\filldraw (-3.1085, 3.00419) circle (1pt); 
\filldraw (-3.52734, 0.612302) circle (1pt); 
\filldraw (-3.72299, -2.86216) circle (1pt);
\filldraw (-3.7,3) circle (1pt);
%
%
\draw[dashed,thick] (3.5,-4) -- (3.5,4);
\node[] at (3.5,-4.2) {$-r_0$};
\draw[dashed,thick] (1,-4) -- (1,4);
\node[] at (1,-4.2) {$-r_1$};
\draw[dashed,thick] (-2,-4) -- (-2,4);
\node[] at (-2,-4.2) {$-r_2$};
\draw[dashed,thick] (-3,-4) -- (-3,4);
\node[] at (-3,-4.2) {$-r_3$};

\end{tikzpicture}
\end{center}
\caption{\label{fig:ht2zeta1}Illustration explaining the notations in Theorem \ref{thm:hk2zeta}: The discrete set $X \subset \CC$ is partitioned into a disjoint sum, $X = \sqcup \, X_k$, with the help of the vertical lines $\Re(z) = -r_k$.}
\end{figure}
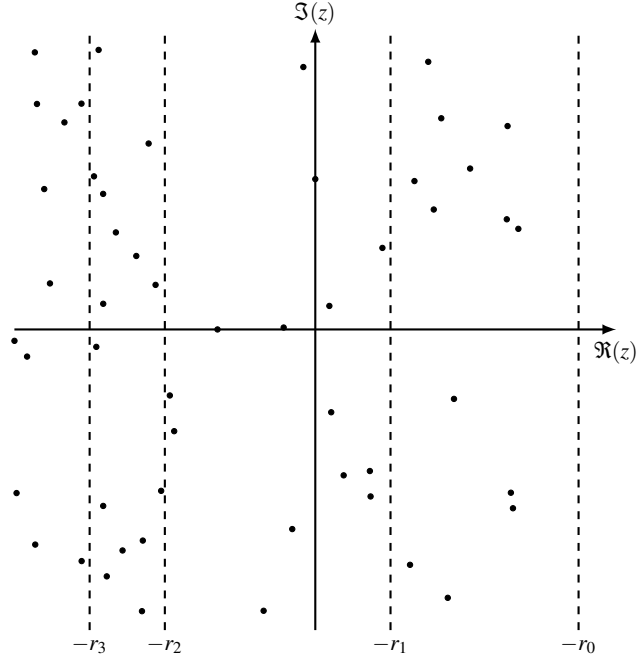

\begin{proof}
For $\Re(s) > p$ we have 
$\zKH(s) = \tfrac{1}{\Gamma(s)} \int_{0}^{\infty} t^{s-1} \, \hKH \, dt$ on the strength of Proposition \ref{prop:hk_Mellin}. For such $s \in \CC$ we can split $\zKH(s) = F_0(s) + F_1(s)$, with
\begin{align*}
F_0(s) \vc \tfrac{1}{\Gamma(s)} \int_{0}^{1} t^{s-1} \, \hKH \, dt, \qquad F_1(s) \vc \tfrac{1}{\Gamma(s)} \int_{1}^{\infty} t^{s-1} \, \hKH \, dt.
\end{align*}
Since $\hKH$ is smooth on $(0,\infty)$ and $\hKH = \Oinf( e^{-\lambda_0(H)\, t})$, the function $F_1$ is actually holomorphic on $\CC$ since $\int_0^\infty t^{s-1}\,e^{-\lambda_0\,t}\,dt=\Gamma(s)\, \lambda_0^{-s}$.

Note first that Proposition \ref{prop:KH} implies $\rho_0(t) = \Oz(t^{-p-\delta})$ for every $\delta >0$, hence $\rho_0(t) = \oz(t^{r_0})$, and we now show that actually $\rho_k(t) = \oz(t^{r_k})$ for every $k \in \N$.\\
To see this let us fix the index $k \in \N$ and let us pick a $t_0 > 0$. As the series over $X_k$ is absolutely convergent for $t>0$, we have, when $t \leq t_0$ and  $n \in \{0,\ldots,d\}$,
\begin{align*}
t^{-r_{k}} \big{\vert} \sum_{z\in X_k} a_{z,n}(K,H) \,t^{-z} \big{\vert} \leq \sum_{z\in X_k} \abs{a_{z,n}(K,H)}\, t^{c_k} \leq \sum_{z\in X_k} \abs{a_{z,n}(K,H)} t_0^{c_k},
\end{align*}
where $c_k \vc - \sup_{z \in X_k} ( r_{k} + \Re (z) ) > 0$.\\Hence, for any $\varepsilon > 0$ there exists $t_0 = [ \varepsilon/\sum \vert a_{z,n} \vert ]^{1/c_k} > 0$ such that for every $t \leq t_0$, $t^{-r_{k}} \big{\vert} \sum_{z \in X_k} a_{z,n}(K,H)\, t^{-z} \big{\vert}  \leq \varepsilon$, 
so that $$\lim_{t \to 0} \; t^{-r_{k}} \vert \sum_{z \in X_k} a_{z,n}(K,H)\, t^{-z} \vert = 0.$$
Since $r_{k} < - \Re(z)$ for $z \in X_k$, we can find $\epsilon' > 0$ such that $r_{k} + \epsilon' < - \Re(z)$ and $\vert \sum_{z \in X_k} a_{z,n}(K,H)\, t^{-z} \vert = \oz(t^{r_{k} + \epsilon'})$. \\
Hence, as $\log^n t = \Oz(t^{-\epsilon'})$ for any $\epsilon'> 0$, $\abs{\log t}^n \, \vert \sum_{z \in X_k} a_{z,n}(K,H)\, t^{-z} \vert = \oz(t^{r_{k}})$ for $n \in \N$. Thus indeed, $\rho_k(t) = \oz(t^{r_{k}})$.

Let us now fix $N \in \N$ and invoke the asymptotic expansion \eqref{hk_general_exp} to provide the meromorphic continuation of $F_0$ to the half-plane $\Re(s) > - r_{N}$. For $\Re(s) > p$,
\begin{align*}
F_0(s) = \sum_{k=0}^{N} \, \int_0^{1} t^{s-1}\, \rho_k(t) \, dt + E_N(s),
\end{align*}
where 
\begin{align*}
E_N(s) \vc \int_{0}^{1} t^{s-1} \,R_N(t) \,dt \,\,\text{ and }\,\, R_N(t) \vc \hKH - \sum_{k=0}^{N} \rho_k(t).
\end{align*}
But since $R_N(t) = \oz(t^{r_{N}})$,  $E_N(s)$ is actually holomorphic for $\Re(s) > - r_{N}$. On the other hand, for $\Re(s) > -r_0 > p$, $\int_0^{1} t^{s-1} \, \rho_k(t) \,dt$ is absolutely convergent for any $k \in \N$ (and so is the series defining $\rho_k$), so that we can swap the integral with the series
\begin{align}
\int_0^{1}t^{s-1}\, \rho_k(t) \, dt &= \sum_{z \in X_k}\,\, \sum_{n=0}^{d} a_{z,n}(K,H) \int_0^{1} t^{s-z-1}\,\log^n t\,  \,\,dt \notag\\
& = \sum_{z \in X_k}\, \,\sum_{n=0}^{d} a_{z,n}(K,H)\, \tfrac{(-1)^n n!}{(s-z)^{n+1}}. \label{zk_mero}
\end{align}

Now, we claim that, the function $s \mapsto \sum_{z \in X_k} \, a_{z,n}(K,H)\, (s-z)^{-n-1}$ 
is holomorphic on $\CC \setminus X_k$ for any $k \in \N, n \in \{0,1,\ldots,d\}$, and hence Formula \eqref{zk_mero} provides a meromorphic extension of the complex function $s \mapsto \int_0^{1}t^{s-1}\, \rho_k(t) \, dt$ to the whole complex plane. Indeed, for any $s\notin X_k$ there exists $\delta>0$ such that $\abs{s-z} \geq \delta$ for all $z\in X_k$, since $X_k$ does not have accumulation points. For any such $s$ we have $$\big{\vert} \sum_{z \in X_k} a_{z,n}(K,H)\, (s-z)^{-n-1} \big{\vert} \leq \delta^{-n-1}\sum_{z \in X_k} \abs{a_{z,n}(K,H)} <\infty.$$

Summarising, for any fixed $N \in \N$ we have
\begin{align*}
\zKH(s) = \tfrac{1}{\Gamma(s)} \sum_{k=0}^{N} \,\sum_{z \in X_k} \,\sum_{n=0}^{d} a_{z,n}(K,H)\, \tfrac{(-1)^n n!}{(s-z)^{n+1}} + E_N(s) + F_1(s),
\end{align*}
which is a meromorphic function for $\Re(s) > -r_{N}$. As $N$ can be taken arbitrarily large and $1/\Gamma$ is meromorphic on $\CC$, we conclude that $\zKH$ is also meromorphic on $\CC$. Moreover, the above reasoning shows that the function $s \mapsto \Gamma(s)\zKH(s)$ can only have poles in $X$ of order at most $d+1$ and \eqref{zKH_res} follows immediately.
\hfill$\Box$
\end{proof}

As illustrated in Figure \ref{fig:ht2zeta1} (see also, for instance, \cite[Figure 2.10]{Lapidus}) the set of poles of a spectral zeta function $\zKH$ does not need, in general, to exhibit any symmetries. However, the geometric origin of the operators $K$ and $H$ can force a highly regular shape of the set $X$.

\begin{example}
\label{ex:PDO_ht2zeta}
In particular, as follows from Example \ref{ex:Gilkey_pseudodiff}, if $K$ and $H$ are classical pseudodifferential operators on a Riemannian manifold $M$ then $\zKH$ admits a meromorphic extension to the whole complex plane with isolated simple poles located in $(\dim M + \ord K - \N)/\ord H \subset \RR$ (cf. \cite{GrubbSeeley}). 

If $K$ is allowed to be a polyhomogeneous pdo, then the poles of $\zKH$ need no longer be simple, but are still contained in the highly regular discrete subset of the real line $(\dim M + \ord K - \N)/\ord H$ --- see \cite{Lesch}.
\hfill$\blacksquare$
\end{example}

\begin{example}
\label{ex:NCtorus_zeta}
Let us illustrate an application of Theorem \ref{thm:hk2zeta} in the context of the noncommutative torus (see Appendix \ref{sec:torus}):

\begin{proposition}
\label{prop:NC_torus}
Let $(\A_{\Theta},\H,\DD)$ be the spectral triple of the noncommutative $d$-torus. \\
Then, for any $a \in \A_{\Theta}$, the spectral zeta function $\zeta_{a,D}$ admits a meromorphic extension to the whole complex plane with a single simple pole located at $s=d/2$. Moreover, $\Res_{s=d/2} \, \zeta_{a,D} = \tau(a) 2^{\lfloor d/2 \rfloor} \pi^{d/2} \Gamma(d/2)^{-1}$ and $\zeta_{a,D}(0) = 0$.
\end{proposition}
\begin{proof}
In order to apply Theorem \ref{thm:hk2zeta} we will show that the corresponding heat trace $\Tr \, a \, e^{-t \,\DD^2}$ admits an asymptotic expansion as $t \downarrow 0$. We have
\begin{align*}
\Tr \, a \, e^{-t \,\DD^2} &=\sum_{k\in \Z^d}\sum_{j=0}^{2^{\lfloor d/2 \rfloor}} \langle U_k\otimes e_j,\,a e^{-t\,\DD^2} \,U_k\otimes e_j \rangle \\&= \sum_{k\in \Z^d}\sum_{j=0}^{2^{\lfloor d/2 \rfloor}}\, e^{-t\norm{k}^2} \langle U_k, a\,U_k\rangle\, \langle e_j,\, e_j \rangle =2^{\lfloor d/2 \rfloor}  \sum_{k\in \Z^d} \tau(U_k^*aU_k)\,e^{-t\norm{k}^2}  \\
&=2^{\lfloor d/2 \rfloor} \,\tau(a) \sum_{k\in \Z^d} e^{-t\norm{k}^2} =\tau(a) \, \Tr \, e^{-t \,\DD^2}.
\end{align*}
Invoking Proposition \ref{prop:seriesintegral} with $g(x)=e^{-\norm{x}^2}$ we obtain
\begin{align*}
2^{-\lfloor d/2 \rfloor} \Tr \,a \,  e^{-t \,\DD^2} & =\tau(a) \, \sum_{k\in \Z^d} \,e^{-\Vert t^{1/2}k \Vert^2}
=\tau(a) \, \sum_k g(t^{1/2} k)\\
& = \tau(a) \, \int_{x\in \RR^d} e^{-\Vert t^{1/2} x\Vert^2}\,dx + \Oz(t^{\infty}) =  \tau(a) \,t^{-d/2} + \Oz(t^{\infty}).
\end{align*}
Since $\dim \Ker \DD = 2^{\lfloor d/2 \rfloor}$ and $\Tr \, P_0 \, a = \tau(a) 2^{\lfloor d/2 \rfloor}$ thanks to \eqref{eq:Pzero}, Formula \eqref{heat_ker} gives:
 $\Tr \,a \,  e^{-t \,D^2} = \tau(a) \, [ \Tr \, e^{-t \,\DD^2} + (e^{-t} - 1) 2^{\lfloor d/2 \rfloor}]$ 
and hence
\begin{align*}
\Tr \,a \,  e^{-t \,D^2} & =\, \tau(a) 2^{\lfloor d/2 \rfloor} \big( \pi^{d/2} \,t^{-d/2} + (e^{-t} - 1) +\OO_0(t^\infty)\big) \\
& \!\!\!\tzero\!\! \tau(a) 2^{\lfloor d/2 \rfloor} \pi^{d/2} \,t^{-d/2}+\tau(a)(e^{-t}-1)\dim\Ker \DD.
\end{align*}
Since the above expansion does not contain any $\log t$ terms, Theorem \ref{thm:hk2zeta} yields a meromorphic extension of $\zeta_{a,D}$ to the whole complex plane with simple poles only. To locate the poles we read out the non-zero coefficients from Formula \eqref{hk_general_exp}
\begin{align*}
a_{d/2,0}(a,D^2) =  2^{\lfloor d/2 \rfloor} \pi^{d/2}\,\tau(a), && a_{-k,0}(a,D^2) = 2^{\lfloor d/2 \rfloor}\, \tfrac{(-1)^k}{k!}\, \tau(a), \, \text{ for } k \in \N^*.
\end{align*}
Thanks to Formula \eqref{zKH_res} we have $$\Res_{s=d/2} \, \zeta_{a,D}(s) =  2^{\lfloor d/2 \rfloor} \pi^{d/2} \Gamma(d/2)^{-1}\,\tau(a).$$
Finally, since $a_{0,0}(a,D^2) = 0$ and $\Res_{s=0} \Gamma(s) = 1$, we conclude from \eqref{zKH_res} that $\zeta_{a,D}(0) = 0$.
\hfill$\Box$
\end{proof}
Remark that no Diophantine restriction on the matrix $\Theta$ (see Definition \ref{ba}) was needed in the previous result --- in contrast with the more general Theorem \ref{zeta(0)}.
\hfill$\blacksquare$
\end{example}

As mentioned on p. \pageref{p:fractals_dimsp}, the complex poles of the function $\zKH$ can appear, for instance, when fractal geometry is involved. We shall see an explicit example in Section \ref{sec:conv}.

Let us now turn to the converse of Theorem \ref{thm:hk2zeta}.

\section{From Zeta Functions to Heat Traces}\label{sec:zeta2hk}

In the context of pseudodifferential operators in classical Riemannian geometry one has at one's disposal the powerful existence theorems about the asymptotic expansion of heat traces. From these one can construct the meromorphic continuations of the corresponding spectral zeta functions via Theorem \ref{thm:hk2zeta}. The original proof for classical pdos (see Appendix \ref{classical tools}) heavily relies on the symbolic calculus and integral kernels of pdos  \cite{Gilkey1}. Unfortunately, these tools are not available in the noncommutative realm. On the other hand, one can hope that given a meromorphic extension of a spectral zeta function $\zKH$, guaranteed, for instance, by the dimension spectrum property, one could deduce the existence (and the form) of an asymptotic expansion of $\hKH$. This is indeed possible, however, one needs to control not only the local structure of $\zKH(s)$ around the poles, but also its asymptotic behaviour on the verticals -- as $\abs{\Im(s)} \to \infty$.

\subsection{Finite Number of Poles in Vertical Strips}

Let us first handle the case when $\zKH$ has a finite number of poles in vertical strips, i.e. for any $U_{a,b} = \{ z \in \CC \, \vert \, a < \Re(z) < b \}$ the set $\PP(\zKH,U_{a,b})$ is finite. This is always the case when $K$ and $H$ are classical pdos, but it often occurs also beyond the classical geometry --- for instance in the case of noncommutative torus (cf. Appendix~\ref{sec:torus}) or the $\suq$ quantum group \cite{ConnesSU2,ILS}.

\begin{theorem}
\label{thm:zeta2hk_finite}
Let $H\in \Tp$ and let $K \in \B(\H)$. Assume that:

\begin{enumerate}[label=\textit{\roman*)},leftmargin=1cm]
\item \label{zeta2hk_finite:assum1} The function $\zKH$ admits a meromorphic extension to $\CC$ with a finite number of poles in vertical strips and of order at most $d$.

\item \label{zeta2hk_finite:assum3} For any fixed $x \leq r_0$ there exists an $\epsilon(x) > 0$, such that 
\begin{align}
\label{ukh_vert}
\UKH(x+iy) = \Oinf(\abs{y}^{-1-\epsilon(x)}).
\end{align}
\end{enumerate}
Fix a sequence $(r_k)_{k \in \N} \subset  \RR$, increasing to $+ \infty$ with $p < -r_0$, giving a partition of $\PP(\UKH)$ as $\PP(\UKH) = \sqcup_k X_k$ with  $$X_k \vc \{ z \in \PP(\UKH) \, \vert \, - r_{k+1} < \Re(z) < - r_{k} \}.$$

Then, there exists an asymptotic expansion with respect to the scale $(t^{r_k})_k$,
\begin{align}
\label{heat_exp_finite}
&\hspace{-0.3cm}\hKH \tzero \sum_{k=0}^{\infty} \, \rho_k(t) , \\
&  \text{ with } \hspace{0.6cm} \rho_k(t) \vc  \sum_{z \in X_k} \,\big[\sum_{n=0}^{d} a_{z,n}(K,H) \log^n t\, \big]\,t^{-z} \label{rho_k2}\\
& \hspace{1.4cm} a_{z,n}(K,H) \vc \tfrac{(-1)^n}{n!} \, \,\Rez{s=z}\, (s-z)^{n} \,\UKH(s). \label{azn}
\end{align}
\end{theorem}

\begin{proof}
We slice the complex plane with the help of the chosen sequence $(r_k)_k$ as shown on Figure \ref{fig:ht2zeta1}. Since $\zKH$ (and hence $\UKH$) has a finite number of poles in vertical strips, for any $k \in \N$ there exists $Y_k > 0$ such that $\UKH$ is holomorphic for $\abs{\Im(s)} \geq Y_k$. For each $k$ let us denote
\begin{align*}
D^k\vc \text{ rectangle }\{s \in \CC \,\vert \,-r_{k+1}\leq \Re(s)\leq -r_{k}, \; -Y_k \leq \Im(s) \leq Y_k \},
\end{align*}
so that $X_k = \PP(\UKH,\,D^k)$.

Let us fix the index $k$. By construction, the function $\UKH$ is regular at the boundary of any $D^k$. Hence, the residue theorem yields
\begin{equation}
\label{finite_contour_k}
\tfrac1{i2\pi}\int_{\partial D^k}\UKH(s)\,t^{-s}\,ds
=\sum_{z\in X_k} \Rez{s=z} \, \UKH(s) \,t^{-s},
\end{equation}

\vspace{-0.2cm}\noindent where the contour $\partial D^k$ is oriented counter-clockwise. Let us decompose the boundary of the rectangle using
\begin{align*}
I_{V}^+(k) & \vc \int_{-r_{k}- i \,Y_k}^{-r_{k}+ i \,Y_k}\UKH(s)\,t^{-s}\,ds, \\
I_{V}^-(k)  &\vc \int_{-r_{k+1}- i \,Y_k}^{-r_{k+1}+ i \,Y_k}\UKH(s)\,t^{-s}\,ds, \\
I_H^\pm(k) & \vc \int_{-r_{k+1} \pm i \,Y_k}^{-r_{k} \pm i \,Y_k}\UKH(s)\,t^{-s}\,ds,
\end{align*}
so that $$\int_{\partial D^k}\UKH(s)\,t^{-s}\,ds = I_{V}^+(k)-I_H^+(k)-I_{V}^-(k)+I_H^-(k).$$
We are now concerned with the limit of Equation \eqref{finite_contour_k} as $Y_k \to \infty$.

Firstly, by assumption \ref{zeta2hk_finite:assum3}, we have,
\begin{align*}
R_k(t) \vc \tfrac1{i2\pi}\,\lim_{Y_k \to \infty} \,I_{V}^-(k) = \tfrac1{2\pi}\int_{- \infty}^{\infty}\UKH(-r_{k+1}+ i y)\,t^{r_{k+1} - i y}\,dy < \infty,
\end{align*}
for any $k \in \N$ and any $t>0$. We also have
\begin{align*}
R_{-1}(t) \vc \tfrac1{i2\pi}\,\lim_{Y_0 \to \infty}\, I_{V}^+(0) = \tfrac1{2\pi}\int_{-\infty}^{\infty}\UKH(-r_0+ i y)\,t^{r_0 -iy}\,dy = \hKH
\end{align*}
on the strength of Corollary \ref{cor:hk_Mellin_inv}.\\ Moreover, for any $k \in \N$, $R_k(t) = \oz(t^{r_{k+1}})$: Indeed, with $\F$ denoting the Fourier transform, we have 
\begin{align}\label{Rk_order}
R_k(t) \, t^{-r_{k+1}}
&=\tfrac1{2\pi }\int_{-\infty}^\infty \UKH(-r_{k+1}+iy)\,t^{-iy}\,dy \notag \\
&=\tfrac1{2\pi }\F[y\mapsto \UKH(-r_{k+1}+iy)]\, (\tfrac{\log t}{2\pi} ) \underset{t\downarrow 0}{\to}\, 0.
\end{align}
The latter statement is a consequence of the Riemann--Lebesgue Lemma, since the function $y \mapsto \UKH(-r_k+iy)$ is in $L^1(\RR,dy)$ for each $r_k$ by hypothesis \eqref{ukh_vert}.

Secondly, hypothesis \eqref{ukh_vert} guarantees that for any $k \in \N$ and any $t>0$
\begin{align*}
\abs{I_{H}^{\pm}(k)} & \leq \int_{-r_{k+1}}^{-r_k} \vert \UKH(x \pm i \, Y_k) \vert \, t^{-x} dx \\
& \leq \sup_{x \in [-r_{k+1},-r_k]} \vert \UKH(x \pm i \, Y_k) \vert \, \int_{-r_{k+1}}^{-r_k} t^{-x} dx \xrightarrow[Y_k \to \infty]{} 0.
\end{align*}

Finally, let us rewrite the residues more explicitly. By assumption \ref{zeta2hk_finite:assum1} and the fact that the function $\Gamma$ has only simple poles, the function $\UKH$ admits a Laurent expansion $\UKH(s) = \sum_{n=-\infty}^{d+1} (-1)^n n ! \, a_{z,n}(K,H) \,(s-z)^{-n}$ in some open punctured disc with the center at any $z \in \PP(\UKH)$, in accordance with Formula \eqref{azn}. On the other hand,
\begin{align*}
t^{-s}=e^{-z\log t}e^{-(s-z)\log t}=t^{-z}\sum_{n=0}^\infty \tfrac{(-1)^n \log^n t}{n!}(s-z)^n, \text{ for any } \; s,z \in \CC, \; t >0.
\end{align*}
Therefore, $$\Rez{s=z} \, \UKH(s)\, t^{-s} = t^{-z} \,\sum_{n=0}^{d} a_{z,n}(K,H) \log^n t.$$
Hence, for any $k \in \N$ and any $t>0$ Equation \eqref{finite_contour_k} yields
\begin{align}
\label{F_k relation_k}
R_{k-1}(t)- R_k(t) = \rho_k(t).
\end{align}

Since $R_k(t) = \oz(t^{r_{k+1}})$, it follows that $\rho_k(t) = \oz(t^{r_{k}})$ for any $k \in \N$. The latter can also be easily deduced from the form of $\rho_k$ given by Formula \eqref{rho_k2}.

Starting with Formula \eqref{F_k relation_k} for $k=0$ and iterating it $N$ times we obtain the announced asymptotic expansion 
$\hKH = \sum_{k = 0}^N  \, \rho_k(t) + R_N(t)$.
\hfill$\Box$
\end{proof}

\begin{remark}
\label{rem:Gk}
Suppose that, for a chosen sequence $(r_k)_k$, we would require only the Lebesgue integrability of $\UKH$ on the verticals $\Re(z) = -r_k$ in place of the stronger constraint \eqref{ukh_vert} --- as we will do in Theorem \ref{thm:zeta2hk}. Then, the theorem would still hold, but the asymptotic expansion \eqref{heat_exp_finite} might get an additional contribution from the horizontal contour integrals, namely $$\hKH \tzero \sum_{k=0}^{\infty} \, \big( \rho_k(t) + g_k(t) \big)$$ with
\begin{align}
\label{Gk}
g_k(t) = \tfrac{1}{i 2\pi} \lim_{y \to \infty} \int_{-r_{k+1}}^{-r_{k}} \big{[}  \UKH(x+ i y)\,t^{- i y} - \UKH(x- i y)\,t^{ i y} \big{]}\, t^{-x} dx,
\end{align}
which is smooth for $t>0$ and $\Oz(t^{r_k})$ as follows from Formula \eqref{finite_contour_k}.

Let us note that even imposing a suitable decay rate of $\UKH$ on $\Re(z) = -r_k$ is not sufficient to get rid of the $g_k$'s, as the Phragm\'en--Lindel\"{o}f interpolation argument (cf. \cite[Section 12.7]{RudinComplex}) firstly requires that $\UKH$ does not grow too fast on any vertical in the segments $\Re(z) \in [-r_{k+1},-r_k]$. This is not a priori guaranteed.
\hfill $\blacksquare$
\end{remark}

In order to estimate the behaviour of $\UKH$ on the verticals, it is useful to recall the vertical decay rates of the $\Gamma$-function.\pagebreak

\begin{lemma}\label{lm:Gamma_vert}
With $x, y \in \RR$ we have
\begin{align*}
\Gamma(x+i y) = \begin{cases} \Oinf (\abs{y}^{x-1/2} e^{-\pi \abs{y}/2} ), & \text{for } x > 1/2, \\
\Oinf ( e^{-\pi \abs{y}/2} ), & \text{for } x \leq 1/2. \end{cases}
\end{align*}
\end{lemma}
\begin{proof}
The case $x \geq 1/2$ is standard \cite[(2.1.19)]{Paris}. To prove the other part we invoke the $\Gamma$ reflection formula
\begin{align}
\label{gamma_reflex}
\Gamma(z)\, \Gamma(1-z) = \pi\sin^{-1}(\pi z), \quad \text{for } z \in \CC \setminus \Z,
\end{align}
along with a lower bound \cite[p. 51]{Carlson}
\begin{align}\label{gamma_lower_bds}
\abs{\Gamma(x+i y)} \geq \cosh(\pi y)^{-1/2}\, \Gamma(x), \quad \text{for } x \geq 1/2, \, y \neq 0.
\end{align}
For $x\leq 1/2,\, x \notin \Z/2$ and $y \neq 0$ we estimate:
\begin{align*}
\abs{\Gamma(x+iy)} & = \pi \,\abs{\sin( \pi (1-x) - i \pi y)}^{-1}\, \abs{\Gamma(1-x-iy)}^{-1} \\
& =\pi\, [\sin^2( \pi x) \cosh^2(\pi y) + \cos^2(\pi x) \sinh^2(\pi y)]^{1/2} \, \abs{\Gamma(1-x-iy)}^{-1} \\
& \leq \tfrac{\pi}{\sqrt{2}} \,[\abs{\sin( \pi x)\,\cos(\pi x)\,\sinh(\pi y)} \, \cosh(\pi y)]^{-1/2} \cosh(\pi y)^{1/2} \Gamma(1-x)^{-1} \\
& \leq  \tfrac{\pi}{\sqrt{2}} \,\Gamma(1-x)^{-1} \, \abs{\sin( \pi x) \cos(\pi x)}^{-1/2} \abs{\sinh(\pi y)}^{-1/2}.
\end{align*}
As $\sinh(y) = \Oinf(e^{\abs{y}})$, it follows that $\Gamma(x+i y) = \Oinf \big( e^{-\pi \abs{y}/2} \big)$ for $x\leq 1/2, \,x \notin \Z/2$ and the constraint $x \notin \Z/2$ can be dropped by continuity.
\hfill$\Box$
\end{proof}
An explicit control on the growth rate of the zeta function $\zKH$ yields the following important corollary of Theorem \ref{thm:zeta2hk_finite}:

\begin{corollary}
\label{cor:heat_q}
Let $H\in \Tp$, $K \in \B(\H)$ and let $\zKH$ meet the assumptions of Theorem \ref{thm:zeta2hk_finite}. Moreover, assume that $\zKH$ is of at most polynomial growth on the verticals\index{polynomial growth on verticals}, i.e. for any fixed $x \leq r_0$ there exists $c(x) < \infty$, such that 
\begin{align}
\label{zkh_vert}
\zKH(x+iy) = \Oinf(\abs{y}^{c(x)}).
\end{align}
Then, for any $q>0$, there exists an asymptotic expansion of $\,\Tr K e^{-t H^q}$ as $t \downarrow 0$ of the form \eqref{heat_exp_finite}.
\end{corollary}

\begin{proof}
Observe that, for $\Re(s) > p/q$, $\zeta_{K,H^q}(s) = \Tr K H^{-qs} = \zeta_{K,H}(qs)$. \\ Then, the assumption of the polynomial growth on the verticals of $\zKH$, together with Lemma~\ref{lm:Gamma_vert}, is sufficient to conclude.
\hfill$\Box$
\end{proof}

There is an important lesson for noncommutative geometers coming from the above result: 

\begin{remark}
\label{rem:zeta_poly}
Let $\ahd$ be a regular spectral triple of dimension $p$. If, for a given $T \in \PDOk{0}$, the zeta function $\zeta_{T,D}$ has a meromorphic extension to $\CC$ with a polynomial growth on the verticals, then the small-$t$ asymptotic expansions of both $\Tr T e^{-t \abs{\DD}}$ and $\Tr T e^{-t \DD^2}$ exist. Nevertheless, the two expansions can be radically different --- compare Examples \ref{ex:even_sphere_2} and \ref{ex:odd_sphere_2} with Example \ref{ex:S1_exact}.
\hfill$\blacksquare$
\end{remark}

We indicated in Section \ref{sec:spec} that the spectral zeta functions are actually complex functions defined by general Dirichlet series. Estimating the growth rate of the meromorphic continuation of these is, in general, a formidable task, which sometimes relates to profound conjectures in number theory \cite[p. 10]{Flajolet96}. For example, it is fairly easy to show \cite[Chapter V]{Titchmarsh} that
\begin{align}
\label{zeta_vert}
\zeta(x+i y) = \begin{cases} \Oinf ( \abs{y}^0 ), & \text{for } 1<x, \\
\Oinf( \abs{y}^{(1-x)/2} ), & \text{for } 0 \leq x \leq 1, \\
\Oinf ( \abs{y}^{1/2-x} ), & \text{for } x < 0. \end{cases}
\end{align}
It implies, in particular, $\zeta(1/2+i y) = \Oinf ( \abs{y}^{1/4} )$.\\ However, the Lindel\"of Hypothesis states that actually $\zeta(1/2+i y) = \Oinf ( \abs{y}^{\epsilon} )$ for any $\epsilon > 0$. The current best estimate of the critical exponent is 13/84 \cite{Bourgain2017}.
\medskip

Although if $H$ and $K$ are classical pseudodifferential operators the small-$t$ asymptotic expansion of $\hKH$ is guaranteed by Theorem \ref{thm:Gilkey}, it is instructive to apply Theorem \ref{thm:zeta2hk_finite} in the context of the classical geometry of 2-sphere.

\begin{example}
\label{ex:even_sphere_2}
Let $\Dslash$ be the standard Dirac operator on $S^2$ (see Appendix \ref{sec:spheres}). From Formula \eqref{eigenvalues_Sd}  we deduce $$\zeta_{\Dslash^2}(s) = 4 \sum_{n=0}^{\infty} (n+1)^{-2s+1} = 4 \zeta(2s-1),$$ which is meromorphic on $\CC$ with a single simple pole at $s=1$. We thus have
\begin{align*}
\Rez{s=1} \ZZ_{\Dslash^2}(s) = 2, \,\,\text{ and }\,\, \Rez{s= -k} \ZZ_{\Dslash^2}(s) = 4 \, \tfrac{(-1)^k}{k!} \, \zeta(-2k-1) = - 4 \, \tfrac{(-1)^k}{k!} \, \tfrac{B_{2k+2}}{2k+2}.
\end{align*}
As the partitioning sequence we can choose, for instance, $r_k = -3/2 + k$, $k \in \N$.\\ Formula \eqref{zeta_vert} guarantees that the assumption \eqref{zkh_vert} is met and Corollary \ref{cor:heat_q} yields $$\Tr e^{-t \,\Dslash^2} \tzero 2t^{-1} - 4 \sum_{k=0}^\infty \tfrac{(-1)^k}{k!}\tfrac{B_{2k+2}}{2k+2}\,t^k.$$ This series is \emph{divergent} for any $t>0$ since $(-1)^{k+1} B_{2k} > 2 (2k)!/(2\pi)^{k}$ for $k \in \N^*$ \cite[(23.1.15)]{Abram}. In fact, the asymptotic expansion of the heat trace associated with the square of the standard Dirac operator on any \emph{even} dimensional sphere is divergent \cite[Proposition 11]{HeatEZ}.
\hfill$\blacksquare$
\end{example}

\subsection{Infinite Number of Poles in Vertical Strips}

For general noncommutative geometries the spectral zeta functions might have an infinite number of poles in the vertical strips. This feature is always present when fractal geometry is in play --- see, for instance, \cite{Guido3} --- and was detected also on the standard Podle\'s sphere (cf. Appendix \ref{sec:Podles} and \cite{PodlesSA}). Theorem \ref{thm:zeta2hk_finite} carries over to this context, but the proof needs to be refined as now each term of the asymptotic expansion \eqref{rho_k2} can itself be an infinite series, the convergence of which is a subtle issue.

\begin{theorem}
\label{thm:zeta2hk}
Let $H\in \Tp$ and let $K \in \B(\H)$. Assume that:

\begin{enumerate}[label=\textit{\roman*)},leftmargin=1cm]

\item \label{zeta2hk:assum1} The function $\zKH$ admits a meromorphic extension to the whole complex plane with poles of order at most $d$.

\item \label{zeta2hk:assum2} There exists a sequence $(r_k)_{k \in \N} \subset \RR$ strictly increasing to $+ \infty$ with $p < -r_0$, such that the function $\UKH$ is regular and Lebesgue integrable on the verticals $\Re(s)=-r_k$.

\item \label{zeta2hk:assum3} For any $k \in \N, \,t>0$ and $n \in \{0,1,\ldots,d\}$ the series
\begin{align}\label{ser_res}
\sum_{z \in X_k} \,\Rez{s=z}\, (s-z)^{n} \,\UKH(s)\,t^{-z},
\end{align} 
with $X_k \vc \{ z \in \PP(\UKH) \, \vert \, - r_{k+1} < \Re(z) < - r_{k} \}$, is absolutely convergent.

\end{enumerate}
Then, there exists an asymptotic expansion, with respect to the scale $(t^{r_k})_k\,$,
\begin{align}
\label{heat_exp_G}
&\hspace{-0.3cm}\hKH \tzero \sum_{k=0}^{\infty} \big( \rho_k(t) + g_k(t)\big), \\
&  \text{ with } \hspace{0.6cm} \rho_k(t) \vc  \sum_{z \in X_k} \,\big[\sum_{n=0}^{d} a_{z,n}(K,H) \log^n t\, \big]\,t^{-z} \label{rho_k3}\\
& \hspace{1.4cm} a_{z,n}(K,H) \vc \tfrac{(-1)^n}{n!} \, \Rez{s=z}\, (s-z)^{n} \,\UKH(s) \label{azn1}
\end{align}
and the functions $g_k$ can be expressed as 
\begin{align}
\label{Gk1}
\!\!\! g_k(t) = \tfrac{1}{i 2\pi} \lim_{m \to \infty} \int_{-r_{k+1}}^{-r_{k}}\hspace{-0.1cm} \big{[}  \UKH(x+ i y_{m}^{(k)})\,t^{- i y_{m}^{(k)}} - \UKH(x- i y_{m}^{(k)})\,t^{ i y_{m}^{(k)}} \big{]}\, t^{-x} dx,
\end{align}

\vspace{-0.2cm}\noindent 
via an (always existing) sequence $(y^{(k)}_m)_{m\in\N}$ strictly increasing to $+\infty$ with $y^{(k)}_0=0$ and such that $\UKH$ is regular on the horizontals $x \pm iy^{(k)}_m$ with $x \in [-r_{k+1},-r_{k}]$.
\end{theorem}

The introduction of the sequences $(y_m^{(k)})_m$, illustrated in Figure \ref{fig:ht2zeta}, is coerced by the need to sum over the residues of an infinite number of poles in each strip. The final result does not depend on the particular choice of the $(y_m^{(k)})_m$'s and, actually, one could allow for curved lines instead of simple straight horizontals.

Observe also that, similarly as in Theorem \ref{thm:hk2zeta}, the only role of the sequence $(r_k)_k$ is to fix an asymptotic scale: If we choose another sequence $(r_k')_k$, for which the assumptions \ref{zeta2hk:assum2} and \ref{zeta2hk:assum3} are met, we will obtain \emph{the same} asymptotic expansion \eqref{heat_exp_G} --- recall Remark \ref{rem:asymp_exp}.

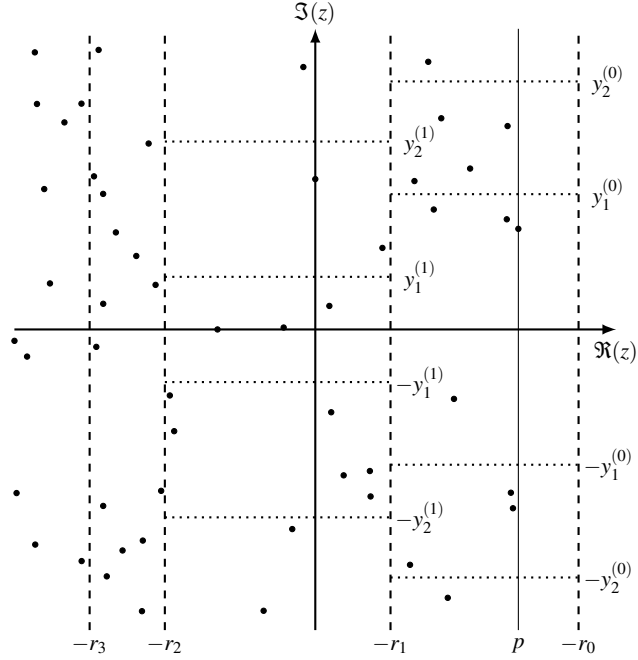
\begin{figure}[ht]
\begin{center}
\begin{tikzpicture}
%
%
\draw[-latex,thick] (-4,0) -- (4,0);
\draw[-latex,thick] (0,-4) -- (0,4);
\node[] at (4,-0.3) {$\Re(z)$};
\node[] at (0,4.2) {$\Im(z)$};
%
%
%
\filldraw (1.25945, -3.13148) circle (1pt); 
\filldraw (2.0582, 2.14066) circle (1pt);
\filldraw (1.67479, 2.80959) circle (1pt);  
\filldraw (2.60021, -2.17227) circle (1pt); 
\filldraw (2.54572, 1.46555) circle (1pt);
\filldraw (2.55604, 2.70598) circle (1pt); 
\filldraw (1.31823, 1.97422) circle (1pt);
\filldraw (2.62783, -2.37889) circle (1pt); 
\filldraw (1.84381, -0.924021) circle (1pt); 
\filldraw (1.57486, 1.59567) circle (1pt);
\filldraw (1.76111, -3.57055) circle (1pt);
\filldraw (2.69955, 1.33859) circle (1pt);
\filldraw (1.50166, 3.56093) circle (1pt);
%
%
\filldraw (0.734873, -2.222) circle (1pt); 
\filldraw (0.725679, -1.88359) circle (1pt);
\filldraw (0.37699, -1.94107) circle (1pt);
\filldraw (0.186981, 0.312932) circle (1pt);
\filldraw (-0.306402, -2.65595) circle (1pt);
\filldraw (-1.93375, -0.878248) circle (1pt); 
\filldraw (-0.158759, 3.49133) circle (1pt); 
\filldraw (0.892321, 1.08486) circle (1pt); 
\filldraw (0.211636, -1.10176) circle (1pt); 
\filldraw (-0.42055, 0.0254455) circle (1pt);
\filldraw (-1.87595, -1.35347) circle (1pt);
\filldraw (-0.687024, -3.74293) circle (1pt);
\filldraw (-1.3, 0) circle (1pt);
\filldraw (0, 2) circle (1pt);
%
%
\filldraw (-2.37965, 0.977474) circle (1pt); 
\filldraw (-2.81933, 0.34137) circle (1pt);
\filldraw (-2.6512, 1.29105) circle (1pt);  
\filldraw (-2.21576, 2.47371) circle (1pt); 
\filldraw (-2.81974, -2.34804) circle (1pt);
\filldraw (-2.88071, 3.71907) circle (1pt); 
\filldraw (-2.82043, 1.80348) circle (1pt);
\filldraw (-2.7722, -3.28512) circle (1pt); 
\filldraw (-2.29309, -2.80882) circle (1pt); 
\filldraw (-2.04921, -2.14728) circle (1pt);
\filldraw (-2.94202, 2.03689) circle (1pt);
\filldraw (-2.12347, 0.593084) circle (1pt);
\filldraw (-2.56238, -2.94058) circle (1pt);
\filldraw (-2.30414, -3.74829) circle (1pt);
\filldraw (-2.91118, -0.231629) circle (1pt);
%
%
\filldraw (-3.97142,-2.17682) circle (1pt); 
\filldraw (-3.10642, -3.08173) circle (1pt);
\filldraw (-3.72845, 3.68643) circle (1pt);  
\filldraw (-3.60357, 1.86766) circle (1pt); 
\filldraw (-3.83147, -0.360911) circle (1pt);
\filldraw (-3.33506, 2.75464) circle (1pt); 
\filldraw (-3.9993, -0.151525) circle (1pt);
\filldraw (-3.1085, 3.00419) circle (1pt); 
\filldraw (-3.52734, 0.612302) circle (1pt); 
\filldraw (-3.72299, -2.86216) circle (1pt);
\filldraw (-3.7,3) circle (1pt);
%
%
\draw[dashed,thick] (3.5,-4) -- (3.5,4);
\node[] at (3.5,-4.2) {$-r_0$};
\draw[dashed,thick] (1,-4) -- (1,4);
\node[] at (1,-4.2) {$-r_1$};
\draw[dashed,thick] (-2,-4) -- (-2,4);
\node[] at (-2,-4.2) {$-r_2$};
\draw[dashed,thick] (-3,-4) -- (-3,4);
\node[] at (-3,-4.2) {$-r_3$};
\draw[] (2.69955,-4) -- (2.69955,4);
\node[] at (2.69955,-4.2) {$p$};
%
%
\draw[dotted,thick] (1,1.8) -- (3.5,1.8);
\node[] at (3.9,1.8) {$y^{(0)}_1$};
\draw[dotted,thick] (1,-1.8) -- (3.5,-1.8);
\node[] at (3.9,-1.8) {$-y^{(0)}_{1}$};
\draw[dotted,thick] (1,3.3) -- (3.5,3.3);
\node[] at (3.9,3.3) {$y^{(0)}_2$};
\draw[dotted,thick] (1,-3.3) -- (3.5,-3.3);
\node[] at (3.9,-3.3) {$-y^{(0)}_{2}$};
\draw[dotted,thick] (-2,0.7) -- (1,0.7);
\node[] at (1.4,0.7) {$y^{(1)}_1$};
\draw[dotted,thick] (-2,-0.7) -- (1,-0.7);
\node[] at (1.4,-0.7) {$-y^{(1)}_{1}$};
\draw[dotted,thick] (-2,2.5) -- (1,2.5);
\node[] at (1.4,2.5) {$y^{(1)}_2$};
\draw[dotted,thick] (-2,-2.5) -- (1,-2.5);
\node[] at (1.4,-2.5) {$-y^{(1)}_{2}$};

\end{tikzpicture}
\end{center}
\caption{\label{fig:ht2zeta}An illustration of the meromorphic structure of a function $\UKH$ together with notations adopted in the proof of Theorem \ref{thm:zeta2hk}.}
\end{figure}

\begin{proof}
For each $k \in \N$, let us choose a sequence $(y^{(k)}_m)_{m\in\N}$ with the properties required by the theorem. Such a sequence can always be found since $\zKH$ is meromorphic (and so is $\UKH$) and hence the sets $X_k$ do not have accumulation points. \\
For each $k$ let us denote $D_0^k \vc \emptyset$ and 
\begin{align*}
& D_m^k\vc \text{ rectangle }\{s \in \CC \,\vert \,-r_{k+1}\leq \Re(s)\leq -r_{k}, \; -y_{m}^{(k)}\leq \Im(s) \leq y_m^{(k)} \},\,\text{ for }m\in \N^*.
\end{align*}

Now, let us fix the indices $k$ and $m$.
\\
By construction of the sequences $(r_k)_k$ and $(y^{(k)}_m)_{m}$, the function $\UKH$ is regular at the boundary of any $D_m^k$. Hence, the residue theorem yields
\begin{equation}
\label{finite_contour_N}
\tfrac1{i2\pi}\int_{\partial D_m^k}\UKH(s)\,t^{-s}\,ds =\sum_{z\in \PP(\UKH,\,D_m^k)} \,\Rez{s=z} \, \UKH(s) \,t^{-s},
\end{equation}
where the contour $\partial D_m^k$ is oriented counter-clockwise. In the above sum only a finite number of residues is taken into account, as the region $D_m^k$ is bounded.

Let us decompose the boundary of the rectangle using
\begin{align*}
&I_{V}^+(k,m)  \vc \int_{-r_{k}- i y_{m}^{(k)}}^{-r_{k}+ i y_{m}^{(k)}}\UKH(s)\,t^{-s}\,ds, \\
&I_{V}^-(k,m)  \vc \int_{-r_{k+1}- i y_{m}^{(k)}}^{-r_{k+1}+ i y_{m}^{(k)}}\UKH(s)\,t^{-s}\,ds, \\
&I_H^\pm(k,m)  \vc \int_{-r_{k+1} \pm i y_{m}^{(k)}}^{-r_{k} \pm i y_{m}^{(k)}}\UKH(s)\,t^{-s}\,ds,
\end{align*}
so that $\int_{\partial D_m^k}\UKH(s)\,t^{-s}\,ds = I_{V}^+(k,m)-I_H^+(k,m)-I_{V}^-(k,m)+I_H^-(k,m)$.\\
We are now concerned with the limit of Equation \eqref{finite_contour_N} as $m \to \infty$.

Firstly, as in the proof of Theorem \ref{thm:zeta2hk_finite}, we deduce from assumption \ref{zeta2hk:assum2} that 
\begin{align*}
& R_k(t)  \vc \tfrac1{i2\pi}\lim_{m \to \infty} \,I_{V}^-(k,m) = \tfrac1{2\pi}\int_{-\infty}^{\infty}\UKH(-r_{k+1}+iy )\,t^{r_{k+1} - iy}\,dy = \oz(t^{r_{k+1}}),\\
& R_{-1}(t) \vc \tfrac1{i2\pi}\lim_{m \to \infty}\, I_{V}^+(0,m) = \hKH.
\end{align*}

Secondly, for any $k \in \N$ and any $t>0$, we are allowed to write
\begin{align*}
\lim_{m \to \infty} \, \sum_{z\in \PP(\UKH,\,D_m^k)}\, \Rez{s=z} \, \UKH(s) \,t^{-s} = \sum_{z\in X_k} \,\Rez{s=z} \, \UKH(s)\, t^{-s} = \rho_k(t),
\end{align*}
as the sum over the (possibly infinite) sets $X_k$ is absolutely convergent by assumption \ref{zeta2hk:assum3}. In particular, the limit $m \to \infty$ of the sum of residues does not depend on the choice of the sequence $(y^{(k)}_m)_{m}$.

Finally, for any $k \in \N$ and any $t>0$ Equation \eqref{finite_contour_N} yields
\begin{align}\label{F_k relation_N}
R_{k-1}(t)- R_k(t) - \rho_k(t) = \tfrac{1}{i2 \pi} \, \lim_{m \to \infty}  \,[ I_H^+(k,m) - I_H^-(k,m) ].
\end{align}

Since the LHS of \eqref{F_k relation_N} is a well-defined function of $k \in \N$ and $t \in (0,\infty)$, the RHS must be so, hence, in particular, the limit $m \to \infty$ exists and is finite for any fixed value of $k$ and $t$. We denote it by $$g_k(t) \vc \tfrac{1}{i2 \pi} \,  \lim_{m \to \infty}  \,[I_H^+(k,m) - I_H^-(k,m) ],$$ which accords with Formula \eqref{Gk1}.

Following the same arguments as in the proof of Theorem \ref{thm:hk2zeta} one can deduce that $\rho_k(t) = \oz(t^{r_{k}})$ for any $k \in \N$. Since $R_k(t) = \oz(t^{r_{k+1}})$, we get $g_k(t) = \oz(t^{r_{k}})$ for any $k \in \N$.

Starting with Formula \eqref{F_k relation_N} for $k=0$ and iterating it $N$ times we obtain the announced asymptotic expansion: 
$\hKH = \sum_{k = 0}^N [\rho_k(t) + g_k(t)] + R_N(t)$.
\hfill$\Box$
\end{proof}

The assumptions in Theorem \ref{thm:zeta2hk} are actually weaker than the ones adopted in Theorem \ref{thm:zeta2hk_finite} as we have only imposed Lebesgue integrability on the verticals $\Re(s) = - r_k$, rather than demanding concrete decay rates in whole segments as in Formula \eqref{ukh_vert}. The price for that are the possible additional terms $g_k$ in the asymptotic expansion \eqref{heat_exp_G} coming from the contour integrals \eqref{Gk1} and not from the poles of $\UKH$. This should not be a surprise in view of Remark \ref{rem:Gk}. It is also clear that if $\abs{\UKH}$ decays on the verticals, i.e. 
\begin{align}
\label{ukh_vert_m}
\lim_{m \to \infty} \, \vert\UKH(x+iy^{(k)}_m)\vert = 0, \quad \text{ for all } \, k \in \N,
\end{align}
then $g_k = 0$ for every $k$.

When \eqref{ukh_vert_m} is assumed, the limit $m \to \infty$ of Formula \eqref{finite_contour_N} reads 
\begin{align}\label{m_conv}
R_{k-1}(t)- R_k(t) = \lim_{m \to \infty}  \sum_{z\in \PP(\UKH,\,D_m^k)}\, \Rez{s=z} \, \UKH(s) \,t^{-s}.
\end{align}
It implies that the sum of residues converges for any $k \in \N$, $t>0$ and, moreover, that the limit $m \to \infty$ does not depend on the chosen sequence $(y^{(k)}_m)_m$. Hence, one might be tempted to conclude that the assumption \ref{zeta2hk:assum3} is actually redundant when the constraint \eqref{ukh_vert_m} is met. However, Formula \eqref{m_conv} guarantees the conditional convergence only and we are not allowed to write the RHS of \eqref{m_conv} simply as $\rho_k(t)$ given by \eqref{ser_res}. More precisely, \eqref{m_conv} says that the series of residues converges if we group the terms into the sets $\PP(\UKH,\,D_{m+1}^k \setminus D_{m}^k)$ and add them subsequently with the increasing counting index $m$. Adopting a more stringent constraint
\begin{align*}
\sup_{x\in[-r_{k+1},-r_k]} \, \vert\UKH(x+iy^{(k)}_m)\vert = \Oinf(\abs{m}^{-1-\epsilon_k}), \quad \text{ for some } \, \epsilon_k > 0,
\end{align*}
which looks natural when compared with assumption \ref{zeta2hk_finite:assum3} of Theorem \ref{thm:zeta2hk_finite}, does imply that
\begin{align*}
\sum_{m = 0}^{\infty} \, \big\vert \,  \sum_{z\in \PP(\UKH,\,D_{m+1}^k \setminus D_{m}^k)}\, \Rez{s=z} \, \UKH(s) \,t^{-s} \, \big\vert < \infty, \text{ for any }k \in \N, \, t > 0
\end{align*}
(cf. \cite[Proposition 4]{HeatEZ}). 
But the grouping of terms into $\PP(\UKH,\,D_{m+1}^k \setminus D_{m}^k)$ might turn out indispensable. Such a situation might produce itself for instance if the meromorphic structure of the function $\UKH$ is as on Figure \ref{fig:ht2zeta_pat}.

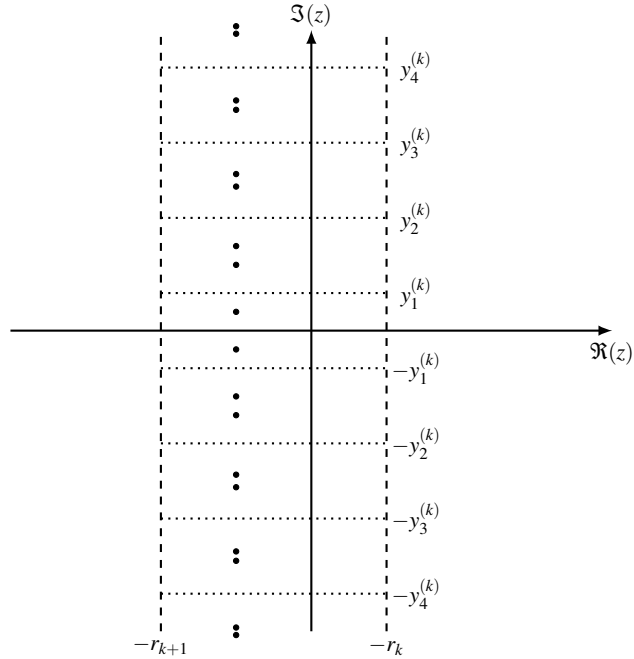
\begin{figure}[ht]
\begin{center}
\begin{tikzpicture}
%
%
\draw[-latex,thick] (-4,0) -- (4,0);
\draw[-latex,thick] (0,-4) -- (0,4);
\node[] at (4,-0.3) {$\Re(z)$};
\node[] at (0,4.2) {$\Im(z)$};
%
%
%
\filldraw (-1,1/4) circle (1pt);
\filldraw (-1,-1/4) circle (1pt);
\filldraw (-1,1+1/8) circle (1pt);
\filldraw (-1,1-1/8) circle (1pt);
\filldraw (-1,-1+1/8) circle (1pt);
\filldraw (-1,-1-1/8) circle (1pt);
\filldraw (-1,2+1/12) circle (1pt);
\filldraw (-1,2-1/12) circle (1pt);
\filldraw (-1,-2+1/12) circle (1pt);
\filldraw (-1,-2-1/12) circle (1pt);
\filldraw (-1,3+1/16) circle (1pt);
\filldraw (-1,3-1/16) circle (1pt);
\filldraw (-1,-3+1/16) circle (1pt);
\filldraw (-1,-3-1/16) circle (1pt);
\filldraw (-1,4+1/20) circle (1pt);
\filldraw (-1,4-1/20) circle (1pt);
\filldraw (-1,-4+1/20) circle (1pt);
\filldraw (-1,-4-1/20) circle (1pt);
%
%
\draw[dashed,thick] (1,-4) -- (1,4);
\node[] at (1,-4.2) {$-r_k$};
\draw[dashed,thick] (-2,-4) -- (-2,4);
\node[] at (-2,-4.2) {$-r_{k+1}$};
%
%
\draw[dotted,thick] (-2,1/2) -- (1,1/2);
\node[] at (1.4,1/2) {$y^{(k)}_1$};
\draw[dotted,thick] (-2,-1/2) -- (1,-1/2);
\node[] at (1.4,-1/2) {$-y^{(k)}_{1}$};
\draw[dotted,thick] (-2,3/2) -- (1,3/2);
\node[] at (1.4,3/2) {$y^{(k)}_2$};
\draw[dotted,thick] (-2,-3/2) -- (1,-3/2);
\node[] at (1.4,-3/2) {$-y^{(k)}_{2}$};
\draw[dotted,thick] (-2,5/2) -- (1,5/2);
\node[] at (1.4,5/2) {$y^{(k)}_3$};
\draw[dotted,thick] (-2,-5/2) -- (1,-5/2);
\node[] at (1.4,-5/2) {$-y^{(k)}_{3}$};
\draw[dotted,thick] (-2,7/2) -- (1,7/2);
\node[] at (1.4,7/2) {$y^{(k)}_4$};
\draw[dotted,thick] (-2,-7/2) -- (1,-7/2);
\node[] at (1.4,-7/2) {$-y^{(k)}_{4}$};

\end{tikzpicture}
\end{center}
\caption{\label{fig:ht2zeta_pat}An illustration of a possible pathological structure of poles of a meromorphic function. The sum of residues might turn out finite if the poles are grouped by two, but infinite if they are added one by one. It is unclear whether such a situation can actually produce itself for a spectral zeta function of geometrical origin.}
\end{figure}

On the other hand, we are not able to cook up a concrete example of a function $\UKH$, for which such a pathology occurs. So, it might well be that the geometric origin of the operators $K, H$ prevents such situations --- see Problem \ref{prob:heat vs zeta} \ref{prob:ht2zeta_pat}. However, we will shortly witness, in Section \ref{sec:conv_non}, a computation of the asymptotic expansion, in which the (vertical) contour integral does give a finite, non-zero, contribution.

We conclude this section with a friendly noncommutative-geometric example.

\begin{example}
\label{ex:Podles_hk_S}
Let $\DD_q^S$ be the simplified Dirac operator on the standard Podle\'s sphere (cf. Appendix \ref{sec:Podles}) and let us denote $u = \abs{w} q (1-q^2)^{-1}> 0$. Using the explicit formula for the eigenvalues of $\vert\DqS\vert$ we can easily compute its spectral zeta function. For $\Re(s) > 0$ we have:
\begin{align*}
\zqS(s) & = \sum_{\{+,-\}} \,\sum_{\ell \in \N + 1/2} \,\,\sum_{m = -\ell}^{\ell} 
\big(\vert w \vert \, \tfrac{q^{-\ell+1/2}}{1-q^2}\big)^{-s} = 2 (u q^{-1})^{-s} \, \sum_{\ell \in \N + 1/2} (2\ell+1) q^{\,s (\ell-1/2)} \\
& = 4 (u q^{-1})^{-s} \, \sum_{n=0}^{\infty} (n+1) \,q^{\,sn} = 4 (u q^{-1})^{-s} \, (1-q^{\,s})^{-2}.
\end{align*}
The latter equality, which is nothing but an elementary summation of a geometric series, provides the meromorphic extension of $\zqS$ to the whole complex plane. It has an infinite number of poles of second order regularly spaced on the imaginary axis. The following meromorphic structure of the function $\UDqS$ emerges:

\begin{itemize}
	\item a third order pole at $s=0$,
	\item second order poles at $s=2 \pi i j/ (\log q)$, with $j \in \Z^*$,
	\item first order poles at $s = - k$, with $k \in \N$.
\end{itemize}

The corresponding residues $a_{z,n}(\bbbone,\vert \DqS \vert)$, given by Formula \eqref{azn}, read
\begin{align*}
& a_{0,2} = \tfrac{2}{\log^2 q}, \qquad a_{0,1} = \tfrac{4}{\log^2 q} \left(\log u + \gamma \right), & \\
& a_{0,0} = \tfrac{1}{\log^2 q} \left( 2 \log^2 u + \tfrac{1}{3} \pi^2 - \tfrac{1}{3} \log^2 q + 4 \gamma \log u + 2\gamma^2 \right), &\\
& a_{2 \pi i j / \log q,1}  = -\tfrac{4}{\log^2 q} \, u^{-2 \pi i j / \log q} \, \Gamma(\tfrac{2 \pi i}{\log q} \,j), & \text { for } j \in \Z^*,\\
& a_{2 \pi i j / \log q,0}  = -\tfrac{4}{\log^2 q} \, u^{-2 \pi i j / \log q} \, \Gamma(\tfrac{2 \pi i}{\log q}\, j) \,[ \log u - \psi(\tfrac{2 \pi i}{\log q}\, j) ], & \text { for } j \in \Z^*,
\end{align*}
where $\gamma$ stands for the Euler's constant and $\psi \vc \Gamma'/\Gamma$ is the digamma function.

To apply Theorem \ref{thm:zeta2hk}, we can choose $r_k = -1/2 + k$. We have on the verticals
\begin{align}
\label{q_zeta_vert}
\vert \zqS(-r_k + iy) \vert = 4 (u q^{-1})^{r_k} \, \abs{1-q^{-r_k+iy}}^{-2} \leq 4 (u q^{-1})^{r_k} \, (1-q^{-r_k})^{-2}.
\end{align}
With the help of Lemma \ref{lm:Gamma_vert}, we see immediately that assumptions \ref{zeta2hk:assum2} and \ref{zeta2hk:assum3} of Theorem \ref{thm:zeta2hk} are fulfilled.

Finally, we choose a convenient sequence $y_m = \tfrac{2 \pi (m + 1/2)}{\log q}$, $m \in \N$, for which $q^{iy_m} = -1$. It yields,  for any $x \in \RR$,
\begin{align*}
\vert \zqS(x + iy_m) \vert = 4 (u q^{-1})^{x} \, \vert 1-q^{x + iy_m} \vert^{-2} = 4 (u q^{-1})^{x} (1+q^{x})^{-2}.
\end{align*}
Hence, no additional contribution from the contour integral \eqref{Gk1} arises.

Summa summarum, Theorem \ref{thm:zeta2hk} yields
\begin{align}
\Tr e^{-t\, \abs{\DqS}}  \tzero \tfrac{1}{\log^2 q} \,
\big[ 2 \log^2 (ut) &+  F_1 \big(\log (ut)\big)  \, \log (ut)  + F_0 \big(\log (ut)\big)\big] \nonumber \\
 &+ \sum_{k=1}^\infty \tfrac{(-1)^k \, q^{-k}}{(k)!(1-q^{-k})^2} \,\,(ut)^{k}, \label{Podles_heat_S_asym}
\end{align}
where $F_0$ and $F_1$  are periodic bounded smooth functions on $\RR$ defined as
\begin{align*}
F_1(x) & \vc  4 \gamma - 4 \sum_{j \in \Z^*} \Gamma (\tfrac{2\pi i}{\log q} \,j) \,e^{2\pi i j x}, \\
F_0(x) & \vc \tfrac{1}{3} ( \pi^2 + 6 \gamma^2 - \log^2 q) - 4
\sum_{j \in \Z^*} \,\Gamma (-\tfrac{2\pi i}{\log q} \,j)\,  \psi (\tfrac{2\pi i}{\log q}\, j)\, e^{2\pi i j x},
\end{align*}
see \cite[Theorem 4.1]{PodlesSA}.

Similarly as in Example \ref{ex:Podles_SA_asym}, it turns out that the expansion \eqref{Podles_heat_S_asym} is actually exact for all $t>0$ --- see Example \ref{ex:Podles_exact}. We remark that these results hold also for the full Dirac operator $\DD_q$ on the standard Podle\'s sphere, though the estimation of the contour integrals is much more tedious --- cf. \cite[Section 4]{PodlesSA}.
\hfill$\blacksquare$
\end{example}

\section{Convergent Expansions of Heat Traces}\label{sec:conv}

As illustrated by Example \ref{ex:even_sphere_2} the small-$t$ asymptotic expansion of a heat trace is typically divergent. We have, however, witnessed the convergence of this expansion in some particular geometrical context. We shall now connect these specific situations to the behaviour of the associated zeta functions.

In the context of Theorems \ref{thm:zeta2hk_finite} and \ref{thm:zeta2hk} the remainder of a convergent expansion (cf. Definition \ref{def:conv}) can be written explicitly as
\begin{align}\label{Rinf}
R_{\infty}(t) = \lim_{N \to \infty} \, \tfrac1{i2\pi}\int_{-r_{N}- i \infty}^{-r_{N}+ i \infty}\UKH(s)\,t^{-s}\,ds.
\end{align}
Indeed, eq. \eqref{heat_exp_G} yields 
\begin{align*}
\hKH -  \lim_{N \to \infty} \,\sum_{k=0}^{N} [ \rho_k(t) + g_k(t) ] = \lim_{N \to \infty} \,R_N(t).
\end{align*}
Hence, we can alternatively deduce the convergence of an asymptotic expansion of the form \eqref{heat_exp_G} by inspecting the expression \eqref{Rinf}.

Namely, if there exists $T \in (0,\infty]$ such that $R_{k}(t)$ converges absolutely/uniformly on $(0,T)$ as $k \to \infty$, then $\sum_{k = 0}^{\infty} ( \rho_k(t) + g_k(t) )$ converges absolutely/uniformly on $(0,T)$ and for $t \in (0,T)$ we can write $\hKH = \sum_{k = 0}^{\infty} [ \rho_k(t) + g_k(t) ] + R_{\infty}(t)$. Moreover, Equation \eqref{Rk_order} then shows that $R_{\infty}(t) = \Oz(t^{\infty})$.

\subsection{Convergent, Non-exact, Expansions}\label{sec:conv_non}

A particular instance of a convergent, but not exact, expansion of a heat trace produces itself when the set of poles of the associated zeta function is finite.

\begin{proposition}\label{prop:almost}
Let $H$ and $K$ meet the assumptions of Theorem \ref{thm:zeta2hk_finite}. \\If the set $\PP(\UKH)$ is finite, then there exists $N \in \N^*$ such that, for all $t>0$,
\begin{align}\label{heat_almost}
\hKH = \sum_{k=0}^{N} \rho_k(t) + R_{\infty}(t), \quad \text{ with } R_{\infty} \neq 0.
\end{align}
\end{proposition}

Note that since $\PP(\Gamma) = - \N$ the hypothesis of Proposition \ref{prop:almost} requires in particular that $\zKH(-k) = 0$ for almost all $k \in \N$.
\begin{proof}
This is a consequence of Theorem \ref{thm:zeta2hk_finite}. If $\PP(\UKH)$ is finite then, for any choice of the partitioning sequence $(r_k)_k$ there exists an $N \in \N^*$ such that $X_k = \emptyset$ for all $k \geq N$, hence $\rho_k = 0$ for $k \geq N$. Observe that $N\neq 0$, since $\zKH$ has at least one pole at $s=p \geq 0$ (cf. p. \pageref{zeta_one_pole}), which cannot be compensated by $\Gamma(s)$, since the function $\Gamma$ has no zeros. \\
Moreover, we cannot have $R_{\infty} = 0$, since $\hKH = \Oinf(e^{-t \lambda_0(H)})$ whereas the sum $\sum_{k=0}^{N} \rho_k(t)$ certainly grows faster than $t^{r_0}$ at infinity.
\hfill$\Box$
\end{proof}

As an illustration let us consider the following example:

\begin{example}\label{ex:odd_sphere_2}
Let $\Dslash$ be the standard Dirac operator on $S^3$ (see Appendix \ref{sec:spheres}). The spectral zeta function associated with $\Dslash^2$ reads
\begin{align*}
\zeta_{\Dslash^2}(s) & = 2 \sum_{n=0}^{\infty} (n+1)(n+2) (n+\tfrac{3}{2})^{-2s} = \sum_{n=0}^{\infty} 2(n+\tfrac{3}{2})^{-2s+2} - \tfrac{1}{2} \sum_{n=0}^{\infty} (n+\tfrac{3}{2})^{-2s} \\
& = 2\zeta(2s-2,3/2) - \tfrac{1}{2} \zeta(2s,3/2),
\end{align*}
where $\zeta(s,a)$\index{_a3zeta_2@$\zeta(s,a)$} for $a \notin - \N$ is the Hurwitz zeta function (see, for instance \cite[Chapter~12]{Apostol}). For any $a$ the latter admits a meromorphic extension to $\CC$ with a single simple pole at $s=1$. Moreover, it can be shown (cf. \cite[Theorem 6]{HeatEZ} and \cite{MatsumotoWeng}) that when $\arg(a) = 0$, $\zeta(s,a)$ is of polynomial growth on the verticals. Hence, Corollary \ref{cor:heat_q} can be applied. We have $$\Rez{s=3/2} \ZZ_{\Dslash^2}(s) = \tfrac{\sqrt{\pi}}{2}\,\text{ and }\,\Rez{s=1/2} \ZZ_{\Dslash^2}(s) = -\tfrac{\sqrt{\pi}}{4}.$$
With $\zeta(s,a+1) = \zeta(s,a) - a^{-s}$, we rewrite $$\zeta_{\Dslash^2}(s) = 2 \zeta(2s-2,1/2) - \tfrac{1}{2} \zeta(2s,1/2).$$ Then, for $k \in \N$, $$\Rez{s= -k} \ZZ_{\Dslash^2}(s) =  \tfrac{(-1)^k}{k!} \, [2 \zeta(-2k-2,1/2) - \tfrac{1}{2} \zeta(-2k,1/2)].
$$
But the properties of the Bernoulli polynomials (cf. page \pageref{eq:Euler Maclaurin}) yield, for any $k \in \N$,
\begin{align*}
\zeta(-2k,1/2) = -\tfrac{B_{2k+1}(1/2)}{2k+1} = \tfrac{1-2^{-2k}}{2k+1} \, B_{2k+1} = 0.
\end{align*}
Hence,
$$ \Tr e^{-t \Dslash^2} \!\!\!\tzero \hspace{-0.1cm}\tfrac{\sqrt{\pi}}{2} t^{-3/2} -\tfrac{\sqrt{\pi}}{4} t^{-1/2}\Longleftrightarrow\Tr e^{-t \Dslash^2} = \tfrac{\sqrt{\pi}}{2} t^{-3/2} -\tfrac{\sqrt{\pi}}{4} t^{-1/2} + \Oz(t^{\infty}),
$$
\vspace{-0.3cm}

\noindent what harmonises with Formula \eqref{SA:S3} for $f(x) = e^{-x^2}$ and $\Lambda = t^{-1/2}$.

As remarked in Example \ref{ex:conv_non}, the asymptotic expansion of the heat trace associated with the square of the standard Dirac operator on any \emph{odd} dimensional sphere is convergent for all $t>0$, but is not exact.
\hfill$\blacksquare$
\end{example}

\begin{remark}
Let us note that Proposition \ref{prop:almost} can be a source of interesting identities for some special functions defined by general Dirichlet series: One checks (cf. \cite{Elizalde93}) that in the context of Example \ref{ex:conv_non} $R_{\infty}(t) = \tfrac{1}{2} \, (\tfrac{\pi}{t})^{1/2} [ \vartheta_3(0; e^{-\pi^2/t}) - 1 ]$ and hence Formula \eqref{heat_almost} yields the Jacobi identity \eqref{Jacobi_id}.
\hfill$\blacksquare$
\end{remark}

\subsection{Exact Expansions}
\label{sec:exact}

Let us now give a sufficient condition for the exactness of the small-$t$ expansion of a heat trace, in terms of the corresponding zeta function.

\begin{theorem}
\label{thm:hk_exact}
Let $H$ and $K$ meet the assumptions of Theorem \ref{thm:zeta2hk} and let the estimate 
\begin{align}\label{exact}
\abs{\UKH(-r_k+iy)}\leq c_k e^{-\epsilon_k\abs{y}}
\end{align}
hold for every $y\in\RR$ and $k\in\N$ with some $c_k, \epsilon_k > 0$. If
\begin{equation}
\label{T}
T \vc \big[\limsup_{k\to\infty}\,(\tfrac{c_k}{\epsilon_k})^{1/r_k}\big]^{-1} > 0 \, ,
\end{equation}
then the series $\sum_{k = 0}^{\infty} ( \rho_k(t) + g_k(t) )$, with $\rho_k$ and $g_k$ given by \eqref{rho_k3} and \eqref{Gk1} respectively, converges to $\hKH$ locally uniformly on $(0,T)$.

If, moreover, $\log k = \oinf(r_k)$ (i.e. $r_k$ grows faster than $\log k$), then the convergence is absolute on $(0,T)$.
\end{theorem}

\begin{proof}
We estimate the remainder $R_k(t)$ as follows
\begin{align*}
\abs{R_k(t)}=\tfrac1{2\pi}\,\vert\int_{-\infty}^\infty \UKH(-r_k+iy) \,t^{r_k-iy}dy\vert
\leq\tfrac1{2\pi}\int_{-\infty}^\infty c_k e^{-\epsilon_k\abs{y}} t^{r_k}dy =\tfrac{1}{\pi} \, \tfrac{c_k }{\epsilon_k} \,t^{r_k}.
\end{align*}
Let $0<T'<T$. For any $t\in(0,T']$ we have $$ \limsup_{k\to\infty} \, t \, \sqrt[r_k]{c_k/\epsilon_k} = t/T \leq T'/T.$$ Hence, for sufficiently large $k$ we have $t\sqrt[r_k]{c_k/\epsilon_k}<a$, where $a\in(T'/T,1)$ is some constant independent of $t$. Then,
\begin{equation}\label{R_k_bound}
\abs{R_k(t)}\leq\tfrac{c_k \, t^{r_k}}{\epsilon_k \pi} < \tfrac{a^{r_k}}\pi \underset{k \uparrow \infty}{\rightarrow} 0,
\end{equation}
so $R_k(t)$ tends to 0 uniformly for $t\in(0,T']$. Since $T'$ can be any number in $(0,T)$, the local uniform convergence is proven.\\
To check the absolute convergence we need to show that $\sum_{k = 0}^{\infty} \abs{ \rho_k(t) + g_k(t) } < \infty$ for $t \in (0,T)$. From the recurrence relation \eqref{F_k relation_N} we obtain that, for any $k \in \N$, $$\abs{ \rho_k(t) + g_k(t) } = \abs{R_{k-1}(t) - R_k(t)} \leq \abs{R_{k-1}(t)} + \abs{R_{k}(t)}.$$ Now, \eqref{R_k_bound} implies
\begin{align*}
\sum_{k = 0}^{\infty} \abs{ \rho_k(t) + g_k(t) } \leq 2\sum_{k = 0}^{\infty} \abs{R_k(t)}  \leq  2\sum_{k = 0}^{\infty} \tfrac{c_k t^{r_k}}{\epsilon_k \pi} < \tfrac{2}{\pi}\sum_{k = 0}^{\infty} a^{r_k},
\end{align*}
with $a \in (T'/T,1)$ for any $T' \in (0,T)$. Therefore, it suffices to show that the last series is convergent for any $a<1$. The latter is a general Dirichlet series in variable $x=-\log a$ with coefficients $a_k=1, b_k=r_k$ for $k\in\N$. By Theorem \ref{thm:abscissa} its abscissa of convergence equals $\limsup_{k\to\infty} (\log k) r_k^{-1}$, which is $0$ by hypothesis. Thus, the series is convergent for $x>0$, i.e. $a<1$.
\hfill$\Box$
\end{proof}

The absolute convergence, on top of the exactness, of a heat trace expansion is necessary to establish an exact large energies expansion of the spectral action (cf. Theorem \ref{thm:f_conv}). Given an exact expansion, the absolute convergence is fairly easy  --- it suffices to verify whether the sequence $(r_k)$ of our choice grows faster than $\log k$. In fact, all of the examples of exact expansions we present in this book are actually absolutely convergent in the same domain. 

Let us also remark that if we have an exact expansion of the heat trace for an open interval $(0,T)$, then $\hKH$ actually provides an analytic continuation of the series $\sum_k \abs{ \rho_k(t) + g_k(t) }$ to the whole half line $(0,\infty)$.

In general, Theorem \ref{thm:hk_exact} provides only a sufficient condition, so Formula \eqref{T} does not necessarily give the maximal range of (absolute) convergence. Nevertheless, the following example shows that the bound \eqref{exact} is often good enough to deduce the actual upper limit.

\begin{example}
\label{ex:S1_exact}
Let $\Dslash$ be the standard Dirac operator on $S^1$ associated with the trivial spin structure (see Appendix \ref{sec:spheres}). As shown in Example \ref{ex:exact}, the associated heat trace can be computed explicitly for any $t>0$ (recall Formula \eqref{heat_S1}) and developed in a Laurent series around $t=0$. The latter converges to $\Tr e^{-t \abs{\Dslash}}$ absolutely for $t \in (0,2\pi)$. Let us now re-derive this result using Theorems \ref{thm:zeta2hk_finite} and \ref{thm:hk_exact}:

The operator $\Dslash$ has a kernel of dimension 1, so Formula \eqref{heat_ker} gives
\begin{align*}
\Tr e^{-t \abs{\Dslash}} = \Tr_{(\bbbone-P_0)\H}\, e^{-t\,\vert\bD\vert} + 1, \quad \text{ for all } t>0.
\end{align*}
We have, $\zeta_{\bD}(s) = 2\sum_{n=1}^{\infty} n^{-s} = 2 \zeta(s)$, which is meromorphic on $\CC$ with a single simple pole at $s=1$ and
\begin{align*}
\Rez{s=1} \,\ZZ_{\bD}(s) = 2, \quad \Rez{s= -k}\, \ZZ_{\bD}(s) = 2 \, \tfrac{(-1)^k}{k!} \, \zeta(-k) = - 2 \, \tfrac{(-1)^k}{k!} \, \tfrac{B_{k+1}}{k+1}, \quad \text{ for } \; k \in \N.
\end{align*}
Since $B_{2k+1} = 0$ for $k \in \N^*$, we actually have $\PP(\ZZ_{\bD}) = \{1, 0, -1, -3, -5, \ldots \}$. \\
As the partitioning sequence we can choose, $r_0 = -3/2, r_1 = -1/2, r_2 = 1/2$ and $r_k = 2(k-2)$, for $k \geq 3$. Then, Formula \eqref{zeta_vert} guarantees that the assumption \eqref{zkh_vert} is met and Corollary \ref{cor:heat_q} yields 
\begin{align*}
\Tr_{(\bbbone-P_0)\H} e^{-t \,\vert \bD\vert} \sim_{t\downarrow 0} 2t^{-1}-1 + 2 \sum_{k=0}^\infty \tfrac{B_{2k+2}}{(2k+2)!}\,t^{2k+1},
\end{align*}
in agreement with Formulae \eqref{heat_S1} and \eqref{coth}.

To check that this expansion is exact we need to find an explicit bound of the form \eqref{exact}. Recall first the Riemann functional equation \cite[Formula (23.2.6)]{Abram}:
\begin{equation*}
\zeta(s)=2^s\pi^{s-1}\sin\left(\tfrac{\pi s}2\right)\Gamma(1-s)\zeta(1-s),
\end{equation*}
which yields $$\ZZ_{\bD}(s)=2\Gamma(s)\zeta(s)=(2 \pi)^{s}\zeta(1-s)\,\big(\sin[\pi(1-s)/2]\big)^{-1}.$$ 
For $s=-r_k+iy$ with $k \geq 3$ the denominator of the above expression equals
\begin{equation*}
\sin[\pi(\tfrac12+k-i\tfrac y2)] = \cos[\pi(k-i\tfrac y2)] = (-1)^k\cosh\left(\tfrac{\pi y}2\right).
\end{equation*}
Thus, for $k \geq 3$ and $y\in\RR$ we have
\begin{align*}
\abs{\ZZ_{\bD}(-r_k+iy)} &= \tfrac{(2 \pi)^{-2k}\,\abs{\zeta(2k+1-iy)}}{\cosh\left(\pi y/2\right)}
\leq 2 (2\pi)^{-2k} \, \zeta(2k+1) \, e^{-\pi\abs{y}/2}.
\end{align*}
Hence, the assumptions of Theorem \ref{thm:hk_exact} are met with $c_k = 2 (2\pi)^{-2k} \zeta(2k+1)$ and $\epsilon_k = \pi/2$. Since $\zeta(x) \to 1$ as $x \to +\infty$, we obtain
\begin{equation*}
T^{-1} = \limsup_{k\to\infty}\, [ 2^{-2k+2} \pi^{-2k-1} \zeta(2k+1) ]^{1/(2k-4)}=\tfrac1{2\pi}.
\end{equation*}
We have recovered the radius of convergence of the Laurent expansion of $\coth (t/2)$. Thus, as $r_k = \Oinf(k)$, we conclude that  this expansion is absolutely convergent. 
\hfill$\blacksquare$
\end{example}

Let us now treat an example of an exact expansion valid for $\,t>0$, i.e. with $T = \infty$.

\begin{example}
\label{ex:Podles_exact}
Let $\DD_q^S$ be the simplified Dirac operator on the standard Podle\'s sphere (cf. Appendix \ref{sec:Podles}) and let us denote $u = \abs{w} q /(1-q^2) > 0$. In Example \ref{ex:Podles_hk_S} we have derived an asymptotic expansion of $\Tr e^{-t \abs{\DqS}}$ and alluded to its exactness. In fact, it is obvious that the series in \eqref{Podles_heat_S_asym} is absolutely and locally uniformly convergent for any $t>0$. But there might a priori be a non-trivial contribution from the contour integral at infinity namely $R_{\infty}$.

To show that this is not the case we will resort to Theorem \ref{thm:hk_exact}. In Example \ref{ex:Podles_hk_S} we chose $r_k = -1/2 + k$ and established an explicit estimate of $\zqS$ on the verticals (cf. \eqref{q_zeta_vert}). In order to proceed, we need a more precise estimate of the Gamma function on the verticals than the one derived in Lemma \ref{lm:Gamma_vert}. The Euler reflection formula \eqref{gamma_reflex} together with inequality \eqref{gamma_lower_bds} and $$\Gamma(x) > (2 \pi)^{1/2} \,x^{x-1/2} e^{-x}\,\text{ for }\,x>0$$ (\cite[p. 253]{WhittakerWatson}) gives
\begin{align}
\abs{\Gamma(-r_k + iy)} & = \tfrac{1}{\abs{\Gamma(k+1/2-iy)}} \, \tfrac{\pi}{\abs{\sin \left[ \pi (k-1/2+iy) \right]}} 
 \leq \tfrac{\pi}{\sqrt{\cosh(\pi y)}} \, \Gamma(k+1/2)^{-1}  \notag\\
 & <  \sqrt{\pi} e^{k+1/2} (k+1/2)^{-k} \, e^{-\pi\abs{y}/2}. \label{gamma_vert}
\end{align}
Thus, the assumptions of Theorem \ref{thm:hk_exact} are met with
\begin{align*}
\epsilon_k = \tfrac{\pi}{2}, \qquad c_k = 4 \sqrt{\pi} e (e u q^{-1})^{k-1/2} (1-q^{1/2-k})^{-2} (k+1/2)^{-k} 
\end{align*}
yielding $$\lim_{k\to\infty} \big[8 \pi^{-1} e \, (e u q^{-1})^{k-1/2} (1-q^{1/2-k})^{-2} (k+1/2)^{-k}  \big]^{1/(k-1/2)} = 0$$ for $0 < q <1$. 
Hence, $T = \infty$ and, with $r_k = \Oinf(k)$, the expansion \eqref{Podles_heat_S_asym} is absolutely convergent and exact for all $t>0$.
\hfill$\blacksquare$
\end{example}

\section{General Asymptotic Expansions}\label{sec:energy}

It turns out that the asymptotic expansion of the heat trace $\hKH$ associated with some operators $K,H$ actually guarantees the existence of an asymptotic expansion for a larger class of spectral functions of the form $\Tr K f(t H)$. This fact is especially pertinent in our quest to unravel a large-energies asymptotic expansion of the spectral action.

Recall that (cf. \eqref{moments_Lap}) a function $f$ on $\RR^+$ belongs to $\Cl_0^p$ iff $f$ is a Laplace transform of a signed measure $\phi$ on $\RR^+$ and $\int_{0}^{\infty} s^{m} d\vert\phi\vert(s) < \infty$ for all $m > -p$.

\begin{theorem}
\label{thm:f_expansion}
Let $H\in \Tp$, $K \in \B(\H)$ and assume that 
\begin{align}\label{hKH_gen_asymp}
\hKH \tzero \sum_{k=0}^{\infty} \rho_k(t), \quad \text{w.r.t. an asymptotic scale } (t^{r_k})_k,
\end{align}
with $\rho_k \in L^{\infty}_{\text{loc}}((0,\infty))$, $\rho_k(t) = \oz (t^{r_k})$ and $\rho_k(t) = \oinf (t^{r_{k+1}})$, for any $k \in \N$.

Then, for any $f=\Lc[\phi] \in \Cl_0^r$ with $r > p$, the operator $K f(t H)$ is trace-class for any $t>0$ and there exists an asymptotic expansion:
\begin{align}
\label{f_asymptotic}
\Tr K f(tH) \tzero \,\sum_{k=0}^{\infty} \psi_k(t), \quad \text{ with } \quad \psi_k(t) = \int_0^{\infty} \rho_k(s\,t)\, d\phi(s) = \oz(t^{r_{k}}).
\end{align}
\end{theorem}

Let us stress that the choice of the asymptotic scale $(t^{r_k})_k$ is merely a matter of convention --- recall Remark \ref{rem:asymp_exp}. In particular, since Proposition \ref{prop:KH} implies that $\rho_0(t) = \Oz(t^{-p-\epsilon})$ for all $\epsilon >0$, we can shift $r_0$ to lie arbitrarily close to $-p$. Hence, given any $r > p$, we actually have $r > - r_0 > p > -r_1$.

\begin{proof}
Recall first that, on the strength of Lemma \ref{traceclass} the operator $f(t H)$ is trace-class for any $t >0$, and hence $K f(t H)$ is so for any $K \in \B(\H)$.

As pointed out in Section \ref{subsec:Laplace} we have $K f(t H) = \int_0^{\infty} K e^{-s t H} d\phi(s)$ in the strong operator sense. The trace being normal, $\Tr K f(t H) = \int_0^{\infty} \Tr K e^{-s t H} d\phi(s)$ for any $t>0$. By assumption we have, for any $N \in \N$, $\hKH = \sum_{k=0}^{N} \rho_k(t) + R_N(t)$ and hence, $\Tr K f(t H) = \sum_{k=0}^{N} \psi_k(t) + \int_0^{\infty} R_N(st)\, d\phi(s)$. 

Let us first show that $\psi_k(t) = \oz(t^{r_{k}})$ for any $k \in \N$. We have
\begin{align*}
\rho_k(t) = \oz (t^{r_{k}}) & \Leftrightarrow \;  \lim_{t\to 0} \, \,t^{-r_{k}} \abs{\rho_k(t)} = 0 \Leftrightarrow \; \forall \, \epsilon >0 \,\,\, \exists \, \delta > 0 \,\,\, \forall \, t \leq \delta, \,\,\,\,\,\abs{\rho_k(t)} \leq \epsilon t^{r_{k}} \\
& \Leftrightarrow \; \forall \, \epsilon,t >0 \,\,\, \exists \, \delta > 0 \,\,\, \forall \, s \leq \delta t^{-1}, \,\,\, \abs{\rho_k(st)} \leq \epsilon t^{r_{k}}s^{r_{k}}.
\end{align*}
On the other hand, $\rho_k(t) = \oinf(t^{r_{k+1}}) \,\, \Rightarrow \,\, \exists \, \delta' > 0, \,\, \forall \, t \geq \delta', \,\,\abs{\rho_k(t)} \leq t^{r_{k+1}}$.
\\
Since $\rho_k$ is locally bounded on $\RR^+$, we actually have a local uniform bound: For any $\delta''$ with $0< \delta'' \leq \delta'$,
\begin{align*}
\forall \, t \geq \delta'', \quad \abs{\rho_k(t)} \leq M_k(\delta'') \, t^{r_{k+1}}\,\, \text{ with } M_k(\delta'') \vc \text{$\sup$}_{t \in [\delta'',\delta']} \, \,\abs{\rho_k(t)}.
\end{align*}
Hence, for any $t>0$, 
$\forall \, s \geq \delta'' t^{-1}, \,\, \abs{\rho_k(st)} \leq M_k(\delta'') \, t^{r_{k+1}} s^{r_{k+1}}.$

We now pick any $\epsilon > 0$ and $\delta'' = \min \{\delta,\delta'\}$. For any $t >0$ we estimate
\begin{align*}
\abs{\psi_k(t)} & \leq \int_0^{\delta'' t^{-1}} \abs{\rho_k(st)} d\aphi(s) + \int_{\delta'' t^{-1}}^{\infty} \abs{\rho_k(st)} d\aphi(s) \\
& \leq \epsilon t^{r_{k}} \int_0^{\delta'' t^{-1}} s^{r_{k}}  d\aphi(s) + M_k(\delta'') \, t^{r_{k+1}} \int_{\delta'' t^{-1}}^{\infty} s^{r_{k+1}} d\aphi(s) \\
& \leq \epsilon t^{r_{k}} \int_0^{\infty} s^{r_{k}}  d\aphi(s) +  M_k(\delta'') \,t^{r_{k+1}} \int_0^{\infty} s^{r_{k+1}} d\aphi(s).
\end{align*}
Now, recall that $r > - r_0 > p > -r_1 > -r_2 > \dotsc$, as $(r_k)_k$ is strictly increasing. Hence, by Proposition \ref{prop:cl}, with $f \in \CM_0^r$ both integrals $\int_0^{\infty} s^{r_{k}} \, d\aphi(s) \cv c_1$ and $\int_0^{\infty} s^{r_{k+1}} \,d\aphi(s) \cv c_2$ are finite. \\
Hence, we have $t^{-r_k} \abs{\psi_k(t)} \leq \epsilon c_1 + t^{r_{k+1} - r_k} c_2$. Since $\epsilon$ can be taken arbitrarily small we conclude that $\psi_k = \oz(t^{r_k})$ for any $k \in \N$.

It remains to show that $\int_0^{\infty} R_N(st) d\phi(s) = \oz(t^{r_{N}})$ 
for any $N \in \N$. This follows by the same arguments since $R_N(t) = \oz(t^{r_{N}})$ and also $R_N(t) = \oinf(t^{r_{N+1}})$, as $\sum_{k=0}^N \rho_k(t) = \Oinf(\rho_N(t)) = \oinf(t^{r_{N+1}})$, whereas $\hKH = \Oinf(e^{-\lambda_0(H) t})$.
\hfill$\Box$
\end{proof}

\begin{remark}
\label{rem:finite_exp}
Note that if we do not insist on having a complete asymptotic expansion of $\Tr \, K f(tH)$, then we can relax the assumption about the heat trace to the following condition: $\hKH = \sum_{k=0}^{N} \rho_k(t) + \oz(t^{r_N})$, for some $N \in \N$, which implies, for a suitable $f$,
$\Tr \, K f(tH) = \sum_{k=0}^{N} \psi_k(t) + \oz(t^{r_N})$.
\hfill$\blacksquare$
\end{remark} 

\begin{example}
As an illustration of Theorem \ref{thm:f_expansion} we consider $H=\Dslash^2$ and $K = \bbbone$, with $\Dslash$ being the standard Dirac operator on $S^2$ (see Appendix \ref{sec:spheres}). We have (cf. Example \ref{ex:even_sphere_2}), 
$$\Tr e^{-t \,\Dslash^2} \tzero 2\,t^{-1} - 2 \sum_{k=0}^\infty \tfrac{(-1)^k}{k!}\tfrac{B_{2k+2}}{k+1}\,t^k.$$
Following Example \ref{ex:(ax+b)-r} let us take $f(x) = (ax+b)^{-r} \in \Cl_0^r$, with some $a,b,r > 0$. Since $\Dslash^2\in \Tp$ with $p=1$ we can apply Theorem \ref{thm:f_expansion} when $r>1$. From Formula \eqref{f_asymptotic} we obtain
\begin{align*}
\psi_0(t) = 2 \tfrac{b^{1-r}}{a(r-1)}\,t^{-1} , && \psi_{k+1}(t) = -2  \tfrac{(-1)^k}{k!}\tfrac{B_{2k+2}}{k+1} \, \tfrac{a^k b^{-k-r} \Gamma (k+r)}{\Gamma (r)}\,t^{k}, \; \text{ for } k \in \N.
\end{align*}
Hence, we have
\begin{align*}
\Tr f(t \Dslash^2) \tzero 2b^{-r} \big[ \tfrac{1}{r-1} (a/b)^{-1}\,t^{-1} - \sum_{k=0}^{\infty} (-1)^k \tfrac{B_{2k+2}}{k+1} \, \tbinom{r+k-1}{k} \, (a/b)^k\,t^k \big].
\tag*{$\blacksquare$}
\end{align*}
\end{example}

Having established the existence theorem for an asymptotic expansion of \linebreak$\Tr K f(tH)$ we can come back to the issue raised at the end of Section \ref{subsec:Laplace}: Why we cannot allow $f$ to be the Laplace transform of a distribution in $\Sclpp$ rather than of a signed measure?

\begin{remark}\label{rem:Laplace_dist}
Observe that one could derive an analogue of the operatorial Formula \eqref{f(H/Lambda)} by regarding $e^{-s P/\Lambda}$ as an element of the space of operator valued test functions $\K(\RR,\B(X,\H))$ (recall Section \ref{sec: On the asymptotics of distributions}). Along the same line, one could prove the normality of the trace functional (cf. \eqref{Tr f(H/Lambda)}), i.e. $$\Tr f(H/\Lambda) = \<\Lc^{-1}[f],\Tr e^{-s H/\Lambda} \>$$ for $\Lc^{-1}[f] \in \K' \cap \Sclpp$, using the convergence of the sequence of partial sums in $\K$ and invoking the weak-$^*$ topology of $\K' \cap \Sclpp$.\\
On the other hand, the topology of $\K'$ obliges us to control the derivatives of the test functions, on which $\Lc^{-1}[f]$ act. But we do not, in general, have such a control on the remainder of the asymptotic series \eqref{hKH_gen_asymp} even in the standard case of $\rho_k(t) = c_{k} \, t^{-r_k}$ (cf. Remark \ref{rem:asympt_diff}).
\hfill$\blacksquare$
\end{remark}

\medskip

Let us now turn to the case of convergent expansions.

\begin{theorem}
\label{thm:f_conv}
Let $H\in \Tp$ and let $K \in \B(\H)$. Assume that for $t \in (0,T)$ with some $T \in (0,+\infty]$,
\begin{align}
\label{assum:conv}
\hKH = \sum_{k=0}^{\infty} \rho_k(t) + R_{\infty}(t) \;\; \text{ and } \;\; \sum_{k=0}^{\infty} \abs{\rho_k(t)} < \infty
\end{align}
with $\rho_k \in L^{\infty}_{\text{loc}}((0,T))$ and $\rho_k(t) = \oz (t^{r_k})$ for any $k \in \N$, and $R_{\infty}(t) = \Oz(t^{\infty})$.

Then, for any $f \in \Cl_c^r$ with $r > -r_0 > p$, the series $\sum_{k=0}^{\infty} \psi_k(t)$, with the functions $$\psi_k(t) \vc \int_0^{\infty} \rho_k(s\,t) \,d\phi(s)$$ is absolutely convergent on the interval $(0,T/N_f)$ with $$N_f \vc \inf \{ N \, \vert \, \supp \, \Lc^{-1}[f] \subset (0,N) \}.$$
 Moreover, for $t \in (0,T/N_f)$, 
\begin{align*}
\Tr K f(tH) = \,\sum_{k=0}^{\infty} \psi_k(t) + R^f_{\infty}(t),\quad \text{where} \quad R^f_{\infty}(t) \vc \int_0^{\infty} R_{\infty}(s\,t) \,d\phi(s) = \Oz(t^{\infty}).
\end{align*}
\end{theorem}

\begin{proof}
On the strength of Formula \eqref{Tr f(H/Lambda)} and by assumption \eqref{assum:conv} we have
\begin{align*}
\Tr K f(t H) = \int_0^{\infty} \Tr K e^{-st H} \, d\phi(s) = \int_0^{N_f} \sum_{k=0}^{\infty} \rho_k(st) \,  d\phi(s)  + \int_0^{N_f} R_{\infty}(st)\, d \phi(s).
\end{align*}
Observe first that the sequence of partial sums is uniformly bounded, i.e. for any $N \in \N$ and any $t \in (0,T/N_f)$, $\big\vert \sum_{k=0}^{N} \rho_{k}(st) \big\vert \leq  \sum_{k=0}^{\infty} \abs{ \rho_{k}(st) } < \infty$. Via the Lebesgue dominated convergence theorem, we obtain $\Tr K f(tH) = \,\sum_{k=0}^{\infty} \psi_k(t) +  \displaystyle{R^f_{\infty}}(t)$ and it remains to show that $\displaystyle{R^f_{\infty}}(t) = \Oz(t^{\infty})$. To this end, let us note that
\begin{align*}
\forall \, k>0 \quad R_{\infty}(t) =\Oz (t^{k})   & \Leftrightarrow \;  \forall \, k >0 \quad\, \exists \, M, \delta > 0 \,\,\, \forall \, t \leq \delta, \,\!\!\!\qquad\abs{R_{\infty}(t)} \leq M t^{k} \\
& \Leftrightarrow \; \forall \, k,\, t >0 \,\,\, \exists \, M, \delta > 0 \,\,\, \forall \, s \leq \delta t^{-1}, \,\,\, \abs{R_{\infty}(st)} \leq M t^{k}s^{k}.
\end{align*}
Then, for any $t \in (0,N_f)$ and any $k >0$, we have
\begin{align*}
\abs{R^f_{\infty}(t)} &= \abs{ \int_0^{\infty} R_{\infty}(s\,t) \,d\phi(s) }  \leq \int_{0}^{\delta t^{-1}} \abs{R_{\infty}(st)} \, d \aphi(s) + \int_{\delta t^{-1}}^{\infty} \abs{R_{\infty}(st)} \, d \aphi(s) \\
&\leq M t^k \int_0^{\infty} s^k \, d\aphi(s) + \int_{\delta t^{-1}}^{\infty} \abs{R_{\infty}(st)} \, d \aphi(s).
\end{align*}
As $f \in \Cl_c$, we have $\int_0^{\infty} s^k \, d\aphi(s) < \infty$ for any $k>0$ and $\int_{\delta t^{-1}}^{\infty} \abs{R_{\infty}(st)} \, d \aphi(s) = 0$ for $t < \delta N_f^{-1}$. Since $k$ can be arbitrarily large we conclude that $\displaystyle{R^f_{\infty}}(t) = \Oz(t^{\infty})$.
\hfill$\Box$
\end{proof}

In general, the compactness of the support of $\Lc^{-1}[f]$ is necessary even if the expansion of the heat trace at hand is actually convergent for all $t>0$. This is because if one would like to allow for $f \in \Cl_0^p$, rather than $f \in \Cl_c^p$, one would need to control the behaviour of $R_{\infty}(t)$ as $t \to \infty$. On the other hand, such a control is (trivially) provided if the expansion of $\hKH$ is exact for all $t>0$.

\begin{corollary}
\label{cor:f_exact}
Let $K$ and $H$ meet the assumptions of Theorem \ref{thm:f_conv} with $T = +\infty$ and $R_{\infty} = 0$.
Then, for any $f \in \Cl_0^r$ with $r > -r_0 > p$, and any $t>0$ we have 
\begin{align*}
\Tr K f(tH) = \,\sum_{k=0}^{\infty} \psi_k(t),\quad \text{with} \quad \psi_k(t) \vc \int_0^{\infty} \rho_k(s\,t) \,d\phi(s).
\end{align*}
\end{corollary}
\begin{proof}
Obviously, $\displaystyle{R_{\infty}^f} = 0$. Hence, it suffices to observe that in the first part of the proof of Theorem \ref{thm:f_conv} we can safely allow for $N_f = +\infty$, if the series $\sum_{k} \rho_k(st)$ is absolutely convergent for any $s,t>0$.
\hfill$\Box$
\end{proof}

On the other hand, we insist that the full force of condition \eqref{moments_Lap} is necessary for Theorem \ref{thm:f_conv} to hold even if the convergent expansion of $\hKH$ has a finite number of terms (recall Section \ref{sec:conv_non}). Let us illustrate this fact:

\begin{example}
\label{ex:SA_almost_non}
In Example \ref{ex:conv_non} we have for any $t>0$, $$\Tr e^{-t\, \Dslash^2} = (\tfrac{\pi}{t})^{1/2} + R_{\infty}(t).$$
Let us try to apply Theorem \ref{thm:f_conv} with $f(x) = e^{-\sqrt{x}}$, which is a completely monotone function, but $f \notin \Cl_0$ (recall Example \ref{ex:exp_sqrt}). Since,
\begin{align*}
\psi_0(t) = \tfrac{1}{2 \sqrt{\pi } } \, t^{-1/2} \, \int_0^{\infty} s^{-1/2} \, s^{-3/2}e^{-1/4 s} \, ds = 2 t^{-1/2},
\end{align*}
we would get $\Tr f(t^2 \Dslash^2) = \Tr e^{-t \abs{\Dslash}} = 2 t^{-1} + \displaystyle{R^f_{\infty}}(t)$. \\
A comparison with Example \ref{ex:exact} gives $\displaystyle{R^f_{\infty}}(t) = t/6 + \Oz(t^3)$. Hence, whereas $R_{\infty}(t) = \Oz(t^{\infty})$, but $\displaystyle{R^f_{\infty}}(t) \neq \Oz(t^{\infty})$.
\hfill$\blacksquare$
\end{example}

An important question is whether one can tailor a cut-off function such that the asymptotic expansion of $\hKH$ turns to a convergent one of $\Tr K f(tH)$. When $\hKH \sim_{t\downarrow 0} \sum_{k=0}^{\infty} a_k t^k$, as happens always in the context of classical differential operators (cf. Example \ref{ex:Gilkey_diff}), one could seek an $f$ with null Taylor expansion at 0 --- as recognised in \cite{ConnesUncanny}. Unfortunately, the natural candidate -- a bump function (cf. \cite[Fig. 1]{ConnesUncanny})  is not a Laplace transform as explained in Remark \ref{rem:Gauss_Laplace}. It is also clear from Proposition \ref{prop:moments} that $f$ cannot be completely monotone. Nevertheless, one can find functions with a null Taylor expansion at 0, which are in $\Cl_0^p$ for some $p$.

\begin{example}
\label{ex:SA_null}
Let $\phi(s) = e^{-s^{1/4}}  \sin (s^{1/4})$ for $s>0$. This non-positive function has all moments vanishing, i.e.
$\int_0^{\infty} s^n \phi(s) ds = 0,\,\forall n \in \N$ (cf. \cite[Example 3.15]{Romano1986}). Moreover, its absolute moments are finite: 
$$\int_0^{\infty} s^n \aphi(s) \,ds \leq \int_0^{\infty} s^n \,e^{-s^{1/4}}\,ds = 4 (4n+3)!$$
The Laplace transform of $f=\Lc[\phi]$ can be obtained with the help of Mathematica yielding an analytic, though rather uninviting, formula
\begin{align*}
f(x) = \tfrac{\Gamma( 5/4)}{x^{5/4}} \, _0F_2\left(\tfrac{1}{2},\tfrac{3}{4};-\tfrac{1}{64
   x}\right)
   -\tfrac{\sqrt{\pi }}{2 x^{3/2}} \, _0F_2\left(\tfrac{3}{4},\tfrac{5}{4};-\tfrac{1}{64 x}\right) - 
   \tfrac{\Gamma (-5/4)}{16 x^{7/4}} \, _0F_2\left(\tfrac{5}{4},\tfrac{3}{2};-\tfrac{1}{64 x}\right),
\end{align*}
with the hypergeometric function $_0F_2(a,b;z) \vc \sum_{k = 0}^{\infty} \tfrac{\Gamma(a) \Gamma(b)}{\Gamma(a+k) \Gamma(b+k)} \,\tfrac{z^k}{k!}$.\\
One can check that $f(x) > 0$ for $x>0$ and, since $\lim_{x \to 0^+} {}_0F_2(b_1,b_2;x) = 1$, we get $f(x) = \Oinf(x^{-5/4})$. With the help of Proposition \ref{prop:cl}, one obtains a cut-off function $f^n$ for some $n \in \N^*$, which has null Taylor expansion at 0.
\hfill$\blacksquare$
\end{example}

Let us emphasise that the cut-off function exemplified above would do the trick of turning an asymptotic expansion to a convergent one only if the expansion of $\hKH$ has a very specific form: $\sum_{k=0}^{\infty} a_k t^k$. In general, the expansion \eqref{f_asymptotic} would involve the full shape of the cut-off function, via the integrals $\int_0^{\infty} s^k \log^n(s) \,d \phi(s)$.

\section{Asymptotic Expansion of the Spectral Action}\label{sec:exp_SA}

Equipped with the general results presented in the preceding section we are at a position to formulate sufficient conditions for the existence of a large-energies asymptotic expansion of the spectral action $\SA$ associated with a given spectral triple.

Recall first that the cut-off function $f \in \Cl_0$ is smooth at 0 (Proposition \ref{prop:moments}). In particular, $f(0) < \infty$ and if the operator $T f(t \awD)$ is trace-class then so is $T f(t \abD)$. Moreover, from the spectral decomposition $\DD=0\,P_0 +\sum_{\lambda_n\neq 0} \lambda_n P_n$, we immediately get $$f(\abD)=f(0)P_0 +\sum_{\lambda_n\neq 0} f(\abs{\lambda_n}) P_n, \qquad f(\awD)=f(1)P_0 +\sum_{\lambda_n\neq 0} f(\abs{\lambda_n}) P_n,$$ and
\begin{align}
\label{SA_ker}
\Tr \, T f(\abD/\Lambda) & = \Tr \, T f(\awD/\Lambda) + \big(f(0) - f(\Lambda^{-1})\big) \, \Tr TP_0,
\\ & = \Tr \, T f(\awD/\Lambda) + \Oinf(\Lambda^{-1}). \notag
\end{align}

\begin{remark}
From this formula we see that the kernel of $\DD$ becomes irrelevant in the physical action at large energies. It also explains why we have not defined the spectral action \eqref{SA} as $\Tr \, T f(\awD/\Lambda)$ despite the simplification of working with an invertible operator. But the kernel pops up if we choose $\bar{D}$ instead of $D$:
\begin{align*}
\Tr \, T f(\abD/\Lambda) - \Tr \, T f(\bar{D}/\Lambda) = f(0)\, \Tr TP_0= \OO_\infty(\Lambda^0).
\tag*{$\blacksquare$}
\end{align*}
\end{remark}

\begin{theorem}\label{thm:ST_SA}
Let $\ahd$ be a $p$-dimensional spectral triple and $T \in \B(\H)$. Assume that there exists $d \in \N$, a sequence $(r_k)_{k \in \N} \subset \RR$ strictly increasing to $+ \infty$ and a discrete set $X \subset \CC$ without accumulation points, such that
\begin{align}
\label{ST_hk_exp}
\hTD \tzero \sum_{k=0}^{\infty} \rho_k(t),\quad \text{ with } \quad \rho_k(t) =  \sum_{z \in X_k} \,\big[\sum_{n=0}^{d} a_{z,n}(T,\abs{D}) \log^n t\, \big]\,t^{-z}
\end{align}
where $X_k \vc \{z \in X \, \vert \, - r_{k+1} < \Re(z) < - r_{k} \}$, $X = \sqcup_{k} X_k$ and the series defining $\rho_k(t)$ is absolutely convergent for any $t>0$ and any $k \in \N$. 

\medskip
\noindent Then,

i) \label{ST_SA:coef} The function $\zTD$ admits a meromorphic extension to the whole complex plane with the poles of order at most $d+1$ and $\PP(\ZZ_{T,D}) \subset X$. 

Moreover, for any $z \in X$ and $n \in \{0,1,\ldots,d\}$,
\begin{align}
\label{coef_a}
a_{z,n}(T,\abs{D}) = \tfrac{(-1)^n}{n!}\,\, \sum_{\ell = n}^{d+1} \,\Gamma_{\ell-n-1}(z)  \, \ncintd{\ell} T \awD^{-z},
\end{align}
with $\Gamma_j(z)$ denoting the $j$-th coefficient of the Laurent expansion around $z$ of $\Gamma$, i.e. $\Gamma_j(z) \vc \Rez{s=z} (s-z)^{-j-1} \Gamma(s)$, for $j \in \N-1$.

 ii) For any $f=\Lc[\phi] \in \Cl_0^r$ with $r > p$, 
\begin{align}
\label{ST_SA:SA} 
 \Tr \, T f(\awD/\Lambda) \Linf \sum_{k=0}^{\infty} \psi_k(\Lambda), \quad \text{ w.r.t. the scale } (\Lambda^{-r_k})_k,
 \end{align}
with
\begin{align}
\psi_k(\Lambda) &= \sum_{z \in X_k} \Lambda^z \, \sum_{n=0}^d (-1)^n \log^n \Lambda \, \sum_{m = n}^d \, \tbinom{m}{m-n}\, a_{z,m}(T,\abs{D}) \, f_{z,m-n} \, , \label{ST_SA:Y}\\ 
f_{z,n} &\vc \int_{0}^{\infty} s^{-z} \log^n(s) \,d\phi(s). \label{ST_SA:f}\index{_zF0z@$f_{z,n}$}
\end{align}
\end{theorem}

Let us emphasise that, as in Theorem \ref{thm:f_expansion}, the onliest role of the sequence $(r_k)_k$ is to fix an asymptotic scale.

\begin{proof}
$i)$ Since the demand \eqref{ST_hk_exp} exactly mimics the assumptions of Theorem \ref{thm:hk2zeta} the first of part of \textit{i)} follows. Then, Formula \eqref{coef_a} arises from Equation \eqref{zKH_res} followed by some easy combinatorics:

\noindent Since $\zeta_{T,D}$ is meromorphic on $\CC$ we have the following Laurent expansions in some punctured neighbourhood around any $z \in X$,
\begin{align*}
& \zeta_{T,D}(s)  = \sum_{\ell=-d-1}^\infty b_\ell(z) (s-z)^\ell, \qquad \Gamma(s) = \sum_{j=-1}^\infty \Gamma_j(z) (s-z)^j,\\
& \ZZ_{T,D}(s)  = \zeta_{T,D}(s)\, \Gamma(s) = \sum_{m=-d-2}^{\infty} c_m(z) (s-z)^m,
\end{align*}
with $c_m(z) = \sum_{\ell = -d-1}^{m+1} b_\ell(z) \,\Gamma_{m-\ell}(z)$. Furthermore,
\begin{align*}
b_\ell(z) = \Rez{s=z} (s-z)^{-\ell-1} \zeta_{T,D}(s) = \Rez{s=0} \, s^{-\ell-1} \Tr T \awD^{-s-z} = \ncintd{-\ell} T \awD^{-z}.
\end{align*}
Now, Formula \eqref{zKH_res} yields, for any $z \in X$, $n \in \{0,1,\ldots,d\}$,
\begin{align*}
a_{z,n}(T,\abs{D})  & = \tfrac{(-1)^n}{n!} \, \Rez{s=z}\, (s-z)^{n} \,\Gamma(s) \zeta_{T,D}(s) = \tfrac{(-1)^n}{n!} \, c_{-n-1}(z) \\
& = \tfrac{(-1)^n}{n!} \, \sum_{\ell = -d-1}^{-n} b_\ell(z) \,\Gamma_{-n-1-\ell}(z) = \tfrac{(-1)^n}{n!} \, \sum_{\ell = n}^{d+1} \Gamma_{\ell-n-1}(z) \ncintd{\ell} T \awD^{-z} \,.
\end{align*}

$ii)$ Let us now turn to assertion \eqref{ST_SA:SA}. We have demonstrated in the course of the proof of Theorem \ref{thm:hk2zeta} that $\rho_k$ defined as in \eqref{ST_hk_exp} satisfies $\rho_k(t) = \oz(t^{r_k})$. Using the same arguments, one shows that $\rho_k(t) = \oinf(t^{r_{k+1}})$. Moreover, $\rho_k \in L^{\infty}_{\text{loc}}((0,\infty))$ and since $\awD \in \Tp$, Theorem \ref{thm:f_expansion} yields, with $t = \Lambda^{-1}$,
\begin{align*}
\Tr \, T f(\awD/\Lambda) \Linf \sum_{k=0}^{\infty} \psi_k(\Lambda), \quad \text{with} \quad \psi_k(\Lambda) = \int_0^{\infty} \rho_k(s / \Lambda)\, d\phi(s).
\end{align*}
It remains to compute the explicit form of $\psi_k$. By assumption, the series defining every $\rho_k(t)$ are absolutely convergent for any $t>0$ so that we can invoke the Lebesgue dominated convergence theorem to commute the sum over $X_k$ with the integral over $d\phi$ (cf. the proof of Theorem \ref{thm:f_conv}). In effect, we obtain
\begin{align*}
\psi_k(\Lambda) & = \sum_{z \in X_k} \Lambda^z \, \sum_{m=0}^d a_{z,m}(T,\abs{D}) \int_0^{\infty} \log^m(s/\Lambda) \,s^{-z} \,d \phi(s) \\
& = \sum_{z \in X_k} \Lambda^z \, \sum_{m=0}^d a_{z,m}(T,\abs{D}) \sum_{j = 0}^{m} \tbinom{m}{j} \log^{m-j}(\Lambda^{-1})\, f_{z,j} \\
& = \sum_{z \in X_k} \Lambda^z \, \sum_{n=0}^d (-1)^n \log^n \Lambda \, \sum_{m = n}^d \, \tbinom{m}{m-n} \,a_{z,m}(T,\abs{D}) \, f_{z,m-n}.
\tag*{$\Box$}
\end{align*}
\end{proof}

Formula \eqref{coef_a} shows how to compute the noncommutative integrals $\ncintd{\ell}T\awD^{-z}$ 
given the coefficients $a_{z,n}(T,\abs{D})$: For instance, since $\Gamma_{-1}(-k) \neq 0$ for $k \in \N$, it implies that $\ncintd{d+1} T \awD^{k} = 0\text{ for }T \in \PDOk{0}\text{and }k \in \N.$ More generally:

\begin{corollary}
Under same hypothesis but with $T\in \Psi^0(\A)$. For $-z = k \in \N$, 
\begin{align*}
&\ncintd{d+1} T \awD^k  = 0,\\
&\ncintd{d} T \awD^k  = (-1)^{d+k} d! \, k! \, a_{-k,d}(T,\abs{D}), \quad\text{ for } d \geq 0,\\
&\ncintd{d-1} T \awD^k = (-1)^{d+k-1} (d-1)! \, k! \big[a_{-k,d-1}(T,\abs{D})  \\
& \hspace{+5cm} + (-1)^{k} d \, k! \, \Gamma_0(-k) \, a_{-k,d}(T,\abs{D}) \big], \,\,\text{ for } d \geq 1;
\intertext{whereas, for $z \in \CC \setminus (-\N)$, in which case $\Gamma_{-1}(z) = 0$ and $\Gamma_0(z) = \Gamma(z)$, we obtain}
&\ncintd{d+1} T \awD^{-z}  = \tfrac{(-1)^{d} d!}{\Gamma(z)} \, a_{z,d}(T,\abs{D}), \quad\text{ for } d \geq 0,\\
&\ncintd{d} T \awD^{-z}  = \tfrac{(-1)^d (d-1)!}{\Gamma(z)^2} \, \big[ \Gamma(z) \, a_{z,d-1}(T,\abs{D}) + \Gamma_1(z) \, d \, a_{z,d}(T,\abs{D}) \big], \quad\text{ for } d \geq 1.
\end{align*}

If $d=0$, then $\zeta_{T,D}$ is regular at $-k \in -\N$ with $$\zeta_{T,D}(-k) = (-1)^k k! a_{-k,0}(T,\abs{D})$$ and we have $\ncint \bbbone =0$ (cf. Example \ref{ex:ncint}).
\end{corollary}

The explicit formulae for the noncommutative integrals become rather involved when $d$ is large, but otherwise can be computed recursively from Formula \eqref{coef_a}.

\begin{remark}
\label{rem:dimsp}
Let $\ahd$ be regular and let the assumptions of Theorem \ref{thm:ST_SA} be met for any $T \in \PDOk{0}$ with a \emph{common} $d$ and $X_k$'s (i.e. there exist $d, X_k$ such that for any $T \in \PDOk{0}$, $d(T) = d$ and $X_k(T) = X_k$). Then, $\ahd$ has a dimension spectrum $\Sd$ of order at most $d+1$ and $\Sd \cup (-\N) \subset X$. 

Let us emphasise that this is a one-way street, even if $\ahd$ is known to possess a simple dimension spectrum included in $\RR$. As we have seen in the previous sections, the existence of a small-$t$ asymptotic expansion of a heat trace is a stronger property than the existence of a meromorphic extension of the associated spectral zeta function. The assumptions about the growth rate of the $\zeta_D$-function on the verticals are indispensable, as they allow us to control the contour integrals involved in the Mellin inversion theorem and hence manage the remainder in the asymptotic series. Its relevance was recognised by Alain Connes and Henri Moscovici in \cite[p. 206]{ConnesMoscovici} in the context of index theorems. See also Problem \ref{prob:heat vs zeta} \ref{prob:zeta 2 heat} in Chapter \ref{chap:open}.
\hfill$\blacksquare$
\end{remark}

\begin{example}
For every commutative spectral triple associated with a Riemannian compact manifold $M$ the assumptions of Theorem \ref{thm:ST_SA} are met --- cf. Example \ref{ex:Gilkey_pseudodiff}. In this context we have $d = 1$, $X = \dim M - \N$ and  $X_k = \{\dim M-k\}$. This extends naturally to almost commutative geometries, as they are embedded in the framework of purely classical pseudodifferential operators \cite{BoeijinkDungen,CacicReconstruction}.
\hfill$\blacksquare$
\end{example}

As a straightforward corollary of Theorem \ref{thm:ST_SA} we obtain the following, slightly enhanced, known result \cite[Theorem 1.145]{ConnesMarcolli}:

\begin{corollary}\label{cor:SA_simple}
Let $\ahd$ be a regular $p$-dimensional spectral triple with a dimension spectrum $\Sd$, denote $\Sd^+ \vc \{z \in \Sd \, \vert \, \Re(z) > 0\}\index{_zSzp@$\Sd^+$}$ and let $T \in \PDOz$. Moreover, assume that
\begin{align}\label{heat_simplified}
\Tr \, T e^{-t \abD} = \sum_{\alpha \in \Sd^+} a_{\alpha}(T,\abs{D}) \, t^{-\alpha} + a_{0}(T,\abs{D}) + \oz(1),
\end{align}
where the, possibly infinite, series over $\alpha$ is absolutely convergent for all $t>0$.\\
Then, the function $\zTD$ is regular at $0$ and for any $f \in \Cl_0^r$ with $r>p$,
\begin{align}
\label{SA_simplified}
\Tr \, T f( \abD / \Lambda ) \!=\! \!\!\!\!\sum_{\alpha \in \Sd^+}\!\! \!\!\Lambda^{\alpha} \!\!\int_0^{\infty}\!\! \!\!x^{\alpha - 1} f(x) dx  \,\ncint T \awD^{-\alpha} + f(0) \zTD(0) + \oinf(1).
\end{align}
Alternatively, if in place of \eqref{heat_simplified} one assumes
\begin{align}\label{heat_simplified_D sq}
\Tr \, T e^{-t \DD^2} = \tfrac{1}{2} \sum_{\beta \in \Sd^+} a_{\beta/2}(T,D^2) \, t^{-\beta/2} + a_{0}(T,D^2) + \oz(1),
\end{align}
then, the expansion \eqref{SA_simplified} holds for $\big[ x \mapsto f(\sqrt{x}) \big] \in \Cl_0^r$ with $r> p/2$.
\end{corollary}
\begin{proof}
The existence of the expansion \eqref{SA_simplified} follows from Theorem \ref{thm:f_expansion} enhanced by Remark \ref{rem:finite_exp}, whereas the explicit form of the coefficients is provided by Formula \eqref{ST_SA:Y} from Theorem~\ref{thm:ST_SA}. The regularity of $\zTD$ at 0 follows from Formula \eqref{coef_a} and $a_{0,1}(T,\abs{D}) = 0$. Furthermore, $$a_{0,0}(T,\abs{D}) = \ncintd{0} T = \Rez{s=0}\, s^{-1}\, \zTD(s) = \zTD(0).$$ Finally, we observe that $f(0) = \int_0^{\infty} d \phi(s)<\infty$ (cf. Proposition \ref{prop:moments}) and moreover, for $\Re(\alpha) > 0$ (see the proof of Proposition \ref{prop:cl}),
\begin{align*}
f_{\alpha,0} = \int_0^{\infty} s^{-\alpha} \, d\phi(s) = \tfrac{1}{\Gamma(\alpha)} \, \int_0^{\infty} x^{\alpha-1} \int_0^{\infty} e^{-sx} d \phi(s) dx = \tfrac{1}{\Gamma(\alpha)} \, \int_0^{\infty} x^{\alpha-1} f(x) dx,
\end{align*}
whereas $$a_{\alpha,0}(T,\abs{D}) = a_{\alpha}(T,\awD) = \Gamma(\alpha) \, \ncint T \awD^{-\alpha}.$$

\smallskip

For the alternative statement with hypothesis \eqref{heat_simplified_D sq}, we have with $f(x) = g(x^2)$, 
\begin{align*}
\Tr \, T f (\abD / \Lambda ) & = \Tr \, T g( \DD^2 / \Lambda^2 ) = \int_0^{\infty} \Tr \, T e^{-s \DD^2/\Lambda^2} d\phi_g(s) \\
& = \tfrac{1}{2}  \sum_{\beta \in \Sd^+} \Lambda^{\beta} a_{\beta/2}(T,D^2) \int_0^{\infty} s^{-\beta/2} d\phi_g(s) + g(0) \zeta_{T,D^2}(0) + \oz(1) \\
& = \tfrac{1}{2}  \sum_{\beta \in \Sd^+} \Lambda^{\beta} \int_0^{\infty} x^{\beta/2-1} g(x)\,dx + f(0) \zeta_{T,D}(0) + \oz(1)  \\
& = \sum_{\beta \in \Sd^+} \Lambda^{\beta} \int_0^{\infty} y^{\beta-1} f(y) \,dy + f(0) \zeta_{T,D}(0) + \oz(1).
\tag*{$\Box$}
\end{align*}
\end{proof}

\begin{remark}\label{rem:f}
Let us emphasise that the remainder in \eqref{SA_simplified} will, in general, involve the full shape of the cut-off function $f$, even if $\ahd$ has a simple dimension spectrum. To conclude that $\oinf(1)$ depends only on the Taylor expansion of $f$ at 0 one would need a considerably stronger assumption
\begin{align*}
\Tr \, T e^{-t \abD} \tzero \sum_{\alpha \in \Sd^+} a_{\alpha}(T,\abs{D}) \, t^{-\alpha} + a_{0}(T,\abs{D}) + \sum_{k=0}^{\infty} a_{-k}(T,\abs{D}) \,t^{-k},
\end{align*} 
in place of \eqref{heat_simplified}. This brings to focus the role of the cut-off function $f$ in the dynamics of the physical system modelled via the triple $\ahd$. See Problem~\ref{The role of f} in Chapter \ref{chap:open}.
\hfill$\blacksquare$
\end{remark}

\begin{remark}
\label{rem:D_vs_Dsq}
As one can see, the swap between the alternative hypotheses \eqref{heat_simplified} and \eqref{heat_simplified_D sq} is not at all innocent. Although they both yield the same Formula \eqref{SA_simplified}, but for $f$ in very different classes $\Cl$. Moreover, the existence of asymptotics of $\Tr \, T e^{-t \DD^2}$ cannot be easily deduced from the asymptotics of $\Tr \, T e^{-t \abD}$ via the Laplace transform --- see Example \ref{ex:exp_sqrt}. The converse implication is also not guaranteed, as observed in Remark \ref{rem:Gauss_Laplace}. The existence of both asymptotic expansions is, however, guaranteed if we know that the meromorphic extension of the corresponding zeta function $\zTD$ is of polynomial growth on the verticals --- cf. Corollary \ref{cor:heat_q}. See Problem \ref{prob:heat vs zeta} \ref{Existence of fluctuated expansion} in Chapter \ref{chap:open}.
\hfill$\blacksquare$
\end{remark}

\medskip

Theorem \ref{thm:ST_SA} allows us to unveil the large-energies asymptotic expansion of the spectral action $\SA$ with a \emph{fixed} operator $\DD$. In order to understand the complete picture of dynamics, we need to fathom out how does the spectral action, and its asymptotics, behaves under the fluctuations of $\DD$ --- recall Section \ref{sec:fluc}. This will be our goal in the next chapter.

%
%
 \bibliographystyle{spmpsci}
 \bibliography{SA_BIB}
%

%
%
%

\chapter{Fluctuations of the Spectral Action}
\label{chap:perturbations}

\abstract{
As we have learned in Section \ref{sec:fluc} a given spectral triple $\ahd$ ought to be considered as a representative of the entire family of triples $\ahda$, which yield equivalent geometries. It is therefore of utmost importance to understand how the spectral action is affected by the fluctuations of geometry.\\
We explore the meromorphic structure of the fluctuated zeta function and, for regular spectral triples with simple dimension spectra, we provide a few formulae for the noncommutative integrals. Finally, we sketch the method of operator perturbations.
}

\section{\texorpdfstring{Fluctuations of the Spectral $\zeta$-Function}{Fluctuations of the spectral zeta-function}}
\label{sec: Fluctuations of spectral functions}

We have seen in Chapter \ref{chap:asymptotic} that the spectral zeta function plays a pivotal role in the asymptotic expansion of the spectral action. In this section we study the relationship of the spectral functions $\zeta_D$ and $\zeta_{D_\Ag}$.\index{_a3zeta_Da@$\zDA(s)$} To this end, we will work with a regular spectral triple $\ahd$ and provide an operator expansion of $\vert D_\Ag \vert^{-s}$ for $s \in \CC$.

Let $\Ag$ be an abstract pdo such that $\Ag=\Ag^* \in \Psi^0(\A)$. This hypothesis encompasses both the standard fluctuations $\Ag= A +\eps JAJ^{*}$, with $A=A^*\in \Omega_\DD^1(\A)$ (cf. p.~\pageref{A without first}) and the more general ones considered in \cite{ConnesFirst}.

Recall from \eqref{def:D_A} that $\DD_\Ag =\DD+\Ag$ and define 
$$D_\Ag \vc \DD_\Ag + P_\Ag,\index{_zD2a@$D_\Ag$}$$
where $P_\Ag$\index{_zP0a@$P_\Ag$} is the projection on $\Ker \DD_\Ag$. We implicitly exclude the case $\DD_\Ag = 0$, which can occur only if $\dim \H < \infty$ (recall p. \pageref{finite triple case}). \\
Let us remark that $\DD_\Ag\in \Psi^1(\A)$ and, as already noted after Definition \ref{def:PDO}, 
\begin{align}
\label{eq:PA smoothing}
V_\Ag:= P_\Ag - P_0\text{  is a smoothing pdo}.
\end{align}
This follows from $\Ker \DD_\Ag \subset \bigcap_{k\geq 1} \Dom \DD_\Ag^k \subset  \bigcap_{k\geq 1} \Dom \abs{D}^k$, the fact that $\abs{D}^r P_\Ag$ and $P_\Ag \abs{D}^s$ are bounded operators for any $r,s\in \RR$ since $\text{Im }P_\Ag$ is finite dimensional, and for any $k\in \N,\, t\in \RR$ the expression $\del^k(P_\Ag)\abs{D}^t$, being a linear combination of terms of the form $\abs{D}^rP_\Ag \abs{D}^s$, is also bounded. Thus, $P_\Ag$ is a smoothing pdo as $P_0$, see Example~\ref{example:pdo} (cf. \cite[Lemma 4.1]{TorusSA} for details).

Using \eqref{OP_z}, define for any $T \in \PDO$, $n\in \N^*$, $s\in \CC$, and $dt:=dt_1\cdots dt_n$,
\begin{align}\label{Kn}
K_n(T,s):=(-\tfrac{s}{2})^n \int_{0\leq t_1\leq\cdots\leq t_n\leq 1}
\sigma_{-s\,t_1}(T)\cdots\sigma_{-s\,t_n}(T)\, \,dt,
\end{align}

Recall that if $T\in \OP^\alpha$, then $\sigma_z(T)\in \OP^\alpha$ for any $z\in \CC$ (cf. Proposition \ref{prop:OP}) and the integrand in \eqref{Kn} is in $\OP^{\alpha n}$, continuous in the strong topology in integration variables varying over a compact set. Hence, $K_n(T,z) \in \OP^{\alpha n}$.

We shall need the following technical operators
\begin{align}
& X  \vc \DD_{\Ag}^2-\DD^2 =\Ag \DD + \DD \Ag + \Ag^2 ,\quad X_V  \vc X+V_\Ag, \label{X}\\
& Y \vc \log(D_\Ag^2) -\log (D^2)=\log(D^2+X_V)-\log(D^2),  \label{Y}
\end{align}
\index{_zxV@$X,\,X_V$}(which are well defined since $\Ag$ preserves $\Dom \DD$ and $D_\Ag^2$ is invertible), so that 
$
D^2+X_V=(\DD+\Ag)^2+P_\Ag=\DD^2_\Ag+P_\Ag= D^2_\Ag.
$\\
We store for the sequel that $X_V \sim X$. In the following we use the multi-index notation $\ell=(\ell_1,\cdots,\ell_n)\in \N^n$ with $\abs{\ell} \vc \ell_1 + \ldots + \ell_n$ and the following complex function $h$
\begin{align}
\label{h}
h_n(s;\ell) & \vc  \left(- \tfrac{s}{2} \right)^n \int_{0\leq t_1 \leq \cdots \leq t_n \leq 1} \tbinom{ -s\,t_1/2}{\ell_1} \cdots \tbinom{ -s\,t_n/2}{\ell_n}\, dt,\quad \ell\in \N^n,
\end{align}
which is a polynomial since $\binom{s}{\ell_i} = \tfrac{ s (s+1) \cdots (s-\ell_i+1)}{\ell_i!}$ for $ s \in \CC, \, \ell_i \in \N^*$ and $\binom{s}{0} = 1$:
\begin{align}
\label{eq: h(s,l)}
h_n(s;\ell) = \sum_{j=0}^{\abs{\ell}} h_j(\ell)\, s^{n+j},  \quad \text{for } \;\ell \in \N^n.
\end{align}

\begin{lemma}
\label{lem:Y is a pdo}
 The operator $Y$ has the following expansion for any $N\in\N^*$,
\begin{align}
\hspace*{-0.2cm}Y &= \sum_{n=1}^{N}\,\sum_{\ell_1, \dotsc,\ell_n=0}^{N-n}\hspace{-0.3cm} \tfrac{(-1)^{\vert \ell \vert+n+1}}{\vert \ell \vert+n} \nabla^{\ell_n}\big(X\nabla^{\ell_{n-1}}(X\cdots \nabla^{\ell_1}(X))\big) \,D^{-2(\vert\ell \vert+n)}\,\,\, \mathrm{mod}\; \OP^{-N-1}.\!\!\!
 \label{Y in OP}
\end{align}
So, $$Y \!= XD^{-2}-\tfrac{1}{2}( \nabla(X) D^{-4} +X^2 D^{-4}) + \dotsc,$$
hence $Y \in\Psi'(\A)\cap \OP^{-1} \subset \Psi^{-1}(\A)$.
\end{lemma}

\begin{proof}
Remark that, although $Y$ depends on $X_V$ we could replace it with $X$ in \eqref{Y in OP} as $X \sim X_V$. Thus, to simplify, we will, by an abuse of notation, work directly with $X$. 

Since for $x>0$, $\log\, x= \int_0^\infty [(1+\la)^{-1}-(x+\la)^{-1}]\,d\la$, we obtain via the functional calculus
\begin{align}
\label{eq:Y}
Y=\int_0^\infty [(D^2+\la)^{-1}-(D^2+X+\la)^{-1}]\,d\la.
\end{align}
Remark that $D^2+X+\la=D_\Ag^2+\lambda$ is invertible for any $\la\in [0,\infty)$. 
Moreover, with $L=L(\la) \vc (D^2+\la)\in \OP^{2}$ for a given $\la$
\begin{align*}
(L+X)^{-1}&=[(\bbbone+XL^{-1})L]^{-1}=L^{-1}[\bbbone+XL^{-1}]^{-1},
\end{align*}
so that for a given $N\in \N^*$,
\begin{align*}
&[\bbbone+XL(\la)^{-1}]^{-1}=\sum_{n=0}^{N} (-1)^n\,[XL(\la)^{-1}]^n + R_{N}(\lambda),  \\
&R_{N}(\lambda) \vc (-1)^{N+1}[XL(\la)^{-1}]^{N+1} [\bbbone+XL(\la)^{-1}]^{-1} \in \OP^{-N-1}.
\end{align*}
We used the fact that $X\in \OP^1$ and Proposition \ref{prop:OP} $viii)$ to get first $L(\la)^{-1}\in \OP^{-2}$ and then $[\bbbone+XL(\la)^{-1}]^{-1} \in \OP^0$, since $\bbbone+XL(\la)^{-1}\in \OP^0$  has a discrete spectrum (because $XL^{-1}=(XD^{-1}) (D^2L^{-1}) \,D^{-1}$ is compact) and $\bbbone+XL^{-1}$ is invertible.
\\
With the definition $$B_n(\la) \vc L(\la)^{-1}[XL(\la)^{-1}]^{n}\in \OP^{-n-2} \text{ for }n\in \N^*,$$
 we get 
\begin{align}
\label{eq:remainder}
Y=\int_0^\infty I(\la)\, d\la, \quad\text{where } \; I(\la) \vc \sum_{n=1}^{N}(-1)^{n+1}\,B_n(\la)  + L(\la)^{-1} R_{N}(\lambda).
\end{align}
The idea is now to move all dependence in $\la$ in the sum to the right to facilitate an explicit integration over $\la$.

Let us define 
$$
B_n(T)\vc L^{-1}[TL^{-1}]^n \, \text{ for }n\in \N^*.
$$
Since $[D^2 + \la,T] = \nabla(T)$, a recurrence gives, for any $q_1\in \N$,
\begin{align*}
&B_1(T)=L^{-1}T L^{-1} = \sum_{\ell=0}^{q_1} (-1)^\ell\nabla^\ell(T) \,L(\la)^{-\ell-2} +r_{1,q_1}(T),\\
& r_{1,q_1}(T)\vc(-1)^{q_1+1} L^{-1} \nabla^{q_1+1}(T) L^{-q_1-2}.
\end{align*}
For $n\geq 2$ we have $B_n(T)=L^{-1}TB_{n-1}$ and another recurrence gives, for any $q_n\in \N$,
\begin{align}
&B_n(T) = \sum_{\ell_n=0}^{q_n}\cdots\sum_{\ell_1=0}^{q_1} (-1)^{|\ell|}\nabla^{\ell_n} \big( T \nabla^{\ell_{n-1}}(T \cdots \nabla^{\ell_1}(T)) \big) \,L^{-(|\ell|+n+1)} +r_{n,q_n}(T),\nonumber\\
& r_{n,q_n}(T)= \!\!\!\!\!\sum_{\ell_{n-1}=0}^{q_{n-1}}\!\!\!\!\!\cdots\sum_{\ell_1=0}^{q_1} (-1)^{q_n+\abs{\ell}+1}\, L^{-1}\nabla^{q_n+1}\big(T\nabla^{\ell_{n-1}}(T\cdots \nabla^{\ell_1}(T))\big)L^{-(q_n+\abs{\ell}+n+1)} \nonumber \\
&\hspace{7cm}+L^{-1}T\,r_{n-1,q_{n-1}}(T). \label{eq:r_n versus r_n-1}
\end{align}
The implementation of $B_n= B_n(X)$ in \eqref{eq:remainder} with $q_i=N-n$ for any $i=1,\cdots,n$, and $s_{N,n}(\la)\vc r_{n,N-n}(X)$  for $n=1,\cdots,N$ gives
\begin{align*}
&I(\la) =\sum_{n=1}^{N} (-1)^{n+1} \hspace{-0.4cm}\sum_{\ell \vert\, \ell_i=0,\cdots,\,N-n}\hspace{-0.5cm} (-1)^{|\ell |}\nabla^{\ell_n}\big(X\nabla^{\ell_{n-1}} (X\cdots \nabla^{\ell_1}(X))\big)\,\, L(\lambda)^{-(\vert \ell \vert+n+1)} \\
 &\hspace{3cm} + \R_{N,1}(\la)+\R_{N,2}(\la), \\ & \hspace{-0.7cm}\text{ with } \R_{N,1}\vc L^{-1} R_{N},\,\,\,\, \R_{N,2} \vc\sum_{n=1}^{N} (-1)^{n+1}s_{N,n}.
\end{align*}

The integration $$\int_{0}^\infty L(\la)^{-(|\ell |+n+1)}\,d\la = (\vert \ell \vert+n)^{-1} \,D^{-2|\ell |-2n}$$ is valid for $n\geq 1$. Remark that the first term of $I(\la)$ is $XL(\la)^{-2}$ (with $N=1$, so $n=1,\,\ell_1=0$) and we deduce that $Y$ is both in $\Psi '(\A)$ and $\OP^{-1}$ with the announced expansion \eqref{Y in OP}, if for all $N\geq 1$ the integrals of the remainders $\R_{N}(\la)$ are in $\OP^{-N-1}$.

We first claim that $\int_0^\infty \R_{N,1}(\la)\,d\la\in \OP^{-N-1+\epsilon}$ for some $\epsilon \in (0,1)$ and to prove it we apply the method of Theorem \ref{thm:integral of OP}.\\
With $$Z(\la)\!\vc\! [\bbbone+XL^{-1}(\la)]^{-1}\!\!\in \!\OP^0,$$
we rewrite the norm of $\abs{D}^{2-\epsilon}\!\R_{N,1}\! \abs{D}^{N-1}$ as
\begin{align*}
\Big\Vert (D^2L^{-1})^{1-\epsilon/2} L^{-\epsilon/2} \prod_{k=0}^{N-1}\Big([\,\abs{D}^{N-1-k}X\abs{D}^{-N+k}](D^2L^{-1}) \Big) \big[\,\abs{D}^{-1}X\big]\big(L^{-1}Z\big)\!\Big\Vert.
\end{align*}
Thus 
$$
\Vert \int_0^\infty \abs{D}^{2-\epsilon}\R_{N,1} \abs{D}^{N-1}\,d\la \Vert \leq c \int_0^\infty \Vert L(\la)^{-\epsilon/2}\Vert\,\Vert (D_\Ag^2+\la)^{-1}\Vert\, d\la,
$$
since the terms in square brackets are in $\OP^0$ and hence bounded, $\Vert D^2L(\la)^{-1}\Vert \leq 1$ and also we have $L^{-1}Z=(D_\Ag^2+\la)^{-1}$ (we use $X_V$ for the equality). \\
This integral is finite for any $\epsilon >0$, since for $\la\in [0,1]$, the integrand is smaller than $\Vert \abs{D}^{-1}\Vert \Vert D_\Ag^{-2}\Vert$, while for $\la\geq 1$ we have $L^{-1}(\la)\leq \lambda^{-1}$ and $ (D^2_\Ag+\la)^{-1} \leq \la^{-1}$.
\\
As explained in Theorem \ref{thm:integral of OP}, to prove the claim it is sufficient to check that the previous integral's estimates remain valid when $L^{-1}Z$ is swapped to $\delta^n(L^{-1}Z)$. Since $$\delta(L^{-1}Z)=L^{-1}\delta(Z)=-L^{-1}Z\,\delta(X)L^{-1} \,Z,$$ this improves previous estimates in $\la$.

We also claim that $\int_0^\infty \R_{N,2}(\la)\,d\la \in \OP^{-N-1+\epsilon}$ and proceed essentially in the same way to prove it. 
We only need to show that $\int_0^\infty \Vert s_{N,n}(\la)\,\abs{D}^{N+1-\epsilon}\Vert \,d\la <\infty$ for each $n$ in $\{1,\cdots,N\}$, and we use the induction to that end: For $n=1$ we have
\begin{align*}
 \Vert s_{N,1}(\la)\,\abs{D}^{N+1-\epsilon} \Vert &= \Vert (-1)^{N}L^{-1}\nabla^{N}(X)L^{-N-1}\abs{D}^{N+1-\epsilon} \Vert \\
 &=\Vert (L^{-1}D^2)[D^{-2} \nabla^{N}(X)\abs{D}^{-N+1}]L^{-N-1}\abs{D}^{2N-\epsilon}\Vert\\
 & \leq c\Vert L^{-1}\abs{D}^{2}\Vert^{N-\epsilon/2} \Vert L^{-1}\Vert^{1+\epsilon/2}=c' \Vert L^{-1}\Vert^{1+\epsilon/2},
\end{align*}
which is integrable on $\RR^+$ for any $\epsilon >0$. 
Now we show that 
$$
\int_0^\infty \Vert s_{N,n-1}(\la)\,\abs{D}^{N+1-\epsilon}\Vert \,d\la <\infty\text{ implies }\int_0^\infty \Vert s_{N,n}(\la)\,\abs{D}^{N+1-\epsilon}\Vert \,d\la <\infty.
$$

Using the relation \eqref{eq:r_n versus r_n-1}, we only need to prove that $L^{-1}X\,r_{n-1}\, \abs{D}^{N+1-\epsilon}$  and the operators $ L^{-1}\nabla^{q_n+1} \big(X\nabla^{\ell_{n-1}} (X\cdots \nabla^{\ell_1}(X)) \big) L^{-(q_n+\abs{\ell}+n+1)}\abs{D}^{N+1-\epsilon}$ are integrable. \\
The first one follows directly from $$\Vert L(\la)^{-1}X\Vert = \Vert L(\la)^{-1} D^2 (D^{-2}X)\Vert \leq \Vert D^{-2}X\Vert$$ and the recurrence hypothesis. For the second one we take the decomposition 
\begin{multline*}
L^{-1}\nabla^{q_n+1}\big(X\nabla^{\ell_{n-1}}(X\cdots \nabla^{\ell_1}(X))\big) L^{-(q_n+\abs{\ell}+n+1)}\abs{D}^{N+1-\epsilon}\\
= L^{-1}D^2\,\big[ D^{-2}\nabla^{q_n+1}\big(X\nabla^{\ell_{n-1}}(X\cdots \nabla^{\ell_1}(X))\big) \abs{D}^{-N-\abs{\ell}+1} \big]\,\times \\
\times \abs{D}^{N+\abs{\ell}-1}L^{-N-\abs{\ell}-1}\abs{D}^{N+1-\epsilon}.
\end{multline*}
The term in the square bracket is in $\OP^0$, so that it remains to show the integrability of $$\Vert \abs{D}^{2N+\abs{\ell}-\epsilon}L^{-N-\abs{\ell}-1}\Vert=\Vert D^2 L^{-1} \Vert^{N + \ell -\epsilon/2} \, \Vert \awD^{-\ell}\Vert  \Vert L^{-1} \Vert^{1+\epsilon/2},$$
which follows as above.

This completes the proof that $$\R_N\vc \int_0^\infty [\R_{N,1}(\la)+\R_{N,2}(\la)] \,d\la \in \OP^{-N-1+\epsilon}.$$
To dispose of the $\epsilon$ we invoke the handy Lemma \ref{lm:polynomial rest}.
\hfill $\Box$
\end{proof}

\begin{theorem}[see \cite{ConnesInner} and \cite{TorusSA}]
 \label{2dev}
 Let $\ahd$ be a regular spectral triple and let $\Ag=\Ag^* \in \Psi^0(\A)$. Then, for any $N\in\N^*$  and any $s\in \CC$,
\begin{align}
\label{eq:expansion for D_A}
|D_\Ag|^{-s} = |D|^{-s} + \sum_{n=1}^N K_n(Y,s) |D|^{-s} \mod \OP^{-(N+1)-\Re(s)},
\end{align}
where $K_n(Y,s) \in \PDOk{-n}$ and we have for each $n\in \{1,\cdots,N\}$
\begin{align}
\label{eq:Kn}
K_n(Y,s) = \sum_{\ell\in \{0,\cdots,\,N-n\}^n}\, h_n(s;\ell) \,\E^{\ell_1}(Y)\cdots \E^{\ell_n}(Y)   \mod \OP^{-N-1}.
\end{align}
\end{theorem}

\begin{proof}
To prove Formula \eqref{eq:expansion for D_A} we observe that 
$$|D_\Ag|^{-s}=e^{C-(s/2)Y}e^{-C}\, |D|^{-s} \,\text{ with }\,C:= (-s/2)\log(D^2).
$$
We are going to use the Duhamel formula \index{Duhamel formula}
\begin{align}
\label{Duhamel}
e^{P+Q}\,e^{-P}=\bbbone  +\int_0^1 e^{\,s(P+Q)}\,Q\,e^{-s\,P}\,ds
\end{align}
leading to the following Volterra series (expansional formula)
\begin{align*}
e^{P+Q}e^{-P}=\bbbone+\sum_{n=1}^\infty\, \int_{0\leq t_1\leq\cdots\leq t_n\leq 1} Q(t_1) \cdots Q(t_n) \,dt, \,\,\text{ where } \,\,Q(t)\vc e^{tP}Q\,e^{-tP}.
\end{align*}
The latter is valid for any (un)bounded selfadjoint operator $P$ and any selfadjoint bounded operator $Q$ in the operator norm topology (see \cite[Theorem 3.5]{Davies}, \cite[Lemma 3.32]{Zagrebnov} for a more general framework).\\
With $P= -\tfrac{s}{2} \log(D^2)$ and $Q=-\tfrac{s}{2} \,Y$, we obtain $Q(t)= -\tfrac{s}{2} \,\sigma_{-s\,t}(Y)$, so that
 \begin{equation}
\label{egalite-DAs}
|D_\Ag|^{-s} = |D|^{-s} + \sum_{n=1}^\infty K_n(Y,s)|D|^{-s},
\end{equation}
with each $K_n(Y,s)$ in $\OP^{-n}$, hence proving \eqref{eq:expansion for D_A}.

Finally, applying the expansion \eqref{OPexpansion} within Equation \eqref{Kn}, with $N-n$ and  $\E^{m}(Y)=\nabla^m(Y)\,\abs{D}^{-2m}$, we get that $K_n(Y,s)$ is equal, modulo $\OP^{-N-1}$, to
\begin{align*}
(-\tfrac{s}{2})^n \int_{0\leq t_1\leq\cdots\leq t_n\leq 1}\, \sum_{\ell\in \{0,\cdots,\,N-n\}^n} \, \tbinom{ -s\,t_1/2}{\ell_1} \cdots
\tbinom{-s\,t_n/2}{\ell_n }\,\E^{\ell_1}(Y)\cdots \E^{\ell_n}(Y)\,dt,
\end{align*}
which is nothing else than \eqref{eq:Kn}.
\hfill $\Box$
\end{proof}
Remark that $\E^\ell(Y)$, $\ell \in \N$, are in $\Psi '(\A)$, see \cite[Corollary 4.4]{TorusSA}.

It turns out that, thanks to Formula \eqref{eq:expansion for D_A}, the fluctuated triple $\ahda$ inherits many of the properties of the original $\ahd$.

\begin{corollary}
\label{cor:pdoA in pdo}
Let $\ahd$ be regular spectral triple and let $\Ag=\Ag^* \in \Psi^0(\A)$. Then, $\ahda$ is a regular spectral triple and $\Psi_{\Ag}(\A) \subset \PDO$, where 
$\Psi_{\Ag}(\A)$ refers to the pseudodifferential calculus defined by $\DD_\Ag$. \index{_apDOpA@$\Psi_{\Ag}(\A)$}
\end{corollary}

\begin{proof}
Recall that $\ahda$ is indeed a spectral triple --- cf. Remark \ref{rem:pert}. \\
Formula \eqref{eq:expansion for D_A} for $s = -1$ yields $\awDA = \awD + B(\Ag)$, with $B(\Ag) \in \PDOz$. In particular $\Dom \delta_\Ag=\Dom \delta$ since $\Psi^0(\A) \subset \B(\H)$. Thus $\OP^0_\Ag =\OP^0$; moreover, if $T \in \OP_\Ag^r$ then $\vert D_\Ag\vert^{-r}T\in \OP^0_\Ag=\OP^0$, so $$\vert D\vert^{-r}T =(\vert D \vert^{-r}\vert D_\Ag\vert^r) \,\vert D_\Ag\vert^{-r} T\in \OP^0$$ as $\vert D \vert^{-r}\vert D_\Ag\vert^r\in \OP^0$, and $T\in \OP^r$. By symmetry, $$\OP_\Ag^r=\OP^r$$ and by extension $\,\Psi_{\Ag}(\A) \subset \PDO$ since we already know that $\DD_\Ag\in \Psi(\A)$ and so is $ \vert D_\Ag \vert$ by Theorem \ref{2dev}, implying that $\vert \DD_\Ag\vert=\vert D_\Ag\vert -P_\Ag$ is also in $\Psi(\A)$ because $P_\Ag$ is smoothing (cf. p.~\pageref{eq:PA smoothing}). As in Lemma \ref{intersection of domains}, $\bigcap{ }_{k = 0}^{\infty} \,\,\Dom \delta_{\Ag}^{k}=\bigcap{ }_{k = 0}^{\infty} \,\,\Dom \delta_{\Ag}^{'\,k}$ with the definition $\delta'_{\Ag}\vc [\DD_\Ag,\,\cdot]$. \\
Since for $a \in \A$, we have $a,\,[\DA,a] \in \OP^0 =\OP^0_\Ag$, the regularity of the fluctuated triple is proved.
\hfill $\Box$
\end{proof}
Consequently, we define $$\Psi_\Ag^\CC(\A) \vc \{ T\abs{D_\Ag}^z\,\vert\, T \in \Psi(\A),\,z\in \CC\}.$$ \index{_apDOza@$\Psi_\Ag^\CC(\A)$}
\begin{theorem}
\label{thm:extension by fluctuation}
Let $\ahd$ be $p$-dimensional regular spectral triple with a dimension spectrum of order $d$ and let $\Ag=\Ag^* \in \Psi^0(\A)$. Then: 

i) The function $s \mapsto \zeta_{D_\Ag}(s) = \Tr \awDA^{-s}$ has a meromorphic continuation to $\CC$.

ii) The triple $\ahda$ is also  $p$-dimensional and has a dimension spectrum $\Sd\ahda$ included in $\Sd\ahd$ and of order at most $d$.

iii) For any $z\in \CC$, we have $\vert D_\Ag \vert^z \in \Psi^\CC(\A)$ and $ \Psi_\Ag^\CC(\A)\subset\Psi^\CC(\A)$. 

iv) For any $k\in \N^*$
\begin{align}
\ncintd{k}{|D_{\Ag}|}^{-p}=\ncintd{k} |D|^{-p}. \label{ncint DA-p}
\end{align}
\end{theorem}

\begin{proof}
\textit{i)} Firstly, let us rewrite Formula \eqref{eq:expansion for D_A}, for any $s \in \CC$ with $\Re(s) > p$, as
\begin{align}
\zeta_{D_\Ag}(s) & = \Tr \awDA^{-s} = \Tr \awD^{-s} + \sum_{n=1}^{N} \Tr \, \left( K_n(Y,s) \awD^{-s} \right) + \Tr \,R_N(s). \label{eq: zeta DA(s)}
\end{align}
Note that it is well-defined for any $N \geq 1$ because $K_n(Y,s)\awD^{-s} \in \OP^{-n-\Re{(s)}}$ and $R_N(s) \in \OP^{-N-1-\Re{(s)}}$, so all of the involved operators are trace-class for $\Re(s) > p$. Furthermore, we can rewrite $\zeta_{D_\Ag}$ using the explicit form of the $K_n$'s \eqref{eq:Kn}:  
\begin{align}
\label{eq:zDA}
\zeta_{D_\Ag}(s)  = \zeta_D(s) + \sum_{n=1}^{N} \,\sum_{\ell\in\{0,\cdots,\,N-n\}^n} \, h_n(s;\ell)\, \zeta_{W_n(\ell),D}(s) + f_N(s),
\end{align}
where $f_N(s) = \Tr \,R_N(s)$ and
\begin{align}\label{W}
W_n(\ell) \vc \E^{\ell_1}(Y)\cdots \E^{\ell_n}(Y) = \nabla^{\ell_1}(Y) \awD^{-2\ell_1} \cdots \nabla^{\ell_n}(Y) \awD^{-2\ell_n}.
\end{align}
Since the functions $\zeta_{W_n(\ell),D}$ admit  meromorphic extensions to $\CC$ by the dimension spectrum hypothesis on $\ahd$ and $f_N$ is actually a holomorphic function for $\Re(s) > p - N -1$, Formula \eqref{eq:zDA} provides a meromorphic continuation of $\zeta_{D_\Ag}$ to the half-plane $\Re(s) > p - N -1$. As $N$ can be taken arbitrarily large, we obtain a meromorphic extension of $\zeta_{D_\Ag}$ to the whole complex plane. Furthermore, since $\ahd$ is $p$-dimensional, $\zeta_D$ is singular at $s=p$ and regular for $\Re(s)>p$, and so is $\zeta_{D_\Ag}$, hence $\ahda$ is also $p$-dimensional. 

\textit{ii)} Secondly, let us multiply Formula \eqref{eq:expansion for D_A} from the left by $T \in \Psi_{\Ag}^0(\A)$ and take the trace as in \eqref{eq: zeta DA(s)}. Then, we have for $N\in \N^*$ and for $\Re (s)>p$ 
\begin{align}\label{eq:zTDA}
\zeta_{T,D_{\Ag}}(s) & = \zeta_{T,D}(s) + \sum_{n=1}^{N}\,\, \sum_{\ell\in\{0,\cdots,\,N-n\}^n}\, h_n(s;\ell) \,\zeta_{TW_n(\ell),D}(s) + \Tr\, T R_N(s).
\end{align}
For any multi-index $\ell$, $W_n(\ell) \in \OP^{-n-\abs{\ell}} \subset \OP^0$, and hence $T W_n(\ell) \in \OP^0$. Since $\Psi_{\Ag}^0(\A) \subset \PDOz$, the dimension spectrum hypothesis on $\ahd$ then assures that the functions $\zeta_{TW_n(\ell),D}$ are meromorphic on $\CC$ and thus we establish a meromorphic continuation of $\zeta_{T,D_\Ag}$ for any $T \in \Psi_{\Ag}^0(\A)$. \\
This fact implies that $\ahda$ has a dimension spectrum and, moreover, $\Sd\ahda \subset \Sd\ahd$, with the former being at most of the order of the latter, since the poles of any $\zeta_{T,D_\Ag}$ can only come from the ones of $\zeta_{T W_n(\ell),D}$\,.

$iii)$ Since $K_n(Y,s) \in \Psi^{-n}(\A)$, Formula \eqref{eq:expansion for D_A} shows that $\vert D_\Ag \vert^z \in \Psi^\CC(\A)$. Furthermore, Corollary \ref{cor:pdoA in pdo} implies the announced inclusion.

$iv)$ Using Formula \eqref{eq:expansion for D_A} we obtain $$(\vert D_\Ag\vert^{-p}-\abs{D}^{-p})\abs{D}^{-s}=K_1(Y,p) \abs{D}^{-(p+s)}\,\text{ mod }\OP^{-(p+2)}.$$
This operator is in $\OP^{-(p+1+\Re(s))}$, as $K_1(Y,p) \in \OP^{-1}$, so it is trace-class in a neighbourhood of $s=0$. Hence, $\ncintd{k} (\vert D_\Ag\vert^{-p}-\abs{D}^{-p})=0$.
\hfill $\Box$
\end{proof}

\section{Fluctuations of Noncommutative Integrals}

Given $\Ag=\Ag^* \in \Psi^0(\A)$ we define, for any $T\in \Psi_\Ag^\CC(\A)$ and $k \in \Z$, 
\begin{align*}
\ncintd{k}_{\hspace{-0.15cm}\Ag} T \vc \Rez{s=0}\,\,s^{k-1} \zeta_{T,D_\Ag}(s).\index{_Ancintz@$\displaystyle\ncintd{k}_{\hspace{-0.2cm}\Ag}$}
\end{align*}

We shall now invoke Proposition \ref{2dev} to compute $\ncintd{k}_{\hspace{-0.08cm}\Ag}T$ in terms of $\ncintd{k}$.
\begin{proposition}
\label{thm:fluctuations}
Let $\ahd$ be a regular spectral triple of finite dimension with a dimension spectrum of order $d$. \\
Then, for $\Ag=\Ag^* \in \Psi^0(\A)$, $T \in \Psi_{\Ag}^{\CC}(\A)$ and $k \in \Z, k < d$ we have: 
\begin{align}
\label{fluc}
\ncintd{k}_{\hspace{-0.15cm}\Ag} T= \ncint^{[k]} T +\, \sum_{n = 1}^{d-k}\,\, \sum_{\ell\in\{0,\cdots,\,d-k-n\}^n} \,\,\sum_{j=0}^{\abs{\ell}} \, h_j(\ell) \ncintd{\,k+n+ j\,} T\,W_n(\ell),
\end{align}
where $h_j(\ell)$ are defined in \eqref{eq: h(s,l)} and $W_n(\ell)$ in \eqref{W}.

Moreover, for any $T\in \Psi^\CC_\Ag(\A)$,
\begin{align}
& \ncintd{d}_{\hspace{-0.15cm}\Ag} T=\ncint^{[d]} T, \label{ncintd DA}
\end{align}
\end{proposition}

\begin{proof}
By exploiting Formula \eqref{eq:zTDA}, we obtain, with $N = d-k$,
\begin{align*}
\ncintd{k}_{\hspace{-0.15cm}\Ag} T -\ncintd{k} T & = \sum_{n = 1}^{d-k}\,\sum_{\ell\in\{0,\cdots,\,d-k-n\}^n} \Rez{s=0}\, s^{k-1} h_n(s;\ell) \, \zeta_{TW_n(\ell),D}(s) \\
& =   \sum_{n = 1}^{d-k} \,\sum_{\ell\in\{0,\cdots,\,d-k-n\}^n} \,\sum_{j=0}^{\abs{\ell}} h_j(\ell) \,\Rez{s=0}\, s^{k+n+j-1} \zeta_{TW_n(\ell),D}(s)  \\
& =  \sum_{n = 1}^{d-k}\, \sum_{\ell\in\{0,\cdots,\,d-k-n\}^n} \,\sum_{j=0}^{\abs{\ell}} h_j(\ell) \,\ncintd{\,k+n+ j\,} T W_n(\ell).
\end{align*}

\noindent Equality \eqref{ncintd DA} results from $\underset{s=0}{\Res}\,\, s^{d+n} \Tr\,T \abs{D}^{-s}=0$, for $T\in \Psi^\CC(\A),\, n\in \N^*$.
\hfill $\Box$
\end{proof}

\begin{corollary}
\label{cor:dominant invariance}
Let $T \in \PDOCA \cap \OP^{-p+\delta}$ for some $\delta < 1$. \\
Then, for any $k \in \N^*$, we have $$\ncintd{k}_{\!\!\Ag} T = \ncintd{k} T.$$
\end{corollary}
\begin{proof}
Recall (cf. \eqref{W}) that $W_n(\ell) \in \OP^{-n-\abs{\ell}}$, so that the operators $TW_n(\ell)$ appearing in Formula \eqref{fluc} are of the order $-p -\epsilon$ with $\epsilon = n+\abs{\ell} - \delta>0$ since $n \geq 1$, $\abs{\ell} \geq 0$ and $\delta < 1$, so $TW_n(\ell)$ are  trace-class and $\ncintd{k} TW_n(\ell) = 0$ for any $k \in \N^*$.
\hfill $\Box$
\end{proof}

\begin{corollary}
\label{cor:difference of fluctuated zeta}
Let $\ahd$ be a regular spectral triple of finite dimension with a simple dimension spectrum and let $\Ag=\Ag^* \in \Psi^0(\A)$. Then:

i) For any $T \in \PDOCA$,
\begin{align}
& \underset{s=0}{\Res}\,\, \zeta_{T,D_\Ag}(s)= \underset{s=0}{\Res}\,\, \zeta_{T,D}(s). \label{eq:zeta TDA} 
\end{align}

ii) $\zeta_{D_\Ag}$ is regular at 0 if and only if  $\zeta_D$ is so and $\zeta_{D_{\Ag}}(0) = \zeta_D(0) -\tfrac 12 \ncint Y$. 

iii) Moreover, with $p$ being the dimension of $\ahd$,
\begin{align}
& \ncint |D_\Ag|^{-(p-1)}= \ncint |D|^{-(p-1)} -\tfrac{p-1}{2}\ncint X|D|^{-p-1},\label{ncint DA-p+1}\\
& \ncint |D_\Ag|^{-(p-2)}= \ncint |D|^{-(p-2)}+\tfrac{p-2}{2}\big(-\ncint X|D|^{-p} + \tfrac{p}{4} \ncint X^2 |D|^{-2-p} \big).\label{ncint DA-p+2}
\end{align}
\end{corollary}

\begin{proof}
$i)$ Equation \eqref{eq:zeta TDA} is just a rewriting of Formula \eqref{ncintd DA} for $d=1$.

$ii)$ The point $i)$ with $T = \bbbone$ implies that $\zDA$ is regular at 0 iff $\zDA$ is so. Moreover, taking $k=0$, $d=1, T=\bbbone$ in Formula \eqref{fluc} and recalling \eqref{h}, \eqref{eq: h(s,l)} we obtain
$$
\Rez{s=0}\, s^{-1} \zeta_{D_{\Ag}}(s) = \Rez{s=0}\, s^{-1} \zeta_D(s) +h_0(0) \ncint Y = \zeta_D(0) -\tfrac{1}{2} \ncint Y.
$$

$iii)$ Using respectively Formulae \eqref{eq:expansion for D_A}, \eqref{eq:Kn} and \eqref{Y in OP} with $N = 1$ we get
\begin{align}
\vert D_{\Ag} \vert^{-s} & = \awD^{-s} + K_1(Y,s) \awD^{-s} + R_1(s) = \awD^{-s} -\tfrac{s}{2} Y \awD^{-s} + R_2(s) \nonumber \\
& =  \awD^{-s} -\tfrac{s}{2} X \awD^{-s-2} + R_3(s), \label{eq:DA for N=1}
\end{align}
for any $s \in \CC$ with $R_i(s) \in \OP^{-\Re(s)-2}$. Taking $s= p-1$ and applying the noncommutative integral yields Equation \eqref{ncint DA-p+1}, since $R_3(p-1) \in \OP^{-p-1}$ is trace-class.\\
The same manoeuvre with $N=2$ gives
\begin{align}
& \vert D_{\Ag} \vert^{-s}  = [\bbbone -\tfrac{s}{2} Y  + \tfrac{s^2}{8} (\nabla(Y)D^{-2}+Y^2)] \awD^{-s} \mod \OP^{-3-\Re(s)},\label{eq:DA for N=2}  \\
& Y=XD^{-2}-\tfrac{1}{2}( \nabla(X) D^{-4} +X^2 D^{-4}) \mod \OP^{-3}. \label{eq:Y mod -3}
\end{align}
Remark that $Y=XD^{-2}\,\text{mod }\OP^{-2}$ and $Y^2=X^2D^{-4} \mod \OP^{-3}$. Since $\ncint$ is a trace, $$\ncint \nabla(Y) \abs{D}^{-r}=\ncint \nabla(X) \abs{D}^{-r}=0 \,\text{ for any }\,r\in \RR$$ and we obtain the Equation \eqref{ncint DA-p+2} for $s=p-2$ --- see \cite[Lemma 4.10]{TorusSA}. 
\hfill $\Box$
\end{proof}

Let us stress that the above corollary explicitly uses the assumption of simplicity of the dimension spectrum of $\ahd$. If this is not the case, then $\zeta_{D_\Ag}$ might fail to be regular at 0, even though $\zeta_D$ is so, because the higher order poles of functions $\zeta_{T\cdot W_n(\ell),D}$ at 0 can render $\ncint_{\hspace{-0.08cm}\Ag} \bbbone$ nonzero --- contrary to the commutative case, see \eqref{int 1=0}. 
To show the next important result on the fluctuation of zeta functions at zero, we need the following lemma which closely follows \cite[Lemma 2.3]{ConnesInner}.

\begin{lemma}
\label{int of log}
Let $(\A,\H,\DD)$ be a regular spectral triple and $\Ag=\Ag^* \in \Psi^0(\A)$. Then:

i) For any $N\in\N$, there exist $B(t)\in \Psi(\A)$ such that for $t>0$, we have $\mod \OP^{-N}$
\begin{align}
\label{eq:Z(t)}
&\tfrac{\partial }{\partial t} [\log(D^2+t X_V)-\log D^2-\log(\bbbone +tX_VD^{-2})]=[D^2+tX,\,B(t)] ,
\\
\label{eq:difference of log}
&\log (D^2+X_V)-\log D^2-\log(\bbbone+X_VD^{-2})=[D^2,\,B_1]+[X,\,B_2],
\end{align}
where $B_1\vc \int_0^1 B(t)\,dt$ and $B_2\vc \int_0^1 t\,B(t) \,dt$ are in $\Psi(\A)$. 

ii) Moreover,
\begin{align}
\label{eq: int of diiference of log}
\ncint Y=\ncint [\log (D^2+X_V)-\log D^2] = \ncint \log(\bbbone+X_VD^{-2}).
\end{align}
\end{lemma}

\begin{proof}
The operator $$C_t\vc D^2+tX_V=tD_\Ag^2+(1-t)D^2$$ is selfadjoint and positive for $t\in [0,1]$. The invertibility of $D^2$ and $D_\Ag^2$ implies that $C_t +\lambda$ is invertible for any $\lambda \geq 0$ and, consequently, that $\bbbone+tX_VD^{-2}=C_t\,D^{-2}$ is also invertible when $t\in [0,1]$. \\
Since $X_VD^{-2}\in \OP^{-1}$ is compact by Proposition \ref{prop:OP} $ix)$, it has a purely discrete spectrum so that the operator $\log (\bbbone +X_VD^{-2})=\log(D_\Ag^2D^{-2})$ exists.

i) As in \eqref{eq:Y}, $$\log (D^2+tX_V)-\log D^2=\int_0^\infty [(D^2+\la)^{-1}-(D^2+tX_V+\la)^{-1}]\,d\la$$ always exists, so we obtain 
$$
\tfrac{\partial }{\partial t} \big(\log(D^2+t X_V)\big)=\int_0^\infty 
\tfrac{1}{C_t+\lambda}\,X_V\, \tfrac{1}{C_t+\lambda}\,d\lambda.
$$
Moreover, 
\begin{align*}
\int_0^\infty X_V\,(C_t+\lambda)^{-2}\,d\lambda &= X_VC_t^{-1}=X_VD^{-2}(\bbbone+tX_VD^{-2})^{-1}=\tfrac{\partial }{\partial t} \log(\bbbone+tX_VD^{-2}).
\end{align*}
Thus, with $L(t)$ being the LHS of \eqref{eq:Z(t)}, $$L(t)=\int_0^\infty \big[(C_t+\lambda)^{-1},\,\,X_V\, (C_t+\lambda)^{-1}]\, d\lambda.$$
We now invoke Lemma \ref{lm:OP_exp} with $D \rightsquigarrow C_t$, $-\lambda \notin \spec \,C_t$ and $n=1$ to deduce, with $\nabla_t(T)\vc [C_t,T]$,
\begin{align}
& [(C_t+\lambda)^{-1},\,X_V(C_t+\lambda)^{-1}]=\sum_{k=1}^N (-1)^k\,\nabla_t^k\big[X_V(C_t+\lambda)^{-(k+2)}\big] +R_N(\la,t) \nonumber\\
&\hspace{1cm}\text{with}\quad R_N(\la,t) \vc  (-1)^{N+1}  \,(C_t+\la)^{-1}\,\nabla_t^{N+1}[ X_V]\, (C_t+\la)^{-(N+2)}. \label{RN(lambda,t)}
\end{align}
Because $(C_t+\lambda)^{-(k+1)}$, being in the centraliser of $\nabla_t$, can be put inside the parenthesis of $\nabla_t$ in the first equality, we get
\begin{align*}
& [(C_t+\lambda)^{-1},\,X_V(C_t+\lambda)^{-1}]=\nabla_t\big[\sum_{k=1}^N (-1)^k\,\nabla_t^{k-1}(X_V)(C_t+\lambda)^{-(k+2)}\big]+R_N(\la,t)
\end{align*}
with $R_N(\la,t)\in \OP^{-N-4}$. The term in the bracket gives, after an integration over $\lambda$, 
$$
B(t)\vc\sum_{k=1}^N \tfrac{(-1)^k}{k+1} \,\nabla_t^{k-1}(X_V) \,C_t^{-(k+1)}.
$$
Since $C_t^{-1}$ is in $\Psi(\A)$ (see Example \ref{example:inverse}), so is $B(t)$.\\
Thus, $$L(t)=\nabla_t[B(t)]+R_N(t) \text{ where } R_N(t)\vc\int_0^\infty R_N(\la,t)\,d\la.$$

We claim that $R_N(t)\in \OP^{-N}$. Using the method of Theorem \ref{thm:integral of OP}, we show that $R_N(t)\,C_t^{N/2} \in \OP^0$. 
With $$E=\nabla_t^{N+1}(X_V)\,C_t^{-1-N/2}\in \OP^0,$$ we get
\begin{align*}
\Vert R_N(t)\,C_t^{N/2}\Vert &\leq\int_0^\infty \Vert R_N(\la,t)\,C_t^{N/2}\Vert \,d\la\\
&\leq \int_0^\infty \Vert(\la+C_t)^{-1}\Vert\,\Vert E\Vert\, \Vert C_t^{1+N/2} (\la+C_t)^{-N-2}\,C_t^{N/2} \Vert \,d\la\\
&\leq \norm{E}\int_0^\infty \Vert C_t(\la+C_t)^{-1}\Vert^{N+1} \,\Vert (\la +C_t)^{-1}\Vert^2\,d\la.
\end{align*}
The latter is finite since $C_t$ is a positive operator and $\Vert C_t(\la+C_t)^{-1}\Vert \leq 1$.
\\
We now show that $$\del(R_N(\la,t)\,C_t^{N/2})=\del(R_N(\la,t))\,C_t^{N/2}+R_N(\la,t)\,\del(C_t^{N/2})$$ has a finite integral over $\lambda \in \Rp$.\\ 
As for the second term, we have $$
\Vert R_N(\la,t)\,\del(C_t^{N/2})\Vert \leq \Vert R_N(\la,t)\,C_t^{N/2}\Vert \,\Vert C_t^{-N/2} \del(C_t^{N/2})\Vert
$$
thus yielding a finite integral. \\
In the first term, we expand the derivation $\del$ on $R_N(\la,t)=A_\la^{-1}\,B\, (A^{-1}_\la)^{N+2}$ as a finite sum of expressions similar to $R_N(\la,t)$ with only one of the $A_\la^{-1}=(\la+C_t)^{-1}$ replaced by $-A_\la^{-1} \del(C_t)A_\la^{-1} $ or $B$ replaced by $\del(B)$. Since $$\Vert A_\la^{-1} \del(C_t)A_\la^{-1} \Vert \leq  \Vert C_t(\la+C_t)^{-1} \Vert\,\Vert C_t^{-1} \del(C_t) \Vert\, \Vert(\la+C_t)^{-1}\Vert,$$
the convergence of the integral is unspoilt. Thus 
$R_N(\la,t)\,C_t^{N/2}\in \Dom \del$ and, with the same arguments, one deduces that $R_N(\la,t)\,C_t^{N/2}\in \cap_n \Dom \del^n$.

The proof of \eqref{eq:Z(t)} is complete because $X_V\sim X$, so one can swap $X_V$ and $X$.\\
Finally, Equation \eqref{eq:difference of log} follows from the definition of $B_1$ and $B_2$.

ii) The tracial property of $\ncint$ (Theorem \ref{thm:ncint_trace}) applied to \eqref{eq:difference of log} gives \eqref{eq: int of diiference of log}.
\hfill $\Box$
\end{proof}

\begin{theorem} [\cite{ConnesInner}]
\label{thm:zeta(0) and fluctuation}
Let $\ahd$ be a $p$-dimensional regular spectral triple with a simple dimension spectrum, $\zeta_D$ regular at zero and let $\Ag=\Ag^* \in \Psi^0(\A)$. Then,
\begin{align}
\label{eq: fluctuation of zeta at zero}
\zeta_{D_\Ag}(0)-\zeta_D(0)= \sum_{k=1}^p \tfrac{(-1)^k}{k}\,\ncint (\Ag \,D^{-1})^k.
\end{align}
\end{theorem}

\begin{proof}
Thanks to Corollary \ref{cor:difference of fluctuated zeta} and \eqref{eq: int of diiference of log}, one has
\begin{align}
\label{eq:difference of zeta}
\zeta_{D_\Ag}(0)-\zeta_D(0)=-\tfrac{1}{2}\ncint Y =-\tfrac{1}{2}\ncint \log(\bbbone+X_VD^{-2}).
\end{align}
We now write $$\log(\bbbone+X_VD^{-2})=\sum_{k=1}^N \tfrac{(-1)^{k+1}}{k}\,(X_VD^{-2})^k+ R_N,$$ 
where the remainder is $$R_N=(-1)^N\, (X_VD^{-2})^{N+1} \int_0^1(1-t)^N\,(\bbbone+tX_VD^{-2})^{-(N+1)}\,dt.$$
We have already seen in the proof of Lemma \ref{int of log} that $\bbbone+tX_VD^{-2}\in \OP^0$ is invertible. Thus, the integral makes sense and since $X_VD^{-2}\in \OP^{-1}$ we have $R_N\in \OP^{-N-1}$. In particular, $R_N$ is trace-class for $N=p$ and $$\ncint \log(\bbbone+X_VD^{-2})=\sum_{k=1}^p \tfrac{(-1)^{k+1}}{k}\,\ncint (X_VD^{-2})^k,$$ where, as usual, we can safely replace $X_V$ by $X$. \\
With $a=D^{-1} \Ag$ and $b=\Ag D^{-1}$, $$X=D(a+b+ab)D \,\text{ and }\,XD^{-2}= D(a+b+ab)D^{-1},
$$
so that, $\ncint$ being a trace, $$\ncint (XD^{-2})^k=\ncint(a+b+ab)^k$$ and we have to compare 
$$
\sum_{k=1}^N\tfrac{(-1)^{k+1}}{k} \ncint (a+b+ab)^k \,\text{ with }\, -\tfrac{1}{2}\sum_{k=1}^N \tfrac{(-1)^{k}}{k}\ncint a^k=\sum_{k=1}^N \tfrac{(-1)^{k+1}}{k} \ncint (a^k+b^k).
$$

Let us, more generally, introduce the following two formal series in $x\in\RR$ within the free algebra generated by $a$ and $b$: 
\begin{align*}
& S(x)\vc \sum_{k=1}^\infty \tfrac{(-1)^{k+1}}{k}\,[x a+x b+x ^2 ab]^k,\\
& T(x) \vc \sum_{k=1}^\infty \tfrac{(-1)^{k+1}}{k}\,[(x a)^k+ (x b)^k].
\end{align*}
We claim that they are equal modulo commutators (denoted by $\doteq$) so that, using the tracial property of $\ncint$, the proof of \eqref{eq: fluctuation of zeta at zero} would be complete.\\
Since $S(0)=T(0)$, it is sufficient to compare the derivatives of $S$ and $T$. Remark first that for the derivative of a term $M(x)^n$, we have $\tfrac{d\,}{dx}M(x)^n \doteq n M(x)^{n-1}M'(x)$ since one can commute $M'(x)$ to the right in $M(x)^kM'(x)M(x)^{n-k-1}$. Thus,
\begin{align*}
S'(x)&\doteq \sum_{k=1}^\infty (-1)^{k+1}(xa+xb+x^2ab)^{k-1}(a+b+2xab)\\
& = (1+xa+xb+x^2ab)^{-1}(a+b+2xab)=[(1+xa)(1+xb)]^{-1}(a+b+2xab)\\
&= (1+xb)^{-1}(1+xa)^{-1}[(1+xa)b+a(1+xb)]\\
&=(1+xb)^{-1}b+ (1+xb)^{-1}(1+xa)^{-1}a(1+xb)\\
&\doteq (1+xb)^{-1}b+(1+xa)^{-1}a =T'(x).
\tag*{$\Box$}
\end{align*}
\end{proof}

\section{Consequences for the Spectral Action}\label{sec:consequences}

Theorem \ref{thm:extension by fluctuation} assures that given a regular $p$-dimensional spectral triple $(\A,\H,\DD)$ with a dimension spectrum, the triple $(\A,\H,\DD_\Ag)$ will also be regular, $p$-dimensio-nal and possessing a dimension spectrum, for any $\Ag = \Ag^* \in \PDOk{0}$. However, to deduce an expansion of $\Tr \, f( \abs{\DD_\Ag} / \Lambda )$ we would need to control the behaviour of (the maximal meromorphic extension of) $\zDA$ on the verticals, as explained in detail in Section \ref{sec:zeta2hk}. This, unfortunately, does not come for free, even if we can control the behaviour of $\zeta_{T,D}$ for every $T \in \PDOk{0}$. Indeed, observe that Formula \eqref{eq:zDA} relating $\zDA$ to $\zeta_{T,D}$'s involves a holomorphic remainder $R_N(s)$, which is harmless when it comes to the poles and residues, but might contribute to the behaviour on the verticals. Also, there is no good reason to believe that a heat trace expansion of the form \eqref{ST_hk_exp} for a given $D$ will imply a similar one for $D_{\Ag}$, although this is indeed the case for commutative spectral triples and also for the noncommutative torus (cf. Lemma \ref{lem:existence of asymp for nct}). In full generality of noncommutative geometry this is a stumbling block and we list it as Problem \ref{prob:heat vs zeta} \ref{Existence of fluctuated expansion} in Chapter~\ref{chap:open}.

After revealing the blot on the landscape, let us enjoy the bright perspective. Given the heat trace asymptotic expansion of $\DD$ and $\DA$ we can:
\begin{enumerate}
\item Deduce the large-$\Lambda$ asymptotic expansion of $\SA$ and $S(\DA,f,\Lambda)$ with the help of Theorem \ref{thm:ST_SA}.

\item Express the coefficients of the expansion of $S(\DA,f,\Lambda)$ as noncommutative integrals of operators polynomial in $\Ag$, via Formulae \eqref{eq:zeta TDA}, \eqref{eq:Kn}, \eqref{Y in OP}, \eqref{X}.
\end{enumerate}

Let us illustrate it in the case of $\ahd$ with a simple dimension spectrum:

\begin{theorem}
\label{thm:SA fluctuated simplified}
Let $(\A,\H,\DD)$ be a regular $p$-dimensional spectral triple with a simple dimension spectrum and let $\Ag = \Ag^* \in \PDOk{0}$. Assume, moreover, that 
\begin{align}
\label{heat_simplified-fluctuated}
\Tr \, e^{-t \abs{D_\Ag}} = \sum_{\alpha \in \Sd^+} a_{\alpha}(\vert D_\Ag \vert) \, t^{-\alpha} + a_{0}(\vert D_{\Ag} \vert) + \oz(1),
\end{align}
where the, possibly infinite, series over $\alpha$ is absolutely convergent for all $t>0$.
\\
Then, the function $\zDA$ is regular at $0$ and, for any $f \in \Cl_0^r$ with $r>p$,
\begin{multline}
\label{eq:SA simple}
\Tr \, f( \abs{\DD_\Ag} / \Lambda ) = \sum_{\alpha \in \Sd^+} \Lambda^{\alpha} \!\!\int_0^{\infty}\!\! \!x^{\alpha - 1} f(x) dx  \sum_{n=0}^{\lfloor p -\Re(\alpha)\rfloor} \,\ncint P_{n}(\alpha,D, D^{-1},\Ag) \awD^{-\alpha} \\
+  f(0) \Big[ \zeta_{D}(0) +\sum_{k=1}^p \tfrac{(-1)^k}{k}\,\ncint (\Ag \,D^{-1})^k \Big]+ \oinf(1),
\end{multline}
where $P_{n} \in \PDOk{-n}$ are polynomials in all variables and of degree $n$ in $\Ag$: 
\begin{align}\label{P}
P_0  = \bbbone, && P_1 = -\alpha \Ag D^{-1}, && P_2 = \tfrac{\alpha}{4}(\alpha+2)(\Ag D^{-1})^2 +  \tfrac{\alpha^2}{4} \Ag^2 D^{-2}, \quad \dotsc
\end{align}
\end{theorem}

\begin{proof}
Firstly, we apply Corollary \ref{cor:SA_simple} to obtain
\begin{align*}
\Tr \, f( \abs{\DD_\Ag} / \Lambda ) = \sum_{\alpha \in \Sd^+} \Lambda^{\alpha} \!\!\int_0^{\infty}\!\! \!x^{\alpha - 1} f(x) dx \,\ncintA \awDA^{-\alpha} +  f(0) \zDA(0) + \oinf(1).\vspace*{-0.1cm}
\end{align*}
Secondly, we invoke Corollary \ref{cor:difference of fluctuated zeta} \textit{i)} to deduce that $\ncint_{\!\!\Ag} \awDA^{-\alpha} = \ncint \awDA^{-\alpha}$. \\
Thirdly, we use Formulae \eqref{eq:expansion for D_A}, \eqref{eq:Kn} and \eqref{Y in OP} with $N = \lfloor p -\Re(\alpha)\rfloor$, as in the proof Equations \eqref{eq:DA for N=2} and \eqref{eq:Y mod -3} to expand $\awDA^{-\alpha} = \sum_n P_n \awD^{-\alpha}$. Moreover, observe that it is sufficient to stop at $N = \lfloor p -\Re(\alpha)\rfloor$, because we have $$P_{N+1} \awD^{-\alpha} \in \OP^{-N-\Re(\alpha)-1} \subset \OP^{-p-\epsilon}\,\text{  for some }\,\epsilon > 0.$$ Hence, $P_{N+1} \awD^{-\alpha}$ is trace-class and $\ncint P_{N+1} \awD^{-\alpha} =0$.

Now, we utilise Formulae \eqref{eq:DA for N=2}, \eqref{eq:Y mod -3} and recall that, since $\ncint$ is a trace, $\ncint \nabla(T) \abs{D}^{-r}=0$ for any $T \in \PDOC$ and any $r\in \RR$. Denoting \vspace*{-0.1cm}
$$T\doteq T'\,\text{ when }\,\ncint T \awD^{-\alpha}=\ncint (T' +R)\awD^{-\alpha}\,\text{ for }\,R\in \OP^{-3}, \vspace*{-0.2cm}$$
we have
\begin{align*}
& \abs{D_\Ag}^{-\alpha} \doteq [\bbbone-\tfrac{\alpha}{2}(XD^{-2}-\tfrac{1}{2}X^2D^{-4})+\tfrac{\alpha^2}{8}(XD^{-2})^2]\,\abs{D}^{-\alpha}, \\
& XD^{-2} =(\Ag D+D\Ag+\Ag^2)D^{-2} \doteq 2\Ag D^{-1}+\Ag^2D^{-2}, \quad (XD^{-2})^2 \doteq X^2D^{-4},\\
& X^2 D^{-4}\doteq 2 \Ag D\Ag D^{-3}+\Ag D^2 \Ag D^{-4} +\Ag^2 D^{-2} \doteq 2(\Ag D^{-1})^2+2 \Ag^2 D^{-2},
\end{align*}
since $$\Ag D\Ag D^{-3}=(\Ag D^{-1})^2+ \Ag D^{-1}[D^2,\Ag] D^{-3}=(\Ag D^{-1})^2 \mod \OP^{-3}$$ and similarly, $$\Ag D^2 \Ag D^{-4}=\Ag^2 D^{-2} + \Ag[D^2,\Ag]D^{-4}=\Ag^2 D^{-2} \mod \OP^{-3}.$$
This gives Formulae \eqref{P} and makes it clear that every $\Ag$ is accompanied by $D^{-1}$, so that $P_n \in \PDO^{-n}$ is a polynomial in $\Ag$ of order $n$. It is also straightforward how to compute the higher $P_n$'s. 

Finally, we apply Theorem \ref{thm:zeta(0) and fluctuation} to express $\zDA(0)$ using Equation \eqref{eq: fluctuation of zeta at zero}.
\hfill $\Box$
\end{proof}

Beyond the case of a triple $\ahd$ with a simple dimension spectrum we can still express the coefficients of the asymptotic expansion of $S(f,\DA,\Lambda)$ in terms of the polynomials $P_n$, provided we have at hand the heat trace expansion \eqref{ST_hk_exp} for \emph{both $\DD$ and $\DA$}. Obviously, the analogue of Formula \eqref{eq:SA simple} would be considerably more involved, but the coefficients can be computed algorithmically from Formulae \eqref{fluc}, \eqref{eq:Kn}, \eqref{Y in OP} and \eqref{X} --- along the same lines as in Theorem \ref{thm:SA fluctuated simplified}.

It is important to stress that if $\ahd$ does not have a simple dimension spectrum then the noncommutative integral is \emph{not} invariant under fluctuations. Indeed, Formula \eqref{fluc} implies that $\ncint_{\!\!\Ag} T - \ncint T$ would, in general, involve the higher order residues. On the other hand, the highest residue $\ncint^{[d]}$ remains insensible to fluctuations --- recall Formula \eqref{ncintd DA}. More generally, we have:

\begin{proposition}
Let $\ahd$ be a regular $p$-dimensional spectral triple with a dimension spectrum of order $d$ and let $\Ag = \Ag^* \in \PDOz$. Assume that
\begin{align}
\Tr \, e^{-t \abs{D}} \;\; & = \sum_{z \in X_0} \big[\sum_{n=0}^{d} a_{z,n}(\abs{D}) \log^n t\, \big]\,t^{-z} + \Oz(t^{-p+1}), & \label{eq:SAA}\\
\Tr \, e^{-t \abs{D_\Ag}} & = \sum_{z \in X_0} \big[\sum_{n=0}^{d} a_{z,n}(\abs{D_{\Ag}}) \log^n t\, \big]\,t^{-z} + \Oz(t^{-p+1}), \notag
\end{align}
where the, possibly infinite, series over the set $X_0 \subset \CC$ are absolutely convergent for all $t>0$ and $\min \{ \Re(z) \, \vert \, z \in X_0 \} > p-1$. \\
Then, for any $f \in \Cl_0^r$ with $r>p$, we have
\begin{align}\label{SAA_inv}
\Tr \, f (\abs{D_\Ag}/\Lambda) - \Tr \, f (\abs{D}/\Lambda) = \Oinf(\Lambda^{p-1}).
\end{align}
\end{proposition}
Note that $X_0$ here coincides with the one in Theorem \ref{thm:ST_SA}, if we choose $r_1 = -p-1$.
\begin{proof}
On the strength of Theorem \ref{thm:ST_SA}, with Formulae \eqref{ST_SA:Y} and \eqref{coef_a}, it is sufficient to check that $\ncintd{k}_{\!\!\Ag} \vert D_\Ag \vert^{-z} = \ncintd{k} \awD^{-z}$ for all $z \in X_0$ and $k\in \{0,1,\dotsc,d\}$. \\
But with $\Re(z) > p -1$, we have, by Corollary \ref{cor:dominant invariance} and Formula \eqref{eq:expansion for D_A} respectively, $$\ncintd{k}_{\!\!\Ag} \vert D_\Ag \vert^{-z} = \ncintd{k} \vert D_\Ag \vert^{-z} = \ncintd{k} \vert D \vert^{-z} + \ncintd{k} R(z) = \ncintd{k} \awD^{-z},$$ since $R(z) \in \OP^{-\Re(z) - 1}$ is trace-class.
\hfill $\Box$
\end{proof}

In the commutative case, as well as on the noncommutative torus, Formulae \eqref{eq:SAA} are simply $\Tr e^{-t \, T} = a_{p,0}(T) \,t^{-p} + \Oz(t^{-p+1})$, with $T \in \{\vert D \vert, \vert D_{\Ag} \vert\}$. In this case, Equation \eqref{SAA_inv} is interpreted as the invariance of the dominant term under the fluctuations. However, we can see that if the dimension spectrum of a given spectral triple is non-simple and/or has poles in the vertical strip $p > \Re(z) > p-1$, then not only the dominant term is immune to the fluctuations.

\begin{definition}
\label{def:tadpole}
Let $(\A,\H,\DD)$ be a regular$p$-dimensional spectral triple with a simple real dimension spectrum.\\ The {\it tadpole} \index{tadpole@tadpole} Tad$_{D_\Ag}(\alpha)$ of order $\alpha$ is the linear term in $\Ag$ in the asymptotics \eqref{eq:SA simple}. 
\end{definition}
We have just shown that 
\begin{align*}
& \text{Tad}_{D_\Ag}(\alpha)=(-\alpha \int_0^\infty x^{\alpha-1}f(x)\,dx)\,\ncint \Ag D^{-1} \,\text{ if }\,\alpha\in \Sd^+,\\
& \text{Tad}_{D_\Ag}(0)=-f(0)\ncint \Ag D^{-1}.
\end{align*}
See also \cite[Proposition 3.5]{Tadpole}.

This notion of a tadpole is borrowed from quantum field theory where $\Ag D^{-1}$ is a one-loop graph with a fermionic internal line and only one external bosonic line $\Ag$, thus looking like a tadpole. There are no tadpoles in commutative geometries on manifolds, also the ones with boundaries and torsion, or on noncommutative tori \cite{Tadpole1,Tadpole,ILVTorsion}. The vanishing of tadpoles means that a given geometry $(\A,\H,\DD)$ is a critical point for the spectral action \cite[p. 210]{ConnesMarcolli}. On the other hand, the existence of tadpoles means a priori the instability of the quantum vacuum --- see \cite{Gere,ILS}.

\section{Operator Perturbations}
\label{sec:pert}

The spectral action \eqref{SA} is only a part of the large programme on the differentiation of operators and on perturbation theory. There exist several strongly related approaches and we only briefly sketch a few of them.

The first approach is the Lifshits formula given, for two selfadjoint, possibly unbounded, operators $H_0$ and $V$ acting on $\H$, by
\begin{align}
\label{trace formula}
\Tr\,[f(H_0+V)-f(H_0)]=\int_{-\infty}^{\infty}  \xi(\lambda) \,f'(\lambda) \,d\lambda,
\end{align}
where $\xi$ is the so-called Krein shift function (see for instance \cite{Chattopadhyay}). If the perturbation $V$ is a trace-class operator and $R(A)$ is the resolvent of $A$, the holomorphic function $$G(z)\vc \det\big(\bbbone +V \,R(H_0)(z)\big)$$ satisfies $$G^{-1}(z)G'(z) =\Tr[R(H_0)(z)-R(H_0+V)(z)]$$ and we define $$\xi(\lambda)\vc \pi^{-1}\lim_{\epsilon \to 0}\, \arg G(\lambda+i \epsilon) \text{ for almost all }\lambda.
$$
Then, we get 
\begin{align*}
&\log G(z)=\int_{-\infty}^\infty \xi(\lambda)\,(\lambda-z)^{-1} \,d\lambda \,\,\text{ when }\Im(z)\neq 0,\\
&\int_{-\infty}^\infty \vert \xi(\lambda) \vert \, d\lambda \leq \Tr V \qquad \text{and} \qquad \Tr(V)=\int_{-\infty}^\infty  \xi(\lambda)\,d\lambda.
\end{align*}
Moreover, \eqref{trace formula} holds true at least for functions $f$ which are smooth and compactly supported. Remark that the function $\xi$ is a spectral shift, because, when $\lambda$ is an isolated eigenvalue of both $H_0+V$ and $H_0$ with respective multiplicities $m$ and $m_0$, then $\xi(\lambda+0)-\xi(\lambda-0)=m_0-m$ and $\xi$ has a constant integer value in any interval located within the resolvents sets of both $H_0+V$ and $H_0$.\\
Formula \eqref{trace formula} is commonly employed in the scattering theory (see \cite{simon2005trace}) and has also been adapted to the noncommutative framework for the computation of the spectral action beyond the weak-field approximation, as in \cite{ILVWeak}.

Another approach is the following: Assume that the unbounded operator $H_0$ is selfadjoint and $V$ is also selfadjoint, but bounded. Let $\{\lambda_k\}_{k=1}^\infty$ be the eigenvalues of $H_0$ counted here with their multiplicities and let $\{\psi_k\}_{k=1}^\infty$ be the corresponding normalised eigenvectors. 
If $$V_{r,s} \vc \langle V\psi_{r},\,\psi_s \rangle,$$ we have the following Taylor asymptotic expansion for $f \in C_c^{N+1}(\RR)$:\vspace*{-0.2cm}
\begin{align}
&\hspace*{-0.6cm} \Tr\,f(H_0+V)=\label{Taylor expansion by divided differences}\\
& \Tr\,f(H_0) + \sum_{n=1}^{N-1} \tfrac{1}{n} \sum_{i_1,\cdots,i_n} V_{i_1,i_2}\cdots V_{i_{n-1},i_{n}}V_{i_n,i_1}\,(f')^{[n-1]}(\lambda_{i_1},\ldots,\lambda_{i_n}) +R_{H_0,f,N},\nonumber \vspace*{-0.2cm}
\end{align}
where $(f')^{[m]}$ is the divided difference of order $m$ of $f'$ and, of course, the difficulty is to control the remainder $R_{H_0,f,N}$. Defining \vspace*{-0.15cm}
$$\R_{H_0,f,N} \vc f(H_0+V)-\sum_{n=0}^{N-1} \tfrac{1}{n!}\, \tfrac{d^n\,}{dt^n}\vert_{t=0}\,f(H_0+tV)
,
$$
where the G\^ateaux derivatives $ \tfrac{d\,}{dt}\vert_{t=0}\,f(H_0+tV)$ are taken in some uniform topology, the strategy is to prove that
\begin{align}
\label{reste}
R_{H_0,f,N}=\Tr\, \R_{H_0,f,N}\quad \text{and}\quad \Vert \R_{H_0,f,N}\Vert = \OO_{\Vert V \Vert \to 0}(\Vert V \Vert^N).
\end{align}
Recall first that the G\^ateaux derivative of a function $f: X \to Y$ between two locally convex topological vector spaces $X,\,Y$, is  $$f'(h)(v) \vc \lim_{t\to 0}\, t^{-1}\,[f(h+tv)-f(h)] \text{ for }h,\,v\in X.$$
It is linear in $v$ if $X,\,Y$ are Fr\'echet spaces.\\
This approach is used in \cite{van2011perturbations} or \cite[Section 7.2.2]{WalterBook} to compute the spectral action $$S_\DD(V)\vc\Tr\,f(\DD+V)$$ via \eqref{Taylor expansion by divided differences} for a finitely summable spectral triple $\ahd$, setting $H_0=\DD$ and under the following assumptions:
\begin{enumerate}
\item $f(x)=\int_0^\infty e^{-tx^2}\,d\phi(t)$ with a positive measure $d\phi$.
\item $\int_0^\infty t^\alpha\,\Tr \,  [\vert\DD \vert^\beta \,e^{-t(\epsilon \DD^2 -\gamma)}] \, d\phi (t) < \infty$ for any $\alpha,\beta,\gamma>0$ and $0\leq \epsilon <1$.
\item Moreover, $V\in \B^2(\H)$ -- as defined below.
\end{enumerate}
The derivation $\delta'$ given in \eqref{regularity} defines a family of seminorms $\{\norm{\delta'^n(T)}\}_{n\in \N}$ and the vector spaces $$\B^n(\H)\vc \set{T\in \B(\H):\, \Vert \delta'^k(T) \Vert <\infty \text{ for all } k\leq n}$$ become Fr\' echet spaces, implying the Fr\' echet differentiability of the spectral action in terms of the perturbation $V$. \\
Remark that if the triple is regular, any selfadjoint one-form in $\Omega_\DD^1(\A)$  is in $\B^2(\H)$. Without entering into the details of \eqref{Taylor expansion by divided differences},  the main steps are an iterated use of the Duhamel formula (see \eqref{Duhamel})
\begin{align*}
e^{-t(\DD+V)^2}=e^{-t(\DD^2 +X)}=e^{-t\DD^2}-t\int_0^1 e^{-st(\DD^2+X)}\,X \,e^{-(1-s)t\DD^2}\,ds\,
\end{align*}
for the perturbation of the heat operator (the second line of the assumption guarantees the convergence of the trace of this integral), and the fact that the Taylor expansion $$S_\DD(V)=\sum_{n=0}^\infty \tfrac{1}{n!}\, S_\DD^{[n]}(0)(V,\cdots,V)$$ has coefficients which are given by $$S_\DD^{[n]}(0)(V,\cdots,V)=(n-1)!\sum_{i_1,\cdots,i_n} V_{i_{n-1},i_{n}}V_{i_n,i_1}\,f'(\lambda_{i_1},\ldots,\lambda_{i_n}).$$

Formula \eqref{Taylor expansion by divided differences} has been generalised in \cite{skripka2014asymptotic} under the only assumption that the selfadjoint operator $H_0$ has a compact resolvent (but no positivity of $H_0$ nor summability condition on its spectrum is required), $V=V^*$ is bounded and $f\in C^{N+1}(\RR)$ with compact support ($f$ is not necessarily positive or even). Even if the proof is more subtle it essentially goes through the quoted steps. For instance, using the inclusion $$C^{N+1}_c(\RR) \subset \set{f:\,f^{(n)} \text{ and }\F(f^{(n)}) \in L^1(\RR), \,n=0,\cdots,N},$$
the exponentials appearing in the iterated Duhamel formula are now unitaries so the second line of the assumption is not necessary.

Further generalisation into the theory of multiple operator integrals, compatible with the formalism of differentiation of operator functions, is possible. It relates to the important notion of spectral flow in the setting of type II von Neumann algebras --- see, for instance, \cite{azamov2007operator}.  

\pagebreak
%
%
 \bibliographystyle{spmpsci}
 \bibliography{SA_BIB}
%

%
%
%

\chapter{Open Problems}
\label{chap:open}

\abstract{As a desert, we serve a number of open problems connected with the subject matter of the book. Some of them consider the general framework of spectral triples and its possible extensions, while the other are more specific and relate to the properties of the spectral action. The problems are essentially of mathematical nature, though, at least in some cases, the conceptual skeleton strongly depends upon the input from physics. To our mind, the solution to each of these stumbling blocks would advance our understanding of the foundations and implications of the Spectral Action Principle. We therefore cordially invite the Reader to contemplate the list below, both from mathematical and physical perspectives.}

\begin{enumerate}[label=\arabic*),wide=0pt,itemsep=4pt]

\item  \textit{Existence of spectral triples.}\label{The space of Dirac operators} Despite quite a few illustrative examples, the territory of spectral triples remains vastly uncharted. Beyond the almost-commutative enclave, the constructive procedures are available in some specific contexts: isospectral deformations \cite{ConnesIsospectral}, AF $C^*$-algebras \cite{Christensen3} or crossed products \cite{ST_crossed_products,Iochum-Masson}.

The first road towards spectral triples is to fix a $C^*$-algebra $\bar{\A}$ \footnote{Such a fixed $C^*$-algebra can come from physics --- as the natural algebra of observables of a given system (see e.g. \cite{Haag,Keyl,Strocchi}).} and then to find (and classify!) all possible smooth dense $^*$-subalgebras $\A$, along with the compatible operators $\DD$ acting on a chosen Hilbert space $\H$. The ultimate aim would be to understand the `space of the operators $\DD$', which is the domain of the spectral action functional \eqref{SA}. This would certainly require the theory of Fredholm modules and $K$-homology \cite{HigsonKhomology}, but surely much more than that.

Alternatively, one can fix a sequence of real numbers seen as the spectrum of $\DD$ and seek a compatible algebra $\A$. Such a situation appears, for instance, in the context of manifolds with boundary \cite{ILVTorsion}.

\item \textit{Non-unital spectral triples.} We have focused exclusively on unital spectral triples, which correspond to compact manifolds in the classical case. Non-unital triples have been studied to some extent \cite{CGVRInt,CGRSIndex,RennieSmooth,RennieSummable}. In such a case the operator $\DD$ does not have a discrete spectrum and one needs to adjust the definition of the spectral action. The easy way is to introduce a `spatial' (or infrared) cut-off: $\Tr\, \Phi f(\abD/\Lambda)$, with a suitable $\Phi \in \A$, but its choice yearns for a deeper physical motivation. One could also promote the energy scale $\Lambda$ to a dilaton field \cite{CCscale,IochumMoyal,MoyalSA,Wulkenhaar}.

\item \textit{Twisted spectral triples.} \label{twisted} 
Motivated by the type III noncommutative geometry, where there are no finitely summable spectral triples \cite{Connes-Mosco}, a few notions of twisted spectral triples emerged. The original one relaxes the condition $[\DD,a] \in \B(\H)$ to the demand of the boundedness of ``twisted commutators'': $\DD a - \sigma(a) \DD \in \B(\H)$, with a given automorphism $\sigma$ of the algebra $\A$ \cite{CNNR,Fath-Khalk,MoscoviciTwisted,Ponge-Wang}.
The notions of reality, pdos, regularity etc. get modified accordingly \cite{TwistedReality,Landitwisted,Mat-Yunck}. A related notion of \textit{modular spectral triples} \cite{CPRmodular} is motivated by deformation of some physical models or quantum groups \cite{BCL,Greenfield,Iochum-Masson,SeniorKaad,Matassa}. There is also an increasing interest for the `twist' in particle physics \cite{Devas-Marti}. 
However, the spectral action has not yet been systematically approached in this context, to our best knowledge.

\item \label{sec:Lorentz} \textit{The Lorentzian signature.} An insistent problem, which has been swept under the carpet is that of the signature: The notion of a spectral triple generalises Riemannian manifolds with the Euclidean signature. Alas, the spacetime we are living in has a Lorentzian signature, instead. At the almost-commutative level, one can bypass the problem using the old (and somewhat murky) trick of the Wick rotation \cite{LizziWick}. However, a rigorous approach requires a deep conceptual change:

A distinctive feature of spaces with the signature $(-,+,+,\dotsc)$ is the existence of a causal structure. Needless to remind that the micro-causality is one of the key axioms in quantum field theory. A rigorous notion of causality suitable for a noncommutative geometry has been proposed \cite{CQG2013} and studied \cite{UNIV2017,PROC2015,PRD2017,SIGMA2014,JGP2015,causMoyal}.

The Dirac operators $\Dslash$ on pseudo-Riemannian manifolds, formally similar to their Euclidean colleagues, are drastically different \cite{Baum}. First of all, $\Dslash$ is \emph{not} selfadjoint in $L^2(M,\SS)$. Secondly, $\Dslash$ has infinite dimensional eigenspaces (in particular, $\dim \Ker \Dslash = \infty$), hence no function $f$ can render $f(\abs{\Dslash})$ trace-class.\\
Let us also point out that in the Lorentzian context a nonunital algebra is mandatory, as compact spacetimes always contain causal loops, which is undesirable. The problem of compactification of spacetimes, or ``attaching a causal boundary to the spacetime'' is an old-standing and a hard one --- both on the conceptual and on the technical side (cf. \cite{MinguzziCompactification} for a nice abstract mathematical formulation).

Yet another serious obstacle is the presence of the notorious spacetime singularities, which seem to be an inherent element of our Universe, as attested by the famous Hawking--Penrose theorems (see \cite{Wald} or any other mathematically oriented textbook on general relativity). On the mathematical side, it creates problems with the incompleteness of spacetime manifolds (see, for instance, \cite{Beem}).

As one can see, the algebraic situation is rather dramatic already at the commutative level. Nevertheless, the programme of ``pseudo-Riemannian spectral triples'' is being systematically developed 
\cite{Besnard16a,TwistedLorentz,Rennie12,KoenIndefinite,UNIV2017,PROC2015,F4,F5,Pas,Stro,WalterCylinder}. 
Its central idea is to work with a Krein space \cite{Bog}.\\
The ``Lorentzian spectral action'' has not yet even emerged from the depth. A result, which might shed some light on it is the Lorentzian version of the index theorem~\cite{BarIndex}.

\item \label{prob:dim_sp} \textit{The dimension spectrum.} The computation of the dimension spectrum of a given spectral triple is a formidable task. Beyond the almost-commutative realm it has been accomplished only for a few examples listed on p. \pageref{dim_sp examples}. Firstly, the existence of the dimension spectrum is by no means automatic --- there exist spectral zeta functions admitting no meromorphic continuation (see \cite{Schrohe} and \cite[Section 5]{HeatEZ}). Secondly, even if we do have a meromorphic extension of $\zD$, the one of $\zTD$ does not come for free, even for $T \in \A$. In the worked out examples, one firstly unravels the meromorphic extension of the basic zeta function $\zD$ and then constructs the ones for $\zTD$ by expressing the operators $T \in \PDOz$ in the eigenbasis of $\DD$. Whereas the poles of $\zeta_D$ and $\zeta_{T,D}$ do not coincide in general (see Example \ref{ex:dim_sp}), in all known cases we actually have $\Sd \subset \mathlarger{\mathlarger{\cup}}_{k \in \N} \PP(\zeta_{\vert D\vert^{-k},D})$. So it seems as if the whole dimension spectrum is actually encoded in the operator $\DD$.\\
Is this a general fact or a specific property of the worked out examples?

\item \label{prob:heat vs zeta} \textit{Heat traces and zeta functions.} In Chapter \ref{chap:asymptotic} we have spied into the intimate interplay between the small-$t$ asymptotic expansion of $\hKH$ and the meromorphic extension of $\zKH$. But on the route we only employed the properties of general Dirichlet series and the geometric origin of the operators $H$ and $K$ remained concealed. It would be highly desirable to understand what impact might the geometrical dwelling of $H$ and $K$ have on the problems we encountered. Concretely:

\begin{enumerate}[label=\alph*)]

\item \label{prob:zeta 2 heat} When is $\zKH$ of polynomial growth on the verticals (in which case $\hKH$ admits an asymptotic expansion with the vanishing contribution \eqref{Gk})?

\item \label{prob:ht2zeta_pat} Are the pathologies illustrated on Figure \ref{fig:ht2zeta_pat} always avoided?

\item When is the asymptotic expansion of $\hKH$ actually convergent?

\item \label{D vs Dsq} When is the existence of the asymptotic expansion of $\Tr \,e^{-t \vert\DD\vert}$ equivalent to the existence of an expansion of $\Tr\, e^{-t \DD^2}$?

\item \label{Existence of fluctuated expansion} When does the existence of an expansion of $\Tr e^{-t \abs{\DD}}$ imply the existence of an expansion of $\Tr e^{-t \abs{\DA}}$ for a suitable fluctuation $\DA = \DD + \Ag$? 

\end{enumerate}

\item \label{The role of f} \textit{The role of the cut-off function.} Arguably, the smooth cut-off function $f$ involved in the definition of the spectral action $\SA$ is of non-geometric origin. It might encode some physical input (such as the parameters of the Standard Model), however, one has to keep in mind that for general noncommutative geometries --- and for the almost-commutative ones, but beyond the asymptotic expansion --- the full shape of $f$ enters into the game. Adopting a puritanical point of view, one should set $f = \chi_{[0,1]}$ and create the tools to study the asymptotic expansion of the raw spectral action $N_{\abD}(\Lambda)$.

\item \textit{The coefficients of the asymptotic expansion.} In the almost-commutative framework the coefficients of the large-energy asymptotic expansion of the spectral action have a pellucid geometric and physical interpretation --- as pictured in formula \eqref{SA_SM}. However, beyond the homely classical ground, the situation is more obscure. The ``curvature'' has been defined \cite[Defintion 1.147]{ConnesMarcolli} and computed on a conformally rescaled noncommutative 2-torus \cite{ConnesModular,ConnesTretkoff,FatKhal}. However, such an interpretation remains controversial, as we have imperceptibly entered into the domain of nonminimal operators, which might also involve the torsion \cite{IMHeat,SitarzNonminimal}.\\
The comprehension of geometry and physics behind the coefficients of the general asymptotic expansion of the form \eqref{ST_SA:Y} is a serious challenge.

\item \label{expansion and distribution} \textit{Distributional approach to the asymptotics.} The framework presented in Section~\ref{sec: On the asymptotics of distributions} is very appealing. It would be desirable to employ it beyond the realm of classical pdos. To that end, one would need to understand how the assumption \eqref{hyp: asymptotic ex} should be reformulated for general elements in $\K'(\RR,\L(\H))$ and when it is met.
      
\item \label{Beyond the asymptotic expansion} \textit{Beyond the asymptotic expansion.} As forewarned on p. \pageref{asymptotics vs exact}, the asymptotic expansion of the spectral action (even in its full glory) might fail to capture the `exponentially small physical phenomena' \cite{Boyd} encoded in the nonperturabtive expression $\Tr f(\abD/\Lambda)$. The pressing question, both for mathematicians and physicists, is: What is actually lost? More concretely (see also Remark \ref{rem:different operators with same asymp}), assume that $\Tr e^{-t\,\abD} - \Tr e^{-t\,\vert \DD'\vert} = \Oz(t^{\infty})$. What can be said about the operators $\DD$ versus $\DD'$?\\
If $\Dslash$ is the standard Dirac operator on the flat 3-torus, then the information about the chosen spin structure $\SS$ is concealed in the $\Oinf(\Lambda^{-\infty})$ term of the spectral action expansion --- see Example \ref{ex:Poisson_tori}. Is this a general feature of commutative geometries? \\
A related problem is the impact of different possible selfadjoint extensions of $\DD$ on the spectral action --- see \cite[Theorem 7.2]{EstradaFulling1999} for the casus of a differential operator.

\item \textit{Spectral action for the noncommutative torus.} 
In the spectral action for the noncommutative torus (Theorem \ref{thm:SA for nc torus}), the constant term in $\Lambda$ coincides with its classical value for the commutative torus after the swap $\tau \leftrightarrow \int_{\TT^d}$. This has been demonstrated in dimensions 2 and 4. We conjecture that this holds true in arbitrary dimension.

The asymptotic expansion of the spectral action for noncommutative tori relied heavily on the Diophantine hypothesis. It would be very instructive to see how does it look like (if it exists at all!) for $\Theta$, which does not meet the Diophantine condition. 

\end{enumerate}

%
%
 \bibliographystyle{spmpsci}
 \bibliography{SA_BIB}
%

%
%
%

\appendix

\chapter{Classical Tools from Geometry and Analysis}
\label{classical tools}

\section{About ``Heat Operators''}
\label{About the exponential of an unbounded operator}

Recall first that if $P$ is an unbounded operator, then $e^{\,t\,P}$ cannot be defined by the series $\norm{\cdot}$-$\lim_{N\to \infty}\sum_{n=0}^N \tfrac{t^n}{n!}\,P^n$. What is needed is that $G:\,t \geq 0 \to e^{\,t\,P}\in \B(\H)$ is a strongly-continuous contraction semigroup (i.e. $G(0)=\bbbone, G(s)\,G(t)=G(s+t)$, $\norm{G(t)}\leq 1$ and the function $t \geq 0 \to G(t)\,\psi$ is norm-continuous for each $\psi\in H$). Then, a closed densely defined operator $P$ is the generator of this semigroup, i.e. by definition $G(t)=e^{\,t\,P}$, if and only if $\RR^+$ is contained in the resolvent set of $P$ and $\Vert(P-\lambda)^{-1}\Vert\leq \lambda^{-1}$ for all $\lambda>0$ (see \cite[Chapter 14]{Grubbbook1}).\\
Moreover, $e^{\,t\,P}=\text{strong-lim}_{n\to \infty} (\bbbone-(t/n)P)^{-n}$ and
\begin{align}
(P-\lambda)^{-1}=-\int_0^\infty e^{-\lambda\, t}\,e^{\,t\,P}\,dt,\quad \text{for } \quad \Re(\lambda)>0
\end{align}
holds true in this generality since actually the right half-plane $\{ \lambda \in \CC \,\,\vert\, \,\Re(\lambda)>0\}$ is in the resolvent set of $P$ and $\Vert(P-\lambda)^{-1}\Vert\leq \Re(\lambda)^{-1}$. The generator $P$ is upper semibounded: $\Re(\langle P\psi,\psi \rangle)\leq 0$, for all $\psi \in \Dom P$.\\
Sometimes, the generator of $G(t)$ is denoted by $-P$ like in \cite[Section X.8]{SimonReed2}.

We can rephrase previous results as a constructive way to get the exponential. Let $P$ be an unbounded operator on the Hilbert space $\H$ such that $P-\lambda$ is invertible in the sector $$
\Lambda_\theta \vc \{r\,e^{i\phi} \,\,\vert\,\, r\geq 0, \,\abs{\phi}\geq \theta\},\, 0<\theta < \tfrac{\pi}{2}$$\index{_a6lambda_n(T)1@$\Lambda_\theta$} and assume there exists $c$ with
\begin{align}
\label{hyp:resolvent decay}
\norm{(P-\lambda)^{-1}} \leq c\,(1+\abs{\lambda}^2)^{-1/2},\quad \forall \lambda \in \Lambda_\theta.
\end{align}
This allows to define 
\begin{align}
\label{heat as integral}
e^{-t\,P} \vc \tfrac{i}{2\pi} \int_\C e^{-t\,\lambda}\,(P-\lambda)^{-1}\,d\lambda,
\end{align}
where $\C=\C_{r_0,\theta}$\index{_zC@$\C_{r_0,\theta}$} (with $\theta<\tfrac{\pi}{2}$) is the path from $\infty$ along the ray $r\,e^{i\theta}$ for $r\geq r_0$ followed by a clockwise circle around zero of radius $r_0$ and ending at infinity along the ray $r\,e^{-i\theta}$. Since the two rays lie in the right-half plane, the exponential decay of $e^{-t\,\lambda}$ guarantees the convergence of the integral. We will need that fact for instance in \eqref{f(H/Lambda)} of Section \ref{subsec:Laplace} on Laplace transform. 

To close these remarks on heat operators, we recall that many functions can be defined through an integral along a curve in $\CC$. For instance, given a selfadjoint operator $P$, one defines (cf. \cite[\S 10]{Shubin}), for any $z \in \CC$ with $\Re(z)>0$, $$P^{-2z} \vc\tfrac{1}{i2\pi} \int_{\lambda \in \C} \lambda^{-z}\,(\lambda-P^2)^{-1}\, d\lambda$$ along the curve $\C$ as in \eqref{heat as integral}. Typically, to control the norm-convergence of the integral, one uses
\begin{align}
\label{eq:estimate resolvent}
&\Vert (\lambda-P)^{-1}\Vert= \sup_{\mu\in \spec P} \vert \mu-\lambda \vert^{-1}= \text{dist}(\lambda, \spec P)^{-1} \leq \abs{\Im (\la)}^{-1},\\
&\Vert P(\la-P)^{-1}\Vert =\sup_{\mu\in \spec P} \vert \mu\vert\,\vert\mu-\lambda \vert^{-1}.\label{eq:estimate resolvent1}
\end{align}
Moreover, if $P=P^*$ is positive,
\begin{align}
\label{eq:estimate resolvent for positive}
\Vert P(\la-P)^{-1}\Vert \leq \left\{ \begin{array}{lr}\abs{\la}\abs{\Im (\la)}^{-1}& \text{ if }\, \Re(\la)\geq 0,\\1 &\text{ if }\,\Re(\la)<0, \end{array} \right.
\end{align}
which follows from $\Vert P(\la-P)^{-1}\Vert =\sup_{\mu\in \spec P} f(\mu)$ with $f(\mu)=\mu\,\abs{\mu-\la}^{-1}$ and the computation of the maximum of $f$.

\section{Definition of pdos, Sobolev Spaces and a Few Spectral Properties}
\label{sec:Def of pdo}

There are several good textbooks on pdos: \cite{Gilkey1,Grubbbook,Hormander,Shubin,Taylorbook}. For the heat trace asymptotics of a pdo we closely follow  \cite[Section 4.2]{Grubbbook} and the nice notes \cite{Schrohe2014}. 
See also \cite{Avramidi,Gilkey2,Kirsten} for the computation of heat kernel coefficients.

To study pdos on $\RR^d$ we need a few basic definitions:

\smallskip

-- $\langle x\rangle\vc (1+\norm{x}^2)^{1/2}$ and $\langle x,\eta\rangle\vc (1+\norm{x}^2+\abs{\eta})^{1/2}$\index{_Ax@$\protect\langle x,\eta\protect\rangle$} for $x\in \RR^n, \,n\in \N^*$, $\eta\in \CC$.

\smallskip

-- $S^m \vc \{p:\,(x,\xi)\in \RR^d\times \RR^d \to\CC\,\text{ such that }\, \vert \partial_\xi^\alpha \partial_x^\beta p(x,\xi) \vert \leq c_{\alpha\beta} \langle \xi\rangle^{m-\abs{\alpha}}\,\}$ is the set of {\it symbols of order $m\in \RR$}. Here $\alpha,\,\beta$ are in $\N^d$ with $\abs{\alpha}=\sum \alpha_i$. \\
This yields a family of seminorms on $S^m$ defined by $$\abs{p}_{m,\alpha,\beta} \vc \sup_{x}  \vert \partial_\xi^\alpha \partial_x^\beta p(x,\xi) \vert  \langle \xi\rangle^{-m+\abs{\alpha}}.$$ 
The set of smoothing symbols \index{symbol!smoothing} is $S^{-\infty}\vc \cap_m \,S^m$.

\smallskip

-- The symbol $p\in S^m$ has the expansion $p \sim \sum_{j=0}^\infty p_{m-j}$, when $p_{m-j} \in S^{m-j}$ and for each $n$, $p-\sum_{j=0}^n p_{m-j} \in S^{m-n}$.\\
It is named {\it classical} \index{symbol!classical} if, moreover, for all $j$, 
$$
p_{m-j}(x,t\xi)=t^{m-j}p(x,\xi),\, \forall (x,\xi)\in \RR^d\times \RR^d,\, \norm{\xi} \geq 1,\,t\geq 1.
$$
It is said to be {\it elliptic} \index{symbol!elliptic}when $p(x,\xi)$ is invertible and $\vert p(x,\xi)^{-1}\vert \leq c\,\langle \xi\rangle^{-m}$ for all $x$ and $\norm{\xi}>r$ for some $r\geq 0$. \\
A classical symbol $p\in S^m$ is elliptic when $p_m(x, \xi)$ is invertible for all $x\in \RR^d$  and $\norm{\xi}=1$. \\
From now on, {\it we assume in this appendix that all symbols are classical.}

\smallskip

--  Every symbol $p$ gives rise to a pdo acting on $u$ in the Schwartz space $\Scl(\RR^d)$ via the inverse Fourier transform by 
$$
(\mathcal{O}(p)\, u)(x)\vc \F^{-1}[(p(x,\cdot)\,\F[u](\cdot)](x) = \int_{\RR^d} e^{i\,2\pi \,x.\xi}\,p(x,\xi)\,\F[u](\xi) \,d\xi.\index{_zOozzzz@$\mathcal{O}(p)$}  \index{symbol!pdo associated to}
$$
This definition is compatible with the product of operators as for $p\in S^{m_1}, \,q\in S^{m_2}$ there exists a symbol called the Leibniz product of $p$ and $q$, denoted $p\circ q\in S^{m_1+m_2}$\index{_Aocirc@$p\circ q$}, such that $\mathcal{O}(p)\,\mathcal{O}(q)=\mathcal{O}(p\circ q)$ with the expansion  
$$(p\circ q)(x,\xi)\sim \sum_\alpha \tfrac{(-i)^{\abs{\alpha}}}{\alpha!} (\partial^\alpha_\xi p)(x,\xi)\,(\partial^\alpha_x q)(x,\xi).$$
In particular, when $p\in S^m$ is elliptic, there exists a symbol $q$ called {\it the parametrix}\index{par@parametrix} such that $p\circ q -1$ and $q\circ p -1$ are both in $S^{-\infty}$. 

\smallskip

-- The Sobolev spaces \index{Sobolev spaces} read $$H^s(\RR^d)\vc \{u\in\Sclp(\RR^d)\,\,\vert\,\, \langle \xi\rangle^s\,\,\F(u)\in L^2(\RR^d)\}$$\index{_zHzzz@$H^s$} for $s\in \RR$, with scalar product $$(u,u') \vc \int \F(u)(\xi)\, \overline{\F(v)(\xi)} \,\langle \xi\rangle^{2s}\,d\xi$$ and complete for the norm $\norm{u}^2_s \vc (u,u)$. For instance $\delta_y\in H^s(\RR^d)$ if $s<-d/2$. \\
We have, $H^0(\RR^d)=L^2(\RR^d)$ and if $s>d/2$ then  $H^s(\RR^d)\subset C(\RR^d)$ (Sobolev embedding theorem). \\
When $s>d/2$, any bounded operator $A:\,H^{-s}\to H^s$ is an integral operator with a Schwartz kernel given by $k_A(x,y)=(A \delta_x,\delta_y)$. 

\medskip

We now adapt previous definitions to 
\begin{center}
\it a compact boundaryless Riemannian manifold $M$ of dimension $d$,
\end{center}
so we need coordinate charts $(U,h)$ where $U$ are open sets in $M$ and $h$ are diffeomorphisms from $U$ to open sets in $\RR^d$. \\
Let $P:\,C^\infty(M)\to C^\infty(M)$; when $\phi,\,\psi\in C^\infty_c(U)$ \index{_zCmin@$C^\infty_c(M)$} (smooth functions on $U$ with compact support), the localised operator $\phi P\psi$ on $C^\infty_c(U)$ is pushed-forward as $h_*(\phi P\psi)$ on $C^\infty_c(h(U))$. The operator $P$ is a pseudodifferential operator of order $m$ when each such localisation is a pdo of order $m$ on $h(U)$. \\
Then, one extends $S^m$ to symbols on $M$ as follows:
$$
S^m(M)\vc \{ p(x,\xi)\in C^\infty(T^*M)\,\,\vert\,\, h(\phi p) \in S^m(h(U)),\,\forall (U,h,\phi)\}.
$$
Of course, $P$ is said to be elliptic or smoothing if all its local symbols have such a property and the set of classical pdos of order $m$ is denoted by $\Psi^m(M)$ \index{_apdoz@$\Psi(M)$} which defines $\Psi(M) \vc \cup_m \Psi^m(M)$.\\
For $x\in (U,h)$ one defines the principal (or leading) symbol of $P\in \Psi^m(M)$ as $$p_m(P)\vc h^*p_m(h_*(\phi P \phi)) \in S^m(M)/ S^{m-1}(M),$$ where one chooses a $\phi \in C^\infty_c(M)$ equal to 1 in vicinity of $x$.\\ One checks that this principal symbol makes sense and is invariantly defined on $T^*M$ while the total symbol is quite sensitive to a change of coordinates. Moreover, for each $p\in S^m(M)$, one constructs $P\in \Psi^m(M)$ with $p_m(P)=p$ via the partition of unity.

A new extension is possible when $P$ acts on sections of a smooth vector bundle $E$ of finite rank over $M$ equipped with a smooth inner product. So, typically, a fiber is acted upon by a matrix. 
By local triviality, on can define $H^s(M,E)$ using a partition of unity. Hence $P:\,C^\infty(M,E)\to C^\infty(M,E)$ is called a pseudodifferential operator of order $m$ if every localisation is a matrix of pdos of order $m$ for all charts $U$ over which $E$ is trivial. 
Such operators, the symbols of which are now matrices, define $\Psi^m(M,E)$\index{_apdozE(M,E)@$\Psi^m(M,E)$}. \\ The properties of classicality and ellipticity are generalised in a straightforward way.\\
In particular, $P$ has a matrix-valued kernel $K_P$, which in local coordinates reads
\begin{align}
\label{kernel of P}
k_P(x,y)=\int_{\RR^d} e^{i\,2\pi(x-y).\xi}\,p(x,\xi)\,d\xi.
\end{align}

Similarly, the Sobolev space $H^s(M)$ is defined as the set of distributions $u$ on $M$ which, in a given local patch $U$, satisfy $u\in \DD'(U)$ with $\psi u\in H^s(\RR^d)$ for all $\psi\in C^\infty_c(U)$. By the Rellich theorem, the inclusion $H^s(M) \hookrightarrow H^{t}$ is compact for any $t<s$ and even trace-class when $t+d<s$.

From the beginning, classical symbols can be seen as objects defined up to $S^{-\infty}$. It has the following consequence: $P\in\Psi(M,E)$ is smoothing (i.e. all of its local symbols are smoothing) if and only if $P$ has a Schwartz kernel $k_P$ which is smooth on $M\times M$. For instance, $\psi P\phi$ is smoothing if $\psi,\phi\in C_c^\infty(M)$ with disjoint supports.

We now recall a few classical results on pdos --- see loc. cit. at the beginning of this section. They provide links between a pdo and the same object, but viewed as an operator which has eventually several closed extensions on a Hilbert space. Recall that a bounded operator between Banach spaces is Fredholm\index{Fredholm operator} if it has a finite dimensional kernel and cokernel and a closed range. \pagebreak

\begin{theorem}
\label{thm:general on pdo}
Let $P\in \Psi^m(M,E)$. Then:

i) The extension of $P:\, H^s(M,E)\to H^{s-m}(M,E)$ is bounded for all $s\in \RR$.

ii) If $P$ is elliptic, all previous extensions are Fredholm operators, which means that there exists a Fredholm inverse which is a pdo of order $-m$. \\
In particular, when $m>0$, $P:C^\infty(M,E)\to C^\infty(M,E)$ acting on the Hilbert space $H^0(M,E)=L^2(M,E)$ has only one closed extension with the domains $H^m(M,E)$ and a spectrum either equal to $\,\CC\,$ or discrete without accumulation points except 0.  

iii) When $P:\, H^s(M,E)\to H^{s-m}(M,E)$ is invertible for some $s$, then we have $P^{-1}\in \Psi^{-m}(M,E)$.

iv) The space $\Psi^0(M,E)$ is an algebra.

v) If $P \in \Psi^m(M,E)$ with $m<-d$, then $P$ has a continuous kernel and its extension $P$ on $L^2(M,E)$ is trace-class with $$\Tr_{L^2(M,E)} P=\int_M \tr_E k_P(x,x)\,dx.$$
\end{theorem}
The inverse in $(ii)$ is obtained by the construction of a local parametrix over a local chart, which, after being patched, gives rise to two  pdos $Q$ and $Q'$ of order $-m$ such that $R=QP-\bbbone$ and $R'=PQ'-\bbbone$ are smoothing pdos. Moreover, $Q-Q'$ is a smoothing pdo. So, modulo smoothing pdos, $Q$ is the left and right inverse of $P$. \\
We now present some details from the constructive proof of this theorem for parameter-dependent symbols.

\section{Complex Parameter-Dependent Symbols and Parametrix}
\label{sec: Complex parameter-dependent symbols}

Let us be given an elliptic pdo $P\in \Psi^m(M,E)$ of order $m>0$ with the matrix symbol $p\sim \sum_{j=0}^\infty p_{m-j}$. Despite the nice unique $L^2$-extension of $P$ provided by Theorem~\ref{thm:general on pdo}, it is still interesting to look at $e^{-t \,P}$  not only as an operator on $L^2(M,E)$ but as a smoothing pdo (or, similarly, to regard the complex power $P^s$ as a pdo of order $m\,\Re(s)$).

The main idea, which we expound in some detail below, is to replace the resolvent $(P-\lambda)^{-1}$ of $P$ by a parameter-dependent parametrix, the symbol of which is under control. Since we want to control the integrand of \eqref{heat as integral}, we assume the following (uniform) parameter-ellipticity of the principal symbol of $P$: 
\begin{hypothesis}
\label{hyp: principal symbol with parameter}
The operator $(P-\lambda)^{-1}$ exists in the left keyhole region $V_{r_0,\theta}$ defined by $\C_{r_0,\theta}$ for $\theta<\tfrac{\pi}{2}$. \\
Moreover, we have the resolvent growth condition: The matrices $p_m(x,\xi)-\lambda$ are invertible for all $x,\xi$ when $\lambda \in V_{r_0,\theta}$ and 
\begin{align*}
\norm{(p_m(x,\xi)-\lambda)^{-1}} \leq (1+\norm{\xi}^2+\abs{\lambda}^{2/m})^{-m/2} =\langle\xi,\lambda^{1/m}\rangle^{-m}.
\end{align*}
\end{hypothesis}
For the principal symbol, let us introduce the strictly homogeneous symbol $p_m^h$: 
$$
p_m^h(x,\xi)\vc \norm{\xi}^m\,p_m(x,\xi/\norm{\xi})
$$
(which coincides with $p_m$ for $\norm{\xi}\geq 1$, but is now homogeneous of degree $m$ for all $\xi\neq 0$) and we can rephrase the hypothesis as: $p_m^h(x,\xi)$ has no eigenvalues in $V_{r_0,\theta}$ for all $\xi\neq 0$ (see \cite[Lemma 1.5.4]{Grubbbook}). Recall that $p_m(x,\xi)$ is homogeneous of degree $m$ only for $\norm{\xi}\geq 1$ and we have to control the integral in $\xi$, as in \eqref{kernel of P}.

Let $(U,h)$ be a fixed coordinate chart. For $x\in U, \,\xi\in \RR^d,\,\lambda= \eta^m\in V_{r_0,\theta}$ and $j\in \N$, we want to generate a parametrix by an inductive sequence (see \cite{Kumano})
\begin{align*}
& q_{-m}(x,\xi,\eta) \vc (p_m(x,\xi)-\eta^m)^{-1},\\
& q_{-m-j}(x,\xi,\eta)\vc -\sum_{k=1}^{j-1}\sum_{\alpha,\ell} \tfrac{(-i)^{\abs{\alpha}}}{\alpha !} \partial^\alpha_\xi q_{-m-k}(x,\xi,\eta)\,\partial^\alpha_x p_{m-\ell}(x,\xi)(p_m(x,\xi)-\eta^m)^{-1}
\end{align*}
where the second sum is over $\alpha \in \N^d,\,\ell \in \N$ such that $k+\ell +\abs{\alpha}=j$.\\
In the scalar case (i.e. the fibers of $E$ are one-dimensional),
\begin{align}
\label{eq: parametrix scalar case}
q_{-m-j}(x,\xi,\eta) = \sum_{k=1}^{2j} p_{j,k}(x,\xi)(p_{m}(x,\xi)-\eta^m)^{-k-1},
\end{align}
where the $p_{j,k}$ are symbols of order $mk-j$ obtained from $p_m,\cdots,p_{m-j}$.

Using $$\partial ((p_m-\eta^m)^{-1})=-(p_m-\eta^m)^{-1}(\partial p_m)(p_m-\eta^m)^{-1},$$ with $\partial =\partial_x$ or $\partial_\xi$, one checks that 
$\partial_\xi^\alpha\partial_x^\beta q_{-m-j}$ is a sum of terms of the form
\begin{align}
\label{eq: derivative of a matrix symbol}
(p_m-\eta^m)^{-1}\partial_\xi^{\alpha_1}\partial_x^{\beta_1} p_{m-k_1}(p_m-\eta^m)^{-1}\cdots\partial_\xi^{\alpha_r}\partial_x^{\beta_r} p_{m-k_r}(p_m-\eta^m)^{-1},
\end{align}
where $\sum_{ \ell=1}^r \,k_\ell+\abs{\alpha_\ell}=j+\abs{\alpha}$. \\
Moreover, $$\partial_\xi^\alpha\partial_x^\beta q_{-m-j}(x,t\xi,t\eta)=t^{-m-j} \partial_\xi^\alpha\partial_x^\beta q_{-m-j}(x,\xi,\eta)  \quad\text{ for } \norm{\xi}\geq 1, \,t\geq 1
$$
and there are at least two factors $(p_m(x,\xi)-\eta^m)^{-1}$ if either $j>0$ or $\abs{\alpha}+\abs{\beta}>0$.\\
This implies the following estimates:
\begin{align}
\!\! \!\Vert\partial_\xi^\alpha\partial_x^\beta q_{-m-j}(x,\xi,\eta)\Vert=
 \left\{\begin{array}{ll}
\!\!\OO_{\Vert \xi\Vert\to \infty}(\langle\xi,\eta\rangle^{-m}\,\langle \xi\rangle^{-j-\abs{\alpha}}) & \text{ for }j\in \N,\label{eq: estimates 1 for q m-j} \\
\!\!\OO_{\Vert \xi\Vert\to \infty}(\langle\xi,\eta\rangle^{-2m}\,\langle \xi\rangle^{m-j-\abs{\alpha}}) & \text{ if } j+\abs{\alpha}+\abs{\beta}>0.
\end{array}\right.
\end{align}
Now, defining $q$ such that $q\sim\sum_{j\in \N}\, q_{-m-j}$, we get 
\begin{align}
\label{eq: finite parametrix}
\partial_\xi^\alpha\partial_x^\beta\big[q(x,\xi,\eta)-\sum_{j<J} q_{-m-j}(x,\xi,\eta)\big]=\OO_{\Vert \xi\Vert\to \infty}(\langle\xi,\eta\rangle^{-2m}\,\langle \xi\rangle^{m-J-\abs{\alpha}}).
\end{align}
We claim that: 
{\it 
If $r(x,\xi,\lambda) \vc q(x,\xi,\lambda^{1/m}) \circ (p(x,\xi)-\lambda)-\id$, then for any $N$, all seminorms in $S^{-N}$ of the symbol $r$ are $\OO_\infty(\langle\lambda\rangle^{-1})$.
}\\
The proof is based on the decomposition of the series defining $p,\,q$ in finite sums and remainders and their Leibniz products via the above estimates.\\
Then, one gets rid of the local chart $U$ by patching previous parametrices to get a parameter-dependent pdo 
$Q(\lambda)$ (associated to the symbol $q(x,\xi,\lambda^{1/m})$), such that $R(\lambda)=Q(\lambda)\,(P-\lambda)-\bbbone$ (with the symbol $r(x,\xi,\lambda) $) is a smoothing pdo, where the seminorms $\norm{\cdot}_{-N,\alpha,\beta}$ of its symbol are $\OO_\infty(\langle\lambda\rangle^{-1})$ for each $N$.\\
Since formally $[Q(\lambda)\,(P-\lambda)]^{-1}=[\bbbone+R(\lambda)]^{-1}=\sum_{j=0}^\infty (-R(\lambda))^j$, we deduce that
\begin{align*}
(P-\lambda)^{-1}-Q(\lambda)&=\big([Q(\lambda)(P-\lambda)]^{-1}-\bbbone\big)\,Q(\lambda)=\sum_{j=1}^\infty(-R(\lambda))^j \,Q(\lambda)
\end{align*}
has a norm which is $\OO_0(\langle\lambda\rangle^{-2})$ since $\norm{Q(\lambda)}=\OO_\infty(\langle\lambda\rangle^{-1})$ so that:

\begin{proposition}
\label{prop: Olambda-2}
We have $\norm{(P-\lambda)^{-1}-Q(\lambda)}=\OO_\infty(\langle\lambda\rangle^{-2})$ for all $\lambda\in \Lambda_\theta$.
\end{proposition}

Remark that we can define similarly $Q'(\lambda)$ such that $(P-\lambda)\,Q'(\lambda) -\bbbone=R'(\lambda)$ for another smoothing pdo $R'(\lambda)$.\\
Thus $Q(\lambda)-Q'(\lambda)=R(\lambda)Q'(\lambda)-Q(\lambda)R'(\lambda)$ is a smoothing pdo. Moreover, the operator $(P-\lambda)^{-1}$ can be seen as an elliptic pdo of order $-m$ since, by \eqref{eq: estimates 1 for q m-j}--\eqref{eq: finite parametrix}, the operator $(P-\lambda)^{-1}-\sum_{j<J} q_{-m-j}(x,\xi,\lambda^{1/m})$ is a pdo of order $m-J$, the seminorms of which  are $\OO_\infty(\langle\lambda\rangle^{-2})$.\\
This explains why the operator $(P-\lambda)^{-1}$ seen as a pdo is nothing else than $Q(\lambda)$. Consequently, $G(t)\vc e^{-t\,P}$ defined by \eqref{heat as integral} is also equal to $\tfrac{i}{2\pi}\int_\C e^{-t\lambda} Q(\lambda)\,d\lambda$.

Another consequence is that $(P-\lambda)^{-1}$ is compact by Theorem \ref{thm:general on pdo} $(v)$ so $P$ has a discrete spectrum without accumulation points (compare with Theorem \ref{thm:general on pdo} $(ii)$).

\section{\texorpdfstring{About $e^{-t\,P}$ as a pdo and About its Kernel}{About exp(-tP) as a pdo and about its kernel}}
\label{sec: The operator exp(-iP) as a pdo}

Let $P\in \Psi^m(M,E)$ be elliptic with $m>0$ and let its principal symbol satisfy \eqref{hyp: principal symbol with parameter}. 

\begin{theorem}
\label{thm: G(t) as pdo}
For $t> 0$, $G(t)\vc e^{-t\,P}$ is a pdo of order zero, the symbol $g(x,\xi,t)$ of which has the expansion $g(x,\xi,t) \sim \sum_{j=0}^\infty g_{-j}(x,\xi,t)$ with
\begin{align}
\label{def: symbols of exp(-tP)}
g_{-j}(x,\xi,t) \vc \tfrac{i}{2\pi}\int_\C e^{-t\lambda} \, q_{-m-j}(x,\xi,\lambda^{1/m}) \,d\lambda, \quad \text{ for } j\in \N.
\end{align}
Moreover, $g_0(x,\xi,0)=1$, while $g_{-j}(x,\xi,0)=0\,$ for $j \in \N^*$.
\end{theorem}

\begin{proof} We follow \cite[Theorem 4.2.2]{Grubbbook}. \\
In a coordinate chart, we have
\begin{align*}
e^{-t\,P}&=\tfrac{i}{2\pi}\int_\C e^{-t\lambda} \mathcal{O}[q(x,\xi,\lambda^{1/m})] \,d\lambda =\mathcal{O} [\tfrac{i}{2\pi}\int_\C e^{-t\lambda} (q(x,\xi,\lambda^{1/m}) \,d\lambda ]  \\
&\sim\sum_{j=0}^\infty \,\,\mathcal{O} [g_{-j}(x,\xi,t)] .
\end{align*}
First, one checks that $g_0=e^{-t\,p_m}$ (by residue calculus) and the homogeneity property
\begin{align*}
g_{-j}(x,r\xi,r^{-m}t)=r^{-j}g_{-j}(x,\xi,t), \quad\text{ for } \norm{\xi}\geq 1,\,r\geq 1. 
\end{align*}
We now want to prove the following estimates: There exists $c>0$ such that
\begin{align}
\label{eq: derivatives of e-j}
\!\!\!\Vert \partial_\xi^\alpha \partial_x^\beta g_{-j}(x,\xi,t) \Vert \leq \langle \xi \rangle^{-j-\abs{\alpha}}\,(t^{1/m} \langle \xi \rangle)^a\,e^{-c\,\langle \xi \rangle^mt}, \,\, \forall a \leq \min(m,j+\abs{\alpha}).
\end{align}
These hold true for $j=0$ and for $j\geq 1$ we begin with the scalar case, cf. \eqref{eq: parametrix scalar case}:
\begin{align}
\label{eq: e-j scalar}
g_{-j}=\sum_{k=1}^{2j} p_{j,k}\, \tfrac{i}{2\pi}\int_\C e^{-t\lambda}(p_m-\lambda)^{-k-1}\,d\lambda=\sum_{k=1}^{2j}p_{j,k}\,\tfrac{1}{k!}\, t^k \,e^{-tp_m}.
\end{align}
We get 
\begin{align*}
\Vert P_{j,k}\,t^k\Vert \leq \langle \xi\rangle^{mk-j}\,t^k=\langle \xi\rangle^{m-j}\,t\,(\langle \xi\rangle^{d}\,t)^{k-1}\text{ with }(\langle \xi\rangle^{m}\,t)^{k-1}\leq 1\text{ when }\langle \xi\rangle^{d}\,t\leq 1.
\end{align*}
Moreover, $\Vert (\langle \xi\rangle^{d}\,t)^{k-1}e^{-tp_m/2}\Vert \leq 1$ if $\langle \xi\rangle^{m}\,t\geq 1$. Thus, the estimates are proved when $\alpha=\beta=0$ with $a=m$. \\
As a consequence, we cover the situation where $a<m$ and for $\alpha,\,\beta$ non-zero one differentiates under the integral of \eqref{def: symbols of exp(-tP)} until the estimate for $\alpha=\beta=0$ applies. For the non-scalar case, one proceeds as in the proof of \eqref{eq: estimates 1 for q m-j} thanks to the expansion \eqref{eq: derivative of a matrix symbol} for the derivatives. \\
The equality $g_{-j}(x,\xi,0)=0$ is a consequence of \eqref{eq: derivatives of e-j}.

We can now conclude the proof thanks to the following argument: The symbol $g(x,\xi,t) \sim \sum_{j=0}^\infty g_{-j}(x,\xi,t)$ can be chosen in such a way that 
$$
\Vert \partial_\xi^\alpha \partial_x^\beta (g-\sum_{j<J} g_{-j})\Vert \leq \langle \xi \rangle^{-J-\abs{\alpha}}\, (t^{1/\mu} \langle \xi \rangle)^a \,e^{-c\,\langle \xi \rangle^m \,t},\,\,\,\forall a \leq \min(m,j+\abs{\alpha}).
$$
So, for any integer $J$, $$\tfrac{i}{2\pi}\int_\C e^{-t\lambda} \,\mathcal{O}[q(x,\xi,\lambda^{1/m})-\sum_{j<J} q_{-m-j}(x,\xi,\lambda^{1/m})] \,d\lambda$$ is a pdo of order zero and the asymptotics of the symbol for $G(t)=e^{-t\,P}$ is fully justified.
\hfill $\Box$
\end{proof}

But since we are interested in $\Tr e^{-t\,P}$ it is worthwhile to control the kernel of $G(t)$ as a function of $t$ and to give an alternative proof of the previous theorem. On the way, it is shown that $e^{-t\,P}$ is a smoothing pdo.

Let $G_{-j}(t)$ be the pdo defined locally by $g_{-j}(x,\xi,t)$ after patching local charts. 

\begin{lemma}
\label{lem: kernel of G}
For any $t>0$:

i) The kernels of $G_{-j}(t)$ satisfy the estimates
\begin{align*}
&\Vert K_{G_0}(x,y,t) \Vert \leq t^{-d/m}\,e^{-c'\,t},\quad\!\!\Vert K_{G_{-j}}(x,y,t)\Vert \leq t^{(j-d)/m}\, e^{-c'\,t}, & \!\!\!\!\!\!\!\! \text{for }\, 0<j<m+d;\\
&\Vert K_{G_{-m-d}}(x,y,t)\Vert \leq t(1+\abs{\log t})\,e^{-c'\,t},\quad\!\!\Vert K_{G_{-j}}(x,y,t) \Vert \leq t\, e^{-c'\,t}, & \text{for } j>m+d.
\end{align*}
Moreover, on the diagonal we have, with $c_{j}(G,x)\vc \int_{\RR^d} g^h_{-j}(x,\xi,1) \,d\xi$,
\begin{align*}
 K_{G_{-j}}(x,x,t) =c_{j}(G,x)\,t^{(j-d)/m}+\OO_0(t), \quad \text{for }\, 0\leq j<m+d.
\end{align*}

ii) The remainder $G^r_J(t) \vc G(t)-\sum_{j<J} G_{-j}(t)$ satisfies
\begin{align*}
& \Vert K_{G^r_J}(x,y,t)\Vert \leq t(1+\abs{\log t}) \,e^{-c'\,t}, &&\!\!\!\!\!\!\!\! \text{for } J>m+d,\\
& K_{G^r_J}(x,y,0)=0, &&\!\!\!\!\!\!\!\! \text{for } J>d.
\end{align*}
\end{lemma}

\begin{proof}
i) We begin with the kernel of $G_0$:  
\begin{align*}
K_{G_0}(x,y,t)=\int_{\RR^d} e^{i2\pi (x-y).\xi}\,e^{-tp_m(x,\xi)}\,d\xi.
\end{align*}
Since $\Vert e^{-t p_m}-e^{-t p_m^h}\Vert=\OO_0(t)$ when $\norm{\xi}\leq 1$, we get
\begin{align*}
\Vert K_{G_0}(x,y,t) \Vert &\leq c_1 \int_{\RR^d} \Vert e^{-t\,p_m^h(x,\xi)} \Vert \,d\xi+ c_1\int_{\norm{\xi}\leq 1} \Vert e^{-tp_m(x,\xi)}-e^{-tp_m^h(x,\xi)}\Vert \,d\xi\\
&\leq c_1\, t^{-d/m}\int_{\RR^d} e^{-c \norm{\eta}^m}\,d\eta + c_2\, t\leq c_3\, t^{-d/m} +c_2\,t,
\end{align*}
and for the diagonal
\begin{align*}
K_{G_0}(x,x,t)&=t^{-d/m} \int_{\RR^d} e^{-p_m^h(x,\eta)}\,d\eta + \int_{\norm{\xi}\leq 1} e^{-tp_m(x,\xi)}-e^{-tp_m^h(x,\xi)} \,d\xi \\
& \cv c_0(G,x) \,t^{-d/m} + \OO_0(t).
\end{align*}
When $j\neq 0$, the estimate $\Vert g_{-j}^h(x,\xi,t)\Vert \leq \norm{\xi}^{m-j} \,t\,e^{-c\,t\norm{\xi}^m}$ follows from \eqref{eq: derivatives of e-j}. Since the last function is $\xi$-integrable when $m-j>-d$, we deduce that the kernel $$K_{G_{-j}}(x,y,t)=\int_{\RR^d} e^{i2\pi (x-y).\xi}\,g_{-j}(x,\xi,t)\,d\xi$$ is $\Oinf(e^{-c\,t/2})$, while for $t>0$,
\begin{align*}
\Vert K_{G_{-j}}(x,y,t) \Vert & \leq c_1 \int_{\RR^d} \Vert g_{-j}^h \Vert \,d\xi +c_1 \int_{\norm{\xi}\leq 1} (\Vert g_{-j} \Vert + \Vert g_{-j}^h \Vert)\,d\xi\\
& \leq c_2 \,t\int_{\RR^d} \norm{\xi}^{m-j}\,e^{-c\,t\,\norm{\xi}^m}\,d\xi + c_3\,t=c_4\, t^{(j-d)/m} + c_3\,t.
\end{align*}
As above, still with $0<j<d+m$, one gets 
\begin{align}
K_{G_{-j}}(x,x,t) & =\int_{\RR^d} g_{-j}^h(x,\xi,t)\,d\xi +c_1 \int_{\norm{\xi}\leq 1} (g_{-j}-g_{-j}^h)(x,\xi,t)\,d\xi \notag\\
& =c_{j}(G,x)\,t^{(j-d)/m}+\OO_0(t). \label{eq: integral kernel on diagonal}
\end{align}
Moreover, using \eqref{eq: derivatives of e-j}, we obtain
\begin{align*}
\Vert K_{G_{-j}}(x,y,t) \Vert &\leq t \Big( \int_{\norm{\xi}\leq1}+\int_{\norm{\xi}\geq1} \Big) \langle \xi \rangle^{m-j}\,e^{-c\,t \langle \xi \rangle^m}\,d\xi \\
& \leq t \big(1+\int_1^\infty r^{m-j-d-1}\,e^{-c\,t\,r^m}\,dr \big)  \\
& \leq \left\{\begin{array}{lr} 
t(c_1+c_2 \abs{\log t}), & \quad \text{for }j=m+d, \\
t(c_1+c_2\,t^{-1+(j-d)/m)}), &\quad \text{for }j>m+d.
\end{array}\right.
\end{align*}
This completes the proof of $i)$.

$ii)$ The remaining symbol $$q^r_J(x,\xi,\lambda) \vc q(x,\xi,\lambda^{1/m})-\sum_{j<J} q_{-m-j}(x,\xi,\lambda^{1/m})$$ gives $$G^r_J(t)=\tfrac{i}{2\pi}\int_\C e^{-t\lambda} \mathcal{O}(q^r_J)(\lambda) \,d\lambda,$$ which is a pdo of order $-m-J<-2m-d$, so has a continuous kernel (and is trace-class) by Theorem \ref{thm:general on pdo} $v)$. Moreover, 
\begin{align*}
K_{G^r_J}(x,y,t)=\tfrac{i}{2\pi}\int_\C e^{-t\lambda} K_{\mathcal{O}(q^r_J)}(x,y,\lambda)\,d\lambda,
\end{align*}
with $\Vert K_{\mathcal{O}(q^r_J)}(x,y,\lambda) \Vert = \OO_\infty(\langle \lambda\rangle^{-2})$ as in the proof of Proposition \ref{prop: Olambda-2}. Thus, the integral over $\C$ converges uniformly for all $t> 0$ since $\vert e^{-t\lambda} \vert \leq e^{-c\,t}$ for some $c>0$. Consequently, $K_{G^r_J}(x,y,t)$ is $\Oinf(e^{-c\,t})$ and as such, it has a continuous extension at zero since $K_{G^r_J}(x,y,0)=0$, because in $\tfrac{i}{2\pi}\int_\C K_{\mathcal{O}(q^r_J)}(x,y,\lambda)\,d\lambda$, 
$\C$ can be deformed into a closed contour around zero. 
\\
We already know from Section \ref{About the exponential of an unbounded operator} that $G^r_J(t)\in C^\infty((0,\infty),\B(\H))$ and, for $t\leq1$,
\begin{align}
\Vert \partial_t \,K_{G^r_J}(x,y,t) \Vert & =\tfrac{1}{2\pi} \big{\Vert} \int_{\C} \lambda e^{-t\lambda} \,K_{G^r_J}(x,y,\lambda)d\lambda \big{\Vert}  \leq c_1 e^{-c\,t} +c_3 \int_{r_0}^\infty e^{-c_2\abs{\lambda}t }\langle\lambda\rangle^{-1}d\lambda \notag\\
& \leq c_4+c_5 \int_{r_0t}^\infty e^{-c_2\,a}\,a^{-1} \,da \leq c_6+c_7 \abs{\log t}. \label{dt of kernel}
\end{align}
(See the definition of $r_0$ after \eqref{heat as integral}.) \\
The Taylor series in $t$ gives $\vert K_{G^r_J}(x,y,t) \vert \leq c\,t(1+\abs{\log t})$ when $t\in (0,1]$ and hence the announced estimate.
\hfill{$\Box$}
\end{proof}

Technically, it is useful to employ $$e^{-t\,P}=\tfrac{i}{2\pi}\int_\C e^{-t\lambda}\,\lambda^{-k} \,P^k\,Q(\lambda)\,d\lambda,\quad \forall k\in \N,$$ which follows from $Q(\lambda)=\lambda^{-1}(\lambda-P+P)Q(\lambda)=-\lambda^{-1}+\lambda^{-1}P\,Q(\lambda)$ which after iteration gives $Q(\lambda)=-\sum_{j=1}^{k} \lambda^{-j}P^{j-1} + \lambda^{-k}P^kQ(\lambda)$ and $\int_\C e^{-t\lambda} \,\lambda^{-j}\,d\lambda=0$.\\
We denote: $$Q^{(k)}(\lambda) \vc P^k\,Q(\lambda),$$ seen as a pdo of order $(k-1)m$ with the symbol $q^{(k)} \sim \sum_{j \in \N} q^{(k)}_{(k-1)m-j}$ and, preserving the notation, $$Q^{(k)}_{(k-1)m-j}\vc \mathcal{O}[q^{(k)}_{(k-1)m-j}]\,\text{and }\,Q^{(k)}_J\vc Q^{(k)}-\sum_{j<J} Q^{(k)}_{(k-1)m-j}.$$
As an example,
\begin{align*}
q&=-\lambda^{-1}-\lambda^{-2} p+\lambda^{-2} q^{(2)}\\
&= \lambda^{-1}-\lambda^{-2}\big(p_m+\dotsb +p_{-m} +\OO_{\Vert \xi \Vert \to \infty}(\langle \xi\rangle^{-m-1})\big)\\
&\hspace{1cm}+ \lambda^{-2}\big(q_{m}^{(2)}+\dotsb+ q_{-m}^{(2)}+\OO_{\Vert \xi \Vert \to \infty}(\langle \xi\rangle^{-1}\langle \xi,\lambda^{1/m}\rangle^{-m})\big).
\end{align*}
By iteration, $q_{-m-j}=-\lambda^{-2}p_{m-j}-\cdots -\lambda^{-k}p^{(k-1)}_{(k-1)m-j}+\lambda^{-k}q^{(k)}_{(k-1)m-j}$, where $p^{(\ell)}$ is the symbol of $P^\ell$. Since these symbols are independent of $\lambda$, we get
\begin{align}
g_{-j}(x,\xi,t)&=\tfrac{i}{2\pi}\int_\C e^{-t\lambda}\,q_{-m-j}(x,\xi,\lambda^{1/m})\,d\lambda \nonumber\\
&=\tfrac{i}{2\pi}\int_\C e^{-t\lambda}\,\lambda^{-k}\,q^{(k)}_{(k-1)m-j}(x,\xi,\lambda^{1/m})\,d\lambda. \label{int of q(k)}
\end{align}
In particular, for $k\in \N$ we rewrite the $Q^{(k)}$'s as
\begin{align}
&G(t)=\tfrac{i}{2\pi}\int_\C e^{-\la t}\,\la^{-k} \,Q^{(k)}(\la)\, d\la, \notag \\
&G_{-j}(t)=\tfrac{i}{2\pi}\int_\C e^{-\la t}\,\la^{-k} \,Q^{(k)}_{(k-1)m-j}(x,\xi,\la)\, d\la, \notag\\
&G^r_J(t)=\tfrac{i}{2\pi}\int_\C e^{-\la t}\,\la^{-k} \,Q^{(k)}_J(\la)\, d\la. \label{eq: GrJ(t) }
\end{align}
With the help of this representation, we can improve the key estimate \eqref{eq: derivatives of e-j} along with the estimates from Lemma \ref{lem: kernel of G}:

\begin{lemma}
\label{lem: improved estimates}
For any $k\in \N$,
\begin{align}
\label{eq: time derivative of g-j}
\Vert \partial_\xi^\alpha \partial_x^\beta \partial_t^k \, g_{-j} (x,\xi,t)\Vert \leq c'(x) \,\langle \xi\rangle^{km-j-\abs{\alpha}}\,e^{-c\langle \xi \rangle^m\,t}.
\end{align}
For the kernel of \,$G_{-j}=\mathcal{O}(g_{-j})$ we have
\begin{align}
\Vert\partial^k_t K_{G_{-j}}(x,y,t) \Vert \leq 
\left\{\begin{array}{ll} 
(1+\abs{\log t}) \, e^{-ct}, &\text{ for } j-km=d, \\
(1+t^{(j-km-d)/m}) \, e^{-ct}, &  \text{ for } j-km\neq d. \label{eq:dt of KG-j}
\end{array}\right.
\end{align}
On the diagonal we get,
\begin{align}
\label{eq: time derivative on diagonal}
\partial^k_tK_{G_{- j }}(x,x,t) =c_{j,k}(x)\,t^{(j-km-d)/m}+\OO_0(t^0),\quad j<(k-1)m+d,
\end{align}
with 
$$
c_{j,k}(x)\vc \int_{\RR^d} \partial_t^k g_{-j}^h(x,\xi,1)\,d\xi.
$$
For the kernel of $G_{-j}$ we have, with $j>km+d$ for some $k\in \N$, 
\begin{align}
\label{eq:Kg-j Taylor exp}
K_{G_{-j}}(x,y,t)=\sum_{\ell=1}^{k-1} \tfrac{t^\ell}{\ell !} \,\partial_t^\ell K_{G_{-j}}(x,y,0)+t^k R(x,y,t)
\end{align}
and $\partial_t^\ell K_{G_{-j}}(x,y,0)$ and $R(x,y,t)$ are continuous in $x,\,y$ and in $t\geq 0$.\\
When $J>(k+1)m+d$ for some $n\in \N$, the kernel of the remainder $G_{J}^r$ satisfies 
\begin{align}
& \Vert \partial_t^kK_{G_{J}^r}(x,y,t) \Vert \leq e^{-c\,t}, \notag \\
& K_{G_{J}^r}(x,y,t)=\sum_{\ell=1}^{k-1} \tfrac{t^\ell}{\ell !} \, \partial_t^\ell K_{G_{J}^r}(x,y,0)+t^k R(x,y,t), \label{eq: Taylor of remainder}
\end{align}
where $\partial_t^\ell K_{G^r_{J}}(x,y,0)$ and $R(x,y,t)$ are continuous in $x,\,y$ and in $t\geq 0$.
\end{lemma}

\begin{proof}
For $\ell\in \N^*$, 
$$
\partial^\ell_t g_{-j}(x,\xi,t)=\tfrac{i(-1)^\ell}{2\pi} \int_\C e^{-t\lambda}\, q^{(\ell)}_{(\ell-1)m-j}(x,\xi,\lambda^{1/m})\,d\lambda
$$
is derived from \eqref{int of q(k)} and the estimate \eqref{eq: time derivative of g-j} is proved in the same way as \eqref{eq: derivatives of e-j}.

Since $\partial^{\ell+1}_t g_{-j}$ is bounded when $t \to 0$ by the above formula, the function $\partial^\ell_t g_{-j}$ is continuous at $t=0$. Since $g_{-j}(x,\xi,0)=0$ by Theorem \ref{thm: G(t) as pdo}, we have the Taylor expansion $$g_{-j}(x,\xi,t)=\sum_{\ell=1}^k \tfrac{1}{\ell !}\partial_t^\ell g_{-j} (x,\xi,0)\, t^\ell + t^k r_{-j,k}(x,\xi,t)$$ where $\partial_t^\ell g_{-j} (x,\xi,0)$ are pdos of order $\ell m-j$.

The operators $\partial_t^k G_{-j}(t)=\mathcal{O}(\partial_t^k g_{-j})$ have kernels, for which we can repeat the same arguments used for the proof of Lemma \ref{lem: kernel of G} $i)$ in order to get \eqref{eq:dt of KG-j}.
\\
If $j>d$, $K_{G_{-j}}(x,y,0)=0$ by Lemma \ref{lem: kernel of G} and by a Taylor expansion, we get \eqref{eq:Kg-j Taylor exp}.

We need to control the remainder. For $k\geq 1$, choose $J>km+d$ and the presentation \eqref{eq: GrJ(t) } for the kernel of $K_{G^r_J}$. Since $G^r_J$ is a pdo of order $(k-1)m -M$, $\Vert q_J^r(x,\xi,\la)\Vert \leq \langle \xi\rangle^{km-J}\,\langle\xi,\la^{1/m}\rangle^{-m}$, thus, after a $\xi$-integration, we get the estimate of the kernel  $$\Vert K_{Q^{(k)}_J}(x,y,\la)\Vert \leq \langle \la \rangle^{-1}$$ and $$K_{G^r_J}(x,y,t)=\tfrac{i}{2\pi}\int_\C e^{-\la t}\,\la^{-k} \,Q^{(k)}_J(\la)\, d\la$$ because all $Q^{(k)}$  are holomorphic in $\la$. \\Thus, for $\ell \leq k-1$,
\begin{align*}
\Vert \partial_t^\ell K_{G^r_J}(x,y,t) \Vert &=\tfrac{1}{2\pi} \Vert \int_\C e^{-\la t}\la^{\ell-k}K_{Q^{(k)}_J}(\la) d\la \Vert \leq \tfrac{1}{2\pi} \vert \int_\C e^{-\la t} \langle \la\rangle^{\ell -k-1} d\la\vert \leq e^{-ct}.
\end{align*}
As a consequence we get: $$K_{G_{J}^r}(x,y,t)=\sum_{\ell=1}^{k-2} \tfrac{t^\ell}{\ell !} \, \partial_t^\ell K_{G_{J}^r}(x,y,0)+t^{k-1} R(x,y,t).$$
This Taylor expansion begins at $\ell=1$ since $K_{G^r_J}(x,y,0)=0$ by Lemma \ref{lem: kernel of G}. Swapping $m-1$ to $m$ completes the proof of the lemma.
\hfill{$\Box$}
\end{proof}

\section{\texorpdfstring{The Small-$t$ Asymptotics of $e^{-t\,P}$}{The small-t asymptotics of exp(-tP)}}
\label{sec: the heat asymptotics of exp(-tP)}

The above estimates can be used to prove that $e^{-t\,P}$ has a smooth Schwartz kernel for any $t>0$. Therefore, $e^{-t\,P}$ is a smoothing pdo, and hence is trace-class.

\begin{theorem}\label{thm:Gilkey}
Let $P\in \Psi^m(M,E)$ be elliptic with $m>0$ and let its principal symbol satisfy \eqref{hyp: principal symbol with parameter}. Then, $G(t)=e^{-t\,P}$ is a smoothing pdo and its kernel has the following asymptotics on the diagonal:
\begin{align*}
K_G(x,x,t) \!\underset{t\to 0}{\sim}\hspace{-0.3cm}\sum_{\substack{n\in \N \\ n-d\notin m\N}} \!\!\!\!c_{n-d}(x,G)\, t^{(n-d)/m} +\hspace{-0.3cm}\sum_{\substack{n\in \N \\ n-d\in m\N}} \!\!\! \!c_{n-d}(x,G) \,t^{(n-d)/m} \log t +\!\sum_{\ell \in \N} r_\ell(x,G) \, t^\ell \!\!\!,
\end{align*}
where the coefficient $c_{n-d}(x,G)\in C^\infty(M)$ depends only on $p_m,\cdots,p_{m-n}$, while $r_\ell(x,G)\in C^\infty(M)$ depends globally on the operator $P$.
\end{theorem}

\begin{proof}
The smoothness of the kernel $K_{-j}(x,y,t)$ in $x,y$ follows from \eqref{eq: time derivative of g-j} and it remains to control the remainder. In fact, $$\Vert (x-y)^\gamma \partial_x^\alpha\partial_y^\beta K_{G^r_J}(x,y,t) \Vert \leq e^{-c\,t}\,\text{ when }J>(k+1)m-\abs{\gamma}+\abs{\alpha}+\abs{\beta}+d$$ which follows, as in the previous lemma, from the estimate 
$$
\Vert (x-y)^\gamma \partial_x^\alpha\partial_y^\beta K_{Q_J^{(k+1)}}(x,y,\la) \Vert \leq \langle \la \rangle^{-1}.
$$
Thus, $G(t)$ is a smoothing pdo.

Now, let us choose a large $k\in \N$. By a successive integration of \eqref{eq: time derivative on diagonal} and using \eqref{eq: integral kernel on diagonal} with $j=d$, we get, for $j < (k-1)m+d$,
\begin{align*}
K_{G_{-j}}(x,x,t)= \left\{\begin{array}{ll} 
c'_{j,k}(x)\,t^{(j-d)/m}+p_{j,k}(x,t)+\OO_0(t^k), &\text{ for } j-d \notin m\Z,\\
c'_{j,k}(x)\,t^{(j-d)/m} \,\log t +p'_{j,k}(x,t)+\OO_0(t^k), &  \text{ for } j-d \in m\Z,
\end{array}\right.
\end{align*}
where $c'_{j,k}(x)$ depends only on $c_{j,k}(x)$ of Lemma \ref{lem: improved estimates}, $p_{j,k}(x,t)$ and $p'_{j,k}(x,t)$ are polynomials of degree $k$ in $t$ and are continuous in $x$ with $p_{j,k}(x,0)=p'_{j,k}(x,0)=0$. \\
Moreover, the remainder in \eqref{eq: Taylor of remainder} for $J=(k-1)m+d>((k-3)+1)m+d$ is $$ K_{G_{J}^r}(x,y,t)=\sum_{\ell=1}^{k-4} \tfrac{t^\ell}{\ell !} \, \partial_t^\ell K_{G_{-J}^r}(x,y,0)+\OO_0(t^{k-3}).$$
Thus, for the full integral kernel,
\begin{align*}
K_G(x,x,t)=\!\!\!\sum_{\substack{0\leq j<J=(k-1)m+d \\ j-d \notin m\N}} \hspace{-0.6cm}c'_{j,k}(x) \,t^{(j-d)/m} +\!\!\!\! \sum_{\substack{j=m\ell+d \\ 1\leq \ell<k-1}} c'_{j,k}(x) \,t^\ell\,\log t +p_k(x,t) +\OO_0(t^{k-3}),
\end{align*}
where the $p_k(x,t)$ are polynomials in $t$ such that  $p_k(x,0)=0$. Sending $k$ to infinity, we get, after a relabeling, the announced asymptotics, because $K_G(x,x,t)$ minus the sum of terms up to $(n-d)/m=N$ and $\ell=N$ is $\OO_0(t^{N+1/d})$. The coefficients $c'_{j,k}$ depend only on $G_{-j}$, thus locally on the symbols of $P$ of orders from $m$ to $m-j$. The $p_k$'s are not easy to characterise, but they are smooth in $x$: The smoothness of $c_{j,k}$ (and so of $c'_{j,k}$) is clear from its definition in Lemma \ref{lem: improved estimates}, while the smoothness of $p_{j,k}$ or $p'_{j,k}$ can be checked at each step of the above integrations in $t$ with $t=1$.~\hfill{$\Box$}
\end{proof}

By taking the trace and relabeling, we immediately get the celebrated expansion:
\begin{corollary}
\label{cor:heat for elliptic pdo}
Let $(M,g)$ be a compact Riemannian manifold of dimension $d$ and let $P\in \Psi^m(M,E)$ be an elliptic pdo with $m>0$, the principal symbol of which satisfies \eqref{hyp: principal symbol with parameter}. Then,
\begin{align*}
\hspace{1cm} \Tr e^{-t\,P} \tzero & \, \sum_{k=0}^\infty a_k(P) \,t^{(k-d)/m} + \sum_{\ell=0}^{\infty} b_\ell(P)  \,t^{\ell} \, \log t, \\
 \text{ with } &
\begin{cases} 
a_k(P) = \int_M \tr c_{k-d}(x,P) \sqrt{g} \, d^d x, & \text{for } k-d \notin m\N, \\
a_k(P) = \int_M \tr r_{(k-d)/m}(x,P) \sqrt{g} \, d^d x, & \text{for } k-d \in m\N, \\
b_\ell(P) = \int_M \tr c_{m\ell}(x,P) \sqrt{g} \, d^d x, & \text{for } \ell \in \N.  \quad \hspace{1.4cm} \Box
\end{cases}
\end{align*}
\end{corollary}
\bigskip

It can be of interest to recall here a few links between the resolvent, complex powers and heat operators (with $s\in \CC,\,k\in \N$, $c>0$)
\begin{align}
& e^{-t\,P}=t^{-k}\tfrac{i}{2\pi} \int_\C e^{-t\lambda}\,\partial_\lambda^k(P-\lambda)^{-1}\,d\lambda=\tfrac{1}{i2\pi} \int_{\Re (s)=c}t^{-s}\,\Gamma(s)\,P^{-s}\,ds, \label{e-tP with derivative}\\
& P^{-s}=\tfrac{1}{(s-1)\cdots (s-k)} \,\tfrac{i}{2\pi} \int_\C \lambda^{k-s}\, \partial^k_\lambda (P-\lambda)^{-1}\,d\lambda =\tfrac{1}{\Gamma(s)} \int_0^\infty t^{s-1}\,e^{-t\,P}\, dt. \nonumber
\end{align}

\section{Meromorphic Extensions of Certain Series and their Residues}
\label{app: residues of series}

We gather below some results on meromorphic extensions of certain series. These will allow us for an extension of Proposition \ref{sa average} and provide tools for the computation of the dimension spectrum of the noncommutative torus in Section \ref{subsec:NC torus dim sp}. On the way, a `Diophantine condition' will pop up guaranteeing a control on the commutation between the residues and the series. For complete proofs, see \cite{TorusSA}.

In the following, $\sum'$ \index{_a9tsigmaprim@$\sum'$}means that we omit the division by zero in the summand.

\begin{theorem}
\label{res-int} 
Let $P$ be a polynomial $P(x)=\sum_{j=0}^{d} P_j(x) \in \CC[x_1,\dotsc,x_n]$, where $P_j$ is the homogeneous part of $P$ of degree $j$ and $d$ is fixed.\\
Then, the function \begin{align*}
f(s)\vc{\sum}'_{k\in\Z^n} P(k) \norm{k}^{-s}
\end{align*}
has a meromorphic extension to $\CC$.\\
Moreover, $f$ is not entire if and only if 
\begin{align*}
\mathcal{P}_P:= \{j \,\,\vert\,\,\int_{u\in S^{n-1}} P_j(u)\, ds(u)\neq 0 \}\neq \varnothing
\end{align*}
and  $f$ has only simple poles at $j+n$, $j\in \mathcal{P}_{P}$, with
\begin{align*}
\underset{s=j+n}{\Res} \, f(s) = \int_{u\in S^{n-1}} P_j(u)\, ds(u).
\end{align*}
\end{theorem}
Here $ds$ is the Lebesgue measure on $S^{n-1}$. The proof is based on the fact that the function 
$$
{\sum}'_{k\in\Z^n} P(k) \norm{k}^{-s}-\int_{\RR^n\setminus B^{n}} P(x) \Vert x \Vert^{-s} \, dx,
$$
 with $B^n$ -- the unit ball in $\RR^n$, originally defined for $\Re(s)>d+n$, extends holomorphically to $\CC$.

This result can be seen as an extension of Proposition \ref{sa average}: Let $D$ be a selfadjoint invertible operator with only discrete spectrum equal to $\Z^n$ such that each eigenvalue $\pm k\in \Z^n$ has multiplicity $p(k)$ for a given polynomial $P\in\N[x_1,\dotsc,x_n]$. Then, $\underset{s=j+n}{\Res} \, \zeta_D(s) = \int_{u\in S^{n-1}} P_j(u)\, ds(u)$.

\begin{example}
Let us consider the functions 
$$
\zeta_{q_{1},\dotsc,q_{n}}(s)\vc{{\sum}'_{k\in\Z^{n}}} \,\, k_1^{q_{1}}\dotsb k_n^{q_{n}} \norm{k}^{-s}, \text{ for } q_{i}\in \N^*.
$$
By the symmetry $k\mapsto -k$, the functions $\zeta_{q_{1},\dots,q_{n}}$ vanish if any $q_{i}$ is odd.
\\
Assume that all of $q_{i}$'s are even, then $\zeta_{q_{1},\dotsc,q_{n}}(s)$ is a nonzero sum of terms $P(k)\Vert k\Vert ^{-s}$, where $P$ is a homogeneous polynomial of degree $q_{1}+\dotsb +q_{n}$. Theorem \ref{res-int} yields the following:  $\zeta_{p_{1},\dotsc,p_{n}}$ has a meromorphic extension to $\CC$ with a unique pole at $n+q_{1} +\dotsb +q_{n}$. This pole is simple and the residue at this pole is
\begin{align}
\label{eq:zetaq1qn}
\hspace*{1.2cm} \underset{s=n+q_{1} +\dotsb +q_{n}}{\Res} \,\zeta_{q_{1},\dotsc,q_{n}}(s)= 2 \, \tfrac{\Gamma[(q_{1}+1)/2] \dotsb \Gamma[(q_{n}+1)/2]} {\Gamma [(n+q_{1}+ \dotsb + q_{n})/2]}. \hspace{1.7cm} \blacksquare
\end{align}
\end{example}

We now recall few notions from the Diophantine approximation theory.

\begin{definition}
\label{ba}
i) Let $\delta >0$. A vector $a \in \RR^n$ is said to be {\it $\delta$-badly approximable} if the Diophantine condition\index{Diophantine condition} holds true:
\\
\centerline{\it There exists $c >0$ such that $|q . a -m| \geq c \,|q|^{-\delta}$, $\forall q \in \Z^n \setminus \set{0}$ and $\forall m \in \Z$.}
Let ${\cal BV}_{\!\!\delta}$ be the set of $\delta$-badly approximable vectors and ${\cal BV} :=\cup_{\delta >0} \,{\cal BV}_{\!\!\delta}$ the set of badly approximable vectors.\par
ii) A matrix $\Th \in {\cal M}_{n}(\RR)$ (real $n \times n$ matrices)\index{_zMn1@${\cal M}_{n}(\RR)$} will be called {\it badly approximable}\index{badly approximable} if there exists $u \in \Z^n$ such that  ${}^t\Th (u)$ is a badly approximable vector of $\RR^n$.
\end{definition}

It is known that for $\delta >n$ the Lebesgue measure of $\RR^n \setminus {\cal BV}_{\!\!\delta}$ is zero (i.e almost any element of $\RR^n$ is $\delta$-badly approximable) and, consequently, almost any matrix in ${\cal M}_n(\RR)$ is badly approximable.

\smallskip

We store below a rather technical result \cite[Theorem 2.6]{TorusSA}, omitting the proof. As compared with the previous theorem, it takes care of the possible oscillations $e^{i\,2\pi k.a}$, where $a$ is a vector in $\RR^n$, which will be allowed to vary later on.
\begin{theorem}
\label{analytic}
Let $P\in \CC[x_1,\cdots,x_n]$ be a homogeneous polynomial of degree $d$ and let $b\in \Scl((\Z^{n})^q)$. Then,

i)
For $a \in \RR^n$, define $$f_a(s):={\sum}'_{k\in \Z^n} P(k) \norm{k}^{-s}\,e^{i\,2\pi  k.a}.$$

\quad 1.
If $a\in \Z^n$, then $f_a$ has a meromorphic extension to the whole space $\CC$.\\
Moreover, 
$$
\text{$f_a$ is not entire if and only if $\int_{u\in S^{n-1}} P(u)\, ds(u)\neq 0.$}
$$
In that case, $f_a$ has  a single simple pole at the point $d+n$, with $$\underset{s=d+n}{\Res} \, f_a(s) = \int_{u\in  S^{n-1}} P(u)\, ds(u).$$

\quad 2.
If $a\in \RR^n\setminus \Z^n$, then $f_a(s)$ extends holomorphically to $\CC$.

ii) Suppose that  $\Th \in {\cal M}_{n}(\RR)$ is badly approximable. For any $(\eps_i)_i\in \{-1,0,1\}^{q}$, the function
$$
g(s):={\sum}_{l\in (\Z^n)^{q}} \, b(l) \,f_{\Th\,
\sum_i \eps_i l_i}(s)
$$
extends meromorphically to $\CC$ with only one possible pole at $s= d+n$.\par Moreover, if we set $${\cal Z}:=\{l\in(\Z^n)^{q} \,\, \vert \,\, \sum_{i=1}^q \eps_i l_i= 0\}\,\text{ and }\,V:=\sum_{l\in {\cal Z}} \, b(l),
$$
then

\quad 1. If\, $V\int_{S^{n-1}} P(u)\, ds(u)\neq 0$, then $s=d+n$ is a simple pole of $g(s)$ and
$$
\underset{s=d+n}{\Res} \, g(s) = V\,\int_{u\in S^{n-1}} P(u)\, ds(u).
$$

\quad 2. If\, $V\int_{S^{n-1}} P(u)\, ds(u)=0$, then $g(s)$ extends holomorphically to $\CC$.

iii) Suppose that $\Th \in \mathcal{M}_n(\RR)$ is badly approximable. For any $(\eps_i)_i\in \{-1,0,1\}^{q}$, the function
$$
g_0(s):={\sum}_{l\in (\Z^n)^{q}\setminus
  {\cal Z}} \,\,
b(l)\,f_{\Th\, \sum_{i=1}^q \eps_i\, l_i}(s),
$$
with $${\cal Z}:=\{l\in(\Z^n)^{q} \,\,\vert\,\, \sum_{i=1}^q \eps_i l_i= 0\},$$ extends holomorphically to $\CC$.
\end{theorem}

Is is unknown whether the Diophantine condition, which is sufficient to get the results of $iii)$, is also necessary --- see nevertheless \cite[Remark 2.9]{TorusSA}.

\bigskip
In the study of the dimension spectrum of the noncommutative torus we will need the following result \cite[Theorem 2.18 (i)]{TorusSA}.\\
It requires some notations: Fix $q\in \N$, $q\geq 2$ and $r=(r_1,\cdots,r_{q-1})\in (\N^*)^{q-1}$. \\
When $(x_1,\dots,x_{2q})$, we set $$\wt x_j:= x_1+\cdots+ x_j+  x_{q+1}+\cdots + x_{q+j}  \,\,\text{ for any }1\leq j\leq q
$$
and we let $P\in \RR[x_1,\cdots,x_n]$ and $d= \deg P$. 

\begin{theorem}
\label{zetageneral}
Let $\tfrac{1}{2\pi}\Th$ be a badly approximable matrix, and $a \in \Scl ((\Z^{n})^{2q})$. Then,
$$
s\mapsto f(s):= \sum_{l\in [(\Z^n)^{q}]^2} a_{l}\ {{\sum_{k\in \Z^n}}\!\!'} \,\, \prod_{i=1}^{q-1}|k+\wt l_i|^{r_i} \,\norm{k}^{-s}\, P(k)\, e^{ik.\Th \sum_1^{q} l_j}
$$
has a meromorphic extension to $\CC$ with at most simple possible poles at the points $s=n+d+\abs{r_1}+\cdots+\abs{r_{q-1}}-m$ where $m\in \N$.
\end{theorem}
An explicit formula for the residues of $f$ is given in \cite[Theorem 2.18 (ii)]{TorusSA}.

\chapter{Examples of Spectral Triples}
\label{chap:appendix}

\section{Spheres}
\label{sec:spheres}

A particularly illustrative example of a commutative spectral triple (recall Example~\ref{ex:commutative}) is provided by the $d$-dimensional unit spheres\index{sphere} $S^d$\index{_zSdim@$S^d$}.

On $S^1$ there are two possible spin structures, where the nontrivial one is associated to functions with antiperiodic boundary conditions. When $d \geq 2$ there is only one spin structure available since $S^d$ is simply connected.\\  Let us equip $S^d$ with the standard round metric and cook up the standard Dirac operator $\Dslash$ acting on the chosen spinor bundle $\SS$. Then, 
$(C^\infty(S^d), L^2(S^d,\SS), \Dslash)$ is a $d$-dimensional regular spectral triple with a simple dimension spectrum $d - \N$ (cf. Example \ref{ex:dimsp_manifold}), for any $d \geq 1$ and any spinor bundle $\SS$.

The spectrum of $\Dslash$ turns out to be very simple \cite{BarSpin,Ginoux}: For the trivial spin structure on $S^1$ we have $\lambda_n(\Dslash) = n$ for $n \in \Z$ and all of the eigenspaces are one-dimensional. In particular, we have $\dim \ker \Dslash = 1$. In the non-trivial case, the spectrum of the Dirac operator agrees with the general pattern for $S^d$ and for $d\geq1$:
\begin{align}
\label{eigenvalues_Sd}
\lambda_n(\Dslash)=\sign(n) (n+\tfrac d2), \qquad M_n(\Dslash)=2^{\floor{\tfrac d2}} \tbinom{\abs{n}+d-1}{d-1},\,\, \text{ with }n\in \Z.
\end{align}
Hence, 
$$\mu_n(\Dslash)= n + \tfrac d2 \,\text{ with }\,M_n(\abs{\Dslash})=2^{\floor{\tfrac d2}+1} \tbinom{n+d-1}{d-1},\, \text{ with }n\in \N.$$

\section{Tori}
\label{sec:tori}

Another commutative spectral triple is given by the flat tori \index{torus} $\TT^d=\RR^d/\Z^d$\index{_zTd@$\TT^d$} and, as above, $(C^\infty(\TT^d), L^2(\TT^d,\SS), \Dslash)$ is a $d$-dimensional regular spectral triple with a simple dimension spectrum $d - \N$. There are $2^d$ different spin structures on $\TT^d$ classified by the twisting of each coordinates of the lattice $\Z^d$: Given a basis $e_1,\cdots,e_d$ of $\Z^d$, this is realised by choosing $s_1,\cdots,s_d\in \{0,1\}$, so that we have the group homomorphism: $e_k\in\Z^d \to(-1)^{s_k}\in \{+1,-1\}$ and all spin structures are given this way by $(s_1,\cdots,s_d)$. \\
The eigenvalues of the Dirac operator $\Dslash$ (endowed with the induced flat metric) depend on the chosen spin structures and are given by (see \cite{BarSpin, Ginoux})
\begin{align}
\label{eq:eigenvalues tori}
\{ \pm2\pi \,\Vert k+ \tfrac12\sum_{j=1}^d s_j\,e_j^* \Vert\,\,\,\vert\,\,\, k\in {\Z^d}^* \}.
\end{align}
where ${\Z^d}^*$ is the dual lattice and $(e_1^*,\cdots,e_d^*)$ is the dual basis of $(e_1,\cdots,e_d)$.\\
The multiplicity of the eigenvalue $0$  (given by $s_j=0,\,\forall \,j$) is $2^{\floor{d/2}}$ while the non-zero eigenvalues have multiplicity $2^{\floor{d/2}-1}$.

\section{Noncommutative Tori}
\label{sec:torus}

The noncommutative $d$-tori \index{noncommutative torus} were introduced by Rieffel \cite{Rieffel} and Connes  \cite{ConnesNCT} as deformations of $\TT^d$ characterised a by non-zero skew-symmetric matrix $\Theta \in M_d(\RR)$. \index{_a4theta@$\Theta$}

Denote by $C^\infty(\TT^d_\Theta)$ \index{_zCminf@$C^\infty(\TT^d_\Theta)$} the algebra generated by $d$ unitaries $u_i$, $i=1,\dots,d$ satisfying
\begin{equation}
\label{rel}
u_\ell\,u_j=e^{i\,\Theta_{\ell j}}\,u_j\,u_\ell,
\end{equation}
and with Schwartz coefficients. So,  $a\in C^\infty(\TT^d_\Theta)$ can be written as
$a=\sum_{k\in\Z^d}a_k\,U_k$, where $\{a_k\}\in \Scl(\Z^d)$ (i.e. $\sup_{k\in \Z^d}\vert k_1\vert^{n_1}\cdots \vert k_d\vert^{n_d}\abs{a_k}<\infty,\,\forall n_i\in \N$) and
$$
U_k\vc e^{-\frac i2 k.\Theta' k}\,\,u_1^{k_1}\cdots u_n^{k_d}, \quad k\in\Z^d,
$$ 
where $\Theta'$ is the restriction of $\Theta$ to its upper triangular part. \\
Relation \eqref{rel} reads
\begin{equation}
\label{rel1}
U_{k}U_{q}=e^{-\frac i2 k.\Theta q} \,U_{k+q}\,\, \text{ or }\,\, U_{k}U_{q}=e^{-i k.\Theta q} \,U_{q}U_{k}\,.
\end{equation}
Thus, the unitary operators $U_{k}$ satisfy 
$$U_{k}^*=U_{-k} \text{ and }[U_{k},U_{l}]=-2i\,\sin(\tfrac 12 k.\Theta l)\,U_{k+l}.$$

Let $\tau$ be the trace on $C^\infty(\TT^d_\Theta)$ defined by 
$$
\tau\big( \,{\sum}_{k\in\Z^d}\, \,a_k\,U_k \big)\vc a_0
$$
and $\H_{\tau}$ \index{_zHHH@$\H_{\tau}$} be the GNS Hilbert space obtained by completion of $C^\infty(\TT^d_\Theta)$
with respect to the norm induced by the scalar product $\langle a,b\rangle\vc \tau(a^*b)$.

On $\H_{\tau}=\set{\sum_{k\in\Z^d}a_k\,U_k \,\, \vert \, \, \{a_{k}\}_{k} \in \ell^2(\Z^d) }$, let $\delta_\mu$, for $\mu\in \set{1,\dots,d}$, be the pairwise commuting canonical derivations, given by
\begin{equation}
\label{dUk}
\delta_\mu(U_k)\vc ik_\mu U_k\,.
\end{equation}
Define now
$$
\A_{\Theta}\vc C^\infty(\TT^d_\Theta) \text{ acting on } \H \vc \H_{\tau} \otimes \CC^{2^m},
\text{ with }m\vc\lfloor d/2\rfloor. \index{_zAtheta@$\A_{\Theta}$}
$$

Each element of $\A_{\Theta}$ is represented on $\H$ as $\pi(a) \vc L(a)\otimes1_{2^m}$ where $L(\cdot)$ (and $R(\cdot)$) is the left (right) regular representation of $C^\infty(\TT^d_\Theta)$. The Tomita conjugation $$
J_{0}(a)\vc a^*$$ satisfies $J_0L(a)=R(a^*)J_0$ and $[J_{0},\delta_{\mu}]=0$ and we define $$J\vc J_{0}\otimes C_{0},$$ where $C_{0}$ is an operator on $\CC^{2^m}$ such that $C_{0}^2=\pm 1_{2^m}$, depending on the parity of $m$. The (flat) Dirac-like operator is given by
\begin{align}
\label{defDirac}
\DD\vc -i\,\delta_{\mu}\otimes \gamma^{\mu},
\end{align}
with hermitian Dirac matrices $\gamma$ satisfying $C_{0}\ga^{\alpha}=\pm \ga^\alpha C_{0}$ (see \cite{ConnesMarcolli,Elements} for details about the signs). The operator $\DD$ is defined and symmetric on the dense subset of $\H$ 
given by $C^\infty(\TT^d_\Theta) \otimes \CC^{2^{m}}$ and we still denote by $\DD$ its selfadjoint extension. Thus
$$
\DD\ U_k \otimes e_i = k_\mu U_k \otimes \gamma^\mu e_i ,
$$
where $(e_i)$ is the canonical basis of $\CC^{2^m}$. \\
Finally, in the even case, the chirality operator reads: $\gamma \vc \id \otimes (-i)^{m} \gamma^1 \cdots \gamma^{d}$.\\
The operator $\DD$ is not invertible: $$\ker \DD=U_0\otimes \CC^{2^m}$$ and this kernel has dimension $2^m$ because if $\psi=\sum_{k,j} c_{k,j}U_k\otimes e_j$, then we have $0=\DD^2\psi=\sum_{k,j}c_{k,j} \norm{k}^2 U_k\otimes e_j$. Thus, 
\begin{align}
\label{eq:Pzero}
P_0=\vert U_0\rangle \langle U_0 \vert \otimes \bbbone_{\CC^{2^m}}.
\end{align}

This yields a spectral triple:
\begin{theorem}
\label{thm:nc torus}
The tuple $(\A_{\Theta},\H,\DD,J,\gamma)$ is a real regular spectral triple of dimension $d$. Its $KO$-dimension is also $d$.
\end{theorem}
Most of the arguments will be revisited in the computation of the dimension spectrum ---  see Theorem \ref{zeta(0)}. For a complete proof see \cite{ConnesMarcolli,Elements}.

We remark that the torus actions on $C^*$-algebras lead to interesting nonunital spectral triples , see \cite[Chapter 5]{CGRSIndex}.

\subsection{Dimension Spectrum}\label{subsec:NC torus dim sp}

\begin{theorem}
\label{zeta(0)}

i) If $\tfrac{1}{2\pi} \Th$ is badly approximable, the dimension spectrum of the triple $\big(\Coo(\TT^d_\Th),\H,\DD\big)$ is equal to the set $\set{d-k \, :\,  k\in \N}$ and all of the poles are simple.

ii) $\zeta_D(0)=0.$
\end{theorem}

\begin{proof}
$i)$ Let $B\in \Pc(\A)$ and $p\in \N$. Suppose that $B$ is of the form
$$
B= a_r b_r \DD^{q_{r-1}}|\DD|^{p_{r-1}} a_{r-1}b_{r-1}\cdots \DD^{q_1}|\DD|^{p_1} a_1 b_1,
$$
where $r\in \N$, $a_i \in \A$, $b_i\in J\A J^{-1}$, $q_i, p_i \in \N$.\\
We decompose $a_i=:\sum_{\ell\in \Z^d} a_{i,\ell}\,U_\ell$ and $b_i=:\sum_{\ell'\in \Z^d} b_{i,\ell'} \,U_{\ell'}$. \\With the shorthands $k_{\mu_1,\mu_{q_i}}\vc k_{\mu_1}\cdots k_{\mu_{q_i}}$ and $\ga^{\mu_1,\mu_{q_i}}\vc \ga^{\mu_1}\cdots \ga^{\mu_{q_i}}$, we get
\begin{align*}
\DD^{q_1}|\DD|^{p_1}  a_1 b_1 \, U_k \otimes e_j \\
=  \sum_{\ell_1,\,\ell'_1} a_{1,\ell_1} &b_{1,\ell'_1} U_{\ell_1}U_k U_{\ell'_1} \,|k+\ell_1+\ell'_1|^{p_1}\,(k+\ell_1+\ell'_1)_{\mu_1,\mu_{q_1}}\otimes \ga^{\mu_1,\mu_{q_1}} e_j,
\end{align*}
which gives, after $r$ iterations,
\begin{align*}
B (U_k \otimes e_j) = \sum_{\ell,\ell'\in \Z^d} &\wt a_{\ell} \wt b_{\ell}\, U_{\ell_r}\cdots U_{\ell_1} U_k U_{\ell'_1}\cdots U_{\ell'_r} \,\prod_{i=1}^{r-1} |k+\wh \ell_i+\wh \ell'_i|^{p_i}(k+\wh \ell_{i} +\wh \ell'_{i})_{\mu^{i}_1,\mu^i_{q_i}}\\
& \hspace{1cm}\otimes \ga^{\mu^{r-1}_1,\mu^{r-1}_{q_{r-1}}}\cdots \ga^{\mu^1_1,\mu^1_{q_1}} e_j,\\
\text{where} \quad&\wt a_\ell : = a_{1,\ell_1}\cdots a_{r,\ell_r},\qquad \wt b_{\ell'} : = b_{1,\ell'_1}\cdots b_{r,\ell'_r},\\
&\wh \ell_i \vc \ell_1+\cdots+\ell_i, \qquad \ga^\mu \vc \ga^{\mu^{r-1}_1,\mu^{r-1}_{q_{r-1}}}\cdots \ga^{\mu^1_1,\mu^1_{q_1}}.
\end{align*}
Let us denote $F_\mu(k,\ell,\ell')\vc \prod_{i=1}^{r-1}|k+\wh \ell_i+\wh \ell'_i|^{p_i}\, (k+\wh \ell_{i} +\wh \ell'_{i})_{\mu^{i}_1,\mu^i_{q_i}}$ . \\With the shortcut $\sim_c$ meaning equality modulo a constant function in the variable $s$, we have
$$
\Tr \big(B|D|^{-p-s}\big) \sim_c \sum _{k\in\Z^d}{\!\!\!'} \, \sum_{\ell,\ell'\in \Z^d}\, \wt a_\ell \wt b_{\ell'} \, \tau\big(U_{-k}U_{\ell_r}\cdots U_{\ell_1} U_k U_{\ell'_1}\cdots U_{\ell'_r}\big) \tfrac{F_\mu(k,\ell,\ell')}{\norm{k}^{s+p}} \tr (\ga^\mu)\, .
$$
Since $U_{\ell_r}\cdots U_{\ell_1} U_k = U_k U_{\ell_r}\cdots U_{\ell_1}\, e^{-i\sum_{i=1}^r \ell_i .\Th k}$, we get
$$
\tau\big(U_{-k}U_{\ell_r}\cdots U_{\ell_1} U_k U_{\ell'_1}\cdots U_{\ell'_r}\big)=\delta_{\sum_{i=1}^r \ell_i+\ell'_i,0} \, e^{i\,\phi(\ell,\ell')} \, e^{-i\,\sum_{i=1}^r \ell_i.\Th k},
$$
where $\phi$ is a real valued function. Thus,
\begin{align*}
\Tr \big(B |D|^{-p-s} \big)&\sim_c  \sum _{k\in\Z^d}{\!\!\!'} \, \sum_{\ell,\ell'\in \Z^d}\, e^{i\phi(\ell,\ell')}\,\delta_{\sum_{i=1}^r \ell_i+\ell'_i,0}\,\, \wt a_\ell \wt b_{\ell'}\, \tfrac{F_\mu(k,\ell,\ell')\,e^{-i\sum_{i=1}^r \ell_i.\Th k}}{\norm{k}^{s+p}} \,\tr(\ga^\mu) \\
&\sim_c f_\mu(s)\tr (\ga^\mu).
\end{align*}
The function $f_\mu(s)$ can be decomposed as a linear combination of zeta functions of the type described in Theorem \ref{zetageneral} (or, if $r=1$ or if all  $p_i$'s are zero, in Theorem \ref{analytic}). Thus, by linearity, $s\mapsto \Tr \big(B |D|^{-p-s}\big)$ has a meromorphic extension to $\CC$ with simple poles located exclusively in $\Z \subset \CC$. \\
Moreover, if $B\in \Psi^0(\A)$ and $q\in \N$ is such that $q>d$, then $B\abs{D}^{-s} \in \OP^{-\Re (s)}$, so it is trace-class around $q$ and hence $\zeta_{B,D}(s)=\Tr \,B\abs{D}^{-s}$ is regular around $q$.

$ii)$ Let 
$$
Z_d(s)\vc\sum _{k\in\Z^d}{\!\!\!'} \,\,\norm{k}^{-s}$$
\index{_zzd@$Z_d$}be the Epstein zeta function associated to the quadratic form $q(x):=x_1^2+...+x_d^2$. Then $Z_d$ enjoys the functional equation: 
\begin{align}
\label{eq:functional eq}
Z_d(s)= \pi^{s-d/2} \Gamma (d/2 -s/2)\Gamma (s/2)^{-1}\,Z_d(d-s).
\end{align}
Since $\pi^{s-d/2} \Gamma (d/2 -s/2) \,\Gamma (s/2)^{-1}=0$ for any negative even integer $d$ and $Z_d(s)$ is meromorphic on $\CC$ with only one pole at $s=d$ (with residue $2 \pi^{d/2} \Gamma (d/2)^{-1}$ according to \eqref{eq:zetaq1qn}), we get $Z_d(0)=-1$.\\ By definition, $$\zeta_D(s)=\sum_{k\in \Z^d}\, \sum_{j=1}^{2^m} \langle U_k\otimes e_j, |D|^{-s}U_k\otimes e_{j}\rangle,$$ so that 
\begin{align}
\label{eq:zeta for nct}
\zeta_D(s)= 2^{\lfloor d/2 \rfloor}(\,Z_d(s)+1)
\end{align}
and the conclusion $\zeta_D(0)=0$ follows.
\hfill{$\Box$}
\end{proof}

\subsection{Heat Kernel Expansion}

\begin{proposition}
\label{eq:heat trace asymptotics for torus}
The heat trace asymptotics is $$\Tr e^{-t\,\DD^2} \underset{t\downarrow 0}{\sim} 2^{\lfloor d/2 \rfloor}\pi^{d/2}\,t^{-d/2}.$$
\end{proposition}

\begin{proof}
Since $\DD^2\, U_k \otimes e_i = \norm{k}^2 U_k \otimes e_i$, we know by Formula \eqref{eq: gaussian series} that 
\begin{align*}
\Tr e^{-t\,\DD^2} = 2^m\sum_{k\in \Z^d} e^{-t\,\norm{k}^2}\underset{t\downarrow 0}{\sim} 2^{\lfloor d/2 \rfloor}\pi^{d/2}\,t^{-d/2}.
\tag*{$\Box$}
\end{align*}
\end{proof}

This result can also be obtained (in an admittedly circuitous way) from Theorem~\ref{thm:zeta2hk_finite}: We have $\ZZ_{D^2} (s)=2^m\Gamma(s) (Z_d(2s)+1)$ and the Epstein zeta function $Z_d$ is meromorphic on $\CC$ with a single simple pole at $s=d$. Moreover, the Epstein zeta function enjoys a polynomial growth on the verticals --- see the non trivial estimates demonstrated in \cite{Essouabri,Mahler}.

Remark that we recover the classical result: For the torus $\TT^d$ with the usual scalar Laplacian  $\Delta=-g^{\mu\nu} \partial_\mu \partial_\nu$, $$\Tr e^{-t\,\Delta}=\Vol (\TT^d)(4\pi)^{-d/2}\,t^{-d/2} +\OO_0(t^{-d/2} e^{-1/4t})$$ 
and $\Vol(\TT^d)=(2\pi)^d$ so that $\Vol (\TT^d)(4\pi)^{-d/2}=\pi^{d/2}= a_{d,0}$ (mod $2^m$).

\subsection{Spectral Action for Noncommutative Tori}

We consider the $d$-dimensional noncommutative torus $(\A_{\Theta},\H,\DD)$ of Theorem \ref{thm:nc torus} with $\A_{\Theta}=C^\infty(\TT^d_\Theta)$, a one-form $A=A^*\in \Omega^1_\DD(\A_{\Theta})$ and 
\begin{align}
\label{def: A symmetric}
\Ag\vc A+\epsilon JAJ^*.
\end{align}
Thus, the constraint $\Ag=\Ag^* \in \Psi^0(\A)$ is satisfied.  \\
Remark that $A$ can be written as $$A\cv L(-iA_\alpha)\otimes \gamma^\alpha\,\text{ where }\,A_\alpha=-A_\alpha^*\in \A_\Theta,$$ so that
\begin{align}
\label{def:DA for nct}
\DD_\Ag=-i[\delta_\alpha +L(A_\alpha)-R(A_\alpha)]\otimes \gamma^\alpha.
\end{align}
As for the commutative torus $\TT^d$, we get
$$
\DD_{\Ag}^2=-\delta^{{\alpha}_{1} {\alpha}_{2}}(\delta_{{\alpha}_{1}}+\Ag_{{\alpha}_{1}})(\delta_{{\alpha}_{2}}
+\Ag_{{\alpha}_{2}})\otimes 1_{2^m} - \tfrac 12 \Omega_{{\alpha}_{1} {\alpha}_{2}} \otimes \gamma^{{\alpha}_{1}{\alpha}_{2}},
$$
where $\delta^{\alpha\beta}$ is the Kronecker symbol and
\begin{align}
&\Ag_\alpha \vc L(A_\alpha)-R(A_\alpha), \quad \gamma^{{\alpha}_{1} {\alpha}_{2}}\vc \tfrac 12(\gamma^{{\alpha}_{1}}\gamma^{{\alpha}_{2}} -\gamma^{{\alpha}_{2}}\gamma^{{\alpha}_{1}}) , \notag\\
&\Omega_{{\alpha}_{1} {\alpha}_{2}}\vc [\delta_{{\alpha}_{1}}+ \Ag_{{\alpha}_{1}},\delta_{{\alpha}_{2}} +\Ag_{{\alpha}_{2}}]\,
=L(F_{{\alpha}_{1} {\alpha}_{2}}) - R(F_{{\alpha}_{1} {\alpha}_{2}}),\notag\\
&F_{{\alpha}_{1} {\alpha}_{2}}\vc \delta_{{\alpha}_{1}}(A_{{\alpha}_{2}})
-\delta_{{\alpha}_{2}}(A_{{\alpha}_{1}})+[A_{{\alpha}_{1}},A_{{\alpha}_{2}}] \in \A_\Theta. \label{Fmunu}
\end{align}
Thus,
\begin{align}
\label{D2}
\DD_{\Ag}^2=-\delta^{{\alpha}_{1} {\alpha}_{2}} [ \delta_{{\alpha}_{1}}+L(A_{{\alpha}_{1}})-R(A_{{\alpha}_{1}})] \, [\delta_{{\alpha}_{2}}+L(A_{{\alpha}_{2}})-R(A_{{\alpha}_{2}})]
\otimes 1_{2^m} \nonumber\\
-\tfrac 12\,\big(L(F_{{\alpha}_{1} {\alpha}_{2}}) - R(F_{{\alpha}_{1} {\alpha}_{2}})\big)
\otimes \gamma^{{\alpha}_{1} {\alpha}_{2}}.
\end{align}

We now prove the existence of the asymptotics of the fluctuated heat trace on the noncommutative torus using the one for the `bare' one given in Proposition \ref{eq:heat trace asymptotics for torus} (cf. also Problem \ref{prob:heat vs zeta} \ref{Existence of fluctuated expansion} in Chapter \ref{chap:open}). To this end, we employ the pseudodifferential calculus introduced by Connes for $C^*$-dynamical system $(A,\RR^d,\alpha)$ (see \cite{ConnesNCT, ConnesTretkoff, ConnesModular}) and follow the arguments given in the proof of \cite[Theorem 4.2]{LeschMoscovici}. The idea is essentially to mimic the classical pdo calculus on a manifold --- cf. Appendix \ref{sec: Complex parameter-dependent symbols} --- improving the Proposition \ref{prop:seriesintegral} to gain control on the series defined by $\Tr e^{-t\DD_\Ag^2}$.

We first quickly summarise this symbolic calculus --- see \cite{ConnesNCT, ConnesModular,ConnesTretkoff} (and especially the complete approach by Lesch and Moscovici \cite{LeschMoscovici}) for details. 
\\
Let $A_\theta$ be the universal $C^*$-algebra generated by the $U_k$ and let $\A_\theta$ consist of those elements in $A_\theta$ for which $a\to \alpha_s(a)$ is $C^\infty$ for each $s\in \RR^d$, with the definition $$\alpha_s(U_k)\vc e^{-i2\pi\, s.k} U_k.$$ 
A smooth map $\rho:\,C^\infty(\RR^d)\to \A_\theta$ is named a symbol of order $m\in\Z$ 
\begin{align*}
&\text{if for any }k,\ell \in \N^d,\,\Vert \delta^k\partial_\xi^\ell\,\rho(\xi) \Vert \leq c_{k,\ell}(1+\Vert \xi\Vert)^{m-\abs{\ell}} \text{ for some constants } c_{k,\ell}\\
&\text{and if there exists } \sigma\in C^\infty(\RR^d ,A_\theta^\infty) \text{ such that }\lim_{\la\to \infty} \la^{-m}\rho(\lambda \xi)=\sigma(\xi),
\end{align*}
where 
\begin{align*}
& \delta^k \vc \delta^{k_1}_1\cdots \delta^{k_d}_d\,\,\text{ for } k\in \N^d, \qquad \partial^\ell_\xi \vc \partial^{\ell_1}_{\xi_1}\cdots\partial^{\ell_d}_{\xi_d},
\end{align*}
Such a symbol is elliptic when $\rho(\xi)^{-1}$ exists and the estimate $$\Vert \rho(\xi)^{-1} \Vert \leq c(1+\norm{\xi})^{-m}$$ holds for $\norm{\xi}$ large enough.

Given a symbol $\rho$, let us define the pdo $P_\rho:\,\Scl(\RR^d,\A_\theta) \to \Scl(\RR^d,\A_\theta)$ by
\begin{align*}
P_\rho (u) &\vc \int\alpha_{-x}(\F^{-1}[\rho](x-y))\,u(y)\,dy =\int \alpha_{-x}(\F^{-1}[\rho](y)) \,u(x-y)\,dy.
\end{align*}
An action of the pre-$C^*$-algebra $\A_\theta \rtimes _\alpha \RR^d$ on the pre-$C^*$-module $\Scl(\RR^d,\A_\theta)$ is given, for $a\in \A_\theta$, by $a\,u(x)=\alpha_{-x}(a) \,u(x)$ and $U_y(u)(x)\vc u(x-y)$. 

Then, $$P_\rho =\int \F^{-1}[\rho](y) \,U_y \,dy$$ can be seen as an element of the multiplier space of $A_\theta \rtimes _\alpha \RR^d$. The GNS representation of $A_\theta$ given by the trace $\tau$ can be extended to $A_\theta \rtimes _\alpha \RR^d$, with the maps $a\in A_\theta\to a$ and $U_x \to \alpha_x$, and $P_\rho=\int \F^{-1}[\rho](y) \,\alpha_y \,dy$ can be seen also as an element of the multiplier of $A_\theta$. Since $$P_\rho\,U_k=\int \F^{-1}[\rho](y) e^{-i2\pi\, y.k}  \,dy \,U_k= \rho(k)\,U_k,$$
we get, for any element $a=\sum_{k\in \Z^d} a_k\,U_k$, $$P_\rho(a)=\sum_{k\in \Z^d} a_k \,\rho(k)\,U_k.$$ 
Thus, $$\Tr P_\rho=\sum_{k\in \Z^d} \langle U_k, P_\rho U_k \rangle = \sum_{k\in \Z^d} \tau (U_k^* \rho(k)U_k)=\sum_{k\in \Z^d} \tau [\rho(k)]$$ is finite if $m<-d$, because $\Vert \rho(k)\Vert \leq c_{0,0} (1+(\sum_{i=1}^d k_i^2)^{1/2})^m$.

A parametric symbol $\rho(\xi,\la)$ with $\la \in V$ -- a region in $\CC$ is defined similarly: 
\begin{align}
\label{eq:def parameter symbol}
\Vert \delta^k \partial^\ell_\xi \partial_\lambda^r \rho(\xi,\la)\Vert \leq c_{k,\ell}(1+\norm{\xi}+\abs{\la})^{m-\abs{\ell}-\abs{r}}.
\end{align}
Cf. \cite[Section 1.7.1]{Gilkey1} and the (slightly different) Hypothesis \ref{hyp: principal symbol with parameter} of Appendix \ref{sec: Complex parameter-dependent symbols}.\\
We now adapt the Proposition \ref{prop:seriesintegral}:

\begin{proposition}
\label{prop:discrete parameter symbol}
If $\rho\in S^m(\RR^d\times V,A_\theta^\infty)$ is a parametric symbol of order $m< -d$, then $$\sum_{k\in \Z^d} \rho(k,\la)=\int_{\RR^d} \rho(\xi,\la)\,d\xi +\OO_\infty(\abs{\la}^{-\infty}).$$
\end{proposition}

\begin{proof}
Using \eqref{Poisson}, $$\sum_{k\in \Z^d} \rho(k,\la)= \int_{\RR^d} \rho(\xi,\la)\, d\xi+\sum_{k\in \Z^d \backslash \{0\}} \F[\rho(\cdot,\la)](k).
$$
Via the Fourier transform of a derivative and \eqref{eq:def parameter symbol}, we get, for each $N\in \N^*$, $\Vert \F[\rho(\cdot,\la)]\Vert \leq \norm{\xi}^{-N} \abs{\la}^{d+m-N}$ and the conclusion follows from: For any $q \in \N^d$,
\begin{align*}
\big \Vert \sum_{k\in \Z^d \backslash \{0\}} \F[\rho(\cdot,\la)](k) \big\Vert_q \leq c_{q,N}\, \big(\sum_{k\in \Z^d \backslash \{0\}} \!\!\norm{k}^{-N}\big)\, \abs{\la}^{d+m-N}.
\tag*{$\Box$}
\end{align*}
\end{proof}

\begin{lemma}
\label{lem:existence of asymp for nct}
For $\DD_\Ag$ as in  \eqref{def:DA for nct}, we have $$\Tr e^{-t\DD_\Ag^2} \,\tzero\,
\sum_{k=0}^\infty a_k(\DD_\Ag)\, t^{(k-d)/2}.$$
Moreover, for $f (\sqrt{\cdot}\,) \in \Cl_0^{r}$  with $r>d/2$, the spectral action has the asymptotics
\begin{align}
\label{eq:sa for torus}
\hspace{-0.2cm}S(\DD_\Ag,f,\Lambda)=\sum_{k=1}^d\!\Lambda^k (\!\int_0^\infty \!\!\!f(t)\,t^{k-1}\,dt)\,\ncint \vert D_\Ag \vert^{-k}+f(0)\,\zeta_{D_\Ag}(0) +\OO_\infty(\Lambda^{-1}).
\end{align}
\end{lemma}

\begin{proof}
The operator $\DD_\Ag$ given in \eqref{def:DA for nct} can be seen as a differential multiplier of order 1. Namely, $\DD_\Ag \in \text{Diff}^1_\sigma(\RR^d, \B_\theta\otimes \CC^{2^m})$, where $\B_\theta$ is the algebra generated by $\A_\theta$ and $J \A_\theta J^{-1}$, following a slight extension of \cite[Definition 3.5]{LeschMoscovici}. \\
The symbol of $\DD_A$ is $\sigma_{\DD_A}(\xi)=  (\xi_\alpha -iA_\alpha -iJA_\alpha J^{-1})\otimes \gamma^\alpha$. Thus, $\DD_\Ag^2$ given in \eqref{D2} has the form $$\sum_{\beta \in \N^d,\,\vert\beta \vert \leq 2} \,b_\beta \xi^\beta \,\text{ where }\, b_\beta \in \B_\theta \otimes \CC^{2^m}\,\text{ and }\,\xi^\beta \vc \xi_1^{\beta_1}\cdots \xi_d^{\beta_d}.$$ 
Thus  $\DD_\Ag^2$ is an elliptic differential multiplier of order 2 and as such, its symbol can be decomposed into a sum of monomials $a_j(\xi)$ of order $j=2,\,1,\,0$: $$\sigma_{\DD_\Ag^2}(\xi)=a_2(\xi)+a_1(\xi)+a_0.$$
If $\lambda$ is the resolvent parameter for $(\DD_A^2-\lambda)^{-1}$, in the search for the resolvent parametrix $B_\lambda$, we need to solve $1=\sigma_\lambda \star \sigma_{B_\lambda}$ for the parameter dependent symbol $$\sigma_\lambda(\xi)\vc \sigma_{\DD_\Ag^2}(\xi)-\lambda= a'_2(\xi)+a_1(\xi)+a_0$$ with $a'_2(\xi)=a_2(\xi)-\lambda$. \\
Since $\sigma_{B_\lambda}=b_{0}+b_{1}+b_{2}+\dotsb$, where $b_{j}(\xi;\lambda)$ is a symbol of order $-2-j$, one can compute recursively the symbols $b_{0},\,b_{1},\,b_{2},\dots$ of the parameter-dependent pseudodifferential multiplier $B_\lambda$ (see \cite{ConnesModular,ConnesTretkoff, FatConnes,FatKhal,LeschMoscovici} for examples of such computations).\\
Then, thanks to  \eqref{e-tP with derivative}, $$\Tr e^{-t\DD_\Ag^2}=t^{-k+1}\tfrac{i}{2\pi}\int_\C e^{-t\lambda}\,\Tr[ (\DD_\Ag^2-\lambda)^{-k}]\,d\lambda$$ (recall that $\Tr[ (\DD_\Ag^2-\lambda)^{-k}]<\infty$ for $k\geq \floor{d/2}+1$), one implements the resolvent expansion into the integral to get the asymptotics $$\Tr e^{-t\DD_\Ag^2} \,\underset{t \,\downarrow \,0}{\sim} \, \sum_{k=0}^\infty a_k(\DD_\Ag^2)\, t^{(k-d)/2}$$ exactly as in \cite[Sections 1.7 and 1.8]{Gilkey1} using the Proposition \ref{prop:discrete parameter symbol}. \\
Thus, $$a_k(\DD_\Ag^2)= \tfrac{i}{2\pi}\int_{\RR^d} d\xi \int_\C e^{-\lambda} \,\tau[ b_k(\xi,\lambda) ] \,d\lambda.$$
\\
Finally, we use Formula \eqref{SA_simplified}, which becomes \eqref{eq:sa for torus}.
\hfill{$\Box$}
\end{proof}

While there is no Diophantine condition to get the asymptotics of the fluctuated heat trace, the computation of the coefficients seems to need one --- this is due to a commutation between a series and a residue.

\begin{theorem}
\label{thm:SA for nc torus}
Assume that $\tfrac{1}{2\pi} \Th$ is badly approximable. Then, the spectral action \eqref{eq:sa for torus} fluctuated  by $\Ag$ as in \eqref{def: A symmetric} reads, for $f (\sqrt{\cdot}\,) \in \Cl_0^{r}$ with $r>d/2$
\begin{align*}
S(\DD_\Ag,f,\Lambda)=\left\{\begin{array}{ll}4\pi f_2\,\Lambda^2+\OO_\infty(\Lambda^{-2}),& \text{ for }d=2,\\
8\pi^2 f_4\,\Lambda^4 -\tfrac{4\pi^2}{3}\,f(0)\,\tau(F_{\mu\nu}F^{\mu\nu})+\OO_\infty(\Lambda^{-2}),&\text{ for }d=4.
\end{array}\right.
\end{align*}
For arbitrary $d$, $$S(\DD_\Ag,f,\Lambda)=\sum_{k=0}^{d-1} f_{d-k}\,c_{d-k}(A) \,\Lambda^{d-k} +\OO_\infty(\Lambda^{-1}), $$
where we have $c_{d-2}(A)=0$ and $c_{d-k}(A)=0$ for $k$ odd. In particular, $c_0(A)=0$ when $d$ is odd.
\end{theorem}

We do not know if, without the Diophantine condition, the above spectral action would stay the same, see \cite[Appendix B]{HeatNCT}.\\The proof goes through several steps and the first one is to identify the noncommutative integrals in \eqref{eq:sa for torus}:

\begin{proposition}
Assume that $\tfrac{1}{2\pi} \Th$ is a badly approximable matrix. \\
Then, we have 
\begin{align*}
& \ncint \vert D_\Ag \vert^{-d}= 2^{\lfloor \tfrac{d+2}{2}\rfloor} \pi^{d/2}\, \Gamma(\tfrac d2)^{-1}, \\
& \ncint \vert D_\Ag \vert^{-d+k}=0 \,\text{ for $k$ odd integer},\\
& \ncint  D_\Ag \vert^{-d+2} =0.
\end{align*}
\end{proposition}

These equalities follow from the explicit computation of \eqref{ncint DA-p}--\eqref{ncint DA-p+2} and the fact that $\ncint P \abs{D}^{-d+q}=0$ for any $P \in \Psi'(\A)$ and any odd integer $q$. 
\\
In a similar way, $\ncint A^pD^{-q}=\ncint (\epsilon JAJ^*)^p D^{-q} =0$ and $\ncint P \,D^{-d+q}=0$ for $p\geq 0$, $1\leq q<d$ and $P$ in the algebra generated by $\A,\, [\DD,\A],\,J\A J^*, \,J[\DD,\A ]J^*$. All these equalities can be proved using the deep results stored in Appendix \ref{app: residues of series}. The case $p=q=1$ corresponds to the fact that there is no tadpole for the noncommutative torus --- see Definition \ref{def:tadpole}. The main point is again to be able to commute an infinite series, like those defining an element of $\A_\Theta$, and a residue given by $\ncint$ and this is where the hypothesis on $\tfrac{1}{2\pi} \Th$ is used.

In the second step, we face explicit computations like, inter alia, the following:

\begin{lemma} 
Under same hypothesis (recall that $A= L(-iA_\alpha)\otimes \gamma^\alpha$),
\begin{align*}
\ncint A^q D^{-q}=\left\{\begin{array}{ll}
-\delta_{q,2} 4 \pi \,\tau(A_\alpha A^\alpha) &\,\,\,\text{for }d=2,\\ 
\delta_{q,4} \tfrac{\pi^2}{12} \,\tau(A_{\alpha_1}
\cdots A_{\alpha_4})\,\tr (\gamma^{\alpha_1}\cdots \gamma^{\alpha_4} \gamma^{\mu_1}\cdots \gamma^{\mu_4}) \delta_{\mu_1,\mu_2,\mu_3,\mu_4} &\,\,\,\text{for }d=4,
\end{array}\right.
\end{align*}
where $\delta_{\mu_1,\mu_2,\mu_3,\mu_4}=\delta_{\mu_1\mu_2}\delta_{\mu_3\mu_4}+\delta_{\mu_1\mu_3}\delta_{\mu_2\mu_4}+\delta_{\mu_1\mu_4}\delta_{\mu_2\mu_3}$.

With $A=L(-iA_\alpha)\otimes \gamma^\alpha=-i\sum_{\ell\in\Z^d} a_{\alpha,\ell}U_\ell \otimes \gamma^\alpha$ and $c\vc \tfrac{4\pi^2}{3}$, we have
\begin{align*}
& +\tfrac 12 \ncint (AD^{-1})^2=  \tfrac 12 \ncint (\epsilon JAJ^*D^{-1})^2=c \,\sum_{\ell\in\Z^d} \, a_{\alpha_{1},\ell} \, a_{\alpha_{2},-\ell} \, \big(l^{\alpha_{1}}\ell^{\alpha_{2}} - \delta^{\alpha_{1}\alpha_{2}} \Vert \ell \Vert^2\big).\\
& -\tfrac 13 \ncint (AD^{-1})^3\!=\!-\tfrac 13 \ncint (\epsilon JAJ^*D^{-1})^3=4c\!\!\!\!
\,\sum_{\ell_i \in \Z^d} a_{\alpha_3,-\ell_1-\ell_2}\,a_{\alpha_1,\ell_2}\,a_{\alpha_1,\ell_1}\,\sin \tfrac{\ell_1.\Th \ell_2}{2}\,\ell_1^{\alpha_3}.\\
& +\tfrac 14 \ncint (AD^{-1})^4= \tfrac 14 \ncint (\epsilon JAJ^*D^{-1})^4 \\
& \hspace{2.15cm} =2c\, \sum_{\ell_i \in \Z^d} a_{\alpha_{1},-\ell_1-\ell_2-\ell_3}\, a_{{\alpha_{2}},\ell_3} \, a_{\alpha_{1},\ell_2} \, a_{\alpha_{2},\ell_1} \sin \tfrac{\ell_1 .\Th (\ell_2+\ell_3)}{2}\, \sin \tfrac{ \ell_2 .\Th \ell_3}{2}.
\end{align*}
\end{lemma}
These equalities follows from tedious computations (cf. \cite{TorusSA} for the details). They are necessary to get $\zeta_{D_\Ag}(0)$ from \eqref{eq: fluctuation of zeta at zero}, using $\zeta_D(0)=0$ obtained in Theorem \ref{zeta(0)}. This gives $\zeta_{D_\Ag}(0)=-c \,\tau (F_{\mu\nu}F^{\mu\nu})$ in dimension $d=4$, while  $\zeta_{D_\Ag}(0)=0$ for $d=2$ or when $d$ is odd. Gathering all these computations brings us to  Theorem \ref{thm:SA for nc torus}.

\section{Podle\'s Spheres}
\label{sec:Podles}
\index{Podle\'s sphere}

Podle\' s has introduced the eponym standard quantum spheres in \cite{Podles} as homogeneous spaces under the action of the quantum deformation of the $SU(2)$ group. They fit into the picture of noncommutative geometry \`a la Connes and concrete spectral triples were constructed in \cite{AllPodles,dab_sit}, see also \cite{DabrowskiQSpheres}.

The algebra $\bar{\A}_q$ \index{_zAq@$\bar{\A}_q$} introduced in \cite{Podles} is a complex $*$-algebra generated, for a parameter $0 < q < 1$, by $A=A^*, \,B,\, B^*$ subject to the relations
\begin{align*}
\!\!AB = q^2 BA,\, \,\,\,\,\,\,AB^* = q^{-2} B^*A, \,\,\,\,\,\,\,BB^* = q^{-2}A(1-A),\,\,\,\,\,\,\, B^*B = A(1- q^2 A).
\end{align*}
As a $C^*$-algebra, it is isomorphic to the minimal unitisation of the algebra of compact operators on a separable Hilbert space and is an invariant subalgebra of the quantum group $SU_q(2)$ under a circle action.  
We will use only the polynomial algebra $\A_q$ \index{_zAq1@$\A_q$} in the above generators, which is a dense subalgebra of $\bar{\A}_q$.

In the following we will employ the \emph{$q$-numbers} \index{qnumber@$q$-number} defined as
\begin{align}
\label{qN}
[n] \vc \tfrac{q^{-n}-q^n}{q^{-1}-q}, \quad \text{for } \; q \in (0,1) \text{ and } n \in \CC.\index{_Azn@$[n] $}\skipindex{_Azn}
\end{align}
Observe that $\lim_{q \to 1} [n] = n$ for any $n \in \CC$.

Let $\H_{1/2}$ \index{_zHHH1@$\H_{1/2}$} be the separable Hilbert space with an orthonormal basis $\ket{l,m}$ for $m \in \{-l,-l+1, \ldots, l\}$ and $l \in \oh + \N$. It is suitable for a representation of $\A_q$, equivariant under the action of the $^*$-Hopf algebra $\UqSU$. There exist two non-equivalent  $\UqSU$-equivariant representations of $\A_q$ on $\H_{1/2}$  \cite{dab_sit}:
\begin{align}
&\pi_{\pm}(A) \ketQS \vc  A^{+}_{l,m,\pm} \, \ket{\lpo,m}_{\pm} + A^{0}_{l,m,\pm} \, \ketQS 
+ A^{-}_{l,m,\pm} \, \ket{\lmo,m}_{\pm}\,, \label{rep}\\
&\pi_{\pm}(B) \ketQS \vc  B^{+}_{l,m,\pm} \, \ket{\lpo,\mpo}_{\pm} + B^{0}_{l,m,\pm} \, \ket{l,\mpo}_{\pm} 
+ B^{-}_{l,m,\pm} \, \ket{\lmo,\mpo}_{\pm}, \notag\\
&\pi_{\pm}(B^*) \ketQS \vc  \wt{B}^{+}_{l,m,\pm}  \ket{\lpo,\mmo}_{\pm} \hspace{-0.1cm}
+ \wt{B}^{0}_{l,m+1,\pm}  \ket{l,\mmo}_{\pm} \hspace{-0.1cm}+ \wt{B}^{-}_{l,m,\pm} \ket{\lmo,\mmo}_{\pm}, \notag
\end{align}
\begin{align}
&A^+_{l,m}  \vc -q^{m+l+\tfrac{1}{2}} \sqrt{ [l-m+1][l+m+1]} \;\alpha^+_l ,\nonumber\\
&A^0_{l,m}  \vc q^{-\oh} \tfrac{1}{1+q^2} \left( [l-m+1][l+m]  - q^2 [l-m][l+m+1]
\right) \alpha^0_l + \tfrac{1}{1+q^2}\,, \label{Arep}\\
& A^-_{l,m}  \vc  q^{m-l-\tfrac{1}{2}} \sqrt{ [l-m][l+m]} \;\alpha^-_l \,,\nonumber
\end{align}
\begin{align}
&B^+_{l,m}  \vc q^{m} \sqrt{ [l+m+1][l+m+2]} \;\alpha^+_l\,, \nonumber\\
&B^0_{l,m}  \vc q^{m} \sqrt{ [l+m+1][l-m]} \;\alpha^0_l\,, \label{Brep}\\
&B^-_{l,m}  \vc q^{m} \sqrt{ [l-m][l-m-1]} \;\alpha^-_l\,,\nonumber\\
\nonumber\\
&\widetilde{B}^+_{l,m}  \vc q^{m-1} \sqrt{ [l-m+2][l-m+1]} \; \alpha^-_{l+1}\,, \nonumber \\
&\widetilde{B}^0_{l,m}  \vc q^{m-1} \sqrt{ [l+m][l-m+1]} \; \alpha^0_l\,, \label{Bsrep}\\
&\widetilde{B}^-_{l,m}  \vc q^{m-1} \sqrt{ [l+m][l+m-1]}\; \alpha^+_{l-1}\,.\nonumber
\end{align}
The coefficients $\alpha_l$ read: $\alpha_l^{-} \vc - q^{2l+2} \,\alpha_l^{+}$ and
\begin{align}
\label{alphaP}
\hspace{-4cm}\text{for }\pi_+: \hspace{1cm}&\alpha^0_l \vc \tfrac{1}{\sqrt{q}} \tfrac{ (q - \oq)
[l - \oh][l+ \frac{3}{2}] + q}{[2l][2l+2]}\ , \\
&\alpha^+_l \vc q^{-l-2}
\tfrac{1}{\sqrt{[2l+2]([4l+4]+[2][2l+2])}}\ ;\\
\hspace{-4cm}\text{for }  \pi_- : \hspace{1cm}&\alpha^0_l \vc \tfrac{1}{\sqrt{q}} \tfrac{ (q - \oq)
[l - \oh][l+ \tfrac{3}{2}] - q^{-1}}{[2l][2l+2]}\ , \\
&\alpha^+_l \vc q^{-l-1} \tfrac{1}{\sqrt{[2l+2]([4l+4]+[2][2l+2])}}\ .
\label{alphaM}
\end{align}
Since $\pi^{\pm}$ are faithful, the algebra $\A_q$ is dense in $\bar{\A}_q$ in the operator norm. 

Let now $\H_q \vc \H_+ \oplus \H_-$ \index{_zHtq@$\H_q$}, where $\H_\pm \vc \H_{1/2}$, with the representation of $\A_q$:
\begin{align}
&\pi(a) \vc \left(\begin{smallmatrix} \pi_+(a)&0\\0&\pi_-(a)\end{smallmatrix} \right). \label{Defpi}
\end{align}

On $\H$ we define an unbounded selfadjoint operator given by 
\begin{align}
\label{Dirac}
&\DD_q \vc \left (\begin{smallmatrix} 0 & \bar{w} \,\,T_q \\  w \,\,T_q & 0 \end{smallmatrix}\right ), \qquad T_q:\, \ket{l,m}\in \H_{1/2} \mapsto 
 [l+\tfrac 12] \, \ket{l,m}\in \H_{1/2}, \index{_zD0q@$\DD_q$}
\end{align}
for an  arbitrary constant $w \in \CC \setminus \{0\}$.

We note that $\ker \DD_q=\set{0}$ and the polar decomposition of $\DD_q$ reads
\begin{align}
\label{AbsDirac}
\DD_q = F |\DD_q|, && \text{with} &&
\vert \DD_q\vert =\vert w \vert\,\left( \begin{smallmatrix} T_q &0 \\ 0 & T_q\end{smallmatrix} \right) && \text{ and } && F \vc \tfrac{1}{|w|} \bigl( \begin{smallmatrix} 0&\bar{w}\\ w&0 \end{smallmatrix} \bigr).
\end{align}
The phase operator satisfies $[F,\,\pi(a)]=0$ for $a \in \A_q$. We have for $l \in \N + \oh$,
\begin{align*}
\vert\Dq\vert \left( \ket{l,m}_+\oplus \ket{l,m}_- \right) = \vert w \vert [l+\half] \left( \ket{l,m}_+\oplus \ket{l,m}_- \right), 
\end{align*}
which gives
\begin{align}
\mu_n(\Dq) = \lambda_n(\vert\Dq\vert) = \vert w \vert [n+1], \quad\text{and} \quad M_n(\vert\Dq\vert) = 4(n+1), \quad n \in \N. \label{Podles_eigenvalues}
\end{align}
Let us recall \eqref{qN} and observe that the singular values of $\Dq$ grow exponentially.

In \cite{dab_sit} it is proven that $(\A_q,\H_q,\DD_q)$ is a spectral triple, which moreover is even for $\gamma \vc \left( \begin{smallmatrix}1&0\\0&-1 \end{smallmatrix} \right)$ and real for the antiunitary operator $J$ on $\H_q$ defined by 
\begin{align*}
J \, \ket{l,m}_{\pm} \vc i^{2m} p^m\,\ket{l,-m}_{\mp}\,\, \text{ with }p\in \RR^+.
\end{align*}
In particular, $J^2=-1$, $J\gamma=-\gamma J$, $JaJ^{-1}$ commutes with $\A_q$ and $[\Dq,J]=0$, so the spectral triple is of $KO$-dimension 2 \cite{Projective}. The operator $\DD$ is the unique $\UqSU$-equivariant operator fulfilling the first-order condition, which makes the spectral triple real \cite{dab_sit}.

In the limit $q \to 1$ one recovers the commutative spectral triple on the $2$-dimensional sphere --- cf. \cite[Remark 2]{Podles} and \cite[p. 8]{dab_sit}.

\smallskip

In \cite{Podles} yet another Dirac operator for the standard Podle\'s sphere was introduced
\begin{equation*} 
\DqS \vc \left (\begin{smallmatrix} 0 & \bar{w} \,T_q^S \\  w \,T_q^S & 0 \end{smallmatrix}\right),
\qquad T^S_q:\, \ket{l,m}\in \H_{1/2} \mapsto  \tfrac{q^{-(l+1/2)}}{q^{-1} - q}  \, \ket{l,m}\in \H_{1/2}.
\end{equation*}
\index{_zD0qs@$\DqS $}
which shares the property of the exponential growth of singular values with $\DD_q$.

The operators $\Dq$ and $\DqS$ commute and are related by $\Dq = \DqS - ( \tfrac{|w| \,q}{1-q^2} )^2 (\DqS)^{-1}$.
Remark that $(\DqS)^{-1}$ is actually trace-class. It turns out that by taking $\DqS$ instead of $\Dq$ we preserve most of the properties of the spectral triple with the exception of the first order condition. Moreover, we have 
\begin{align*}
|\Dq| = |\DqS| - ( \tfrac{|w| \,q}{1-q^2})^2 |\DqS|^{-1}.
\end{align*}

\begin{proposition}
\label{thm:podles as triples}
The triples $(\A_q, \H_q, \Dq)$ and $(\A_q,\H_q,\DqS)$ are $0$-dimensional. 
\end{proposition}
\begin{proof} This is a direct consequence of the exponential growth of singular values of $\Dq$ and $\DqS$. Indeed, $\mu_n(\Dq) = \Oinf(q^{-n})$ and $\mu_n(\DqS) = \Oinf(q^{-n})$, whereas $M_n(\vert\Dq\vert) = M_n(\vert\DqS\vert) = \Oinf(n)$, which means that for any $\epsilon > 0$ both $\Tr \vert\Dq\vert^{-\epsilon} $ and $\Tr \vert\DqS\vert^{-\epsilon}$ are finite.
\hfill{$\Box$}
\end{proof}

Observe that, whereas the limit $q \to 1$ of the spectral triple $(\A_q,\H_q,\DD_q)$ is well defined and yields the standard round geometry on $S^2$, the spectral dimension jumps abruptly from 0 to 2. This phenomenon, known as the `dimension drop', has also its impact on (co)homologies of $\bar{\A}_q$ \cite{MasudaPodles} and provided inspiration for `twisted' noncommutative geometries (see Problem \ref{twisted} in Chapter \ref{chap:open}). The potential physical implications of the dimension drop in the context of quantum gravity are the subject of an intensive study --- see, for instance, \cite{Benedetti}. 

Another drastic consequence of the exponential growth of singular values is that neither of the spectral triples $(\A_q, \H_q, \Dq)$, $(\A_q,\H_q,\DqS)$ is regular. 
In fact, already $\big[|\DD|,[|\DD|,a] \big]$ is an unbounded operator and more generally $\delta^n(a)$ is an operator of order $n-1$ for a generic element $a \in \A$. This could be seen by computing a fixed matrix element of the operator $\delta^n(A)$, for instance: 
\begin{align*}
{}_{\pm} \bra{l+1,\,m} \, \delta^n(A) \ket{l,m}_{\pm} & = \left( [l+\tfrac{3}{2}] - [l+\tfrac{1}{2}] \right)^n A^+_{l,m}\,.
\end{align*}
The behaviour of $A^+_{l,m}$ for large $l$ is $\Oinf(q^l)$, which can be read from the explicit Formula \eqref{Arep}, while $\left( [l+\tfrac{3}{2}] - [l+\tfrac{1}{2}] \right)^n =\Oinf(q^{-nl})$, and the above expression is unbounded for $n>0$; moreover, generally, $\delta^n(A)$ is in $\OP^{n-1}$. \\
Note that, in general, $\delta^0(a)$ and $\delta^1(a)$ are bounded for any $a \in \A_q$, but not $\delta^2(a)$.

Nevertheless, the spectral triple $(\A_q,\H_q,\DD_q)$ is quasi-regular (recall Section~\ref{sec:quasi_reg}). Within this extended framework one discovers that \cite[Corollary 3.10]{PodlesSA}:
\begin{theorem}\label{thm:Podles_dimsp}
For any $0<q<1$, the dimension spectrum of $(\A_q,\H_q,\DD_q)$ is of the order 2 and equals to $-\N + i \tfrac{2\pi}{\log q} \Z$.
\end{theorem}
\begin{proof}[sketch]
Firstly, one shows (using a simple summation of geometric series \cite[Proposition 3.2]{PodlesSA}) that the basic zeta function reads
\begin{align*}
\label{zeta_Podles}
\zeta_{\Dq}(s) = \Tr \vert \Dq \vert^{-s} =  4 (\tfrac{1-q^2}{\vert w \vert})^{s} \sum_{n=0}^{\infty} \, \tfrac{\Gamma (s+n)}{n! \, \Gamma (s)} \,
\tfrac{q^{2n}}{(1-q^{s+2n})^2}\,, \quad \text{ for } \quad \Re(s) > 0.
\end{align*}
This formula leads us to the realm of `$q$-zeta functions' (see \cite{Kaneko} and other references on \cite[p. 633]{PodlesSA}). It yields a meromorphic extension of $\zeta_{\Dq}$ to the whole complex plane with second order poles in $-2\N + i \tfrac{2\pi}{\log q} \Z \subset \CC$.

Secondly, one shall consider general functions $\zeta_{P,\Dq}$ with $P \in \PDOqk{0}$ (cf. \cite[Proposition 3.8]{PodlesSA}). To start, one notices that it is sufficient to consider elements of the form $P = T \vert\Dq\vert^{-p}$, with $T \in \Pc(\A_q) \cap \op^n$ for some $n \in \N$ and $n \leq p \in \N$ and uses the quasi-regularity \cite[Lemmas 3.4 \& 3.6]{PodlesSA} to commute all of the operators $\Dq$ and $\vert \Dq \vert$ from $T$ to the right. Then, one observes that the generators of $\A_q$ are represented via $\pi$ in terms of weighted shift operators on $\H_q$. Moreover, the weights are analytic functions of (bounded) variables $q^{l+m}, q^{l-m}$ and $q^l$. Rewriting these in terms of infinite convergent Taylor series one arrives at a formula for $\zeta_{T,\Dq}$, which involves a (multiple) infinite series in $q^l$. Finally, a resummation over $l \in \N + 1/2$ yields the desired meromorphic extension of $\zeta_{P,\Dq}$ to $\CC$ --- cf. \cite[Eq. (30)]{PodlesSA}.
\hfill{$\Box$}
\end{proof}

The zeta function associated with the simplified operator $\DqS$ was presented in Example \ref{ex:Podles_hk_S}. Although the spectral triple $(\A_q,\H_q,\DqS)$ has not been studied extensively in \cite{PodlesSA}, one can show along the same lines that it is quasi-regular and has a dimension spectrum  of second order equal to $-\N + i \tfrac{2\pi}{\log q} \Z$.
\smallskip

The exponential growth of singular values of $\Dq$ has also some pros: It leads to the following spectacular result highlighted at the begining of Section \ref{sec:conv_exp}:
\begin{theorem}\label{thm:SA_Podles}
Let $f \in \Cl_0^r$ for some $r>0$ and denote $\kappa \vc \tfrac{2 \pi i}{\log q}$. \\
Then, for any $\Lambda>0$,
\begin{align*}
S(\DD_q,f,\Lambda) \!=\! \sum_{k = 0}^{\infty} \,\, \sum_{j \in \Z} \,\, \sum_{n = 0}^{2} \, a_{-2k+ \kappa j,n} \, \sum_{m = 0}^{n} (-1)^{n-m}  \tbinom{n}{m}  \,f_{-2k+ \kappa j,m} \,(\log \Lambda)^{n-m} \,\Lambda^{-2k+ \kappa j}.
\end{align*}
\end{theorem}
\begin{proof}[sketch]
The theorem is proven in detail in \cite{PodlesSA}, via the heat trace expansion . It uses the full force of Theorems \ref{thm:zeta2hk}, \ref{thm:hk_exact}, \ref{thm:f_conv} and Corollary \ref{cor:f_exact}. The estimation of the contour integrals is rather subtle and arduous --- cf. \cite[Proposition 4.3]{PodlesSA}.
\hfill{$\Box$}
\end{proof}

This result is remarkable for two reasons:
Firstly, the formula for the spectral action contains terms proportional to $\log^2 \Lambda$ and oscillating with $\Lambda$, which is a sign that the geometry of Podle\'s sphere lies outside of the kingdom of classical pdos (recall Example \ref{ex:Gilkey_pseudodiff}).
Secondly, the formula is exact for all $\Lambda > 0$ and for a fairly general class of cut-off functions. Recall that in the classical (pseudo)differential geometry one is bound to use the asymptotic expansion, which might obscure some important information (see p. \pageref{asymptotics vs exact}).

As for the fluctuations, these are much more tedious to control in the quasi-regular case. In \cite[Theorem 5.6]{PodlesSA} it was shown that the leading term of the spectral action $S(\DD_q + \Ag,f,\Lambda)$ does not depend on $\Ag$, when the fluctuation is `small', but a deeper understanding of the problem is missing. In particular, it is not clear whether an explicit exact formula for the fluctuated action is available at all .

%
%
 \bibliographystyle{spmpsci}
 \bibliography{SA_BIB}

\backmatter
\printindex


\end{document}